\documentclass[10pt]{article}
\usepackage{tikz}
\usetikzlibrary{arrows.meta,bending}
\usepackage{appendix}
\usepackage{amsmath}
\usepackage{upgreek}
\usepackage[T1]{fontenc}
\usepackage{geometry}
\usepackage{booktabs}
\usepackage{amsthm}
\usepackage{amsmath}
\usepackage{fontenc}
\usepackage{bm,upgreek}
\usepackage{setspace}
\usepackage[affil-it]{authblk}
\usepackage[english]{babel}
\usepackage{blindtext}
\usepackage{graphicx,psfrag,epsf}
\usepackage{amssymb}
\usepackage[authoryear]{natbib}
\usepackage{mathtools}
\usepackage[unicode=true,
 bookmarks=false,
 breaklinks=false,pdfborder={0 0 1},colorlinks=true,citecolor = blue]{hyperref}
\usepackage{caption}
\usepackage{subcaption}
\usepackage{xcolor}
\usepackage{bm}
\usepackage{color}
\usepackage{tikz}
\usepackage{stmaryrd}
\usepackage{url}
\usepackage{mathrsfs}  
\usepackage{enumerate}
\usepackage{colortbl}
\usepackage{comment}
\usepackage{gensymb}
\usepackage[tableposition=top]{caption}
\newcommand{\ind}{\mbox{$\perp\!\!\!\perp$}}
\allowdisplaybreaks
\usepackage[colorinlistoftodos,textwidth=2.1cm]{todonotes}

\usepackage{times}
\usepackage{bm}
\usepackage{natbib}
\usepackage[ruled,vlined]{algorithm2e}
\usepackage{xcolor}
\usepackage{hyperref}      
\DeclareMathOperator*{\argmin}{arg\,min}
\makeatletter
\renewcommand{\algocf@captiontext}[2]{\quad #1\algocf@typo. \AlCapFnt{}#2} 
\def\@algocf@capt@plain{top}
\renewcommand{\algocf@makecaption}[2]{%
  \addtolength{\hsize}{\algomargin}%
  \sbox\@tempboxa{\algocf@captiontext{#1}{#2}}%
  \ifdim\wd\@tempboxa >\hsize
    \hskip .5\algomargin%
    \parbox[t]{\hsize}{\algocf@captiontext{#1}{#2}}
  \else%
    \global\@minipagefalse%
    \hbox to\hsize{\box\@tempboxa}
  \fi%
  \addtolength{\hsize}{-\algomargin}%
}
\makeatother

\def\T{{ \mathrm{\scriptscriptstyle T} }}

\newtheorem{proposition}{Proposition}
\newtheorem{theorem}{Theorem}
\newtheorem{lemma}{Lemma}

\usepackage{xr}
\makeatletter
\newcommand*{\addFileDependency}[1]{
  \typeout{(#1)}
  \@addtofilelist{#1}
  \IfFileExists{#1}{}{\typeout{No file #1.}}
}
\makeatother
\newcommand*{\myexternaldocument}[1]{%
    \externaldocument{#1}%
    \addFileDependency{#1.tex}%
    \addFileDependency{#1.aux}%
}
\myexternaldocument{./supp}

\newtheorem{assumption}{Assumption}
\newtheorem{definition}{Definition}
\newcommand{\p}{{P}}
\newcommand{\bpsi}{\Psi}
\newcommand{\QED}{$\hfill\square$}

\newcommand{\RR}{\mathbb{R}}
\newcommand{\ee}{\end{aligned} \end{equation}}
\newcommand{\eq}{\end{quote}}
\newcommand{\diag}{\text{diag}}
\newcommand{\ep}{\end{parts}}
\newcommand{\bqp}{\begin{quote}\begin{parts}}

\newcommand{\epq}{\end{parts}\end{quote}}

\newcommand{\M}{}
\newcommand{\E}{E}

\newcommand{\var}{\text{Var}}
\newcommand{\lp}{\mathcal{L}_{\mathcal{P}}^2}

\def\T{{ \mathrm{\scriptscriptstyle T} }} 
\newcommand{\Rom}[1]{\text{\uppercase\expandafter{\romannumeral #1\relax}}}
\newcommand{\bee}{\begin{equation}\begin{aligned}}

\renewcommand*{\Pr}{\text{pr}}
\newtheorem{remark}{Remark}
\title{Towards R-learner 
with Continuous Treatments}
 
\author{Yichi Zhang\thanks{Department of Statistics, Indiana University Bloomington; Email: \texttt{yiczhan@iu.edu}}\,\,,  Dehan Kong\thanks{Department of Statistical Sciences, University of Toronto; Email: \texttt{dehan.kong@utoronto.ca}}\,\,,   and   Shu Yang\thanks{Department of Statistics, North Carolina State University; Email: \texttt{syang24@ncsu.edu}}}

\date{}
\begin{document}
\maketitle
\begin{abstract}
The R-learner is widely used in causal inference due to its flexibility and efficiency in estimating the conditional average treatment effect. However, extending the R-learner framework from binary to continuous treatments introduces a non-identifiability issue, as the functional zero constraint inherent to the conditional average treatment effect cannot be directly imposed in the R-loss under continuous treatments. To address this, we propose a two-step identification strategy: we first identify an intermediary function via Tikhonov regularization, and then recover the conditional average treatment effect using a zero-constraining operator. Building on this strategy, an $\ell_2$-regularized R-learner framework is developed to estimate the conditional average treatment effect for continuous treatments. The new  framework accommodates modern, flexible machine learning algorithms to estimate both nuisance functions and target estimand. Theoretical properties are demonstrated when the target estimand is approximated by sieve approximation with B-splines, including error rates, asymptotic normality, and confidence intervals.  
\end{abstract}


\section{Introduction}
Estimating heterogeneous treatment effects is fundamental in causal inference and provides insights into various fields, including precision medicine, education, online marketing, and offline policy evaluation. Let $T$ be a treatment, $Y^{(t)}$ be the potential outcome had a subject received treatment level $T = t$, and $X$ be  pre-treatment covariates.  The treatment effect heterogeneity can be quantified by 
\begin{gather}\label{def:CATEtau}
\tau(x, t) = E\big(Y^{(t)} - Y^{(0)}\mid  X = x\big),
\end{gather} 
where $t = 0$ is a reference treatment level. Early works of conditional average treatment effect estimation focus on semiparametric models, including partially linear models \citep{robinson1988root} and  structural nested models \citep{robins1994correcting}. Recent years have witnessed the rapid growth of newly-developed methods with flexible models; see, e.g., \citet{chernozhukov2018double, wager2018estimation, kennedy2020optimal} and the references therein. One prevailing stream of works includes nonparametric \textit{meta-learners} including S- and X-learners \citep{kunzel2019metalearners} and R-learner \citep{nie2021quasi}, which are model-free and can be implemented via any off-the-shelf regression algorithm. S- and X-learners are tied to approximating the potential outcome surfaces using, e.g., the Bayesian additive regression trees \citep{hill2011bayesian}, deep learning \citep{shalit2017estimating}, and the causal random forest \citep{wager2018estimation}. However, they are not directly estimating the treatment effect. On the contrary, the R-learner and its variants \citep{kennedy2020optimal} target the treatment effect estimation. The R-learner capitalizes on the decomposition of the outcome model initially proposed by \citet{robinson1988root} in partially linear models and extends for machine learning-based treatment effect estimation \citep{nie2021quasi}. Notably,  when using the  two nuisance functions estimated under flexible models, the R-learner preserves the oracle property of  treatment effect estimation as though the nuisance functions were known. Despite these advantages, the current R-learner framework applies only to binary or categorical treatments. 
\par
In this article, we extend the R-learner framework to estimate the conditional average treatment effect flexibly with continuous treatments. This extension is nontrivial in both identification and estimation. 
Echoing the approach of \citet{nie2021quasi}, we focus  on adapting the generalized  R-learner loss function with continuous treatments.  Unlike the binary-treatment case, we demonstrate that directly minimizing the generalized R-loss does not uniquely identify $\tau(x,t)$ but instead identifies a broad class of functions. 
This is because  
the zero condition of $\tau(x,t)$: $\tau(x,0)\equiv 0$, cannot be easily encoded into the  R-loss when the treatment is continuous. We resolve this non-dentification issue by introducing a two-step identification strategy. This strategy is actualized through our $\ell_2$-regularized R-learner, leveraging the principles of Tikhonov regularization \citep{tikhonov1963solution}. It first approximates an intermediary
$\tilde\tau(x,t)=\tau(x,t)-E\{\tau(X,T)\mid X = x\}$, and then estimates $\tau(x,t)$ by transforming the intermediary estimation through a    zero-constraining operator, whose  output functions always satisfy the same zero condition as $\tau(x,t)$.  We elucidate the new R-learning framework through the method of sieves and provide a thorough investigation of the asymptotic properties. Unlike the classical sieve regression, theoretical analysis of the sieve R-learner involves low-rank matrices inherited from the non-identification nature of the generalized R-loss, which utilizes the toolkit in the matrix perturbation theory and  spectral analysis  \citep{bhatia2013matrix}. Whenever the nuisance functions can be approximated under the $o_P(n^{-1/4})$-convergence rate, the convergence rate of our proposed  estimator does not rely on the smoothness of the outcome model but relies only on the smoothness of the conditional average treatment effect  and propensity score functions---the two intrinsic components in $\tilde\tau(x,t)$. We   derive asymptotic normality of the R-learner, under which we propose a closed-form variance estimator and confidence intervals for inference. Numerical experiments  show the valid performance of our proposed R-learner in both estimation and inference.  
\subsection{Setup and notation}\label{sec:pre}
Let $\{Z_i = ( X_i, T_i, Y_i)\}_{i = 1}^n$ be independent and identically distributed samples from the distribution of $(X, T, Y)$, where $X = (X^{(1)},\dots,X^{(d)})$ is a $d$-dimensional vector of covariates. 
Under Rubin's causal model framework \citep{rubin1974estimating}, $Y^{(t)}$ is the potential outcome had the unit received treatment level $T =  t\in \RR$. The causal estimand is {\color{black}$\tau(x,t)$} defined in \eqref{def:CATEtau}. Due to the fundamental problem in causal inference that not all potential outcomes can be observed for a particular unit, $\tau(x, t)$ is not identifiable without further assumptions. 
We employ  common assumptions for continuous treatments \citep{kennedy2017non}.
\begin{assumption}[No unmeasured confounding]\label{A:UNC} We have
$\{Y^{(t)}\}_{t\in\mathbb{T}}\ind T \mid X$.
\end{assumption}
\begin{assumption}[Stable unit and treatment value]\label{A:NI}
When $T = t\in\mathbb{T}$, we have $Y = Y^{(t)}$. 
\end{assumption}
\begin{assumption}[Positivity]\label{A:CS}
There exists an $\varepsilon > 0$ such that the generalized propensity score $f(T = t\mid  X=x) \in(\epsilon,1/\epsilon)$ for any $(x,t)\in \mathbb{X}\times\mathbb{T}$. 
\end{assumption}
\par We summarize the notation used throughout the paper. For any vector $ v$, 
$\|v\|$ denotes its 
$\ell_2$ norm.   For any random variable $W\in\mathbb{W}$, $f(w)$ and $\mathcal{P}(w)$ denote its probability density function and probability measure. For any  function $g(w)$, $P_n \{g(W)\} = \sum_{i = 1}^n g(W_i)/n$ denotes its empirical expectation and $\|g\|_{\mathcal{L}^2} = \{\int_{w \in \mathbb{W}}g^2(w)dw\}^{1/2}$, $\|g\|_{\mathcal{L}^2_{\mathcal{P}}} = \{\int_{w \in \mathbb{W}}g^2(w)d\mathcal{P}(w)\}^{1/2}$, $\|g\|_{\mathbb{W}} = \sup_{w\in\mathbb{W}}|g(w)|$  denote its $\mathcal{L}^2$, $\mathcal{L}^2_{\mathcal{P}}$ and $\mathcal{L}^{\infty}$ norms. $\mathcal{L}_{\mathcal{P}}^2(W)$ represents the function space of all $g(w)$ with a bounded $\mathcal{L}^2_{\mathcal{P}}$ norm. When $g(w)$ is a multivariate function,  denote $\|g\|_{\mathbb{W}} = \sup_{w\in\mathbb{W}}\|g(w)\|$.  We require two nuisance functions, the conditional outcome mean and generalized propensity score: 
\bee\nonumber
m(x)= E(Y\mid X=x ),\quad\varpi(t\mid x)=f(T = t\mid X  =x).
\ee
We denote the full conditional outcome mean model $\mu(x,t) = E(Y\mid X = x , T = t)$, and hereby define the observation noises,
\bee\label{obvi}
\varepsilon_i = Y_i - \mu(X_i,T_i), \quad i = 1,\dots, n,
\ee 
where $E(\varepsilon_i\mid X_i,T_i) = 0$, following the definition of $\mu(x,t)$.  
\section{Continuous-treatment R-learner}\label{sec:main}
\subsection{The generalized R-loss}\label{sec:grloss}
We first generalize the idea of the Robinson's residual \citep{robinson1988root,nie2021quasi} to the continuous-treatment scenario. The unconfoundedness and stable unit and treatment value imply
\bee\label{key:ob}
Y_i^{(T_i)} &=  \mu(X_i,T_i) + \varepsilon_i = \mu(X_i,0) + \tau(X_i,T_i) + \varepsilon_i
,
\ee
where the first equality follows from  Assumption \ref{A:NI} and equation \eqref{obvi}, and the second equality follows from Assumption \ref{A:UNC} and the definition of $\tau(x,t)$. Model (\ref{key:ob}) is nonparametric and free of any additional structural assumptions.
Given $X_i$,  taking the conditional expectation on (\ref{key:ob}) leads to
\bee\label{key:ob2}
m(X_i)  = E\big(Y^{(T_i)}\mid X = X_i\big)  
=\mu(X_i,0) + E_{\varpi}\{\tau(X,T)\mid X = X_i\},
\ee
where the last equality is followed by the law of  total expectation such that
$
E(\varepsilon_i\mid X_i) = E\big\{E(\varepsilon_i\mid X_i,T_i)\mid X_i\big\} = E\big(0\mid X_i\big) = 0.
$
The notation $E_{\varpi}\{\tau(X,T)\mid X = X_i\}$ in \eqref{key:ob2} highlights the dependency of the conditional expectation on the generalized propensity score as $E_{\varpi}\{\tau(X,T)\mid X = X_i\} = \int_{t\in\mathbb{T}}\tau(X_i,t)\varpi(t\mid X_i)dt$.
By subtracting \eqref{key:ob2}  from \eqref{key:ob} on both left- and right-hand sides, we have 
\begin{gather}\label{con:decom1}
Y_i^{(T_i)}  - m(X_i) = \tau(X_i,T_i)  - E_{\varpi}\{\tau(X,T)\mid X = X_i\} + \varepsilon_i.
\end{gather}
By treating the left-hand side of (\ref{con:decom1}) as the response and the right-hand side except $\varepsilon_i$ as the mean function, we derive the following population loss function,
\bee\label{con:decom2}
L_c(h) = E\big[Y - m(X)-h(X,T) +E_{\varpi}\{h(X,T)\mid X\}\big]^2,
\ee
which is minimized at $h = \tau$. The above derivation parallels that of the binary-treatment R-learner. In fact, a similar loss function to $L_c(h)$ appears in \citet[$\mathsection$7]{nie2021quasi} under the multi-treatment setting. We view $L_c(h)$ as a natural generalization of the binary-treatment R-loss function \citep[$\mathsection$2]{nie2021quasi} to the continuous-treatment setting, and thus refer to $L_c(h)$ as the generalized R-loss.   In particular, under the binary-treatment case, $\tau(x,t)$ reduces to $\{\tau(x,0),\tau(x,1)\}$, where $\tau(x,0) = E(Y^{(0)} - Y^{(0)}\mid X = x) = 0$ for any $x \in \mathbb{X}$, and $\tau(x,1)$ is the conditional average treatment effect of interest. It suffices to estimate $\tau(x,1)$  by solving the $h(\cdot,1)$ that minimizes \eqref{con:decom2}, after imposing a zero condition of $h(\cdot,0)$:
\bee\label{eq:shape}
h(x,0) = 0, \quad \text{for  any } x\in \mathbb{X}.
\ee
More specifically, observing that under \eqref{eq:shape} one has $h(X,T) - E_{e}\{h(X,T)\mid X\}
 =\{T - e(X)\}h(X,1)$ a.s., where 
$e(x) = \Pr(T = 1\mid X = x)$ is the propensity score, 
 the R-loss function  \eqref{con:decom2} reduces to
\bee\label{rloss:binary}
L_{b}(h) = E\big[Y - m(X)-\{T - e(X)\}h(X,1)\big]^2
\ee
as   in \citet{nie2021quasi}, 
which is also minimized at $h = \tau$. 
\subsection{{\color{black}Non-identifiability of R-learner with continuous treatments}}\label{sec:simple}
The generalization of the R-loss from the binary treatment
to the continuous treatment is natural, which however results in a transition of the identifiability of $\tau(x,t)$.  
Suppose we construct $\hat{\tau}(x,t)$ by directly minimizing the empirical analogy of $L_c(\cdot)$ following \citet{nie2021quasi}. 
The R-learner for   continuous treatment will have poor estimation performance, due to the non-unique identifiability of the generalized R-loss. 
To illustrate, we conduct a simple simulation study where we use B-splines for the nonparametric approximation of $\hat{\tau}(x,t)$. 
The R-learner approximates $\tau(x,t)$ well when the treatment is binary (Fig.~\ref{fig:example}a) and poorly when the treatment is continuous (Fig.~\ref{fig:example}c). The simulation details are deferred to $\mathsection$\ref{sec:simeg}. 
 \begin{figure}[tp]\centering
      \begin{subfigure}[b]{0.19\textwidth}\centering\includegraphics[width=1\linewidth]{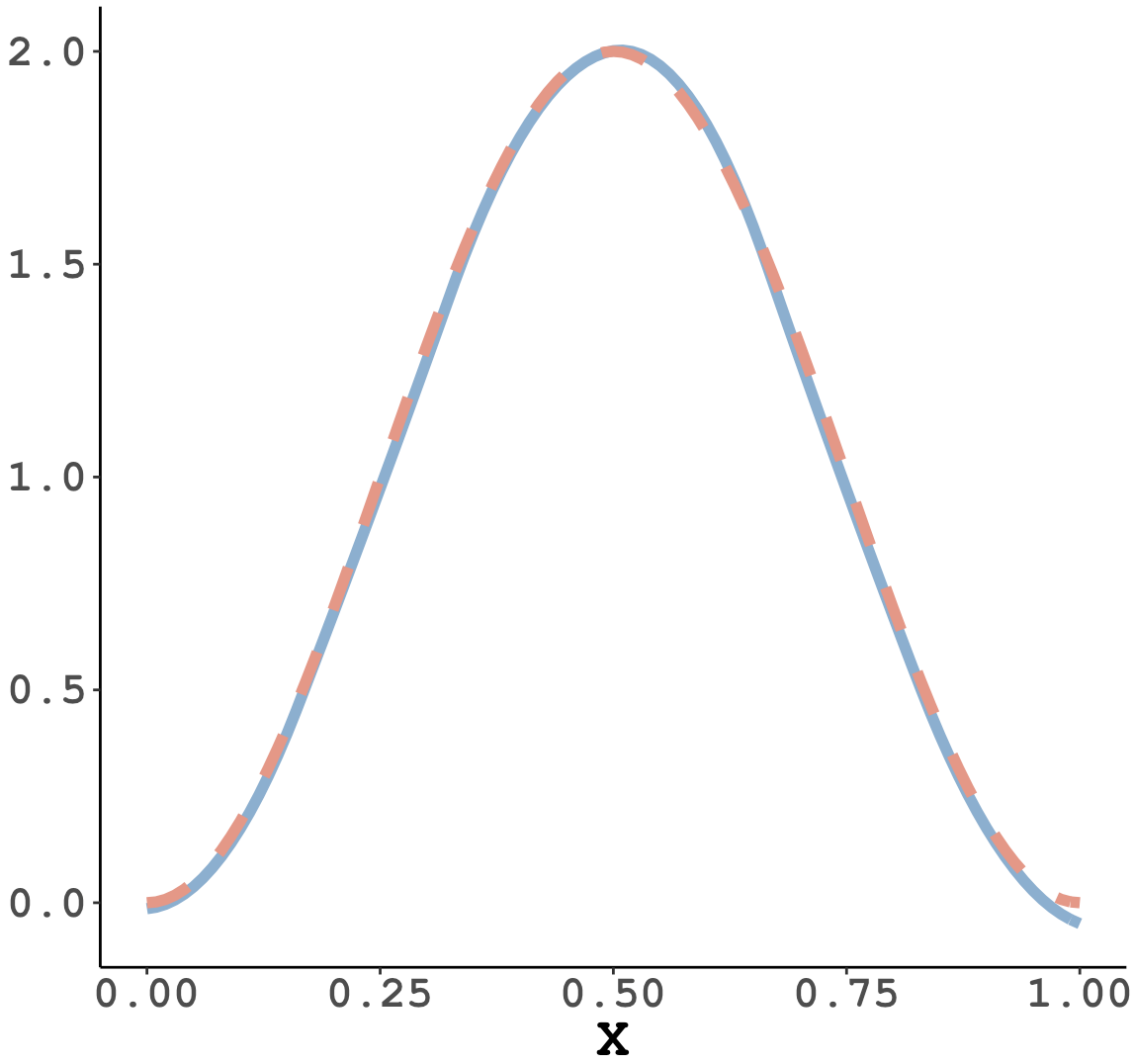}
         \caption{}
     \end{subfigure}
           \begin{subfigure}[b]{0.19\textwidth}\centering\includegraphics[width=1\linewidth]{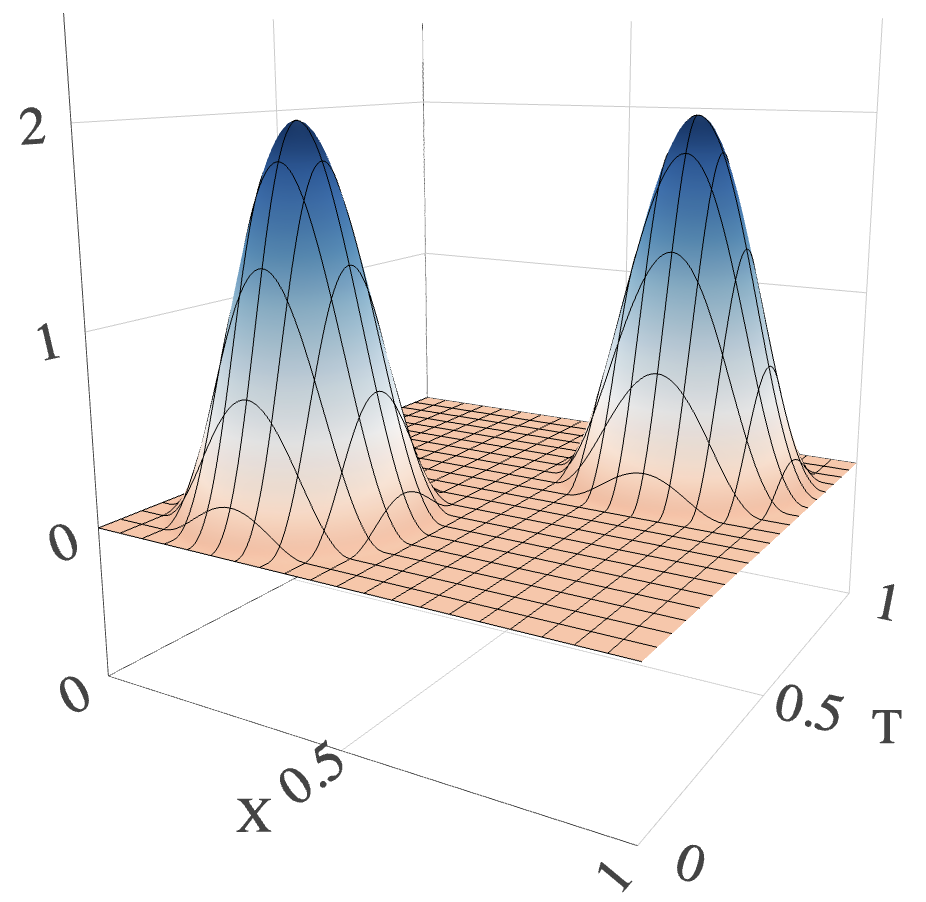}
         \caption{}
 
     \end{subfigure}
                \begin{subfigure}[b]{0.19\textwidth}\centering\includegraphics[width=1\linewidth]{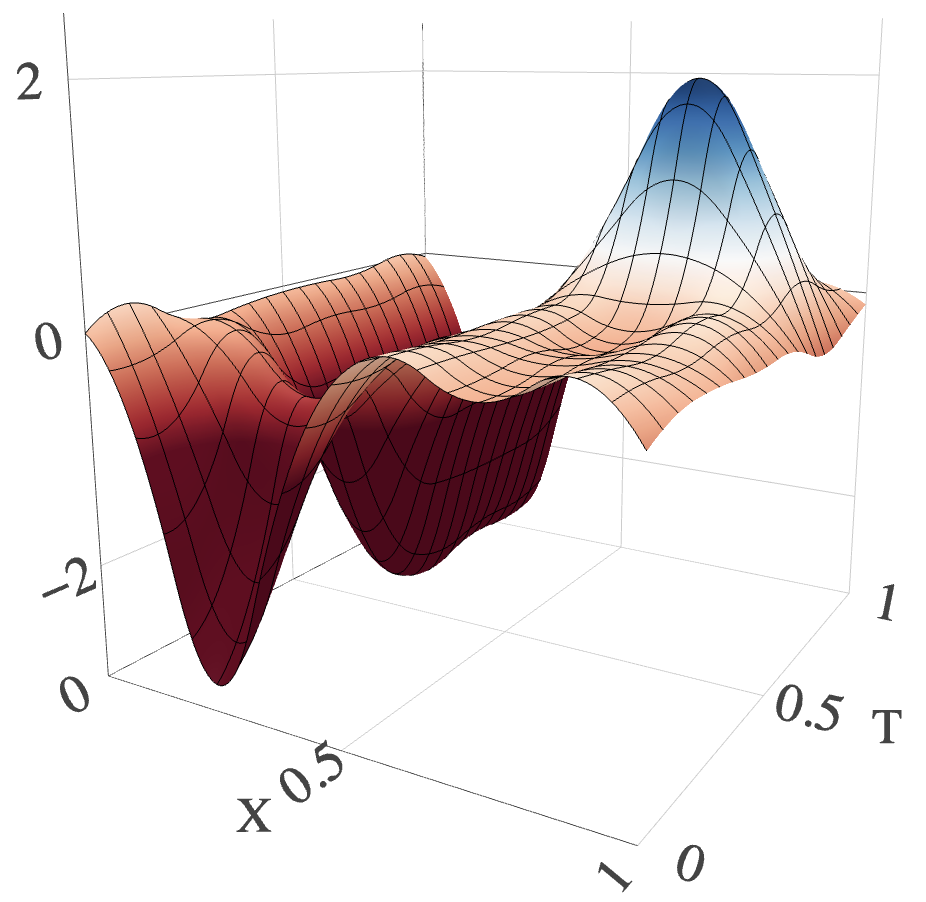}
         \caption{}
     \end{subfigure}
                \begin{subfigure}[b]{0.19\textwidth}\centering\includegraphics[width=1\linewidth]{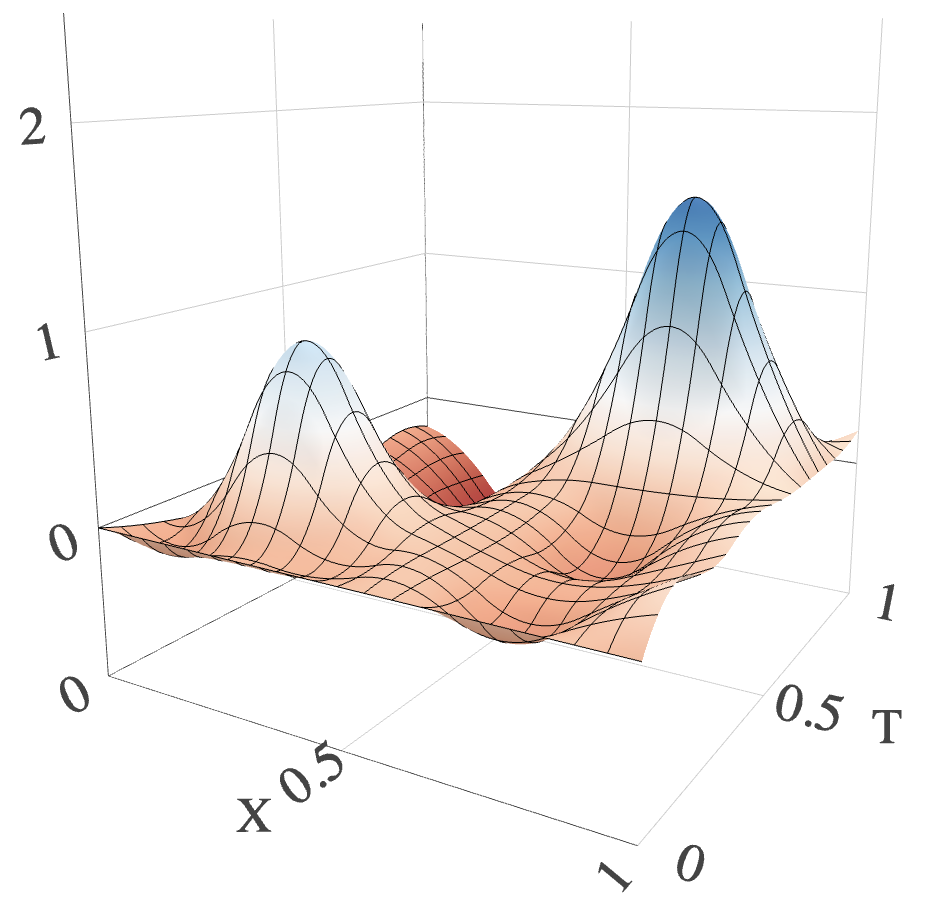}
         \caption{}
     \end{subfigure}
     \begin{subfigure}[b]{0.19\textwidth}\centering\includegraphics[width=1\linewidth]{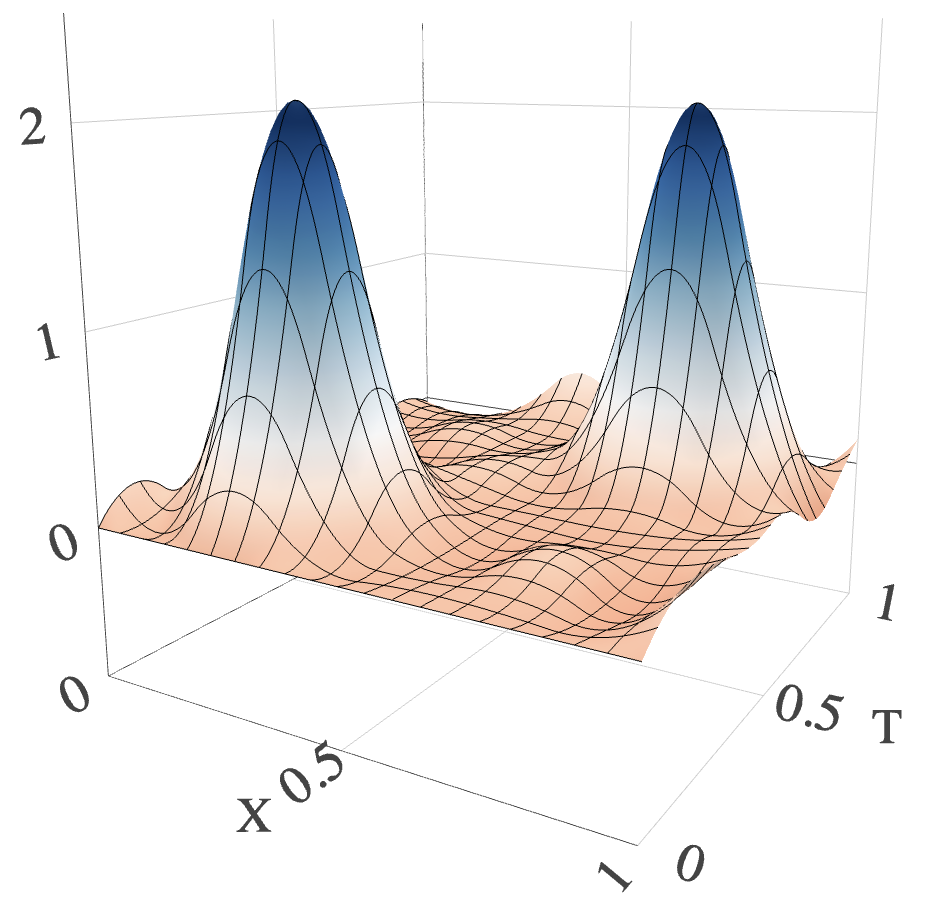}
         \caption{}
     \end{subfigure}
\caption{A simple simulation   with $n = 5000$. When the treatment is binary (Panel a),  the R-learner  (the  red  dashed curve) is consistent for $\tau(\cdot,1)$ (the blue curve).  
When the treatment is continuous, the directly generalized R-learner (Panel c) and the one-step generalized R-learner with a functional zero constraint (Panel d) are far away from $\tau$ (Panel b).  The  two-step $\ell_2$-regularized R-learner (Panel e) recovers $\tau$ well. The panels will be referenced  with  the corresponding estimators introduced in sequence in~$\mathsection$\ref{sec:main}.}
\label{fig:example}
\end{figure}
\par
We provide a theoretical explanation for the success and failure of the identifiabilities of the R-learner for binary treatment and  generalized R-learner for continuous treatment, respectively. 
Denote 
\bee\label{Sdef}
\mathcal{S} = \{h \mid h(X,T) = \tau(X,T) + s(X)\,\text{a.s., for any }  s\in \mathcal{L}^2_{\mathcal{P}}(X)\}.
\ee
 It is easy to check that for any $h \in \mathcal{S}$,
 $$
Y - m(X)-[h(X,T) - E_{\varpi}\{h(X,T)\mid X\}] = Y - m(X)-[\tau(X,T) - E_{\varpi}\{\tau(X,T)\mid X\}] \quad \text{a.s.}.
 $$
From \eqref{con:decom2}, any function {\color{black}$h\in\mathcal{S}$} minimizes the generalized R-loss $L_c(\cdot)$. Therefore, when $T$ is continuous, directly minimizing the generalized R-loss fails to uniquely identify the target estimand $\tau(x,t)$, as there are infinitely many solutions in $\mathcal{S}$. This result theoretically substantiates the ill-posedness of estimating $\tau(x,t)$ by minimizing the empirical counterpart of $L_c(\cdot)$ using nonparametric estimators, and also explains the failure-to-estimate issue illustrated in Fig.~\ref{fig:example}.
\par  Part (i) of Proposition \ref{thm:id} below rigorously proves that $\mathcal{S}$ in fact contains all  minima of  $L_c(\cdot)$ in  $\mathcal{L}_{\mathcal{P}}^2(X,T)$. In contrast, minimizing the binary-treatment R-loss \eqref{rloss:binary} which incorporates the zero condition \eqref{eq:shape},  can successfully  identify $\tau $, because  \eqref{eq:shape} narrows  the general solution set $\mathcal{S}$ into  
\bee\label{solutionsetbinary}\mathcal{S}^{\natural}=\{h\mid h(X,T)=\tau(X,T) \text{ a.s.}\}
\ee
with a formal proof relegated to  $\mathsection$\ref{sec:partIIint} in the Supplementary Material. Despite the popular use of  R-loss \eqref{rloss:binary} in  literature \citep{zhao2022selective,nie2021quasi}, the corresponding identification problem has not been rigorously discussed. Part (ii) of Proposition \ref{thm:id} fulfills this gap.
\begin{proposition}\label{thm:id}
Suppose Assumptions \ref{A:UNC}--\ref{A:NI} hold. We have the following identification results.
\begin{itemize}
\item[(i)] Suppose $T$ is a continuous treatment and $\tau\in\mathcal{L}_{\mathcal{P}}^2(X,T)$. Then $\mathcal{S}$  is the solution set of the following optimization problem,
\bee\label{opt:1}
\argmin_{h\in \mathcal{L}_{\mathcal{P}}^2(X,T)}L_c(h).
\ee
\item[(ii)] Suppose $T$ is a binary treatment and $\tau(\cdot,1)\in \mathcal{L}_{\mathcal{P}}^2(X)$. Additionally assume the positivity assumption such that  $e(x)$ satisfying $e(x) \in (\epsilon',1- \epsilon')$ for some fixed $\epsilon' > 0$ and for all $x\in \mathbb{X}$. Then among the set of interested functions, $\mathcal{L}_b = \{h\mid  h(\cdot,1)\in\mathcal{L}_{\mathcal{P}}^2(X) \text{ and }h(X,0) = 0 \text{ a.s.}\}$, {\color{black}$\mathcal{S}^\natural$ in \eqref{solutionsetbinary}} is the solution set of the   optimization problem:
$
\argmin_{h\in\mathcal{L}_b}L_b(h).
$
\end{itemize}

\end{proposition}
 \subsection{One-step nonparametric identification with a functional zero constraint}\label{sec:fakeid}
{\color{black}The contrast between Proposition~\ref{thm:id} (i) and (ii) suggests that, in order to identify $\tau(x,t)$ nonparametrically for the continuous-treatment R-learner, it can be beneficial to leverage the zero condition \eqref{eq:shape} satisfied by $\tau(x,t)$. 
A straightforward way to impose the zero condition 
 is to solve \eqref{opt:1} among all functions satisfying \eqref{eq:shape}.  From a population level, one may consider to solve the following  one-step optimization problem with  a functional zero constraint:
\bee\label{opt:1:rec}
\argmin_{h\in \mathcal{L}_{\mathcal{P}}^2(X,T)\cap\{h\mid h(x,0) = 0\text{ for any }x\in\mathbb{X}\}}L_c(h).
\ee
However, Proposition \ref{po:nonunique} shows that, this strategy continues to fail in achieving a nonparametric identification of $\tau(x, t)$.
\begin{proposition}\label{po:nonunique}
Suppose Assumptions \ref{A:UNC}--\ref{A:NI} hold, and $(X,T)$ has a bounded density function, i.e., $\sup_{(x,t)\in\mathbb{X}\times \mathbb{T}} |f(x,t)|< \infty$. Let $\check{\tau}(x,t\mid s)$ be a function taking the following form,
\bee\label{checktau}
\check{\tau}(x,t\mid s) = \begin{cases}
\tau(x,t) + s(x) & \text{when } t\neq 0
\\
0 & \text{when } t = 0
\end{cases},
\ee
where $s $ can be any function in $\mathcal{L}^2_{\mathcal{P}}(X)$. Then    $\check{\tau}(x,t\mid s)$ with any $s\in\mathcal{L}_{\mathcal{P}}^2(X)$ solves \eqref{opt:1:rec}. 
\end{proposition}
The failure of \eqref{opt:1:rec} to   identify $\tau$,  stems from the fundamental nature that the solutions of \eqref{opt:1}  in $\mathcal{S}$  
take the form of $\tau  + s $ in an almost-surely sense rather than in an exact sense.  Consequently, imposing the zero condition \eqref{eq:shape}  is insufficient to eliminate the solutions  in $\mathcal{S}$ that manifest as $\check{\tau}(x,t\mid s)$  as long as the density function of $(X,T)$ is bounded.  More specifically, any $\check{\tau}(x,t\mid s)$   can 
 satisfy the zero condition  \eqref{eq:shape} while still belong to $\mathcal{S}$. Thus the functional minima of \eqref{checktau} are non-unique, and the subsequent empirical nonparametric estimation procedure based on \eqref{opt:1:rec} still suffers severe ill-posedness due to such non-uniqueness.  
To be more specific, the B-spline functions satisfying the zero condition, can approximate any minimum $\check{\tau}(x,t\mid s)$ with smooth $\tau(x,t)$ and $s(x)$ arbitrarily well, 
 as the number of basis grows; see Proposition~\ref{prop:zerospline}(ii) 
 for details.  So for an R-learner formulated through  the empirical resolution of \eqref{opt:1:rec} with B-spline functions, although it can satisfy the zero condition \eqref{eq:shape}, it is still  ill-posed.  
  We demonstrate such ill-posedness by continuing our numerical experiment in $\mathsection$\ref{sec:simple}. The new   R-learner  satisfies the zero condition yet still fails to approximate $\tau(x,t)$ well  (Fig. \ref{fig:example}d). 
As shown in $\mathsection$\ref{sec:simeg} in the Supplementary File, after increasing the sample size, it continuously yields poor estimation performance. 
 
 \subsection{Two-step Tikhonov identification and $\ell_2$ regularization R-learner}\label{sec:id}
To resolve the ill-posedness of nonparametrically estimating $\tau (x,t)$ through the one-step R-loss optimization \eqref{opt:1:rec}, it is necessary to address the non-uniqueness   of the optimization solutions. Tikhonov  $\ell_2$-regularization \citep{tikhonov1963solution}, originally developed in non-linear functional analysis, aims to identify a specific solution for the functional least-squares problems with non-unique solutions \citep[$\mathsection$37.14]{zeidler2013nonlinear}. Tikhonov regularization, together with the idea of the functional zero constraint in $\mathsection$\ref{sec:fakeid}, inspires our two-step identification strategy for $\tau(x,t) $.
\par
In {Step I}, from the population level, we solve the 
$\ell_2$-regularized variant of \eqref{opt:1} with some given $\rho > 0$,
\bee\label{id:eq2} 
{\tau}_\rho = \argmin_{h\in \mathcal{L}_{\mathcal{P}}^2(X,T)} L_{c,\ell_2}(h\mid \rho) &= \argmin_{h\in \mathcal{L}_{\mathcal{P}}^2(X,T)} L_c(h) +\rho \|h\|_{\mathcal{L}_{\mathcal{P}}^2}^2.
\ee  
The new loss $L_{c,\ell_2}(h\mid \rho)$ is strictly convex over $\mathcal{L}_{\mathcal{P}}^2(X,T)$ due to the addition of a strictly convex functional $\rho\|h\|_{\mathcal{L}^2_{\mathcal{P}}}^2 = \rho E\{h^2(X,T)\}$. Thus minimizing $L_{c,\ell_2}(h\mid \rho)$  becomes  well-posed 
and yields a unique functional minimum ${\tau}_\rho $. 
    Theorem~\ref{pro:tech} explicitly characterizes this unique minimum. 
 
\begin{theorem}\label{pro:tech} Define the following  intermediary function in  $\mathcal{S}$:
\bee\label{ttaudef}
\tilde{\tau}(x,t) = \tau(x,t) - \E\{\tau(X,T)\mid X = x\}.
\ee  When Assumptions \ref{A:UNC}--\ref{A:NI} hold and $\tau\in\mathcal{L}_{\mathcal{P}}^2(X,T)$, given $\rho>0$, the solution set of \eqref{id:eq2}
is $\mathcal{S}_\rho = \{h\mid h(X,T) ={\tau_\rho}(X,T) \ \text{a.s.}\}$ with 
$
 {\tau_\rho}(x,t) = (1+\rho)^{-1}\tilde{\tau}(x,t).
$
\end{theorem}
 Theorem~\ref{pro:tech} implies that we can identify an intermediary function in  $\mathcal{S}$, namely, $\tilde{\tau}    $ in \eqref{ttaudef}, by augmenting~$\tau_\rho$ with a factor of $(1+\rho)$, i.e., $\tilde{\tau} = (1+\rho)\tau_\rho$.  
The $\tilde{\tau}    $ is a solution for the original     R-loss~$L_c(h)$. We defer a more detailed discussion on the intuition behind why minimizing the \(\ell_2\)-regularized R-loss helps to identify this specific solution \(\tilde{\tau}\) for the original R-loss, to Section~\ref{details:R-learner} in the Supplementary File.

In {Step II},  we transform $\tilde{\tau}  $ through a zero-constraining operator $\mathscr{C} (\cdot): \mathcal{L}_{\mathcal{P}}^2(X,T) \mapsto \mathcal{L}_{\mathcal{P}}^2(X,T)$ such that $\mathscr{C} (h)(x,t) = h(x,t) - h(x,0)$ for any $h\in\mathcal{L}_{\mathcal{P}}^2(X,T)$. Then we have
\[
\mathscr{C} (\tilde{\tau})(x,t) 
= \mathscr{C} ((1+\rho)\tau_\rho)(x,t) =  (1 + \rho) \left\{ {\tau}_\rho(x,t)  -  {\tau}_\rho(x, 0)\right\}.
\]
The operator \(\mathscr{C} (\cdot)\)  ensures that any function undergoing its transformation  will satisfy the zero condition~\eqref{eq:shape}. 
Because ${\tau}$ is the only function in $\mathcal{S}$ that satisfies the zero condition \eqref{eq:shape}, transforming any function in the solution set $\mathcal{S}$ will indeed identify $\tau$. Thus with $\tilde{\tau}\in\mathcal{S}$, $\mathscr{C} (\tilde{\tau})$ in our second step ultimately identifies~$\tau$. Formally, we have the following theoretical justification. 

\begin{theorem}\label{thm:onetwoid} 
Suppose Assumption~\ref{A:CS} holds  and  $\tau(x,t)$ is continuous at $t = 0$ for any $x\in\mathbb{X}$. Then $\mathscr{C} (h )  = \tau $, a.s., for any $h\in\mathcal{S}$.  As a special case,  $\mathscr{C}(\tilde{\tau}) = \tau$, a.s..
\end{theorem}
Our two-step identification strategy for $\tau(x,t)$ based on Tikhonov $\ell_2$-regularization, differs from the classical expectation-based identification approach, where the causal estimand is identified by the expectation of a specific estimating function, and the estimator is   the empirical average of that function; see e.g., the augmented inverse
propensity score weighting estimator for average treatment effect \citep{robins1994estimation}, or the debiased estimator for a general   estimand \citep{chernozhukov2023simple}. Our two-step identification relies on a sequence of functionals indexed by $\rho$, and any   
  fixed and positive $\rho$ with $n$ is sufficient to develop a well-posed empirical R-learner. On the other hand, a fast vanishing $\rho$ will   introduce large estimation variance. Intuitively, a small $\rho$ will result in a weak $\ell_2$-regularization. Then with finite samples, such weak $\ell_2$-regularization can not make the minimization of the empirical analogy of $L_{c,\ell_2}(h\mid \rho)$ significantly different from naively minimizing the empirical analogy of $L_{c}(h)$, which is ill-posed as shown in $\mathsection$\ref{sec:simple}. 
Such phenomenon will also be revealed by Theorem \ref{thm:main:mainpaper}, where a fast vanishing $\rho$   cannot yield a well-controlled error rate of our proposed estimator, while a fixed $\rho > 0$ could.

\par
Our identification strategy   leads to an  $\ell_2$-regularized R-learning procedure, briefly, the $\ell_2$-regularized R-learner. Formal algorithm  and implementation details of $\ell_2$-regularized R-learner are deferred to $\mathsection$\ref{details:R-learner} in   Supplementary File.    
   The newly proposed R-learner inherits many practical and theoretical advantages from the original R-learner \citep{nie2021quasi}. Practically,  minimizing the empirical analogy of $L_{c,\ell_2}(h\mid \rho)$ separates the process of estimating the nuisance functions and that of estimating the target estimand, both of which can be implemented by flexible machine learning methods. {\color{black}Theoretically,    when $\tau(x,t)$ is approximated by the method of sieve, our proposed R-learner is  theoretically robust to slow convergence rates of  the nuisance estimators.  Notably, this paper mainly implements our $\ell_2$-regularized R-learner with the method of sieve for our asymptotical and numerical analysis. Nevertheless, our estimation strategy and general Algorithm \ref{ago:1} are not tied on the method of sieve. In fact, any loss-minimization method, e.g., the regularized nonparametric regression, deep neural networks, and boosting can be flexibly used to specify $\mathcal{H}$ and minimize \eqref{emp:loss}  in Algorithm \ref{ago:1}, for $\tau(x,t)$ estimation. For future research, it is interesting to explore the practical and theoretical performances of the $\ell_2$-regularized R-learner, implemented with other machine learning algorithms for $\tau(x,t)$ approximation.}   
\par {Continuing with the numerical experiment in $\mathsection$\ref{sec:simple}, 
we demonstrate the effectiveness of the proposed $\ell_2$-regularized  R-learner (Fig. \ref{fig:example}e), which well approximates 
the true $\tau(x,t)$ (Fig. \ref{fig:example}b). We further generalize the S-learner and X-learner proposed in \citet{kunzel2019metalearners} into the continuous-treatment case, with details introduced in $\mathsection$\ref{sec:SX}. The $\tau(x,t)$ estimators by the generalized S- and X-learners for the numerical experiment in $\mathsection$\ref{sec:simple}, are reported in Fig.~\ref{fig:exampless}, and we can see our proposed R-learner also outperforms these two generalized meta-learners.
\section{Implementation with the method of sieve and theoretical properties}\label{sec:sieve}
To obtain in-depth understanding of  our  $\ell_2$-regularized R-learner, we study its implementation details and theoretical properties.   
The sieve approximation \citep{geman1982nonparametric} has been broadly studied and applied for nonparametric estimation due to its good interpretability and  theoretical properties \citep[e.g.,][]{chen2007large}. In particular, we consider a triangular array  expansion of $h(x,t)$,
\bee\label{h:bspline}
h(x,t) = \phi^\T\begin{bmatrix}
\Psi_1(x,t),&\cdots &,\Psi_K(x,t)
\end{bmatrix}^\T = \phi^\T\Psi(x,t),
\ee
where $\Psi(x,t)$ is the vector of basis functions, $\phi\in\RR^{K}$ is the coefficient vector, and $K$ is the number of basis functions.
The empirical counterpart of $L_{c,\ell_2}(h\mid \rho)$ with $h$ in \eqref{h:bspline}, is 
\bee\nonumber
\hat{L}_{c,\ell_2}(h\mid \rho,\hat{\varpi},\hat{m})  = &P_n \left[\Big[Y - \hat{m} (X) -\phi^\T\Psi(X,T)  + \phi^\T\E_{\hat{\varpi}}\{\Psi(X,T)\mid  X\}\Big]^2 + \rho\big\{\phi^\T\Psi(X,T)\big\}^2\right],
\ee
where $\hat{m}$ and $\hat{\varpi}$ are nuisance function estimators for ${m}$ and ${\varpi}$  trained via any generic and fine-tuned machine learning method, respectively. Following Step~I in $\mathsection$\ref{sec:id}, we estimate ${\tau}_\rho$ by solving the empirical counterpart of \eqref{id:eq2}, which has a closed-form solution: 
\bee\nonumber
\hat{{\tau}}_\rho(x,t) = \argmin_{h = \phi^\T\Psi}\hat{L}_{c,\ell_2}(h\mid \rho,\hat{\varpi},\hat{m}) = \hat{\phi}^\T\Psi(x,t).
\ee
Here we denote $\hat{R}_n = P_n[\{\M\Psi(X,T) - \hat{\Gamma}(X)\}\{\M\Psi(X,T) - \hat{\Gamma}(X)\}^\T]$, $\hat{Q}_n = P_n[\M\Psi(X,T)\M\Psi^\T(X,T)]$, and
\bee\label{def:phi}
\hat{\phi} = \big(\hat{R}_n + \rho\hat{Q}_n\big)^{-1} P_n\Big[\big\{\Psi(X,T) - \hat{\Gamma}(X)\big\}\big\{Y - \hat{m}(X)\big\}\Big].
\ee
Following  Step~II in $\mathsection$\ref{sec:id}, our proposed estimator for $\tau(x,t)$ is
\bee\nonumber
\hat{\tau}(x,t) = \mathscr{C} \left((1+\rho)\hat{{\tau}}_\rho\right)(x,t) = (1 + \rho)\hat{\phi}^\T\big\{\Psi(x,t) -  \Psi(x,0)\big\}.
\ee
The formal algorithm to implement  the $\ell_2$-regularized R-learner with sieve approximation, incorporated with a sample splitting procedure and more details about nuisance function training, are  provided in $\mathsection$\ref{sec:realsieveRlearner} in the Supplementary File. For theoretical analysis, we choose $\Psi(x,t)$ as the tensor-product of B-splines;~see e.g.,
\citet{chen2015optimal}. 
We provide technicality including the theoretical properties of the B-spline basis in $\mathsection$\ref{sec:bspline} and  regularity conditions in $\mathsection$\ref{sec:prea}. To address the effect of nuisance function estimation, we consider the following concentration conditions for   $\hat{m}(x)$ and $\hat{\varpi}(x)$ with $r_m,r_{\varpi} \precsim  1$, 
\begin{eqnarray}\label{rate:m:main}
 &&\big\|\hat{m} - m\big\|_{\mathcal{L}^2_{\mathcal{P}}}= {o}_{{P}}(r_m), \quad \sup_{x\in\mathbb{X}}\|\hat{\varpi}(\cdot\mid x) - {\varpi}(\cdot\mid x)\|_{\mathcal{L}^2}= o_\p(r_{\varpi}).
\end{eqnarray}
We present the asymptotic properties of the $\ell_2$-regularized R-learner $\hat{\tau}(x,t)$ with sieve approximation, as obtained by Algorithm~\ref{ago:1}. The following theorem follows as a special case of the more general Theorem~\ref{thm:main}, which is established under weaker assumptions and deferred to the Supplementary File.
\begin{theorem}\label{thm:main:mainpaper}
Suppose all   conditions in the main Theorem~\ref{thm:main} hold,  
and $\hat{m}(x),\hat{\varpi}(x)$ satisfy \eqref{rate:m:main}.   When   $r_{m}, r_{\varpi} \precsim n^{-1/4}$ and $\tilde{\tau}(x,t)$ belongs to the H{\"o}lder class $ \Lambda(p,c,\mathbb{X}\times\mathbb{T})$ for some $p,c > 0$ (c.f., Definition~\ref{def:holder}), we have the following results. 
\begin{itemize}
\item\text{(Consistency). }When choosing $K$ and $\rho$ such that $K\asymp n^{(d + 1)/(2p)}$ and 
$  n^{-1 +(d + 1)/(2p)}\log n\prec \rho\precsim 1$, we have
\bee\label{main:consist}
\big|\hat{\tau}(x_0,t_0) - \tau(x_0,t_0)\big|=\mathcal{O}_P(n^{-1/2 + (d + 1)/(4p)}).
\ee
\item\text{(Limiting distribution). }Suppose further the $(2 + c_0)$-order moment condition in \eqref{clt:con}
holds for some fixed $c_0 > 0$. Choosing $K\asymp n^{\epsilon_{\text{clt}} + (d + 1)/(2p)}$ for any $\epsilon_{\text{clt}} \in (0,1/2 - (d + 1)/(2p))$ and $ \rho \asymp n^{-1/2}$, and with $\tilde{\sigma}$    defined in \eqref{def:tsigma}, we have 
\bee\label{main:clt}
{\sqrt{n}}\tilde{\sigma}^{-1}\big\{\hat{\tau}(x_0,t_0) - \tau(x_0,t_0)\big\}\leadsto \mathcal{N}(0,1).
\ee
\item(Confidence interval). Let $\hat{\sigma}$ be obtained by  Algorithm \ref{alg:sigma} with  
$
\|\hat{\mu} - \mu\|_{\mathbb{X}\times\mathbb{T}} = o_P(1).
$
Then we have
\bee\label{main:ci}
{\sqrt{n}}\hat{\sigma}^{-1}\big\{\hat{\tau}(x_0,t_0) - \tau(x_0,t_0)\big\}\leadsto \mathcal{N}(0,1).
\ee
\end{itemize}
\end{theorem}
 Theorem~\ref{thm:main:mainpaper} establishes that, given reasonably well-estimated nuisance functions at the $\mathcal{O}(n^{-1/4})$ rate and under certain smoothness conditions, our $\ell_2$-regularized R-learner is consistent and satisfies pointwise asymptotic normality, facilitating valid statistical inference. Notably, the error rate in \eqref{main:consist} is attainable as long as $\rho$ is not too small, with $\rho \asymp 1$ and $\rho \asymp n^{-1/2}$ serving as two special cases. This aligns with the identification result in $\mathsection$\ref{sec:id}, demonstrating that our two-step identification strategy successfully identifies $\tau$ even when $\rho$ remains fixed at certain positive value. For   simplicity, we focus on the scenario of $ \rho \asymp n^{-1/2}$ for our inference results, but such condition can also be relaxed as shown in Theorem~\ref{thm:main}. Further technical discussions regarding (i)  the conditions for Theorem~\ref{thm:main:mainpaper}, (ii) the order requirements for tuning parameters, and (iii) the construction of confidence intervals, are provided in Section~\ref{sec:thm}.
\bibliographystyle{chicago}
\bibliography{ref}
\newpage
\makeatletter
\renewcommand \thesection{S\@arabic\c@section}
\renewcommand\thetable{S\@arabic\c@table}
\renewcommand \thefigure{S\@arabic\c@figure}
\renewcommand \theequation{S\@arabic\c@equation}
\makeatother
\setcounter{section}{0}
\setcounter{figure}{0}
\setcounter{table}{0}
\setcounter{equation}{0}
\section*{\huge Supplementary material for ``Towards R-learner 
with Continuous Treatments''}
 Supplementary material includes formal algorithms and implementation details of the proposed $\ell_2$-regularized R-learner, general asymptotic results, generalized S- and X-learners with continuous treatments, a cross-validation-based tuning parameter selection method, details of all numerical experiments, an introduction to B-splines, and all technical proofs.
\section{$\ell_2$-regularized R-learner: Formal algorithm and more discussions}\label{details:R-learner}
 
In this section, we introduce the general $\ell_2$-regularized R-learner in details, developed based on the identification strategy described in Section~\ref{sec:id}. Let $\mathcal{H}$ denote the function class used to approximate $\tau(x,t)$, which can be implemented using various machine learning algorithms, such as linear regression, random forests, or neural networks. 

To approximate the population loss
\bee\label{L2loss:sup}
L_c(h) = E\big[Y - m(X) - h(X,T) + E_{\varpi}\{h(X,T)\mid X\}\big]^2,
\ee
we first estimate the two nuisance functions $m(x)$ and $\varpi(t\mid x)$ using generic machine learning methods. Incorporating the standard cross-fitting procedure \citep{chernozhukov2018double,schick1986asymptotically} with $\mathcal{J}$ sample splits, we summarize the general $\ell_2$-regularized R-learner in Algorithm~\ref{ago:1}.

Similar to the original R-learner, the $\ell_2$-regularized R-learner provides a general estimation framework that accommodates any off-the-shelf machine learning algorithms for nuisance function estimation, as well as any loss-minimization-based machine learning algorithms for estimating $\tau(x,t)$.  A variant of the cross-fitting procedure \citep[Definition 3.2]{chernozhukov2018double} can also be used in Algorithm~\ref{ago:1}. Specifically, in Step~3, we may compute $\mathcal{J}$ separate estimators of ${{\tau}}_\rho(x,t)$ by minimizing the empirical loss within each subsample, and then aggregate them by averaging to obtain the final estimator $\hat{{\tau}}_\rho(x,t)$.  Although the estimators produced by different cross-fitting variants may differ slightly, their asymptotic properties remain equivalent \citep[Remark 3.1]{chernozhukov2018double}.

  \par
Next, we provide further discussion on the role of $\ell_2$-regularization in our proposed identification and estimation procedures. $\ell_2$-regularization has been extensively employed in classical statistical methods, including ridge regression for linear and logistic models \citep{hoerl1970ridge}, regularized spline regression \citep{o1986statistical}, and kernel ridge regression \citep{aizerman1964theoretical}, among others. In these methods, the $\ell_2$ penalty primarily serves to control the complexity of the estimator, thereby yielding stable and theoretically well-behaved estimators.  However, these prior works assume that the target estimand is at least a locally unique minimizer of the corresponding population loss function. To the best of our knowledge, our paper is the first to utilize $\ell_2$-regularization not only to regularize  complexity of the estimator, but also to resolve a fundamentally different challenge: a population-level ill-posedness and non-identification problem, where the original population-level loss function has infinitely many minimizers. In this setting, $\ell_2$-regularization plays a dual role by both enforcing identifiability and controlling the complexity of estimation.
\par
Finally, Theorem~\ref{pro:tech} demonstrates that the population-level optimization solution $\tau_\rho$ of $L_{c,\ell_2}(h\mid \rho)$ takes the form of $(1+\rho)^{-1}\tilde{\tau}$ and thus will approach $\tilde{\tau}$ as $\rho\rightarrow 0$. We give some intuitive explanations on why minimizing the \(\ell_2\)-regularized R-loss helps to  identify this specific solution \(\tilde{\tau}\) of the original   continuous-treatment R-loss $L_c(h)$. First, the new loss $L_{c,\ell_2}(h\mid \rho)$ becomes strictly convex over $\mathcal{L}_{\mathcal{P}}^2(X,T)$ due to the addition of a strictly convex functional $\rho\|h\|_{\mathcal{L}^2_{\mathcal{P}}}^2 = \rho E\{h^2(X,T)\}$. Thus, its minima equal to a unique function $\tau_\rho\in\mathcal{L}_{\mathcal{P}}^2(X,T)$ a.s.. Second, 
intuitively, when $\rho\rightarrow 0$, the difference between two loss functions, $L_{c,\ell_2}(h\mid \rho)$ and $L_c(h)$, vanishes. Therefore, $\tau_\rho$ shall approach $L_c(h)$'s solution set $\mathcal{S}$ as $\rho\rightarrow 0$. On the other hand, among $\mathcal{S}$, $\tilde{\tau}$ has the smallest value in terms of  $L_{c,\ell_2}(h\mid \rho)$  for any $\rho > 0$. This is because any function $h$ in  $\mathcal{S}$ produces the same value of the first term on the right-hand side of  \eqref{L2loss:sup} in $ L_c(h)$, while only $\tilde{\tau}$ can minimize the second term, namely $\rho \|h\|_{\mathcal{L}_{\mathcal{P}}^2}^2$, as it has the smallest $\mathcal{L}^2_{\mathcal{P}}$ norm. Thus, $\tau_\rho$  particularly  approaches $\tilde{\tau}$ in $\mathcal{S}$ as $\rho \rightarrow 0$. 
\par

There might be other regularization terms that could also be added to the generalized R-loss and resolve the non-identification issue with continuous treatments, similar to the Tikhonov regularization. This may further motivate other regularized R-learner for continuous treatments, and we leave the exploration along this direction for future research.

\begin{algorithm}[t]
\caption{General $\ell_2$-regularized R-learner}\label{ago:1}
\SetAlgoLined
\textit{Step 1}.\ Split $\{Z_i\}_{i = 1}^{n}$ into $\mathcal{J}$ mutually exclusive ($\mathcal{J} > 1$), and equally sized or nearly  equally sized sub-sample sets $\mathcal{S}_{1},\dots,\mathcal{S}_{\mathcal{J}}$, such that $\cup_{j = 1}^{\mathcal{J}}\mathcal{S}_{j} = \{Z_i\}_{ i = 1}^n$; { let $i$th sample belongs to  $j_i$th subset}\;
\textit{Step 2}.\  For each $j\in[\mathcal{J}]$,  train the nuisance function estimator $ \{\hat{m}^{(-j)}(x),\hat{\varpi}^{(-j)}(t\mid x)\}$ with all data except $\mathcal{S}_j$ via any generic and fine-tuned machine learning method\;
\textit{Step 3}.\ Choosing $\rho > 0$, estimate $\hat{{\tau}}_\rho(x,t)$ by 
\bee
\label{emp:loss}
\qquad \hat{{\tau}}_\rho(x,t)&=\argmin_{h(\cdot,\cdot)\in\mathcal{H}}\hat{L}_{c,\ell_2}\{h(\cdot,\cdot)\mid \rho\}
\\
&= \argmin_{h(\cdot,\cdot)\in\mathcal{H}}\frac{1}{n}\sum_{i = 1}^n \Big[Y_i - \hat{m}^{(-j_i)}(X_i) -h(X_i,T_i) + E_{\hat{\varpi}^{(-j_i)}}\{h(X_i,T_i)\mid X_i\}\Big]^2 
\\
&\quad + \rho P_n\{h^2(X,T)\}.
\ee
\textit{Step 4}.\ Output $\hat{\tau}(x,t)   =  (1 +\rho)\cdot \{\hat{{\tau}}_\rho(x,t) - \hat{{\tau}}_\rho(x, 0)\}$.
\end{algorithm}
\section{$\ell_2$-regularized R-learner implemented with the method of sieve}\label{sec:realsieveRlearner}

\begin{algorithm}[t]
\caption{The general $\ell_2$-regularized R-learner with sieve approximation}\label{ago:1:sieve}
\SetAlgoLined
\textit{Step 1}.\ Split $\{Z_i\}_{i = 1}^{n}$ into $\mathcal{J}$ mutually exclusive ($\mathcal{J} > 1$), and equally sized or nearly  equally sized sub-sample sets $\mathcal{S}_{1},\dots,\mathcal{S}_{\mathcal{J}}$, such that $\cup_{j = 1}^{\mathcal{J}}\mathcal{S}_{j} = \{Z_i\}_{ i = 1}^n$; { let $i$th sample belongs to  $j_i$th subset}\;
\textit{Step 2}\  For each $j\in[\mathcal{J}]$,  obtain $\{\hat{m}^{(-j)}(x),\hat{\Gamma}^{(-j)}(x)\}$ using all data except $\mathcal{S}_j$ by the method of sieves (see $\mathsection$\ref{rk:gamma:est} for two options of $\hat{\Gamma}^{(-j)}(x)$ training)\;
\textit{Step 3}\ Obtain $\hat{\phi}$ from \eqref{def:phi} with $\rho > 0$ and obtain 
$
\hat{{\tau}}_\rho(x,t) = \hat{\phi}^\T\Psi(x,t);
$

\textit{Step 4}.\ Output $\hat{\tau}(x,t) = (1 + \rho)\cdot\{\hat{{\tau}}_\rho(x,t) - \hat{{\tau}}_\rho(x,  0)\}$.
\end{algorithm}

We consider to minimize \eqref{emp:loss} in Algorithm~\ref{ago:1} with the method of sieve. 
Then solving \eqref{emp:loss}  in our proposed algorithm  becomes solving ${\phi}$ from 
\bee\nonumber
\argmin_{\phi\in\RR^K} \frac{1}{n}\sum_{i = 1}^n \Big\{Y_i - \hat{m}^{(-j_i)}(X_i) -\phi^\T\Psi(X_i,T_i) + \phi^\T\hat{\Gamma}^{(-j_i)}(X_i)\Big\}^2 + \frac{\rho}{n}\sum_{i = 1}^n\big\{\phi^\T\Psi(X_i,T_i)\big\}^2,
\ee
where $\hat{\Gamma}^{(-j_i)}(X_i)$ is an estimator of $\Gamma(X_i) = {E}_{{\varpi}}\{\Psi(X,T)\mid X_i\}$; See $\mathsection$\ref{rk:gamma:est} for the details of its estimation, either through estimating $\varpi$ or through a coordinate-wise regression.  Then by straightforward algebra, one has
\bee\label{def:phi}
\hat{\phi} = \frac{1}{n}\big(\hat{R}_n + \rho\hat{Q}_n\big)^{-1} \sum_{i = 1}^n\big\{\Psi(X_i,T_i) - \hat{\Gamma}^{(-j_i)}(X_i)\big\}\big\{Y_i - \hat{m}^{(-j_i)}(X_i)\big\} ;
\ee
here we denote  
\bee\nonumber
\hat{R}_n &= \frac{1}{n}\sum_{i = 1}^n\Big\{\M\Psi(X_i,T_i) - \hat{\Gamma}^{(-j_i)}(X_i)\Big\}\Big\{\M\Psi(X_i,T_i) - \hat{\Gamma}^{(-j_i)}(X_i)\Big\}^\T ,
\\
\hat{Q}_n &=\frac{1}{n} \sum_{i = 1}^n\M\Psi(X_i,T_i)\M\Psi^\T(X_i,T_i).
\ee
The generalized R-learner with sieve approximation are formally provided in Algorithm~\ref{ago:1}.
\subsection{Two options for $\hat{\Gamma}$ training}\label{rk:gamma:est}
We discuss two options for obtaining the nuisance vector-valued function 
$\hat{\Gamma}^{(-j)}(x)$ for all $j\in[\mathcal{J}]$.
\begin{enumerate}
\item[(i)] $\hat{\Gamma}^{(-j)}(x)={E}_{\hat{\varpi}^{(-j)}}\{\Psi(X,T)\mid X = x\}$, where  $\hat{\varpi}^{(-j)}(t\mid x)$ is the generalized propensity score estimator  over $\{Z_i\}_{i = 1}^n\setminus \mathcal{S}_j$. One-dimensional numerical integrations can be implemented to approximate all conditional expectations in $\hat{\Gamma}^{(-j)}(x)$. 
\item[(ii)]  $\hat{\Gamma}^{(-j)}(x)$ consists of coordinate-wise nonparametric regressions. For simplicity, we consider $\Psi(x,t)$ as the tensor-product of B-splines (c.f.,~$\mathsection$\ref{sec:bspline}) and . First, observing that $\Gamma(x) = E\{\psi(T)\mid X = x\}\otimes \Psi(x)$, where $\Psi(x)$ denotes $\psi(x^{(1)})\otimes\cdots\otimes\psi(x^{(d)})$ and the dimension of $\psi(T)$ is $k_T$, we can estimate $\hat{\Gamma}^{(-j)}(x)$ by setting
$$
\hat{\Gamma}^{(-j)}(x) = \hat{E}^{(-j)}\{\psi(T)\mid X = x\}\otimes\Psi(x),
$$
where 
$$
\hat{E}^{(-j)}\{\psi(T)\mid X= x\} = \left(\hat{E}^{(-j)}\{\psi^{(1)}(T)\mid X= x\},\cdots\cdots,\hat{E}^{(-j)}\{\psi^{(k_T)}(T)\mid X= x\}\right)^\T,
$$ 
are obtained by coordinate-wise  regressions over the covariates-response pairs $\{X_i,\psi^{(k)}(T_i)\mid Z_i \in\{Z_i\}_{i = 1}^n\setminus \mathcal{S}_j\}$ for any $k \in[k_T]$. Suppose $\psi(t)\in \RR^{k_T}$ is a $k_T$-dimensional B-spline basis for the treatment variable, and let $1_{k_T} = (1,\dots,1)^\T\in\RR^{k_T}$. Inspired by the sum invariance property of B-splines (Lemma \ref{po:bspline}) such that $$1_{k_T}^\T E\{ \psi(T)\mid X = x\} =  E\{ 1_{k_T}^\T\psi(T)\mid X = x\} =  E(\sqrt{k_T}\mid X = x) = \sqrt{k_T}$$ for any $x\in\mathbb{X}$, we further impose a shape constraint for  $\hat{E}\{\psi(T)\mid X = x\}$, 
\bee\label{shape:con}
1_{k_T}^\T \hat{E}^{(-j)}\{\psi(T)\mid X = x\} = \sqrt{k_T},\quad\text{for any }x\in\mathbb{X}, j\in[\mathcal{J}].
\ee
A simple strategy can be used to address the above shape constraint. 
First obtain $\hat{E}^{(j)}\{\psi(T)\mid X = x \}$ by  coordinate-wise regression when $j\in[k_T-1]$, and then obtain
\bee\label{ek:train}
\hat{E}^{(-k_T)}\{\psi(T)\mid X = x\} = \sqrt{k_T} - \sum_{j = 1}^{k_T -1}\hat{E}^{(-j)}\{\psi(T)\mid X = x \}.
\ee
\end{enumerate}
The above approaches apply similarly to other types of basis functions. Conditional density estimation in method (i) is often difficult, especially when $X$ is high dimensional. In this case, method (ii) might be  a more flexible alternative, which replaces conditional density estimation with a series of regressions.
\subsection{Main asymptotic results}\label{sec:thm}
For theoretical analysis, we choose $\Psi(x,t)$ as the tensor-product of B-splines  
\citep{chen2015optimal} in this paper. 
We provide technicality including the theoretical properties of the B-spline basis in $\mathsection$\ref{sec:bspline} and  regularity conditions in $\mathsection$\ref{sec:prea}. 
 Classic nonparametric sieve regression \citep{newey1997convergence} often assumes a full-rank gram matrix $Q_n = \E\{\Psi(X,T)\Psi^\T(X,T)\}$ where the dimension of $\Psi$ depends on $n$. 
In stark contrast, theoretical analysis of the proposed R-learner involves a low-rank gram matrix, 
\bee\label{def:rn}
R_n = E\big[\{\Psi(X,T) - \M\Gamma(X)\}\{\Psi(X,T)- \M\Gamma(X)\}^\T\big] = \begin{pmatrix}
U & U_\perp
\end{pmatrix}
\begin{pmatrix}
\Sigma & 
\\
&  0
\end{pmatrix}
\begin{pmatrix}
U^\T
\\
U_{\perp}^\T
\end{pmatrix},
\ee
where the right-hand side of (\ref{def:rn}) is the singular value decomposition of $R_n$, with $\text{rank}(R_n) = \zeta$ and $\Sigma = \diag(\sigma_1,\dots,\sigma_{\zeta})$ such that    $\sigma_1 \geq \dots\geq\sigma_{\zeta} > 0$. Each entry of ${R}_n$ is the probability limit of the corresponding entry of $\hat{R}_n$  in \eqref{def:phi}. Intuitively, the low rank of $R_n$ is tied to the non-identification issue of the generalized R-loss in $\mathsection$\ref{sec:id} when setting $\rho = 0$. In this case,  $\hat{\phi}$ is  asymptotically unsolvable, or equivalently $\hat{R}_n$ in \eqref{def:phi} is asymptotically non-invertible; i.e., $R_n$ is low-rank.  We denote the smallest positive singular value of $R_n$ by $\beta_n = \sigma_{\zeta} > 0$, which plays an essential quantity in our theoretical results. See Lemma \ref{lm:svdR} for  some theoretical justifications of the  low-rankness and detailed spectral properties of $R_n$. 

To address the effect of nuisance function estimation, we consider the following concentration conditions for   $\hat{m}(x)$ and $\hat{\Gamma}(x)$ with $r_m,r_{\gamma},r_{\gamma}'\precsim  1$, 
\begin{eqnarray}\label{rate:m}
 &&\big\|\hat{m} - m\big\|_{\mathcal{L}^2_{\mathcal{P}}}= {o}_{{P}}(r_m),
\\
\label{rate:gamma}
 &&\big\|\hat{\Gamma} - \Gamma\big\|_{\mathbb{X}}/\sqrt{K} = o_P(r_{\gamma}'),
 \\\label{rate:gamma2}
 && \Big\|\int_{x\in\mathbb{X}}\{\hat{\M\Gamma}(x) - {\M\Gamma}(x)\}\{\hat{\M\Gamma}(x) - \M\Gamma(x)\}^{\T}d\mathcal{P}(x)\Big\|^{1/2}_2  = {o}_{{P}}(r_{\gamma}).
\end{eqnarray}
The convergence rate condition of $\hat{m}(x)$ is commonly assumed; see, e.g.,  \citet{kennedy2017non}. The following proposition further implies that, if we obtain $\hat{\Gamma}(x)$ through $\hat {\varpi}(t\mid x)$ by method (i) in $\mathsection$\ref{rk:gamma:est}, the convergence rates in \eqref{rate:gamma} and \eqref{rate:gamma2} are simultaneously attained as long as  $\hat{\varpi}(t\mid x)$ satisfies the corresponding $\mathcal{L}^{2}$-convergence rate uniformly for all $x\in\mathbb{X}$.
\begin{proposition}\label{po:nui:equal2}
Suppose regularity conditions \ref{am:bspline} and \ref{am:densX} hold, and $\hat{\Gamma}(x) = E_{\hat{\varpi}}\{\Psi(X,T)\mid X = x\}$. When $\hat{\varpi}(t\mid x)$ satisfies
$
\sup_{x\in\mathbb{X}}\|\hat{\varpi}(\cdot\mid x) - {\varpi}(\cdot\mid x)\|_{\mathcal{L}^2}= o_\p(r_{\varpi})
$, then 
 \eqref{rate:gamma} and \eqref{rate:gamma2} hold with 
$$
r_{\gamma}' = r_{\gamma} = r_{\varpi}.
$$
\end{proposition}
Theorem \ref{thm:main} summarizes the asymptotic results for the R-learner $\hat{\tau}(x,t)$ obtained by Algorithm \ref{ago:1}.
\begin{definition}[H{\"o}lder class]\label{def:holder}
For any function $h(w)$ with $w\in\RR^{d'}$, we denote its $\alpha$-derivative by $D^{\alpha} h(w)=\partial^{\alpha_1+\dots+\alpha_{d'}}h(w)/{\partial ^{\alpha_{1}}w_{1} \cdots \partial ^{\alpha_{d'}}w_{d'}}$, where $\alpha = (\alpha_1,\dots,\alpha_{d'})$ is a vector of positive integers. For asymptotic analysis, we restrict the target functional estimand to the popular $p$-smooth H{\"o}lder class \citep{newey1997convergence},
\bee\label{def:holder}
\Lambda(p, c,\mathbb{W})=\Bigg\{&h( w):\sup _{\|\alpha\|_{1} \leq\lfloor p\rfloor} \sup _{ w \in  \mathbb{W}}\left|D^{\alpha} h( w)\right| \leq c, 
\sup _{\|\alpha\|_{1}=\lfloor p\rfloor\atop w_1,w_2 \in  \mathbb{W}} \sup _{ w_1 \neq  w_2} \frac{\left|D^{\alpha} h( w_1)-D^{\alpha} h( w_2)\right|}{\| w_1- w_2\|_{2}^{p-\lfloor p\rfloor}} \leq c\Bigg\},
\ee
where $c,p>0$ are fixed, and $h(w)$ belongs to the class of all $\lfloor p\rfloor$-times differentiable functions  over $\mathbb{W}$.
\end{definition}
\begin{theorem}\label{thm:main}
Suppose Assumptions \ref{A:UNC}--\ref{A:CS} and  regularity conditions in $\mathsection$\ref{sec:prea} hold,  
and $\hat{m}(x),\hat{\Gamma}(x)$ satisfy \eqref{rate:m}--\eqref{rate:gamma2}. Suppose further the conditions hold: (i) $r_{\gamma}^2\precsim \sqrt{K\log n/n} \prec 1$; (ii) $\sqrt{K\log n/n} \prec \beta_n$;   (iii) ${ \tilde{\tau}(x,t)}\in \Lambda(p,c,\mathbb{X}\times\mathbb{T})$ for some $p,c > 0$; (iv) $\hat{\Gamma}(x)$ is trained via either one of the methods in $\mathsection$\ref{rk:gamma:est}; (v) $0<\rho\precsim 1$. Then for any $(x_0,t_0)\in\mathbb{X}\times\mathbb{T}$, we have the general upper bound,
\bee\label{thm:gb}
\big|\hat{\tau}(x_0,t_0) - \tau(x_0,t_0)\big| \leq r(n,K,\beta_n,\rho,r_m,r_\gamma,r_{\gamma}'),
\ee
as $n\rightarrow \infty$. An explicit form of $r(n,K,\beta_n,\rho,r_m,r_\gamma,r_{\gamma}')$ is given in \eqref{bigboy}. When $\beta_n\asymp 1$, $p > d +1$, and $r_{m}, r_\gamma,r_{\gamma'}\precsim n^{-1/4}$, we have the following results.
\begin{itemize}
\item\text{(Consistency). }When choosing $K\asymp n^{(d + 1)/(2p)}$ and 
$1\succsim\rho\succ n^{-1 +(d + 1)/(2p)}\log n$, the  rate in \eqref{thm:gb} can be minimized by
\bee\label{main:consist}
\big|\hat{\tau}(x_0,t_0) - \tau(x_0,t_0)\big|=\mathcal{O}_P(n^{-1/2 + (d + 1)/(4p)}).
\ee
\item\text{(Limiting distribution). }Suppose further the $(2 + c_0)$-order moment condition,
\bee\label{clt:con}
\sup_{(x,t)\in\mathbb{X}\times\mathbb{T}}E\big\{|Y - E(Y\mid X,T)|^{2 + c_0} \mid X = x, T = t\big\}< +\infty,
\ee 
holds for some fixed $c_0 > 0$. When choosing $K$ and $\rho$ such that, $K\asymp n^{\epsilon_{\text{clt}} + (d + 1)/(2p)}$ for some fixed $\epsilon_{\text{clt}} \in (0,1/2 - (d + 1)/(2p))$ and  $ n^{-1 + (d + 1)/(2p) + (1 +\delta)\epsilon_{clt}}\prec \rho \prec \sqrt{K\log n/n}$ for some small enough and fixed $\delta > 0$, we have 
\bee\label{main:clt}
{\sqrt{n}}\tilde{\sigma}^{-1}\big\{\hat{\tau}(x_0,t_0) - \tau(x_0,t_0)\big\}\leadsto \mathcal{N}(0,1),
\ee
where $\tilde{\sigma}$  is defined in \eqref{def:tsigma} in the Supplementary Material. Notably, $\rho\asymp n^{-1/2}$ always satisfies the above condition. 
\item(Confidence interval). Let $\hat{\sigma}$ be obtained by  Algorithm \ref{alg:sigma} with $\hat{\mu}(x,t)$ satisfying
$
\|\hat{\mu} - \mu\|_{\mathbb{X}\times\mathbb{T}} = o_P(1).
$
Then we have
\bee\label{main:ci}
{\sqrt{n}}\hat{\sigma}^{-1}\big\{\hat{\tau}(x_0,t_0) - \tau(x_0,t_0)\big\}\leadsto \mathcal{N}(0,1).
\ee
\end{itemize}
\end{theorem}

In Theorem \ref{thm:main}, we first state the upper bound of the $\ell_2$-regularized R-learner with sieve approximation. Condition (i) holds whenever $r_{\gamma}\precsim n^{-1/4}$. Conditions (ii)   holds when $\beta_n$  decays slowly with $n$. Condition (iii) specifies the smoothness of $\tilde{\tau}(x,t)$. 
Recall that $\tilde{\tau}(x,t) = \tau(x,t) - E\{\tau(X,T)\mid X = x\}$. Two sufficient conditions for condition (iii) to hold are (a) both $\tau(x,t)$ and $\varpi(t\mid x)$ are in $\Lambda(p,c,\mathbb{X}\times\mathbb{T})$, and (b) $\varpi(t\mid x)$ does not depend on $x$ in a completely randomized experiment. Overall, similar to the original R-learner in the binary-treatment case \citep{nie2021quasi}, our theoretical results do not rely on the smoothness of the outcome model when all nuisance functions can be estimated with $o_{P}(n^{-1/4})$ rates. Condition (iv) covers the two training strategies of $\hat{\Gamma}(\cdot)$ introduced in $\mathsection$\ref{rk:gamma:est}. Condition (v) allows a very general regime of $\rho$, and the convergence is still valid even $\rho$ is fixed and positive. This result complies with our identification result in Theorem~\ref{pro:tech} stating that our two-step identification strategy can identify $\tilde{\tau}$ and $\tau$ with any fixed $\rho > 0$.
\par
The general upper bound $r(n,K,\beta_n,\rho,r_m,r_\gamma,r_{\gamma}')$ has a  complicate form, thus we defer its explicit form to the Supplementary Material. To ease the exposition, we consider a simple but reasonable condition such that $\beta_n\asymp 1$, $p > d +1$ and $r_{m}, r_\gamma,r_{\gamma'}\precsim n^{-1/4}$ to elucidate results in \eqref{main:consist}, \eqref{main:clt}, and \eqref{main:ci}. The asymptotic behavior of $\beta_n$ relies on the joint design of $(X,T)$. One sufficient condition for  $\beta_n \asymp 1$ to hold is that  $T$ follows complete randomization \citep{splawa1990application}; see Lemma \ref{lm:svdR} (iii). The relationship $p > d + 1$ means that $\tilde{\tau}(x,t)$ is smooth enough with respect to its dimension. A similar condition is also considered in the theoretical analysis of other sieve-type estimators; see e.g., \citet{shi2021statistical}.  We emphasize that assumptions like $\beta_n\asymp 1$ and $p > d +1$ are made mainly for succinct conditions and results. 
{\color{black}With more careful bookkeeping}, 
the consistency result remains valid, yet with a more complicated form of the convergence rate, when $p$ is slightly smaller than $d + 1$, and the limiting distribution results still hold when $\beta_n$ is slowly decaying. Finally, the nuisance functions are estimated with  the $o_P(n^{-1/4})$ rate, which are relaxed from the $\mathcal{O}_P(n^{-1/2})$ rate, demonstrating the robustness of our proposed estimator to slower rates of convergence of nuisance function estimators.  Similar conditions are  assumed in \citet{nie2021quasi,yadlowsky2018bounds}, among others.
\par
Our consistency result  \eqref{main:consist} is attained after choosing $K$  to balance the bias-variance tradeoff. If a higher-order moment condition \eqref{clt:con} holds, we further have the limiting distribution result \eqref{main:clt}, as long as we slightly increase $K$ resulting in an undersmoothing estimator. Similar conditions appear when studying the limiting distribution of the classic nonparametric sieve regression \citep{newey1997convergence}. 
Finally, we propose a closed-form  variance estimator $\hat{\sigma}^2$ of $\hat\tau(x_0,t_0)$ in Algorithm \ref{alg:sigma}, for any given $(x_0,t_0)\in\mathbb{X}\times\mathbb{T}$. The variance estimator requires an additional consistent estimator for $\mu(x,t)$ under $\mathcal{L}^{\infty}$ norm. Such condition is common  \citep[e.g.,][]{kennedy2017non}. Then we construct a $(1 - c)$-confidence interval:
\bee\label{form:CI}
\Big(\hat{\tau}(x_0,t_0) - z_{c/2}\frac{\hat{\sigma}}{\sqrt{n}},\,\hat{\tau}(x_0,t_0) + z_{c/2}\frac{\hat{\sigma}}{\sqrt{n}}\Big).
\ee
We use $z_{c}$ to represent the {\color{black}$(1 - c)$}-quantile of the standard normal distribution for any $c\in(0,1)$.
\begin{algorithm}[t]
\caption{Cont'd Algorithm \ref{ago:1}: Variance estimator of  $\hat{\tau}(x_0,t_0)$}\label{alg:sigma}
\SetAlgoLined
\textit{Step 1}.\ Obtain top-$(K-K/k)$ singular value decomposition $\hat{U}\hat{\Sigma}\hat{U}^\T$ of $\hat{G}_n = \hat{R}_n + \rho\hat{Q}_n$\;
\textit{Step 2}.\ Set $\hat{A}_n = \hat{U}\hat{\Sigma}^{-1}\hat{U}^\T$\;
\textit{Step 3}.\ For each $j\in[\mathcal{J}]$,  obtain $\hat{\mu}^{(-j)}(x,t)$ with all data except $\mathcal{S}_j$, via flexible machine learning methods. Then obtain 
$$
\hat{B}_n^{(-j)} = \frac{1}{|\mathcal{S}_j|}\sum_{Z_i\in\mathcal{S}_j}\{Y_i - \hat{\mu}^{(-j)}(X_i,T_i)\}^2\big\{\Psi(X_i,T_i) - \hat{\Gamma}^{(-j)}(X_i)\big\}\big\{\Psi(X_i,T_i) - \hat{\Gamma}^{(-j)}(X_i)\big\}^\T;
$$
\textit{Step 4}.\ Set $\hat{B}_n = \mathcal{J}^{-1}\sum_{j = 1}^{\mathcal{J}}\hat{B}_n^{(-j)}$\;
\ \ \textit{Step 5}.\ Output $\hat{\sigma}^2 = \{\Psi(x_0,t_0) - \Psi(x_0,0)\}^\T\hat{A}_n\hat{B}_n\hat{A}_n\{\Psi(x_0,t_0) - \Psi(x_0,0)\}$.
\end{algorithm}
\begin{remark}[Tuning parameter selection]\label{rk:para}\color{black}
Theorem \ref{thm:main} indicates that, to attain the best convergence rate of our proposed method, one needs to carefully select the tuning parameters $K$ and $\rho$. Especially, we need to select $K$ to balance the bias-variance tradeoff, and for $\rho$, according to both our theoretical results and empirical experiments, it will not significantly affect the estimation error of our proposed estimator as long as it is not too small.  
The R-learner with sieve approximation permits a closed-form solution (\ref{def:phi}), and therefore allows a fast generalized cross validation-based algorithm for selecting $K$ and $\rho$; see $\mathsection$\ref{sec:gcv}.     
For inference, as illustrated in Theorem \ref{thm:main}, one could slightly increase the optimal $K$ selected by generalized cross validation  and slightly decrease $\rho$ before constructing the confidence interval \eqref{form:CI}; The effectiveness of such strategy is further demonstrated by our numerical experiments in $\mathsection$\ref{sec:simu}.
\end{remark}
\section{Generalized cross validation-based tuning parameter selection}\label{sec:gcv}
Parameters $K$ and $\rho$ for our sieve-type $\ell_2$-regularized R-learner (Algorithm \ref{ago:1}) control the bias-variance tradeoff. An essential problem in practice is how to select the optimal $K$ and $\rho$ based on samples and minimize the finite-sample error. For sieve-type estimators, many data-driven methods have been considered to select the optimal number of basis functions, including AIC, BIC, cross validation, Lepski's method, Mallows criterion. We refer interested readers to \citet{ichimura2007implementing,hansen2014nonparametric}.
\par
In this section, we adapt the  generalized cross validation \citep{golub1979generalized} as the parameter selection method for our sieve-type $\ell_2$-regularized R-learner. The generalized cross validation is known to be efficient and valid in many model selection problems.  In comparison to ad-hoc cross validation methods that have been popularly studied for sieve-type estimators \citep{belloni2015some, belloni2019conditional}, the generalized cross validation is preferable under our framework as it permits a closed-form solution and avoids  introducing additional computational complexity.
\par
For given samples $\{\M Z_i\}_{i = 1}^n$, we first set the candidate pool for relative parameters $K$ and $\rho$
$$
\mathcal{L} = \{(K,\rho)\mid K = K_1,\dots,K_{n_1}, \rho = \rho_1,\dots,\rho_{n_2}\}.
$$ 
Following the convention and without loss of generality, we assume that for each $K = K_a$ with $a\in[n_1]$, the corresponding basis function $\Psi(x,t)$ is determined. We then select the optimal  $(K,\rho)$ from $\mathcal{L}$ by the generalized cross validation. Recalling that in Algorithm \ref{ago:1}, we split the full data into $\mathcal{J}$ folds $\mathcal{S}_1,\dots,\mathcal{S}_{\mathcal{J}}$ with sample sizes $n_1,\dots,n_{\mathcal{J}}$, respectively. Without loss of generality, we assume $\mathcal{S}_{j} = \{\M Z_{1 + \sum_{l = 1}^{j-1} n_l},\dots,\M Z_{\sum_{l = 1}^{j} n_l}\}$ for each $j = 1,\dots,\mathcal{J}$.
\par Given specific $(K,\rho) = (K_a,\rho_b)\in\mathcal{L}$, we define the following  quantities:
\bee\nonumber
\M Y &= \big(Y_1,\dots,Y_{n}\big)^\T
\\
\hat{m} &= \big\{\hat{m}^{(-1)}(\M X_1),\dots,\hat{m}^{(-\mathcal{J})}(\M X_{n})\big\}^\T
\\
\M\Lambda &= \big\{\M \Psi(\M X_1,T_1) - \hat{\M \Gamma}^{(-1)}(\M X_1),\dots,\M \Psi(\M X_{n},T_{n}) - \hat{\M \Gamma}^{(-\mathcal{J})}(\M X_n)\big\}
\\
\hat{\bar{\tau}} &= \big\{(\rho + 1)\hat{\phi}^\T\M \Psi(\M X_1,T_1) -(\rho + 1)\hat{\phi}^\T \hat{\M \Gamma}^{(-1)}(\M X_1),\cdots,(\rho + 1)\hat{\phi}^\T\M \Psi(\M X_n,T_n) -(\rho + 1)\hat{\phi}^\T \hat{\M \Gamma}^{(-\mathcal{J})}(\M X_n)\big\}^\T,
\ee
for all $j\in[\mathcal{J}]$. 

 Based on Algorithm \ref{ago:1}, we can write
\bee\label{GCV:main}
\hat{\bar{\tau}} = \underbrace{(\rho + 1){n^{-1}\M\Lambda^\T  \hat{G}_n^{-1}\Lambda}}_{S_{(K_a,\rho_b)}}(\M Y - \hat{ m} ),
\ee   
where $S_{(K_a,\rho_b)}$ is the so-called smoothing matrix of the generalized cross-validation  for a particular parameter group $(K_a,\rho_b)$; see, e.g., \citet[Section 5.3]{wasserman2006all}. Recall~\eqref{con:decom1}, we consider $Y_i  - m(X_i)$ as the residuals. Following \citet{wasserman2006all}, the corresponding empirical error criterion for the generalized cross-validation is
\bee\label{err:gcv}
\text{Error}(K_a,\rho_b) = \frac{1}{n}\sum_{i = 1}^n\Bigg[\frac{Y_i - \hat{m}^{(-j_i)}(\M X_i) - (\rho + 1) \{\hat{\phi}^\T\Psi(\M X_i,T_i) - \hat{\phi}^\T\hat{\M\Gamma}^{(-j_i)}(\M X_i)\}}{1 - \text{tr}\{S_{(K_a,\rho_b)}\}/n}\Bigg]^2,
\ee
where $j_i$ is the splitting fold of $i$th sample. Finally, we choose the optimal $(K,\rho)$ among $\mathcal{L}$ that minimizes \eqref{err:gcv}.
\begin{algorithm}[t]
\caption{Generalized S-learner and X-learner for continuous treatments}\label{alg:sx}
\SetAlgoLined
\textit{Step 1}.\  Obtain $\hat{\mu}(x,t) = \hat{E}(Y\mid X = x,T = t)$ over all $\{Z_i\}_{i = 1}^n$ via a flexible machine learning algorithm\;
\textit{Step 2}.\ S- and X-learners proceed as follows.
\begin{itemize}
\item(S-learner). Obtain $\hat{\tau}_{SL}(x,t) = \hat{\mu}(Y\mid X = x,T = t) - \hat{\mu}(Y\mid X = x,T = 0)$;
\item (X-learner). Construct pseudo-individual treatment effect  $D_i = Y_i - \hat{\mu}(X_i,0)$, based on which fit $$\hat{\tau}_{XL}(x,t) = \hat{E}(D\mid X = x,T = t),$$ via a flexible machine learning algorithm.
\end{itemize}
\end{algorithm} 
\section{Generalized S- and X-learners}\label{sec:SX}
\citet{kunzel2019metalearners} proposed two meta-learners for estimating the conditional average treatment effect (CATE) under the binary-treatment setting, namely, the S-learner and the X-learner. These two learners adopt different strategies for CATE estimation by leveraging flexible machine learning algorithms. Building on similar ideas, our Algorithm~\ref{alg:sx} provides  natural extensions of the original S- and X-learners to the continuous-treatment setting.

\section{The details of the simple simulation in $\mathsection$\ref{sec:simple} and additional results with increasing sample sizes}\label{sec:simeg}
\begin{figure}[t]
\centering

\begin{subfigure}[b]{0.25\textwidth}\centering\includegraphics[width=1\linewidth]{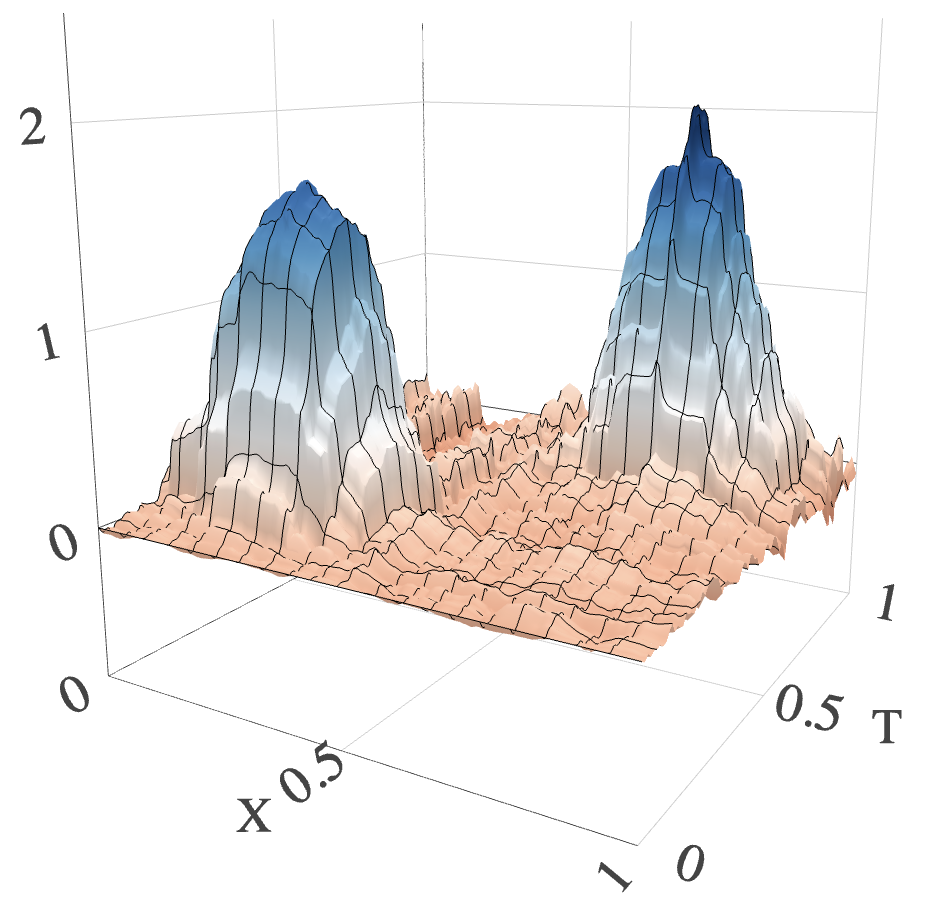}
         \caption{}
     \end{subfigure}
     \begin{subfigure}[b]{0.25\textwidth}\centering\includegraphics[width=1\linewidth]{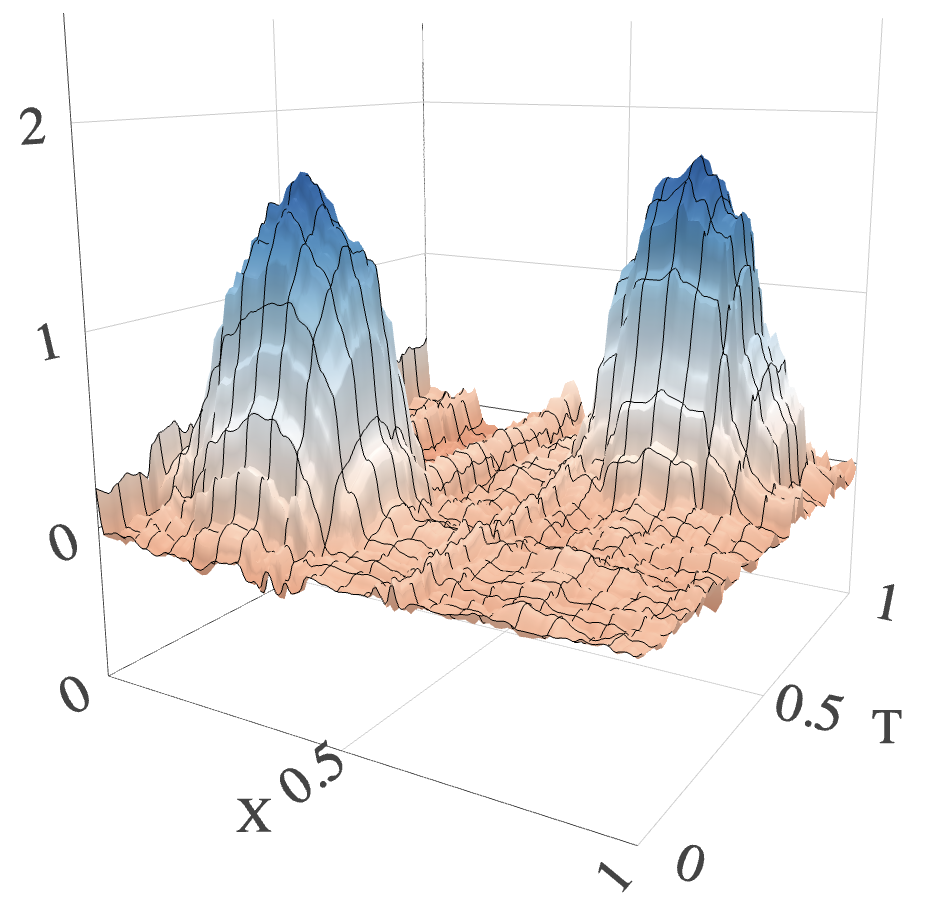}
         \caption{}
     \end{subfigure}
\caption{{When the treatment is continuous, the generalized S-learner (Panel a) and the generalized X-learner (Panel b)  roughly captures $\tau(x,t)$ (Fig. \ref{fig:example}b)} for the simple simulation   in $\mathsection$\ref{sec:simple}. However, their performances are not as good as the proposed R-learner (Fig. \ref{fig:example}e).}
\label{fig:exampless}
\end{figure}

\begin{figure}[tp]
\centering

\begin{subfigure}[b]{0.24\textwidth}\centering\includegraphics[width=1\linewidth]{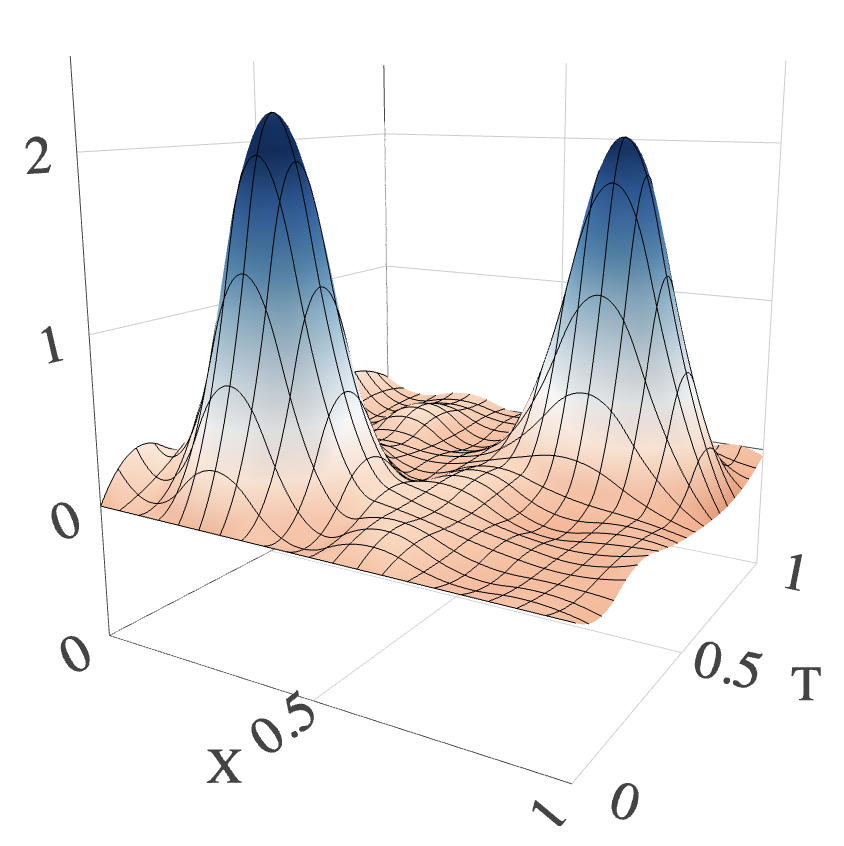}
         \caption{$n = 10000$}
     \end{subfigure}
\begin{subfigure}[b]{0.24\textwidth}\centering\includegraphics[width=1\linewidth]{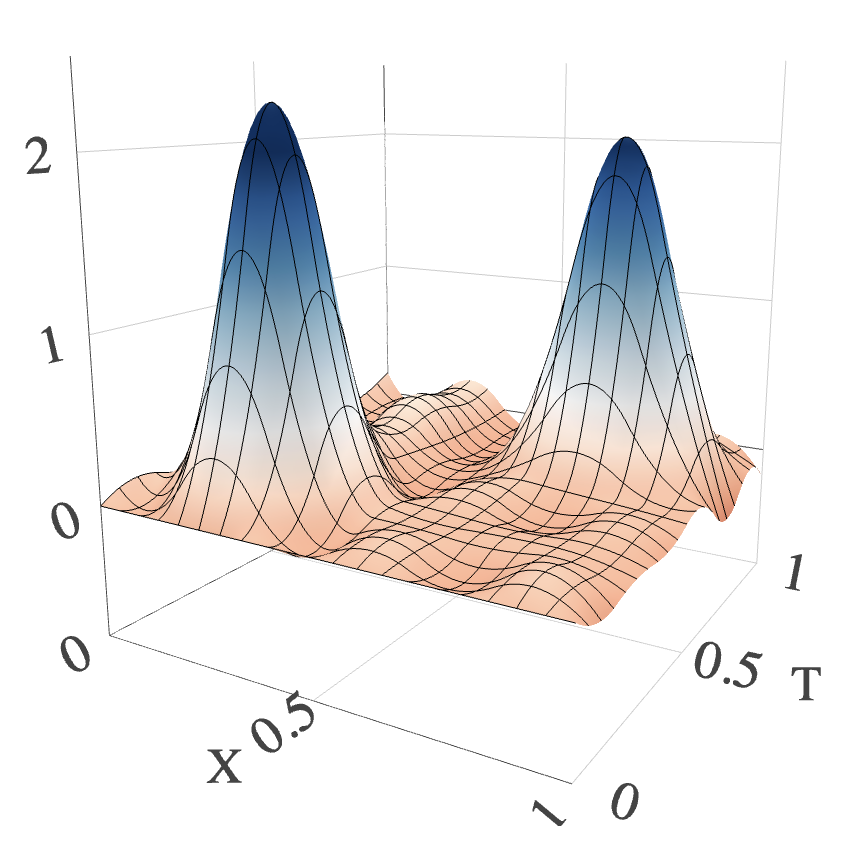}
         \caption{$n = 20000$}
     \end{subfigure}
\begin{subfigure}[b]{0.24\textwidth}\centering\includegraphics[width=1\linewidth]{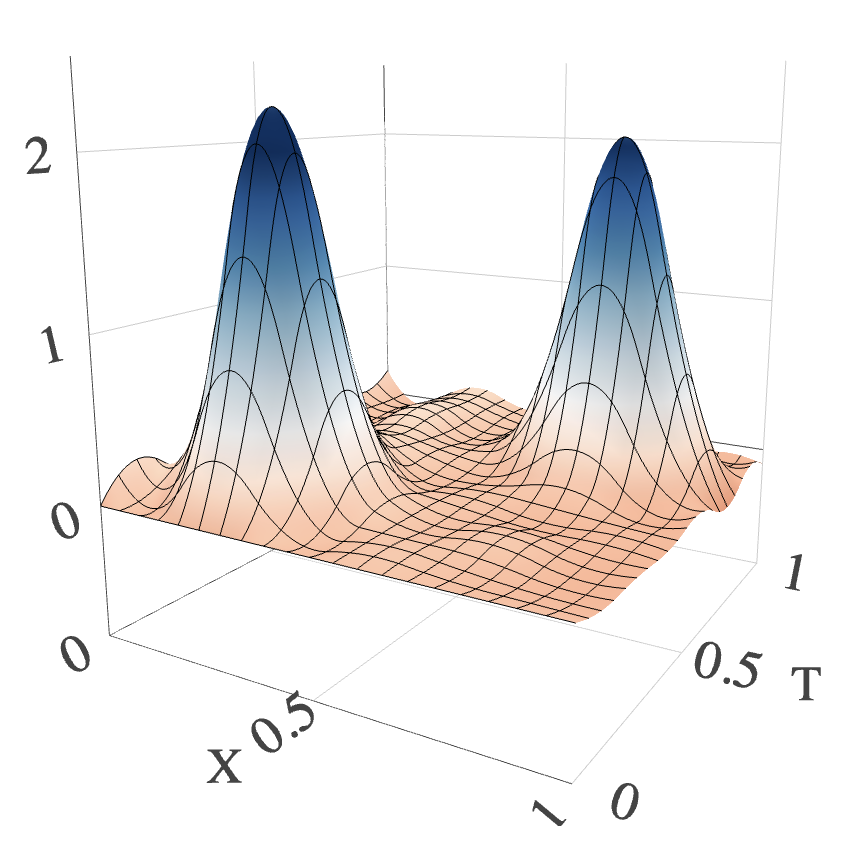}
         \caption{$n = 30000$}
     \end{subfigure}
\begin{subfigure}[b]{0.24\textwidth}\centering\includegraphics[width=1\linewidth]{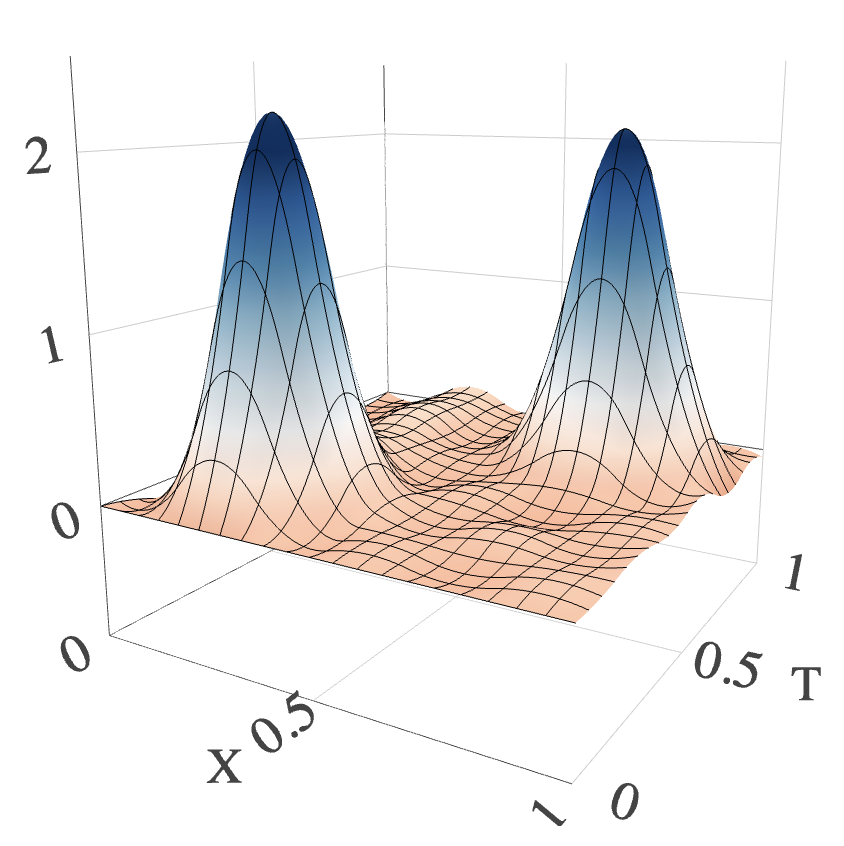}
         \caption{$n = 40000$}
     \end{subfigure}

    \begin{subfigure}[b]{0.24\textwidth}\centering\includegraphics[width=1\linewidth]{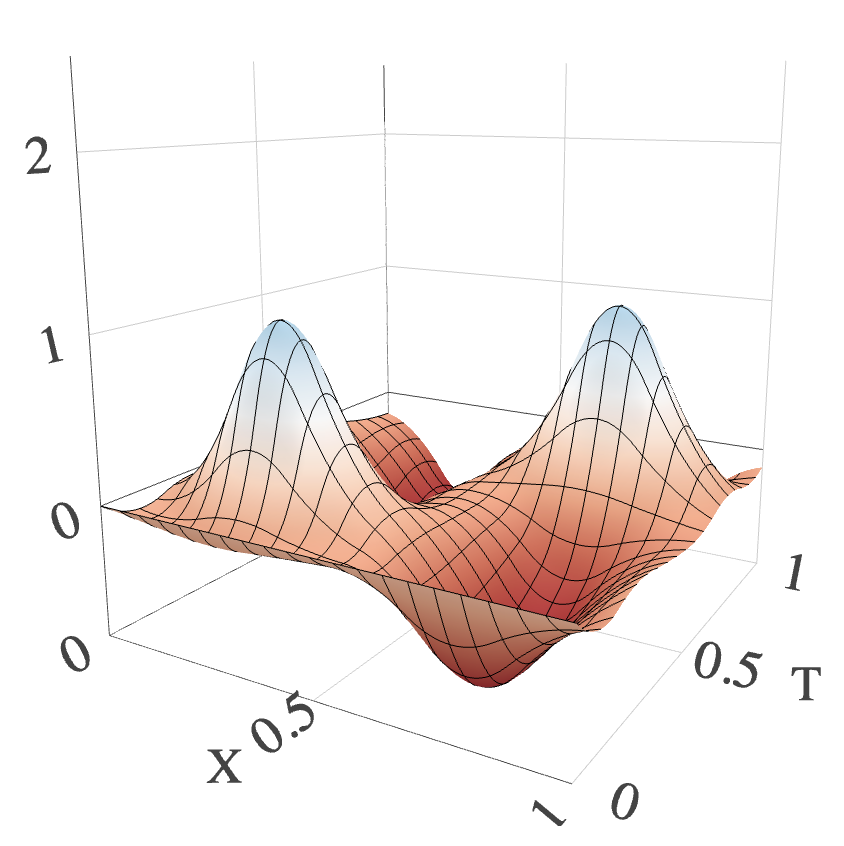}
     \end{subfigure}
\begin{subfigure}[b]{0.24\textwidth}\centering\includegraphics[width=1\linewidth]{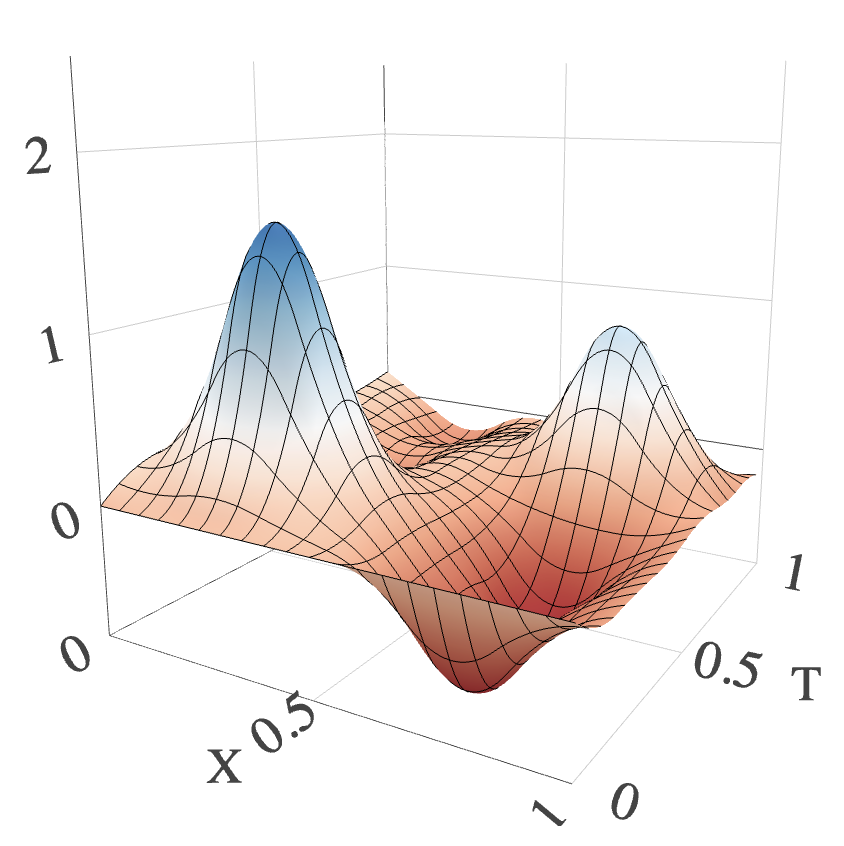}
     \end{subfigure}
\begin{subfigure}[b]{0.24\textwidth}\centering\includegraphics[width=1\linewidth]{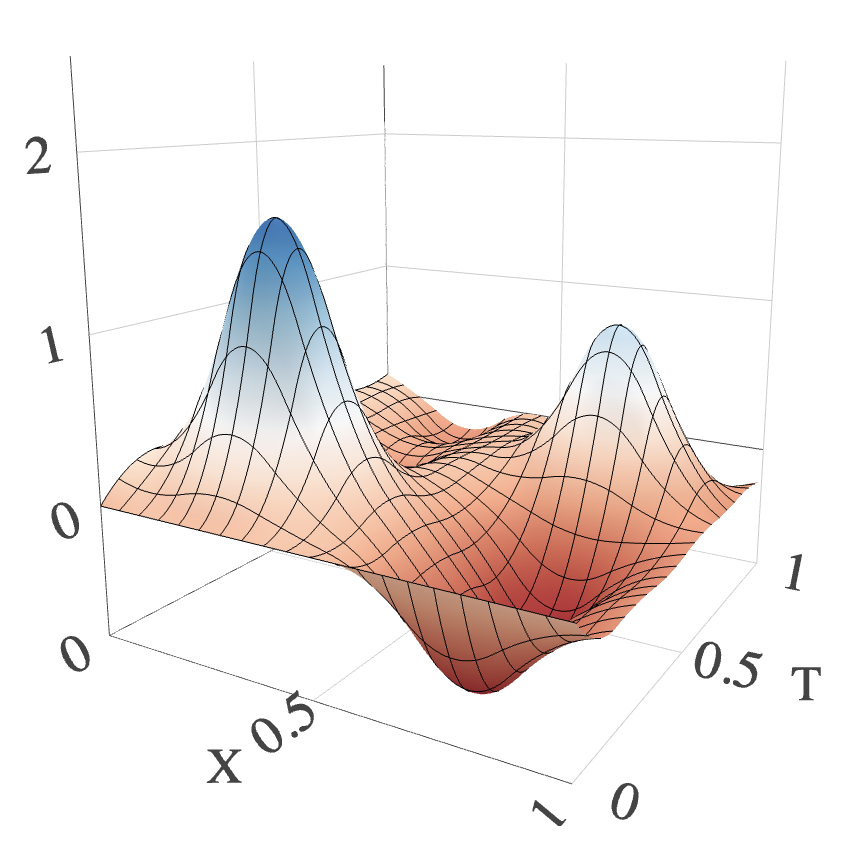}
     \end{subfigure}
\begin{subfigure}[b]{0.24\textwidth}\centering\includegraphics[width=1\linewidth]{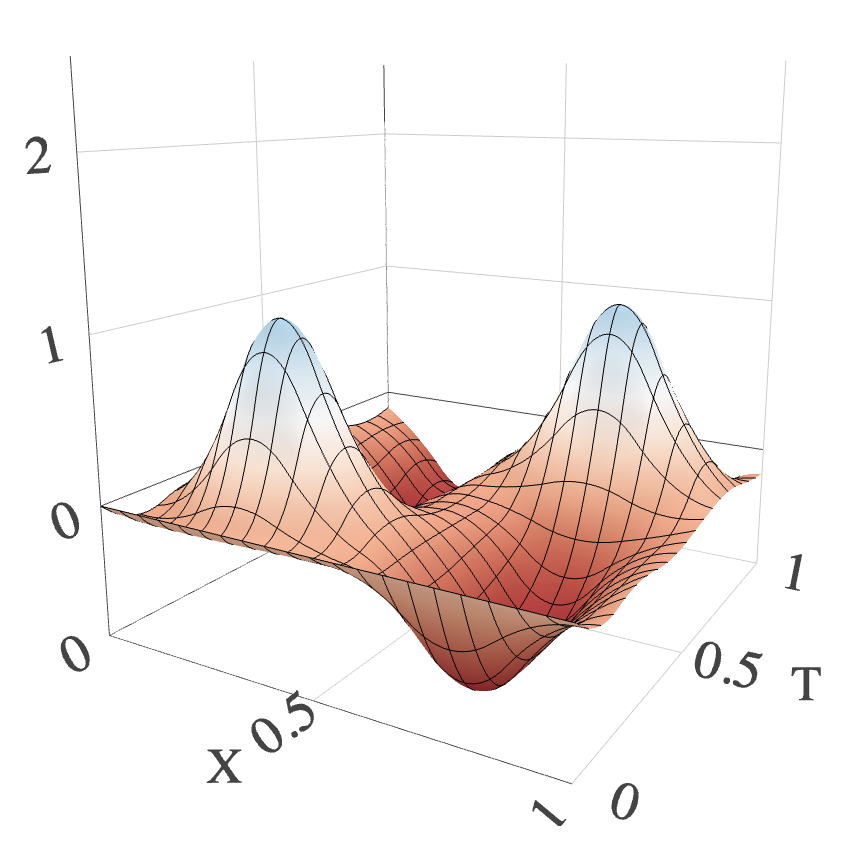}
     \end{subfigure}
\caption{The $\ell_2$-regularized R-learner (the top row) and the one-step generalized R-learner under zero condition (the bottom row) with increasing sample sizes under the simple simulation setting.}
\label{fig:exampless}
\end{figure}
We consider two data examples with one binary-treatment setting and one continuous-treatment setting, respectively, for our simple simulation in $\mathsection$\ref{sec:main}.
\begin{itemize} 
\item(A binary-treatment setting) Generate a single covariate $X_i\sim\text{Uniform}(0,1)$, a binary treatment $T_i\sim \text{Bernoulli}(0.5)$, and 
$Y_i = \sin(2X_i) + X_i^2 + \tau(X_i,T_i) + \varepsilon_i$ with  $\tau(X_i,0) = 0$, $\tau(X_i,1) = \tau(X_i)= \sin(2\pi r_i + \pi/2) + 1$,  $r_i = |X_i - 0.5|$, and $\varepsilon_{i}\mid X_i,T_i\sim \mathcal{N}(0,0.3^2)$. 
\item(A continuous-treatment setting) Generate a single covariate $X_i\sim\text{Uniform}(0,1)$, a continuous treatment $T_i\sim \text{Uniform}(0,1)$, and  $Y_i = \sin(2X_i) + X_i^2 + \tau(X_i,T_i) + \varepsilon_i$ with 
$$
\tau(X_i,T_i) = \begin{cases}
\sin\left(4\pi \sqrt{(X - 0.25) ^ 2+(T - 0.25)^2} + \pi/2\right) + 1& \sqrt{(X - 0.25) ^ 2+(T - 0.25)^2} < 0.25,
\\
\sin\left(4\pi \sqrt{(X - 0.75) ^ 2+(T - 0.75)^2} + \pi/2\right) + 1& \sqrt{(X - 0.75) ^ 2+(T - 0.75)^2} < 0.25,
\\
0&\text{otherwise},
\end{cases}
$$ 
and $\varepsilon_{i}\mid X_i,T_i\sim \mathcal{N}(0,0.3^2)$.  
\end{itemize}
Both settings can be seen as completely randomized experiments and the sample size is $n = 5000$ for the results in Fig.~\ref{fig:example}.

 For the binary-treatment setting, we approximate $\tau(x,t = 1)=\tau(x)$, with an $8$-dimensional B-spline basis $\psi(x)$ such that $\tau(x)\approx \phi^\T \psi(x)$. The $\psi(\cdot)$ has equally-spaced knots over $[0,1]$; see $\mathsection$\ref{sec:bspline} for more  details of B-splines. We construct the empirical analogy of the binary-treatment R-loss $L_{b}(h)$ in \eqref{rloss:binary}, with $h(x) = \phi^\T \psi(x)$ and the true nuisance functions. We then estimate $\hat{\tau}(x)= \hat{\phi}_b^\T \psi(x)$ as the minimum of this empirical R-loss. 
\par
For the continuous-treatment setting, firstly, we use a $64$-dimensional tensor-productive B-spline function $\Psi(x,t)$ to approximate $\tau(x,t)$, where $\Psi(x,t) = \psi(t)\otimes \psi(x)$ and $\psi(\cdot)$ is defined in the same way as before. Specifically, the empirical analogy $\hat{L}_c(h)$ of our generalized R-loss $L_c(h)$ in  \eqref{con:decom2} is constructed through the sample average and the true nuisance functions.  The directly generalized R-learner  $ \hat{\phi}_c^\T\Psi(x,t)$ is the minimum of this empirical generalized R-loss $\hat{L}_c(h)$ among all $h(x,t) = \phi^\T\Psi(x,t)$ and is presented in Fig.~\ref{fig:example}c. 
\par
Secondly, to build the one-step generalized R-learner with the functional zero constraint, we use the same empirical R-loss $\hat{L}_c(h)$ as the directly generalized R-learner  with true nuisance functions. In particular, we impose the zero constraint in \eqref{opt:1:rec} on the class of the B-spline estimators that we consider for our directly generalized R-learner, and focus on the  estimator class
$$
\mathcal{H}_0 = \{h(x,t)  = \phi^\T\Psi(x,t)\mid \phi^\T\Psi(x,0) = 0 \text{ for all }x\in\mathbb{X}\text{ and }\phi\in\RR^{64} \}.
$$
By the basic property of the B-spline function as we will show in Proposition~\ref{prop:zerospline}, we can further write  $\mathcal{H}_0$ as the following function class. First, define the $7$-dimensional B-spline $\tilde{\psi}(t)$  by dropping the intercept spline in  ${\psi}(t)$, i.e., the $\psi^{(1)}(t)$ defined in $\mathsection$\ref{sec:bspline}, and keep all other spline functions unchanged. Then, by Proposition~\ref{prop:zerospline}, we can write $\mathcal{H}_0$  as follows:
$$
\mathcal{H}_0 = \{h(x,t) = \tilde{\phi}^\T\tilde{\psi}(t)\otimes \psi(x)\mid \tilde{\phi}\in\RR^{56}\},
$$ 
 where we can check that $\tilde{\psi}(0)\otimes \psi(x)$ is always  a zero vector for any $x\in\mathbb{X}$. 
We then minimize $\hat{L}_c(h)$ over $\mathcal{H}_0$ and obtain our one-step generalized R-learner with the functional zero constraint, which is presented in Fig.~\ref{fig:example}d.
\par
Thirdly, we build our proposed estimator in Algorithm \ref{ago:1:sieve}. 
For simplicity, we do not impose sample splitting, i.e., set $\mathcal{J} = 1$. We use the same B-spline basis functions $\Psi(x,t)$ as the directly generalized R-learner.  Nuisance functions are trained via strategy (i) based on the \textsc{SuperLearner}, which is considered in the simulation study in $\mathsection$\ref{sec:simu}. The $\rho$ has been selected as $0.05$. Our proposed estimator is presented in  Fig.~\ref{fig:example}e. 
\par
Finally, we assess the generalized S- and X-learners defined in Algorithm \ref{alg:sx}. All regression procedures in Algorithm \ref{alg:sx} are carried out using \textsc{SuperLearner} \citep{van2007super}, which combines random forest, XGboost, Bayesian Additive Regression Trees, adaptive polynomial splines, adaptive regression splines, single-layer neural networks, linear regression, and recursive partitioning-based regression trees. We show the results of $\hat{\tau}_{SL}(x,t)$ and $\hat{\tau}_{XL}(x,t)$ in Fig.~\ref{fig:exampless}. Comparing Fig.~\ref{fig:example} and Fig.~\ref{fig:exampless},  the proposed R-learner $\hat{\tau}(x,t)$ is the best among all comparative approaches.

\section{Numerical experiment}\label{sec:simu}
\begin{figure}[tp]
\centering
\includegraphics[width=0.9\linewidth]{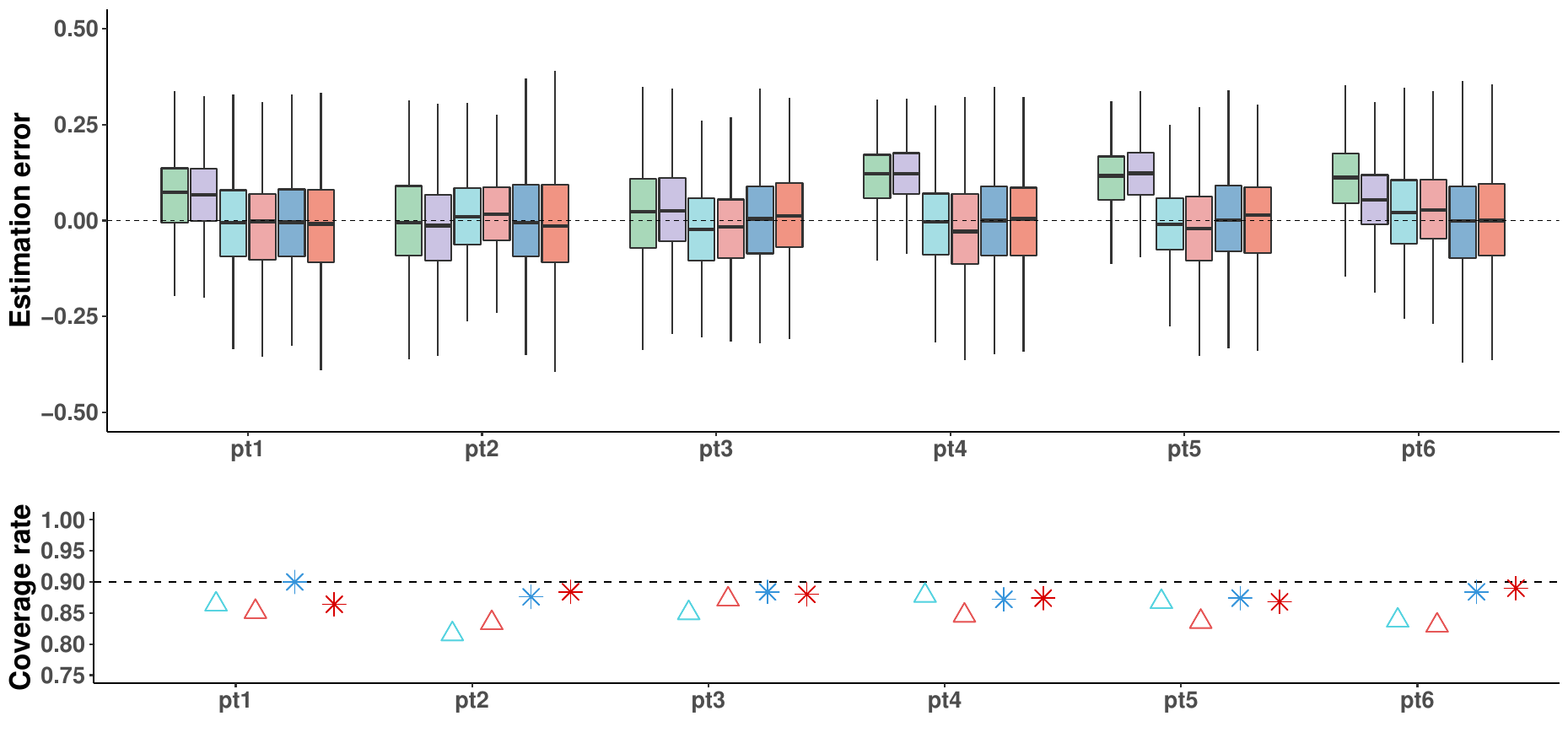}
\caption{Simulation results under the setting when $n = 2000$ and $T$ is generated from the uniform distribution  based on $500$ Monte Carlo simulated datasets. Different colors represent simulation results of different methods. The top penal shows the box-plots of pointwise errors for all comparative estimators, and the bottom panel shows the empirical coverage rates for the proposed estimators excluding S- and X-learners for which confidence intervals are not available. Different methods are labeled in different colors and/or shapes: S-learner (Green), X-learner (purple), proposed R-learner with the nuisance function estimated by Strategy (i) (light blue/triangle) or (ii) (light red/triangle) under Scenario (A) and with the nuisance function estimated by Strategy (i) (dark blue/star) or (ii) (dark red/star) under Scenario (B).}
\label{fig:simu}
\end{figure}

We run simulations to evaluate the finite-sample performance of our proposed methods. With sample size $n$, we generate random samples $\{\M Z_i\}_{i = 1}^n$, where $\M Z$ takes the form of 
$
\M Z  = \{\M X = (X^{(1)},X^{(2)}), T,Y\}.
$
The generating process of $\M Z$ is as follows:
\begin{itemize}
\item Generate  $X^{(1)}\sim\text{Bernoulli}(0.5)$ and $X^{(2)}\sim\text{Uniform}(0,1)$.
\item Generate $T$ from one of the following  generalized propensity score distributions:
\begin{itemize}
\item(Complete randomized experiment).  $\text{Uniform}(0,1)$;
\item (Beta distribution).  $\text{Beta}(\lambda_{\M X}, 1 - \lambda_{\M X}),
$ where we set $\text{logit}(\lambda_{\M X}) = (1,X^{(1)},X^{(2)})\eta$ with $\eta = (0.8 ,0.2, - 0.8)^\T$. Such distribution has been considered by \citet[Section 4]{kennedy2017non}. 
\end{itemize} 
\item Generate $
Y = \mu(\M X,0) + \tau(X,T) + \mathcal{N}\big(0,0.3^2\big),
$
where we set $ \mu(\M X,0) = X^{(1)}X^{(2)} + \sin(2X^{(2)}) + (X^{(2)})^2$, and
\bee\nonumber
\tau(X,T) = \begin{cases}
I(r_{0.5} \leq 0.5)v(r_{0.5})& X^{(1)} = 1
\\
I(r_{0.5} \leq 0.5)v(r_{0.5}) + I(r_{0.8} \leq 0.5)v(r_{0.8})/2& X^{(1)} = 0
\end{cases},
\ee 
where $r_a$ is the Euclidean distance between $(X^{(2)},T)$ and $(a,a)$, and $v(r) = \sin(2\pi r + \pi/2) + 1$ and $I(\cdot)$ is the indicator function.
\end{itemize}
\par
To fit data, we first stratify the simulated data  by $X^{(1)}$.	 We then fit each stratum by the $\ell_2$-regularized R-learner with sieve approximation stated in Algorithm \ref{ago:1}, where $\Psi(x^{(2)},t) =\psi(t)\otimes\psi(x^{(2)})$ and $\mathcal{J} = 6$; here $x^{(2)}$ corresponds $X^{(2)}$. 
Both $\psi(t)$ and $\psi(x^{(2)})$ are B-spline functions with a quadratic degree  and equally spaced knots over $[0,1]$. We select the optimal basis numbers $k_{T,opt}$ and $k_{X^{(2)},opt}$ for $\psi(t)$ and $\psi(x^{(2)})$ and regularization parameter $\rho_{opt}$ by the generalized cross validation-based algorithm in $\mathsection$\ref{sec:gcv}, among a candidate pool $\mathcal{L} = \{(k_{X^{(2)}},k_T,\rho)\mid 4\leq k_{X^{(2)}},k_T \leq 6, \rho = 0.005,0.01,\dots,0.5\}$. We choose basis numbers for our proposed R-learner under two scenarios:
\begin{itemize}
\item[(A)] The dimensions of $\psi(t)$ and  $\psi(x^{(2)})$ are $k_{T,opt}$ and $k_{X^{(2)},opt}$. \item[(B)] The dimensions of $\psi(t)$ and $\psi(x^{(2)})$ are  $k_{T,opt} + 1$ and $k_{X^{(2)},opt} + 1$.  
 \end{itemize}
 The nuisance functions are estimated via one of the following two strategies.
	\begin{itemize}
	\item[(i).] Estimate $\hat{m}^{(-j)}(x^{(2)})$ via \textsc{SuperLearner}, combining adaptive polynomial splines, adaptive regression splines, single-layer neural networks, linear regression, and recursive partitioning-based regression trees. Estimate $\hat{\Gamma}^{(-j)}(x^{(2)})$ through method (i) in $\mathsection$\ref{rk:gamma:est}. When $T$ is generated from the complete randomized experiment, $\varpi(t\mid x^{(2)})$ is assumed to be known as $1$ over $[0,1]^2$. When  $T$ is generated from the Beta distribution, we estimate $\hat{\varpi}^{(-j)}(t\mid x^{(2)})$ by estimating the parameters $\eta$ of  mean function $\lambda_X$ through logistic regression following \citet{kennedy2017non}.
	\item[(ii).]Estimate $\hat{m}^{(-j)}(x^{(2)})$  in the same way as strategy (i). Estimate $\hat{\Gamma}^{(-j)}(x^{(2)})$ through method (ii) in $\mathsection$\ref{rk:gamma:est}, where the coordinate-wise regressions are conducted by the same \textsc{SuperLearner} algorithm as for $\hat{m}^{(-j)}(x^{(2)})$.
	\end{itemize}
\begin{figure}[htbp]
\centering
 \begin{subfigure}[b]{0.7\textwidth}\centering\includegraphics[width=1\linewidth]{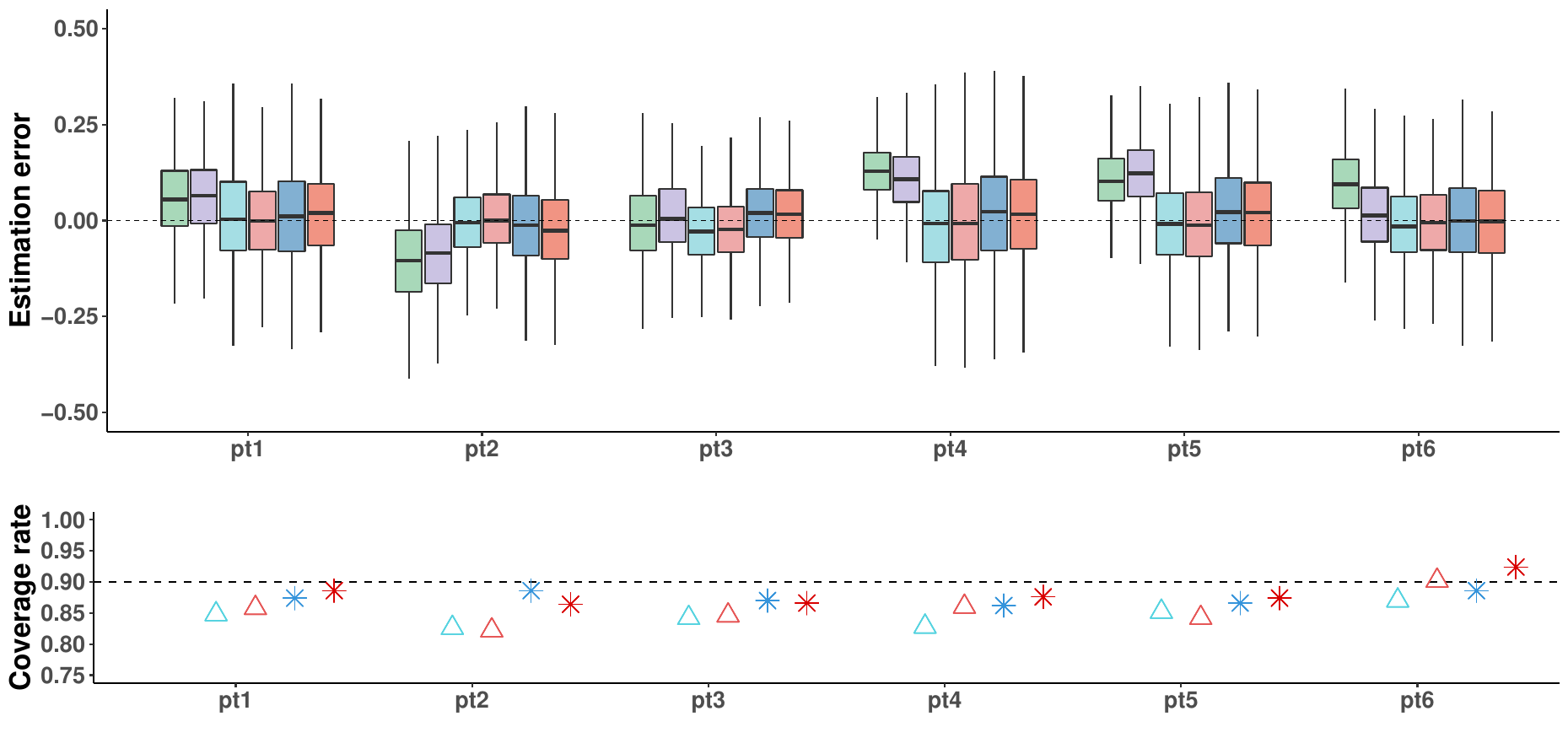}
         \caption{$n = 2000$ and $T$ is generated from the Beta distribution}
     \end{subfigure}
\par
 \begin{subfigure}[b]{0.7\textwidth}\centering\includegraphics[width=1\linewidth]{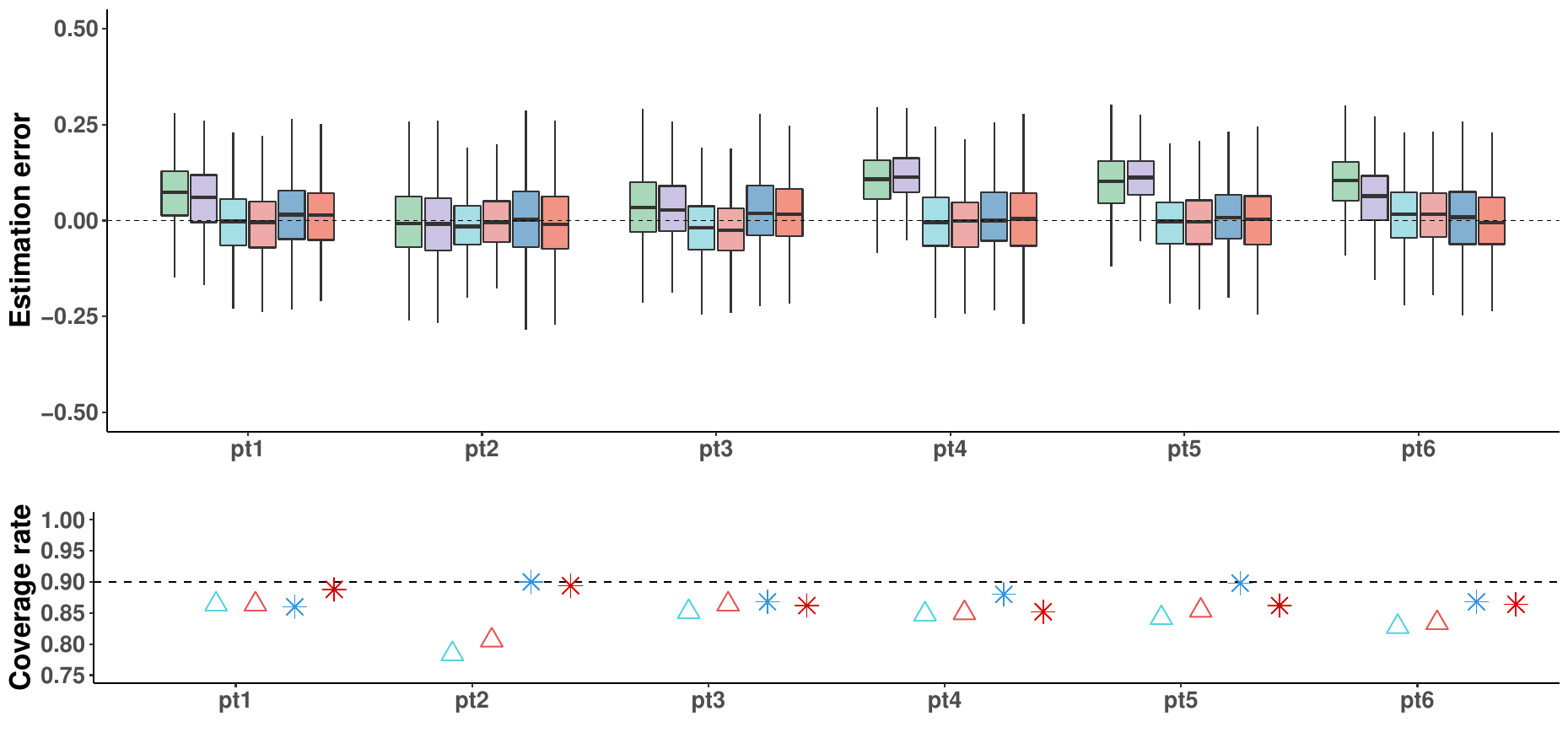}
         \caption{$n = 4000$ and $T$ is generated from the uniform distribution}
     \end{subfigure}
\par
 \begin{subfigure}[b]{0.7\textwidth}\centering\includegraphics[width=1\linewidth]{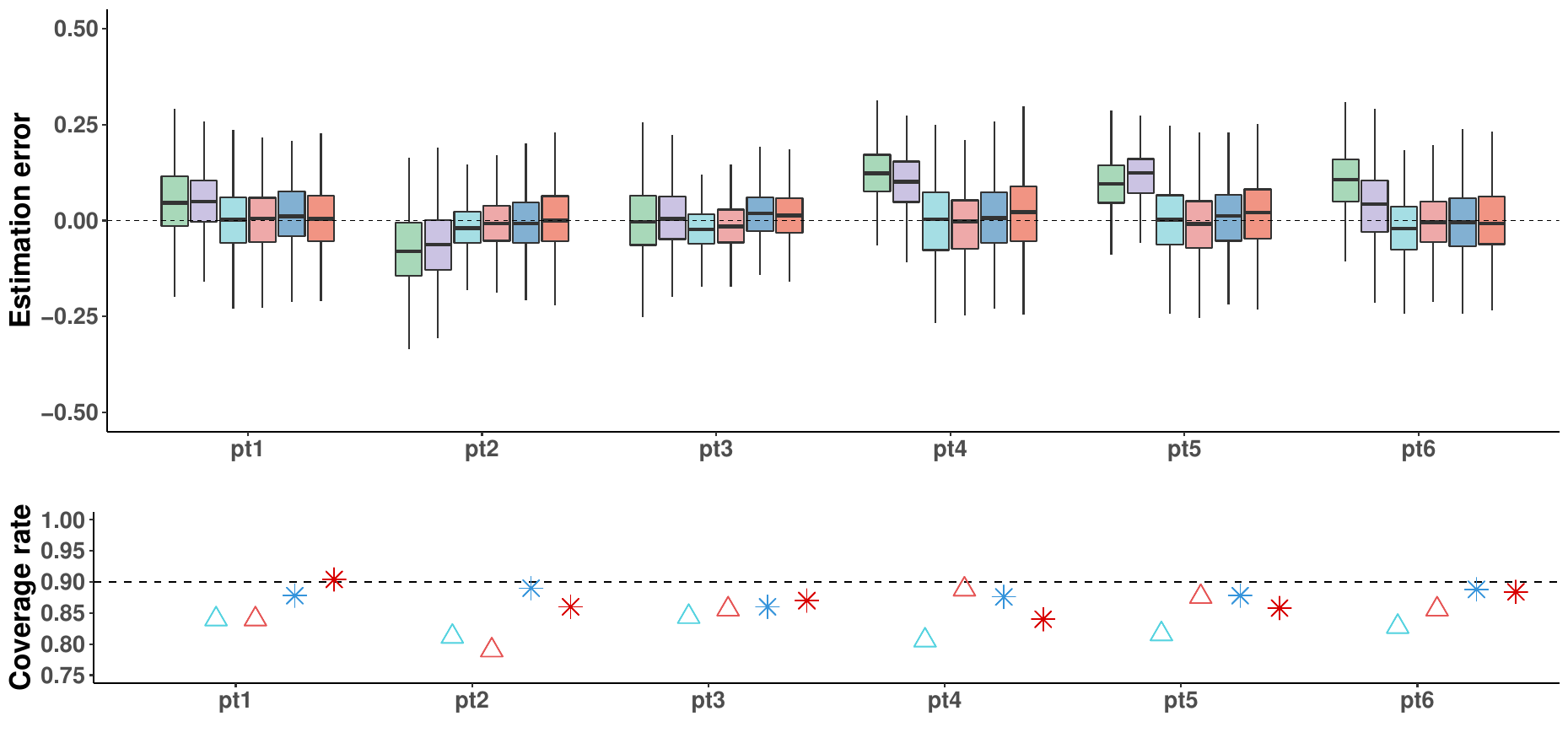}
         \caption{$n = 4000$ and $T$ is generated from the Beta distribution}
     \end{subfigure}
\caption{From top to bottom, different subfigures report simulation results under different settings. In each subfigure, different colors represent simulation results of different methods. The top penal shows the box-plots of pointwise errors for all comparative estimators, and the bottom panel shows the empirical coverage rates for the proposed estimators excluding S- and X-learners for which confidence intervals are not available. Different methods are labeled in different colors and/or shapes: S-learner (Green), X-learner (purple), proposed R-learner with the nuisance function estimated by Strategy (i) (light blue/triangle) or (ii) (light red/triangle) under Scenario (A) and with the nuisance function estimated by Strategy (i) (dark blue/star) or (ii) (dark red/star) under Scenario (B).}
\label{fig:simu:add}
\end{figure}
We test the performance of our proposed R-learners at six different points of $(x^{(1)},x^{(2)},t)$, namely, $(0,0.25,0.25),$ $(0,0.5,0.5),(0,0.75,0.75),(1,0.25,0.25),(1,0.25,0.75)$, $(1,0.5,0.25)$;
we term these points pt1--pt6 for brevity. In each round of the  simulation, we report the point-wise difference between the estimated $\tau(x,t)$ and the true $\tau(x,t)$ as the estimation error. We further construct $90\%$ confidence interval \eqref{form:CI} for each tested point, and check if it covers the truth; here $\hat{\sigma}$ is estimated following Algorithm \ref{alg:sigma} and $\hat{\mu}^{(-j)}(x^{(2)},t)$ is trained by same  \textsc{SuperLearner} algorithm as for $\hat{m}^{(-j)}(x^{(2)})$. For comparison, we additionally report estimation errors of  generalized S-learner and X-learners \citep{kunzel2019metalearners}, which are introduced in details in $\mathsection$\ref{sec:SX}. In our simulations, all regressions in S- and X-learners are performed by  \textsc{SuperLearner}  similar to $\hat{m}^{(-j)}(x^{(2)})$, yet additionally combining with random forest, XGboost, and Bayesian Additive Regression Trees.
\par
We run $500$ Monte Carlo simulations under each setting. Simulation results when $n = 2000$ and $T$ is generated under the complete randomized trial, are presented in Fig. \ref{fig:simu}. The top panel of Fig. \ref{fig:simu} shows the superior performances of our proposed method in terms of small estimation errors over all simulation settings and tested points. Both Strategies (i) and (ii) for nuisance function training produce valid results. {\color{black}The bottom panel of Fig. \ref{fig:simu} shows that when conducting proposed methods under Scenario (B)}, the empirical coverage rates are close to $90\%$ for all points, meanwhile the estimation variance is slightly larger. Such observations verify the theoretical claims in Remark \ref{rk:para}, and show the effectiveness of our proposed confidence interval after undersmoothing. Additional numerical results with $T$ generated under the Beta distribution and/or $n = 4000$, are contained  
  in Fig.~\ref{fig:simu:add}. Across all different settings, our proposed estimators present superior performance with small estimation errors, and the empirical coverage rates are closed to the nominal level $90\%$ after undersmoothing. With larger $n$, our proposed estimators have smaller estimation variances, which is consistent with our theoretical results for consistency.

\section{Proofs of propositions and theorems}
Before delving into the main proofs, we first simplify the setting for theoretical analysis, define some shorthand notation, and state some regularization conditions in $\mathsection$\ref{sec:prea}. We briefly review some technical details of  B-spline functions in $\mathsection$\ref{sec:bspline}. We list useful lemmas in $\mathsection$\ref{sec:lemma} and give the proofs for our main propositions and theorems in the remaining of the sections.
\subsection{Preliminaries}\label{sec:prea}\noindent\textbf{Simplification of the sample-splitting procedure:} Recall $\hat{\phi}$ and $\hat{\sigma}$ are constructed based on the double machine learning framework, with $K$-fold sample splitting. During the proof, we consider a simple $2$-fold training scenario such that,  $\hat{\phi}$ and $\hat{\sigma}$ are fit based on $n$ i.i.d. samples, while the nuisance functions $\hat{\Gamma}$, $\hat{m}$ and $\hat{\mu}$ are all trained via another independent $n$ i.i.d. samples. In the sense of the asymptotic theoretical analysis considered in this paper, such simplification does not lose generality, and it has been popularly employed by the splitting-based estimators \citep{nie2021quasi,kennedy2020optimal,kennedy2020sharp}.
\par
\
\par
\noindent\textbf{Notation: }For simplicity, we occasionally omit some function arguments during the proof. For any function of random variable $ h(W)$, we may use $h$ to represent it. Similarly, we denote $ h(W_i)$ by $ h_i$. For random variable $W$, we denote $P\{\hat{f}(W)\} = \int_{w\in \mathbb{W}}\hat{f}(w)d\mathcal{P}(w)$ as the expectation for $W$ treating $\hat{f}$ as a fixed function or matrix. Thus if $\hat{f}$ is sample-based, $P\{\hat{f}(W)\}$ is a random variable or random matrix. The constants in the form of $C,C_1,C_2,\dots$ can change the meanings in different proofs. We also write ``wpa1'' as a short notation of ``with probability approaching 1''. For two matrices $A$ and $B$, we write $A\succeq B$, if $A - B$ is positive definitive. We use $\|A\|_2$ to represent the spectral norm of $A$. 
\par
We now introduce some basic matrix notation  involved in our main proof,
\bee\nonumber
{Q}_n &= E\big\{\bpsi(X,T)\bpsi^\T(X,T)\big\},
\\
{R}_n &= E\big[\{\Psi(X,T) - \M\Gamma(X)\}\{\Psi(X,T) - \M\Gamma(X)\}^\T\big],
\\
G_n &= {R}_n + \rho{Q}_n.
\ee
It is easy to see both $R_n$ and $G_n$ are positive semi-definitive matrices. Let the singular value decompositions   of ${R}_n$ and $G_n$ be
\bee\label{Rdecom}
R_n &= 
\begin{pmatrix}
U & U_\perp
\end{pmatrix}
\begin{pmatrix}
\Sigma & 
\\
&  0
\end{pmatrix}
\begin{pmatrix}
U^\T
\\
U_{\perp}^\T
\end{pmatrix},
\\
G_n &= 
\begin{pmatrix}
\tilde{U} & \tilde{U}_\perp
\end{pmatrix}
\begin{pmatrix}
\tilde{\Sigma} & 
\\
& \tilde{\Sigma}_\perp
\end{pmatrix}
\begin{pmatrix}
\tilde{U}^\T
\\
\tilde{U}_{\perp}^\T
\end{pmatrix},
\ee
where $\Sigma = \diag(\sigma_1,\dots,\sigma_{\zeta})$ such that $\zeta = \text{rank}(R_n)$ and $\sigma_1 \geq \dots\geq\sigma_{\zeta}$; $\tilde{\M\Sigma} = \diag(\tilde{\sigma}_1,\dots,\tilde{\sigma}_{\zeta})$, $\tilde{\M\Sigma}_{\perp} = \diag(\tilde{\sigma}_{\zeta + 1},\dots,\tilde{\sigma}_{K})$ and $\tilde{\sigma}_1\geq \cdots\geq\tilde{\sigma}_K$. Similarly, we write the empirical versions of $Q_n$ and $R_n$ as,
\bee\nonumber
\hat{Q}_n &= \p_n\big\{\bpsi(X,T)\bpsi^\T(X,T)\big\},
\\
\hat{R}_n &=\p_n\big[\{{\Psi}(X,T) - \hat{\M\Gamma}(X)\}\{{\Psi}(X,T) - \hat{\M\Gamma}(X)\}^\T\big]
\ee
Recalling the definition of $\hat{G}_n$, we then have $\hat{G}_n = \hat{R}_n + \rho\hat{Q}_n$. We also define, 
\bee\nonumber
\bar{G}_{n} &=\p\big[\{\Psi(X,T) - \hat{\M\Gamma}(X)\}\{\Psi(X,T)  - \hat{\M\Gamma}(X) \}^\T\big] + \rho  Q_n
\\
\bar{R}_n &= \p\big[\{\Psi(X,T) - \hat{\M\Gamma}(X)\}\{\Psi(X,T)  - \hat{\M\Gamma}(X)\}^\T\big].
\ee Finally, we write the SVD form of $\hat{G}_n$,
\bee\nonumber
\hat{G}_n = \begin{pmatrix}
\hat{U} & \hat{U}_\perp
\end{pmatrix}
\begin{pmatrix}
\hat{\Sigma} & 
\\
&  \hat{\Sigma}_\perp
\end{pmatrix}
\begin{pmatrix}
\hat{U}^\T
\\
\hat{U}_{\perp}^\T
\end{pmatrix},
\ee
where $\hat{\M\Sigma} = \diag(\hat{\sigma}_1,\dots,\hat{\sigma}_{\zeta})$, $\hat{\Sigma}_{\perp} = \diag(\hat{\sigma}_{\zeta + 1},\dots,\hat{\sigma}_{K})$, and $\hat{\sigma}_1\geq\dots\geq\hat{\sigma}_K$.
\par
\
\par
\noindent\textbf{True coefficient vector $\phi^*$: }For our theoretical analysis, we need $\hat{\phi}$  converging to some target vector. We choose the least squares approximation coefficients of $\tilde{\tau}(X,T)$ as this target vector. In particular, consider the population-level least squares approximation,
\bee\nonumber
\phi^* = \argmin_{\phi\in\RR^{K}} E\big[\{\tilde{\tau}(X,T) - \phi^\T\Psi(X,T)\}^2\big].
\ee
With simple algebra, one has 
\bee\label{def:phi*}
\phi^* = Q_n^{-1}E\big[\tilde{\tau}(X,T)\Psi(X,T)\big].
\ee For the tensor product of B-spline basis, \citet{huang2003local,belloni2015some} have shown the $\mathcal{L}^{\infty}$ norm approximation power of least-square approximation $(\phi^*)^\T\Psi(X,T)$, when $\tilde{\tau}(X,T)$ is in the $p$-smooth H{\"o}lder class. The following proposition is a direct application of \citet{huang2003local,belloni2015some}'s general results; see, e.g., the Appendix of \citet{huang2003local}, and \citet[Proposition 3.1 \& Example 3.8]{belloni2015some}.
\begin{proposition}\label{lm:psieapprox}
Suppose $\tilde{\tau}\in \Lambda(p,c,\mathbb{X}\times\mathbb{T})$ for some fixed $p,c > 0$, and Assumptions \ref{am:bspline} holds. We have $\|\tilde{\tau} - (\phi^*)^\T\Psi\|_{\mathbb{X}\times\mathbb{T}} \precsim K^{-p/(d + 1)}$.
\end{proposition}
\par
\
\par
\noindent\textbf{Regularization conditions: }We summarize the regularity conditions for our asymptotic results as follows.
\begin{assumption}\label{am:bspline}
Assume the basis $\M\Psi(x,t)\in \RR^{K}$ takes the tensor-product form,
$$
\M\Psi(x,t) = \psi(t)\otimes\psi(x^{(1)})\otimes\cdots\otimes \psi(x^{(d)}).
$$
Here $ \psi(\cdot) = \begin{bmatrix}\psi^{(1)}(\cdot) ,&\dots&,\psi^{(k)}(\cdot)\end{bmatrix}^\T\in \RR^{k} $ is a vector of the B-spline functions for a single variable over $[0,1]$; thus $K = k^{d + 1}$. The splines in $ \psi(\cdot)$ are of a fixed degree  $r  \geq \lceil p\rceil$. The  generation intervals (see $\mathsection$\ref{sec:bspline}) of $\psi(\cdot)$ are $I_1 = [0,s_1), I_2 = [s_1,s_2),\dots, I_m=[s_m,1]$ with $0< s_1< \dots< s_m<  1$, $m\geq 2$ and $k = r + m$. Finally, we assume there exists a $\kappa > 0$ not dependent on $n$, such that for any  $n\in\mathbb{Z}^+$,
\bee\label{am:mesh}
\frac{\max_{0\leq m'\leq m}(s_{m' + 1} - s_{m'})}{\min_{0\leq m'\leq m}(s_{m' + 1} - s_{m'})} \leq \kappa.
\ee
\end{assumption}
\begin{assumption}\label{am:compact}
$\mathbb{X}\times \mathbb{T}\subseteq \RR^{d + 1}$ is compact. Without loss of generality and for simplicity, we let $\mathbb{X} = [0,1]^d$ and $\mathbb{T} = [0,1]$.
\end{assumption}
\begin{assumption}\label{am:Kn}
For some $0<a_1\leq a_2<1$, one has $n^{a_1}\precsim K\precsim n^{a_2}$.
\end{assumption}
\begin{assumption}\label{am:densX}
There exists some fixed constants $0<c_f<C_f<+\infty$, such that $ c_f \leq  \inf_{x\in\mathbb{X}}f(x) \leq \sup_{x\in\mathbb{X}}f(x)\leq C_f$.
\end{assumption}
\begin{assumption}\label{am:moment}
(i) $0<\inf_{(x,t)\in \mathbb{X}\times \mathbb{T}}\text{Var}\big(Y\mid X = x, T = t\big)\leq\sup_{(x,t)\in \mathbb{X}\times \mathbb{T}}\text{Var}\big(Y\mid X = x, T = t\big) <+\infty$; (ii) $\sup_{x\in\mathbb{X}}\var(Y\mid X = x) < +\infty$; (iii) $\sup_{x\in\mathbb{X}}|E(Y\mid X = x)|<+\infty$.
\end{assumption}
\begin{assumption}\label{rate:sup}
(i) $\big\|\hat{m} - m\big\|_{\mathbb{X}} = {o}_{{P}}(1)$; (ii) $\|\hat{\Gamma}^\T\phi^* - \Gamma^\T\phi^*\|_{\mathbb{X}} = o_P(1)$.
\end{assumption}
For the ease of exposition, we use the same univariate B-spline function $\psi(\cdot)$ for all variables $X^{(1)},\dots,X^{(d)}$ and $T$ in Assumption \ref{am:bspline}. We note the theoretical results remain the same when dimension $k$ and $I_1,\dots,I_m$ vary for each variable, as long as the dimensions of all univariate B-splines are on the same asymptotic order and their mesh ratios  given in \eqref{am:mesh} are bounded by an uniform constant $\kappa > 0$. Overall the Assumption \ref{am:bspline} is mild; see, e.g., \citet{douglas1975optimal, huang2003local}.  Assumption \ref{am:compact} is a common assumption for the sieve-type estimator \citep{newey1997convergence,huang2003local,belloni2015some,shi2021statistical}. Assumption \ref{am:Kn} mildly controls the growing speed of the basis number, so that we can approximate $\tau(X,T)$ nonparametrically. Assumption \ref{am:densX} is a standard condition for sieve method \citep{huang2003local, chen2015optimal}. Assumption \ref{am:moment} contains some weak moment conditions; Similar assumptions are often employed in sieve estimation and causal inference literature   \citep{chen2015optimal, cui2020semiparametric}.  Assumption \ref{rate:sup} states that the nuisance functions $\hat{m}(x)$ and $\hat{\Gamma}^\T(x)\phi^*$ are both uniformly consistent estimators for their true counterparts.   Assumption \ref{rate:sup} (i) has previously been considered by \citet{kennedy2017non} for the doubly robust average treatment effect estimation with continuous treatments. Moreover, one sufficient condition that makes Assumption \ref{rate:sup} (ii) hold is that $\hat{\Gamma}(x)$ is estimated by the generalized propensity score $\hat{\varpi}(t\mid x)$, and $\hat{\varpi}(t\mid x)$ satisfies a similar uniform consistency as $\hat{m}(x)$,
\bee\label{gps:rate2}
\|\hat{\varpi} - {\varpi}\|_{\mathbb{X}\times\mathbb{T}} = o_P(1).
\ee
Condition \eqref{gps:rate2} has previously been considered by  \citet{kennedy2017non} as a regularization condition. Proposition \ref{po:gpsrate} formalizes the justification that Condition \eqref{gps:rate2} is a sufficient condition for Assumption  \ref{rate:sup} (ii).
\par
Classic nonparametric sieve regression \citep{newey1997convergence} often assumes a full-rank gram matrix $Q_n = \E\{\Psi(X,T)\Psi^\T(X,T)\}$ where the dimension of $\Psi$ depends on $n$. 
In stark contrast, theoretical analysis of the proposed R-learner involves a low-rank gram matrix, 
\bee\label{def:rn}
 R_n = E\big[\{\Psi(X,T) - \M\Gamma(X)\}\{\Psi(X,T)- \M\Gamma(X)\}^\T\big] = \begin{pmatrix}
U & U_\perp
\end{pmatrix}
\begin{pmatrix}
\Sigma & 
\\
&  0
\end{pmatrix}
\begin{pmatrix}
U^\T
\\
U_{\perp}^\T
\end{pmatrix},
\ee
where the right-hand side of (\ref{def:rn}) is the singular value decomposition of $R_n$, with $\text{rank}(R_n) = \zeta$ and $\Sigma = \diag(\sigma_1,\dots,\sigma_{\zeta})$ such that  and $\sigma_1 \geq \dots\geq\sigma_{\zeta} > 0$. Each entry of ${R}_n$ is the probability limit of the corresponding entry of $\hat{R}_n$  in \eqref{def:phi}. Intuitively, the low rank of $R_n$ is tied to the non-identification issue of the generalized R-loss in $\mathsection$\ref{sec:id} when setting $\rho = 0$. In this case,  $\hat{\phi}$ is  asymptotically unsolvable, or equivalently $\hat{R}_n$ 
in \eqref{def:phi} is asymptotically non-invertible; i.e., $R_n$ is low-rank.  We denote  $\beta_n = \sigma_{\zeta} > 0$, which plays an essential quantity in our theoretical results. See Lemma \ref{lm:svdR} for detailed spectral properties of $R_n$.
\begin{assumption}
As $n\rightarrow \infty$, we have $\beta_n\asymp 1$.
\end{assumption}
\begin{proposition}\label{po:gpsrate}
Suppose $\tilde{\tau}\in \Lambda(p,c,\mathbb{X}\times\mathbb{T})$ for some $p,c > 0$,  and Assumptions \ref{am:bspline}, \ref{am:compact}, \ref{am:densX} hold. When $\hat{\Gamma}(x) = E_{\hat{\varpi}}\{\Psi(X,T)\mid X = x\}$, then we have Assumption \ref{rate:sup} (ii) holds, whenever \eqref{gps:rate2} holds.
\end{proposition}
\begin{proof}[of Proposition \ref{po:gpsrate}]By Proposition \ref{lm:psieapprox}, we have $\|\tilde{\tau} - \{\phi^*\}^\T\Psi\|_{\mathbb{X}\times\mathbb{T}} \precsim 1$. Also, $\|\tilde{\tau}\|_{\mathbb{X}\times\mathbb{T}}\leq C$ for some $C > 0$ due to $\tilde{\tau}\in \Lambda(p,c,\mathbb{X}\times\mathbb{T})$. We thus have 
\bee\label{pos2:add}
\|\{\phi^*\}^\T\Psi\|_{\mathbb{X}\times\mathbb{T}}&\leq \|\tilde{\tau} - \{\phi^*\}^\T\Psi\|_{\mathbb{X}\times\mathbb{T}} + \|\tilde{\tau}\|_{\mathbb{X}\times\mathbb{T}}
\\
&\precsim 1.
\ee
 On the other hand, we note that
\bee\nonumber
\hat{\Gamma}^\T(x)\phi^* - \Gamma^\T(x)\phi^* &= E_{\hat{\varpi}}\{\Psi^\T(X,T)\phi^*\mid X=x\} - E_{{\varpi}}\{\Psi^\T(X,T)\phi^*\mid X = x\}
\\
&=\int_{\mathbb{T}} \Psi^\T(x,t)\phi^*\Big\{\hat{\varpi}(t\mid x)- {\varpi}(t\mid x)\Big\}dt.
\ee
Thus we have
\bee\nonumber
\|\hat{\Gamma}^\T\phi^* - \Gamma^\T\phi^* \|_{\mathbb{X}} &= \sup_{x\in\mathbb{X}}\Big|\int_{\mathbb{T}} \Psi^\T(x,t)\phi^*\Big\{\hat{\varpi}(t\mid x)- {\varpi}(t\mid x)\Big\}dt\Big|
\\
&\leq \|\hat{\varpi}- {\varpi}\|_{\mathbb{X}\times\mathbb{T}}\cdot\sup_{x\in\mathbb{X}}\int_{\mathbb{T}} \Big|\Psi^\T(x,t)\phi^*\Big|dt
\\
&\leq  \|\hat{\varpi}- {\varpi}\|_{\mathbb{X}\times\mathbb{T}} \cdot\int_{\mathbb{T}} \|\{\phi^*\}^\T\Psi\|_{\mathbb{X}\times\mathbb{T}}dt
\\
&\precsim o_{P}(1),
\ee
by \eqref{pos2:add}, whenever \eqref{gps:rate2} holds and $\mathbb{T}$ is compact. This completes the proof.
\end{proof}
\subsection{Brief review of the B-spline functions}\label{sec:bspline}
We give some precise descriptions of how to construct  $\Psi(x,t)$ and the corresponding theoretical properties. For simplicity, we focus on the case that the domains of $t$, $x^{(1)},\dots,x^{(t)}$  are all $[0,1]$.  We first construct $\psi(w) = \big[\psi^{(1)}(w),\dots,\psi^{(k)}(w)\big]^\T$ for $w\in[0,1]$. Define the knot set 
\bee
0 = s_{-(r - 1)}= \dots =s_{-1} = s_0 \leq s_1\leq\dots \leq s_m \leq s_{m + 1}=\dots = s_{m + r} = 1
\ee
 Define
\bee\nonumber
N^{(1)}_{j}(w) ={I}(w \in I_j),
\ee
for $I_1 = [s_0 = 0,s_1),\dots, I_m = [s_m,s_{m + 1} = 1)$ and $j = 0,\dots,m$. For $r \geq 2$ and $j = -(r - 1),\dots,m$, define 
\bee\label{Njr:def}
N_j^{(r)}(w) = \frac{w -s_j}{s_{j + r - 1} - s_j}N^{(r - 1)}_j(w) + \frac{s_{j + r} - w}{s_{j + r} - s_{j + 1}}N^{(r - 1)}_{j + 1}(w),
\ee
where we set ${1}/{0}= 0$ as the convention. We thus have the relationship that $k = r + m$. We define $\psi(w)$ by normalizing $N_j^{(r)}(w)$ as,
\bee\label{defpsiw}
\psi(w) &= \begin{bmatrix}\psi^{(1)}(w)&,\dots,&\psi^{(k)}(w)\end{bmatrix}^\T
\\
&= \sqrt{k}\cdot\begin{bmatrix}N_{-(r - 1)}^{(r)}(w)&,\dots,&N^{(r)}_m(w)\end{bmatrix}^\T.
\ee
This constitutes the univariate B-spline function. Finally, we form the multivariate B-spline functions $\M\Psi(x,t)$ by the tensor product of $\psi(t),\psi(x^{(1)}),\dots,\psi(x^{(d)})$ as
\bee\label{Psixt}
\M\Psi(x,t) = \psi(t)\otimes\psi(x^{(1)})\otimes\cdots\otimes \psi(x^{(d)});
\ee
here $x^{(1)},\dots,x^{(d)}$ correspond to the coordinate-wise random variables $X^{(1)},\dots,X^{(d)}$ of $X$. 
\par
The following lemmas state some standard results for the B-spline functions, which will be useful in our theoretical analysis. In particular, letting $1_k = (1,\dots,1)^\T\in\RR^k$, Lemma \ref{po:bspline} shows that $ 1_k^\T\psi(w)$ will be a fixed constant for any $w \in[0,1]$.  Lemma \ref{po:bs:other} provides  explicit properties of $N^{(r)}_j(w)$. Lemma \ref{am:psi} states  asymptotic properties of $\Psi(x,t)$.
\begin{lemma}\label{po:bspline}
Given any fixed $r \in {Z}^+$, for the $\{N^{(r)}_j(w)\}_{j = 1}^m$ defined in (\ref{Njr:def}), we have $$\sum_{j = -(r- 1)}^m  N_{j}^{(r)}(w)\equiv 1$$ for any $w\in[0,1]$. Furthermore, for any $w \in [0,1]$, 
$
1_k^\T\psi(w) \equiv \sqrt{k}.
$
\end{lemma}
\begin{proof}[of Lemma \ref{po:bspline}]
When $r = 1$, we have $$\sum_{ j=1}^m N_j^{(1)}(w) = \sum_{j = 1}^m {I}(w\in I_j) = {I}(w \in[0,1]) \equiv 1.$$
\par
When $r \geq 2$, by \eqref{Njr:def}, one has
\bee\label{deg:p5}
\sum_{j= -(r - 1)}^m N_j^{(r)}(w) = &\frac{w -s_{-(r - 1)}}{s_{0} - s_{-(r - 1)}}N^{(r - 1)}_{-(r - 1)}(w) + \left\{\sum_{j= -(r - 1)}^{m-1} \frac{s_{j + r} - w}{s_{j + r} - s_{j + 1}}N^{(r - 1)}_{j + 1}(w) + \frac{w -s_{j + 1}}{s_{j + r} - s_{j + 1}}N^{(r - 1)}_{j + 1}(w) \right\}
\\
& + \frac{s_{m + r} - w}{s_{m + r} - s_{m + 1}}N^{(r - 1)}_{m + 1}(w)
\\
=&  \sum_{j= -(r - 1)}^{m-1} \frac{s_{j + r} - w}{s_{j + r} - s_{j + 1}}N^{(r - 1)}_{j + 1}(w) + \frac{w -s_{j + 1}}{s_{j + r} - s_{j + 1}}N^{(r - 1)}_{j + 1}(w) 
\\
= &\sum_{j= -\{(r - 1) - 1\}}^{m} N^{(r - 1)}_{j}(w), 
\ee
where the first equality follows by $s_0 - s_{-(r- 1)} = s_{m+1} - s_{m + 1} = 0$ and  ${1}/{0} = 0$. Thus we have $\sum_{j= -1}^m N_j^{(2)}(w) = \sum_{j = 0}^m N_j^{(1)}(w)= 1$ for any $w\in[0,1]$. Finally, by \eqref{deg:p5} and mathematical induction, one has
$
\sum_{j= -(r - 1)}^m N_j^{(r)}(w)  \equiv 1
$
for all $w\in[0,1]$ and $r \geq 2$. By \eqref{defpsiw}, we conclude that for any $w\in[0,1]$,
\bee\nonumber
(1,\dots,1)\cdot\psi(w) = \sqrt{k}\sum_{j = -(r - 1)}^{m} N^{(r)}_{j}(w) = \sqrt{k},
\ee
which completes the proof.
\end{proof}
\begin{lemma}\label{po:bs:other}
Given any fixed $r \in {Z}^+$, for the $\{N^{(r)}_j(w)\}_{j = 1}^m$ defined in (\ref{Njr:def}), we have
\begin{enumerate}
\item[(i)*] For any given $j = -(r-1),\dots,m$, we have $N_{j}^{(r)}(w) =0$ for all $w \notin[s_j,s_{j + r}]$.
\item[(ii)*] For any given $j = -(r-1),\dots,m$, we have $N_{j}^{(r)}(w)  >0$ for all $w \in (s_j,s_{j + r})$.
\item[(iii)*] For any given $\ell = 0,\dots,m$, when $w \in [s_\ell,s_{\ell + 1})$, we have
$$
\sum_{j = -(r- 1)}^m \tilde{\beta}_j N_j^{(r)}(w) = \sum_{j = -(r - 1) + \ell}^{\ell}\tilde{\beta}_jN_j^{(r)}(w),
$$
for any $(\tilde{\beta}_{-(r - 1)},\dots,\tilde{\beta}_m)^\T\in \RR^{m + r}$.
\item[(iv)*] Under Assumption \ref{am:bspline}, we have $\sup_{j = -(r - 1),\dots,m\atop w\in[0,1]}|N^{(r)}_j (w)| = \mathcal{O}(1)$ when $m \rightarrow +\infty$.
\end{enumerate}
By \eqref{defpsiw}, we then have that $ \psi(w)$ satisfies the corresponding properties:
\begin{enumerate}
\item[(i)] For any given $\tilde{k} = 1,\dots,k$, we have $\psi^{(\tilde{k})}(w) =0$ for all $w \notin[s_{\tilde{k} - r},s_{\tilde{k}}]$.
\item[(ii)] For any given $\tilde{k} = 1,\dots,k$, we have $\psi^{(\tilde{k})}(w)  >0$ for all $w \in (s_{\tilde{k} - r},s_{\tilde{k}})$.
\item[(iii)] For any given $\ell = 0,\dots,m$, when $w \in [s_\ell,s_{\ell + 1})$, we have
$$
\sum_{\tilde{k} = 1}^{k} {\beta}_{\tilde{k}} \psi^{(\tilde{k})}(w) = \sum_{\tilde{k} = \ell + 1}^{\ell + r}{\beta}_{\tilde{k}} \psi^{(\tilde{k})}(w),
$$
for any $({\beta}_{1},\dots,\beta_k)^\T\in \RR^k$.
\item[(iv)] Under Assumption \ref{am:bspline}, we have $\sup_{\tilde{k} = 1,\dots,k\atop w\in[0,1]}\big|\psi^{(\tilde{k})}(w)\big| = \mathcal{O}(\sqrt{k})$ when $m \rightarrow +\infty$.
\end{enumerate} 
\end{lemma}
\begin{proof}[of Lemma \ref{po:bs:other}] 
The results in (i)--(iv) follow from (i)*--(iv)*  by recalling \eqref{defpsiw} and the relationship $j = k -r$. So it suffices to show (i)*--(iv)*. The results in (i)*, (ii)* and (iii)* are adapted from (1.6), (1.7) and (1.36) in \cite{kunoth2018splines}, respectively, with their $B_{j + r,r-1, \xi}(w)$ being our notation $N_j^{(r)}(w)$ and $ \xi = (s_{-(r - 1)},\dots, s_{m + r})$. For the B-spline functions with regularization conditions assumed in Assumption \ref{am:bspline}, the (iv) and (iv)* are standard results; see, e.g., \cite{newey1997convergence, belloni2015some}.
\end{proof}
\begin{lemma}\label{am:psi}
Suppose $\M\Psi(x,t)$ satisfies Assumption \ref{am:bspline}. When $n\rightarrow +\infty$, we have
\begin{enumerate}
\item[(i)](Uniform boundedness): 
$
\|\Psi\|_{\mathbb{X}\times\mathbb{T}}\precsim\sqrt{K}.
$
\item[(ii)](Bounded spectrums): Let $J_{X,T} = \int_{\mathbb{X}\times \mathbb{T}}\Psi(x,t)\Psi^\T(x,t) dxdt$. One has 
$
\lambda_{\min}(J_{X,T})$  and $\lambda_{\max}(J_{X,T}) 
$ are bounded away from both $0$ and $+\infty$. Similar upper and lower bounds also hold for $J_{T} = \int_{\mathbb{T}}\psi(t)\psi^\T(t) dt$ and $\int_{\mathbb{X}}\Psi(x)\Psi^\T(x)dx$.
\item[(iii)](Belonging to $\lp(X,T)$ space): $\phi^\T\Psi$ is in $\lp(X,T)$ for any given $n$ and $\phi\in\RR^{K}$.
\end{enumerate} 
\end{lemma}
\begin{proof}[of Lemma \ref{am:psi}] These are standard results for B-splines under Assumption \ref{am:bspline}; see, e.g., \citet{belloni2015some,chen2015optimal}.
\end{proof}
\par
Our next proposition demonstrates two facts about the B-spline approximation function $\phi^\T\Psi(x,t)$ satisfying the zero condition:
\bee\label{bspline:zeroconstrain}
\phi^\T\Psi(x,t = 0) = 0,\text{ for any }x\in\mathbb{X}.
\ee
\begin{itemize}\item[(i)] B-spline approximation $\phi^\T\Psi(x,t)$ satisfies  \eqref{bspline:zeroconstrain}
if and only if all coefficients in $\phi$ associated with the ``intercept'' function, namely $\psi^{(1)}(t)$ in $\psi(t) = [\psi^{(1)}(t),\dots,\psi^{(k)}(t)]^\T$, are zeros.
\item[(ii)] B-spline approximation $\phi^\T\Psi(x,t)$ satisfying  \eqref{bspline:zeroconstrain} can still approximate any function $\check{\tau}(x,t)$ in \eqref{checktau} with a smooth $s(x)$ arbitrarily well under the $\mathcal{L}_{\mathbb{X}\times \mathbb{T}}^2$ norm, as $K$ grows. 
\end{itemize} 
The $\psi^{(1)}(t)$ is referred as the intercept function, because  it is the only one among $\psi^{(1)}(t),\dots,\psi^{(k)}(t)$ having nonzero value at $t = 0$, which can be observed by the construction in \eqref{Njr:def}. Recalling \eqref{Psixt},  we can write 
\bee\label{psi:newdecom}
\M\Psi(x,t) =( \psi^{(1)}(t)\cdot\Psi^\T(x), \dots,\psi^{(k)}(t)\cdot\Psi^\T(x))^\T,
\ee where $\Psi(x) = \psi(x^{(1)})\otimes\cdots\otimes \psi(x^{(d)})$, and thus the first $K/k$ coefficients of $\phi$ are associated with the intercept function $\psi^{(1)}(t)$. 
\par
We formalize the Properties~(i) and (ii) above into the Proposition~\ref{prop:zerospline}. Property~(i) figures out the concrete form of the B-spline approximation functions satisfying the zero condition~\eqref{bspline:zeroconstrain}, with the tensor-product B-spline basis introduced at the beginning of this section. Property~(ii) implies the ill-posedness of using B-spline approximation $\phi^\T\Psi(x,t)$ under zero condition to approximate $\tau$, when minimizing the generalized R-loss $L_c(h)$. In particular, Proposition~\ref{po:nonunique} shows that $L_c(h)$ has infinite many minima in the form of \eqref{checktau} from the population level, and our Property~(ii) further shows that, with sufficiently large $K$, $\phi^\T\Psi(x,t)$ with zero condition can approximate any such minima with some smooth $s$, arbitrarily well under the $\mathcal{L}_{\mathcal{P}}^2$ norm as $K\rightarrow \infty$. Thus with sufficient large $n$ and $K$ for nonparametric estimation, minimizing the generalized R-loss through $\phi^\T\Psi(x,t)$ under zero condition, will fail to help us identify and   approximate the minima of the true CATE $\tau$ in $L_c(h)$, as there are infinitely many minima of $L_c(h)$ can be well approximated by  $\phi^\T\Psi(x,t)$ under zero condition.

\begin{proposition}\label{prop:zerospline} 
Suppose Assumption~\ref{am:bspline} holds. We have: 
\begin{itemize}
\item[(i)] Spline function $\phi^\T\Psi(x,t)$ satisfies  \eqref{bspline:zeroconstrain} if and only if $\phi_1 = \dots = \phi_{K/k} = 0$ for $\phi = (\phi_1,\dots,\phi_K)$; 
\item[(ii)] Suppose ${\tau}\in \Lambda(p,c,\mathbb{X}\times\mathbb{T})$ and $s\in\Lambda(p',c',\mathbb{X})$ for some positive and fixed constants $p,c,p',c'$, and Assumptions~\ref{A:CS} and \ref{am:densX} hold. Then we   have
\bee\nonumber
\inf_{\text{All }\phi^\T\Psi\text{ satisfy \eqref{bspline:zeroconstrain}}}\|\phi^\T\Psi - \check{\tau}(\cdot\mid s)\|_{\mathcal{L}_{\mathcal{P}}^2} \rightarrow 0,\text{ as $K\rightarrow \infty$}.
\ee
\end{itemize}
\end{proposition}\begin{proof}[of Proposition~\ref{prop:zerospline}]We first show (i). By the construction of \eqref{Njr:def}, one can verify that when $s_0 = 0<s_1<\cdots<s_m < s_{m + 1} = 1$ under Assumption~\ref{am:bspline}, we have $\psi^{(1)}(0)  >  0$ and 
\bee\label{psi2k0}
\psi^{(2)}(0) = \cdots = \psi^{(k)}(0) = 0.
\ee Thus when $\phi_1 = \cdots = \phi_{K/k} = 0$, recalling \eqref{psi:newdecom} we have
\bee\nonumber
\phi^\T\Psi(x,t = 0) &= (0,\dots,0)\cdot \psi^{(1)}(0)\cdot\Psi(x) + (\phi_{K/k + 1},\dots,\phi_K)\cdot\underbrace{( \psi^{(2)}(0)\cdot\Psi^\T(x), \dots,\psi^{(k)}(0)\cdot\Psi^\T(x))^\T}_{ = (0,\dots,0)^\T}
\\
&= 0 ,
\ee 
for any $x\in\mathbb{X}$. 
\par
On the other hand, we want to prove  that \eqref{bspline:zeroconstrain} holds only if $\phi_1 = \dots=\phi_{K/k} = 0$. We prove this through the contradiction. Supposing $\phi_1,\cdots, \phi_{K/k}$ are not all zero, i.e., $(\phi_1,\dots,\phi_{K/k})^\T = \tilde{\phi}  \in\RR^{K/k}$ for some $\|\tilde\phi\| > 0$, we attempts to show that $\phi^\T\Psi(x,0)\neq 0$ for some $x\in\mathbb{X}$. Now by \eqref{psi2k0}, we have
\bee\nonumber
\phi^\T\Psi(x,t = 0) = \psi^{(1)}(0)\cdot \tilde{\phi}^\T\Psi(x), 
\ee 
and $\psi^{(1)}(0) > 0$. It is left to show that for any given $\tilde{\phi}$ with $\|\tilde{\phi}\| > 0$, there always exists some $x\in[0,1]^d$ such that
$
\tilde{\phi}^\T\Psi(x)\neq 0
$. By Lemma~\ref{am:psi}(ii), we have 
$$
\int_{x\in\mathbb{X}}\{\tilde{\phi}^\T\Psi(x)\}^2dx = \int_{x\in\mathbb{X}}\tilde{\phi}^\T\Psi(x)\Psi^\T(x)\tilde{\phi}dx > 0,
$$ for any  non-zero $\tilde{\phi}$, which directly implies that $\tilde{\phi}^\T\Psi(x)\neq 0$ for some $x\in[0,1]^d$.
\par
Next, we prove (ii). Consider 
\bee\nonumber
\check{\tau}^{\natural}(x,t\mid s) = \tau(x,t) + s(x),
\ee
for all $(x,t)\in[0,1]^{d + 1}$. Recall the definition of the H{\"o}lder class in \eqref{def:holder}. We can check that  $\check{\tau}^{\natural}(x,t\mid s) \in \Lambda(\min\{p,p'\},c+ c')$. Proposition~\ref{lm:psieapprox} guarantees that 
$$
\phi^\natural = Q_n^{-1}E\left\{\check{\tau}^\natural(X,T\mid s)\Psi(X,T)\right\}\in\RR^K
$$ satisfying
\bee\label{supnorm:cov:zerocondition}
\|(\phi^\natural)^\T\Psi - \check{\tau}^\natural\|_{\mathbb{X}\times \mathbb{T}} \precsim K^{-\min\{p',p\}/(d + 1)}\rightarrow 0,
\ee
as $K\rightarrow \infty$. By Lemma~\ref{am:svb}, we further have $\|\phi^\natural\|\precsim 1$. Next denote $\phi^\natural = (\phi_1^{\natural},\dots,\phi_K^{\natural})$, 
$$
\phi^{\natural}_0 = (\underbrace{0,\dots,0}_{K/k\text{ coordinates}},\phi^\natural_{K/k + 1},\dots,\phi^\natural_{K}),\text{ and }I_{0} = \mathrm{diag}((\underbrace{0,\dots,0}_{K/k\text{ coordinates}},\underbrace{1,\dots,1}_{K - K/k\text{ coordinates}}).
$$
By (i), it is clear that $(\phi^
\natural_0)^
\T\Psi $ satisfies \eqref{bspline:zeroconstrain}. By Lemma~\ref{po:bs:other}, we know $\psi^{(1)}(t) = 0$ when $t\in[s_1,1]$, and thus when $(t,x)\in[s_1,1]\times [0,1]$, we have
$$
(\phi^\natural)^\T\Psi(x,t) =(\phi^\natural_0)^\T\Psi(x,t).
$$ 
On the other hand, by \eqref{psi2k0} we have
\bee\nonumber
(\phi^\natural)^\T\Psi(x,0)& = (\phi^\natural)^\T\Psi(x,0)  - (\phi_0^\natural)^\T\Psi(x,0) 
\\
&= \sum_{j = 1}^{K/k} \phi_j^\natural\cdot\psi^{(1)}(0) \cdot\Psi_j(x)
\\
& = \sqrt{k}\cdot \sum_{j = 1}^{K/k} \phi_j^\natural \Psi_j(x),
\ee
where we recall $\Psi(x)$ can be written as $\Psi(x) = (\Psi_1(x),\dots,\Psi_{K/k}(x))^\T$, and the last equality is by Lemma~\ref{po:bs:other} such that $\sum_{j\in[k]}\psi^{(j)}(0) = \psi^{(1)}(0) = \sqrt{k}$. By \eqref{supnorm:cov:zerocondition}, we know 
\bee\label{boundatpoint0:illplemma}
\sqrt{k}\cdot\sup_{x\in[0,1]^d}\left| \sum_{j = 1}^{K/k} \phi_j^\natural\cdot \Psi_j(x)\right| = \sup_{x\in[0,1]^d}|(\phi^\natural)^\T\Psi(x,0)| \leq \|\check{\tau}^\natural\|_{\mathbb{X}\times\mathbb{T}} + o(1) \precsim 1,
\ee
as $K\rightarrow \infty$. Now uniformly for all $(t,x)\in[0,s_1]\times [0,1]$, we have $0\leq \psi^{(1)}(t)\leq \sqrt{k} - \sum_{j = 2}^k\psi^{(j)}(t)\leq \sqrt{k}$, and thus
\bee\nonumber
\left|(\phi^\natural)^\T\Psi(x,t)  - (\phi_0^\natural)^\T\Psi(x,t)\right| &= \left|\sum_{j = 1}^{K/k} \phi_j^\natural\psi^{(1)}(t) \Psi_j(x)\right|
\\
&=\psi^{(1)}(t)
\cdot \left|\sum_{j = 1}^{K/k} \phi_j^\natural\cdot \Psi_j(x)\right|
\\
&\leq  \sqrt{k}\cdot\left|\sum_{j = 1}^{K/k}\phi_j^\natural\cdot \Psi_j(x)\right|
\\
&\precsim 1,
\ee
as $K\rightarrow \infty$, where the last inequality is by \eqref{boundatpoint0:illplemma}.
\par
 Summarizing all results above and noting that the density functions of $X$ and $T\mid X$ are uniformly bounded, we have 
\begin{align}\nonumber
\|(\phi^\natural_0)^\T\Psi - \check{\tau}(\cdot\mid s)\|^2_{\mathcal{L}_{\mathcal{P}}^2} &= \|(\phi^\natural_0)^\T\Psi - \check{\tau}^\natural(\cdot\mid s)\|^2_{\mathcal{L}_{\mathcal{P}}^2}
\\\nonumber
& =  \int_{(t,x)\in[0,1]^{d + 1}}\left\{(\phi^\natural_0)^\T\Psi(x,t) - \check{\tau}^\natural(x,t\mid s)\right\}^2 f(x)\varpi(t\mid x)dxdt
\\\nonumber
&=\int_{(x,t)\in[s_1,1]\times[0,1]^{d }}\left\{(\phi_0^\natural)^\T\Psi(x,t) - \check{\tau}^\natural(x,t\mid s)\right\}^2 f(x)\varpi(t\mid x)dxdt
\\\nonumber
&\quad +\int_{(x,t)\in[0,s_1]\times[0,1]^{d }}\left\{(\phi_0^\natural)^\T\Psi(x,t) - \check{\tau}^\natural(x,t\mid s)\right\}^2 f(x)\varpi(t\mid x)dxdt
\\\nonumber
&\leq\int_{(x,t)\in[s_1,1]\times[0,1]^{d }}\left\{(\phi^\natural)^\T\Psi(x,t) - \check{\tau}^\natural(x,t\mid s)\right\}^2 f(x)\varpi(t\mid x)dxdt
\\\nonumber
&\quad +2\int_{(x,t)\in[0,s_1]\times[0,1]^{d }}\left\{(\phi_0^\natural)^\T\Psi(x,t) - (\phi^\natural)^\T\Psi(x,t) \right\}^2 f(x)\varpi(t\mid x)dxdt
\\\nonumber
&\quad+2\int_{(x,t)\in[0,s_1]\times[0,1]^{d }}\left\{\check{\tau}^\natural(x,t\mid s) - (\phi^\natural)^\T\Psi(x,t) \right\}^2 f(x)\varpi(t\mid x)dxdt
\\\nonumber
&\precsim K^{-\min\{p',p\}/(d + 1)} + k^{-1} + k^{-1}\cdot  K^{-\min\{p',p\}/(d + 1)} \rightarrow 0,
\end{align}
as $K\rightarrow \infty$, where the last inequality is by \eqref{supnorm:cov:zerocondition}, \eqref{supnorm:cov:zerocondition} and the fact that $|s_1|\precsim k^{-1}$ by \eqref{am:mesh}.
\end{proof}

\subsection{Technical lemmas}\label{sec:lemma}
In this section, we present all the technical lemmas for the proofs of our main propositions and theorems.
\begin{lemma}\label{lm:svdR}Suppose Assumptions \ref{A:CS}, \ref{am:bspline}, \ref{am:compact} and \ref{am:densX} hold. Let $\text{span}(U_{\perp})$ be the linear subspace in $\RR^{K}$ which is spanned by the column vectors in $U_{\perp}$, and $\text{span}(U)$ is defined correspondingly. We  have the following spectral properties of $R_n$.
\begin{enumerate}
\item[(i)] We have $\text{span}(U_\perp) = \{ u\mid  u^\T\Psi(x,t)\text{ is free of }t\}$. Specifically, if $ u\in\text{span}(U_\perp)$, we have $ u^\T\Psi(x,t) = u^\T \Gamma(x)$ for any $(x,t)\in\mathbb{X}\times \mathbb{T}$.
\item[(ii)] Let $ f_j = (1, 0_{j - 1}^\T,-1, 0_{k-j-1}^\T)^\T\in \RR^{k}$ for $j = 1,\dots, k - 1$, and $F = [f_1,\dots, f_{k-1}]$. The $\text{span}(U)$ and $\text{span}(U_\perp)$ can be represented as follows:
\begin{enumerate}
\item $\text{span}(U) = \text{span}\big\{ v_T \otimes  v\big\}$ where \text{(i) }$ v_T = \tilde{ \beta}^\T F;$\, \text{(ii) }$\tilde{ \beta} \text{ and }  v$ can be any vectors in  $\RR^{k - 1}$ and $\RR^{K/k }$, respectively. Here $\text{span}\big\{ v_T \otimes  v\big\}$ represents  the linear subspace spanned by all vectors taking the form of  $v_T \otimes  v$.
\item $\text{span}(U_\perp) = \big\{ 1_{k}\otimes  v\mid  v_{}\text{ can be any vector in }\RR^{K/k}\big\}$.
\end{enumerate}
With the dimensions of $\text{span}(U)$ and $\text{span}(U_\perp)$ being specified, we can further conclude that $$\zeta = K - K/k.$$
\item[(iii)] When $\varpi(t\mid x)$ is free of $x$, i.e., $T$ is completely random, one has $$\sigma_{\zeta} = \sigma_{K - K/k}\succsim 1.$$
\item[(iv)] Suppose $\hat{\Gamma}(x)$ is trained via the methods according to $\mathsection$\ref{rk:gamma:est}. We have $U_{\perp}\{\Gamma(x) - \hat{\Gamma}(x)\} = 0$ for any $x\in\mathbb{X}$.
\end{enumerate}
\end{lemma}
\begin{proof}[of Lemma \ref{lm:svdR}]
 We prove the four parts of Lemma \ref{lm:svdR} in order.
\par
\noindent\fbox{Proof of Lemma \ref{lm:svdR} (i)}
By the basic property of singular value decomposition, we have $ u\in\text{span}(U_\perp)$ if and only if $ u^\T R_n  u = 0$. Then, if $ u\in\text{span}(U_{\perp})$, one has,
\bee\label{lm:svd:equiv}
E\Big[\big\{{u}^\T\bpsi(X,T) - \E[{ u}^\T\bpsi(X,T)\mid X]\big\}^2\Big] =  u^\T R_n  u = 0,
\ee
which is equivalent to that  
\bee\label{core:lemma1}
 u^\T\Psi(X,T) = \E\{u^\T\Psi(X,T)\mid X\}
\ee a.s.. This implies that, $$ u^\T\Psi(x,t) = \E\{u^\T\Psi(X,T)\mid X = x\}$$ holds almost everywhere on the Lebesgue measure over $[0,1]^{d + 1}$, since $f(x,t) = f(x)f(x\mid t)$ is an upper and lower bounded density function over $\mathbb{X}\times\mathbb{T} = [0,1]^{d + 1}$, under Assumptions \ref{A:CS}, \ref{am:compact} and \ref{am:densX}. By the continuity of the B-spline function ($\mathsection$\ref{sec:bspline}), we further have that $$ u^\T\Psi(x,t) = \E\{ u^\T\Psi(X,T)\mid X= x\}$$ holds for any $(x,t)\in[0,1]^{d + 1}$.  In addition, since $\E\{u^\T\Psi(X,T)\mid X = x\}$ does not contain $t$, we conclude from \eqref{core:lemma1} that $ u^\T\Psi(x,t)$ now is a function dependent only on $x$, over $(x,t)\in[0,1]^{d + 1}$. 
\par
On the other hand, suppose $ u$ satisfies that $ u^\T\bpsi(x,t)$ is a function only of $x$ over $[0,1]^{d + 1}$. We then have, $\E\{u^\T\M\Psi(X,T)\mid X = x\} =  u^\T\M\Psi(x,t)$ and thus
\bee\nonumber
E\Big[{u}^\T\bpsi(X,T) - \E\{{u}^\T\bpsi(X,T)\mid X\}\Big]^2  = \E\Big[\{{ u}^\T\bpsi(X,T) -{u}^\T\bpsi(X,T)\}^2\Big] = 0,
\ee
which implies $u^\T R_n  u = 0$ and thus $ u\in\text{span}(U_{\perp})$. Summarizing the two sides of the equivalence shown above, the result (i) is hence proved.
\par
\noindent\fbox{Proof of Lemma \ref{lm:svdR} (ii)} Let $S_{U_{\perp}} = \big\{ 1_{k} \otimes  v\mid  v\text{ can be any vector in }\RR^{K/k}\big\}$ and $S_{U} = \text{span}\big\{v_T \otimes  v\big\}$ where $\text{(i) } v_T = \tilde{ \beta}^\T F;\,\text{(ii) }\tilde{ \beta} \text{ and }  v\text{ can be any vectors in } \RR^{k - 1}\text{ and } \RR^{K/k} \text{, respectively}$.  It is easy to verify $S_{U_{\perp}}$ and $S_{U}$ are both linear subspaces in $\RR^{K}$. Specifically, let $1_{k} \otimes  v_1$ and $1_{k} \otimes  v_2$ be two arbitrary vectors in $S_{U_{\perp}}$. We note their linear combination takes the form of 
\bee\nonumber
c_1 (1_{k} \otimes  v_1) + c_2 (1_{k} \otimes  v_2) =  1_{k}\otimes (c_1v_1 + c_2v_2),
\ee
which is also in $S_{U_\perp}$. Similar arguments hold for $S_U$ as well. In addition, by checking the definition, one can also see that $\{ f_{j_1}\otimes  e_{j_2}, j_1 = 1,\dots,k - 1, j_2 = 1,\dots,K/k\}$ and $\{ 1_{k}\otimes  e_{j_2}, j_2 = 1,\dots,K/k\}$ form basis of $S_{U}$ and  $S_{U_\perp}$, respectively. Here $\{e_{j_2}\}_{j_2 = 1}^{K/k}$ are the standard basis of $\RR^{K/k}$. Therefore, $\text{dim}(S_{U}) = K - K/k$, $\text{dim}(S_{U_\perp}) = K/k$, and $\text{dim}(S_{U}) + \text{dim}(S_{U_\perp}) =K$; here $\text{dim}(\cdot)$ is the dimension of the corresponding linear subspace.
\par
We first show $S_{U_\perp} = \text{span}(U_{\perp})$. Denote $\Psi(x) = \psi(x^{(1)})\otimes \cdots\otimes \psi(x^{(d)})$. By Lemma \ref{po:bspline} and the basic property of Kronecker product multiplication \citep[e.g.,][]{horn2012matrix}, we observe for any $ v\in \RR^{K/k}$,
\bee\nonumber
\{ 1_{k} \otimes  v\}^\T \M\Psi(x,t)
&=\{ 1_{k} \otimes  v\}^\T\{\psi(t)\otimes\Psi(x)\}
\\
& =  1^\T_{k}\psi(t)\cdot v^\T\Psi(x) 
\\
&= \sqrt{k} v^\T\Psi(x),
\ee
which is free of $t$, thus $ 1_{k} \otimes  v \in \text{span}(U_\perp)$ and $$S_{U_{\perp}}\subseteq \text{span}(U_{\perp}).$$ Since both $S_{U_{\perp}}$ and $\text{span}(U_\perp)$ are linear subspaces, to show $S_{U_\perp} = \text{span}(U_{\perp})$, it is left to show 
\bee\label{lm1:etarget}
\text{dim}(S_{U_\perp}) = K/k \geq\text{dim}\{\text{span}(U_{\perp})\}.\ee
\par
Now for any $ f_{j_1}\otimes  e_{j_2}$ as one basis function of $S_{U}$, we have for $j_1 = 1,\dots, k - 1$,
\bee\label{fee}
\{ f_{j_1}\otimes  e_{j_2} \}^\T\M\Psi(x,t) = \big\{\psi^{(1)}(t) - \psi^{(j_1 + 1)}(t)\big\}\cdot e_{j_2}^\T\Psi(x).
\ee
The above function depends on $t$ and thus is not in $\text{span}(U_\perp)$ by the result in (i). This is because by Lemma \ref{po:bs:other}, $\psi^{(1)}(t) = 0$ when $t \in (s_{j_1}, s_{j_1 + 1})$, while $\psi^{(1)}(t)  > 0$ when $t \in (s_{j_1},s_{j_1 + 1})$, and thus 
$$
\psi^{(1)}(t) - \psi^{(j_1 + 1)}(t) < 0
$$ when $t\in(s_{j_1},s_{j_1 + 1})$. On the other hand, if $j_1 < k - 1$, we have 
$$
\psi^{(1)}(t) = \psi^{(j_1 + 1)}(t) = 0
$$ when $t \in(s_{j_1 + 1},1)$. If $j_1 = k - 1$, we have $\psi^{(1)}(t) = \psi^{(j_1 + 1)}(t) = 0$ when $t\in(s_1,s_{m})$; note here $m \geq 2$ and $s_m \geq s_1$ under Assumption \ref{am:bspline}. Therefore for any $j_1 = 1,\dots,k - 1$, there exists some $a,b \in[0,1]$ such that $\psi^{(1)}(t = a) - \psi^{j_1 + 1}(t = a) < 0$ while $\psi^{(1)}(t = b) - \psi^{j_1 + 1}(t = b) = 0$, which directly implies from \eqref{fee} that $\{ f_{j_1}\otimes  e_{j_2} \}^\T\M\Psi(x,t)$ is a function that can change with $T$, for any given basis function $ f_{j_1}\otimes  e_{j_2}$ of $S_{U}$.
By results in (i), we know all $K - K/k$ linearly independent vectors in $\{ f_{j_1}\otimes  e_{j_2} \mid j_1= 1,\dots,k - 1, j_2 = 1,\dots,K/k \}$ are not in \text{span}$(U_\perp)$. By the basic property of the linear space, we conclude $\text{dim}\{\text{span}(U_{\perp})\} \leq K - (K - K/k) = \text{dim}(S_{U_{\perp}})$, which verifies \eqref{lm1:etarget} and thus shows $$S_{U_{\perp}} = \text{span}(U_{\perp})$$
\par
Finally, observing that for any basis function $ f_{j_1}\otimes  e_{j_2}$ of $S_{U}$ and any vector $ 1_{k}\otimes  v\in S_{U_{\perp}} = \text{span}(U_{\perp})$, we have
\bee\nonumber
( f_{j_1}\otimes  e_{j_2})^\T \cdot  1_{k}\otimes  v = ( f_{j_1}^\T\cdot  1_{k})\cdot( e_{j_2}^\T\cdot v) = 0,
\ee
as $ f_{j_1} \cdot  1_{k} = 1 - 1 = 0$ by definition. We thus have 
\bee\label{sperp2}
S_{U}\perp\text{span}(U_{\perp})
\ee and $S_{U}\oplus \text{span}(U_{\perp}) = \RR^K$, since $\text{dim}(S_{U}) + \text{dim}\{\text{span}(U_\perp)\} = K - K/k + K/k = K$; here $\oplus$ denotes the direct sum of two linear spaces. With same argument, we can also show that 
\bee\label{sperp1}
\text{span}(U_{})\perp\text{span}(U_{\perp})
\ee and $\text{span}(U_{})\oplus \text{span}(U_{\perp}) = \RR^K$. Since the orthogonal complement of $\text{span}(U_{\perp})$ in $\RR^K$ is unique, by \eqref{sperp2}--\eqref{sperp1}, we conclude that,
$$
\text{span}(U_{}) = S_{U},
$$
which completes the proof. Since $\text{dim}\{\text{span}(U)\} = K - \text{dim}\{\text{span}(U_\perp)\} = K - K/k$, we know the number of column vectors in $U$ is $K - K/k$, and thus $\zeta = K - K/k$.
\par
\noindent\fbox{Proof of Lemma \ref{lm:svdR} (iii)} We first present the general Weyl's inequality of matrix eigenvalue perturbation, which will be frequently used in the paper. The proof can be found in \citet[Theorem 3.3.16]{horn2012matrix}.
\begin{proposition}[General Weyl's inequality]\label{po:gweyl}
Let ${\mathcal{R}}$ be any matrix in $\RR^{d_1\times d_2}$ and $\hat{\mathcal{R}}$ be its perturbed version such that
$
\hat{\mathcal{R}} = {\mathcal{R}} + \mathcal{E}.
$ For any $i,j$ such that $1\leq i,j\leq \min\{d_1,d_2\}$ and $i + j \leq \min\{d_1,d_2\} + 1$, we have
\bee\label{wy:res:1}
\sigma_{i+j - 1}(\hat{\mathcal{R}}) - \sigma_{i}({\mathcal{R}}) \leq \sigma_j(\mathcal{E}).
\ee
Specifically, for any $i\leq \min\{d_1,d_2\}$, \bee\label{wy:res:2}
|\sigma_i(\hat{\mathcal{R}}) - \sigma_i({\mathcal{R}})| \leq \|\mathcal{E}\|_2.
\ee
\end{proposition} 
We now get into our main proof. Note by the law of total expectation, we have
\bee\nonumber
Q_n = R_n + \E(\M\Gamma\M\Gamma^\T) ;
\ee 
see \eqref{rew:Rn} in what follows. For simplicity, we denote $\gamma = E\{\psi(T)\mid X = x\}$. Note $\gamma$ is free of $x$ as $\varpi(t\mid x)$ is free of $x$. We thus also write $\varpi(t\mid x) = \varpi(t)$ for abbreviation.  By the basic property of Kronecker product \citep{schacke2004kronecker},  we then have
\bee\label{ggkron}
\E(\M\Gamma\M\Gamma^\T) &= E\Big[\big[E\{\psi(T)\mid X\}\otimes \Psi(X)\big]\big[E\{\psi(T)\mid X\}\otimes \Psi(X)\big]^\T\Big]
\\
&=E\big\{\gamma\gamma^\T \otimes \Psi(X)\Psi^\T(X)\big\}
\\
&= \gamma\gamma^\T  \otimes E\big\{ \Psi(X)\Psi^\T(X)\big\}.
\ee
Since $\gamma\gamma^\T$ is a rank-one matrix, by e.g. \citet[Theorem 4.2.12]{horn2012matrix} and \eqref{ggkron}, we have 
\bee\nonumber
\text{rank}\big\{\E(\M\Gamma\M\Gamma^\T)\big\} &\leq  1\cdot \text{rank}\big[E\big\{ \Psi(X)\Psi^\T(X)\big\}\big] 
\\
& \leq K/k,
\ee
which implies $\sigma_{K/k + 1}\{\E(\M\Gamma\M\Gamma^\T)\} = 0$. Taking $\hat{\mathcal{R}} = Q_n$, ${\mathcal{R}} = R_n$, $\mathcal{E} = E(\Gamma\Gamma^\T)$, $j = K/k + 1$, and $i = K - K/k$ in Proposition \ref{po:gweyl}, we have that,
\bee\nonumber
\sigma_{K - K/k}(R_n) &\geq \sigma_{K}(Q_n)-\sigma_{K/k + 1}\{E(\Gamma\Gamma^\T)\}
\\
&= \sigma_{K}(Q_n)
\\
&\succsim 1,
\ee
where the last inequality is because the smallest singular value of $Q_n$ is bounded away from $0$; see Lemma~\ref{am:svb}.
\par
\noindent\fbox{Proof of Lemma \ref{lm:svdR} (iv)} By the result in (i), we have $U_\perp\Psi(x,t)$ is free of $t$. We then have
\bee\nonumber
U_{\perp}\{\Gamma(x) - \hat{\Gamma}(x)\} &= E\{U_{\perp}\Psi(X,T)\mid X = x\} - E_{\hat{\varpi}}\{U_{\perp}\Psi(X,T)\mid X = x\}
\\
&=U_{\perp}\Psi(x,t) - U_{\perp}\Psi(x,t)
\\
&= 0.
\ee
\end{proof}
\begin{lemma}\label{am:svb}
 Assumptions \ref{A:CS}, \ref{am:bspline}, \ref{am:densX}  hold. When $n\rightarrow +\infty$, we have following bounds.
\begin{enumerate}
\item[(i)] The eigenvalues of ${Q}_n$ are bounded away from $0$ and $+\infty$, and $\|{R}_n\|_2,\|\E(\M\Gamma\M\Gamma^\T)\|_2,\|E[E\{\psi(T)\mid X\}E\{\psi^\T(T)\mid X\}]\|_2\precsim 1$.
\item[(ii)] For any $ h\in \Lambda(p,c,\mathbb{X}\times\mathbb{T})$ for some fixed $p,c > 0$, let
$$
\phi_h = Q_n^{-1}E\{h\cdot\Psi(X,T)\}.
$$
 We have $\|\phi_h\|\precsim 1$.
\end{enumerate}   
\end{lemma} 
\begin{proof}[of Lemma \ref{am:svb}] We prove the two parts of Lemma \ref{am:svb} in order.
\par
\noindent\fbox{Proof of Lemma \ref{am:svb} (i)} 
 First we note that by the forms of $Q_n$, $R_n$, $\E(\M\Gamma\M\Gamma^\T)$, and $E\big[E\{\psi(T)\mid X\}E\{\psi^\T(T)\mid X\}\big]$, it is clear to see they are all symmetric and positive semi-definitive. 
 \par
 Let $ v\in \RR^{K}$ be any vector with $\| v\| = 1$.  One has
\bee\label{uppperbound:Q}
&\lambda_{\max}(Q_n) 
\\
&= \sup_{\| v\| = 1}\big| v^\T E\big[\bpsi(T,X)\bpsi(T,X)^\T\big] v\big|
\\
&=\sup_{\| v\| = 1}\int_{\mathbb{X}\times \mathbb{T}} \big\{{ v}^{\T}\bpsi(x,t)\big\}^2 f(t\mid x)f(x)dxdt
\\
&\leq C_f/\epsilon\cdot\sup_{\| v\| = 1} \int_{\mathbb{X}\times \mathbb{T}} \big\{{v}^{\T}\bpsi(x,t)\big\}^2 dxdt \quad\text{(Assumptions \ref{am:bspline} and \ref{am:densX})}
\\
&\precsim 1,
\ee
where the last inequality follows by Lemma \ref{am:psi}. Similarly, we have $\lambda_{\min}(Q_n)  \geq c_f/\epsilon\cdot \inf_{\| v\| = 1} \int_{\mathbb{X}\times \mathbb{T}} \big\{{ v}^{\T}\bpsi(x,t)\big\}^2 dxdt \succsim 1$ also by Lemma \ref{am:psi} and the corresponding assumptions.
\par
 By the property of spectral norm  \citep[e.g.,][]{golub2013matrix}, one has
\bee\nonumber
\|\E(\M\Gamma\M\Gamma^\T)\|_2 &= \sup_{\|{ u}\| = 1}\big|{ u}^\T\E(\M\Gamma\M\Gamma^\T){ u}\big|
\\
&=\sup_{\|{ u}\| = 1}\Big|E\Big[\big[E\{{ u}^\T\M\Psi(X,T)\mid X\}\big]^2\Big]\Big|
\\
&\leq \sup_{\|{ u}\| = 1}\Big|E\Big[E\big[\{{ u}^\T\M\Psi(X,T)\}^2\mid X\big]\Big]\Big|
\\
& = \sup_{\|{ u}\| = 1}\Big|E\big[\{{ u}^\T\M\Psi(X,T)\}^2\big]\Big|
\\
&=\|Q_n\|_2 \precsim 1,
\ee
where the first inequality follows by Cauchy-Schwarz inequality and the last inequality follows by \eqref{uppperbound:Q}. The $\big\|E\big[E\{\psi(T)\mid X\}E\{\psi^\T(T)\mid X\}\big]\big\|_2$ can be bounded by similar arguments. This is because $\|E(\M\Gamma\M\Gamma^\T)\|$ is actually the same type of matrix as $\E(\Gamma\Gamma^\T)$, which only replaces the $\M\Gamma(x) = \E\{\M\Psi(X,T)\mid X=x\}$ with $\E\{\psi(T)\mid X=x\}$. Finally, rewrite $R_n$ as
\bee\label{rew:Rn}
R_n &= E\big\{(\Psi - \M\Gamma)\{\Psi - \M\Gamma\}^\T\big\} 
\\
&=Q_n  - \E\{\M\Psi(X,T)\M\Gamma^\T(X)\} - \E\{\M\Gamma(X)\M\Psi^\T(X,T)\} +  \E[\M\Gamma\M\Gamma^\T]
\\
&= Q_n - \E[\M\Gamma\M\Gamma^\T],
\ee 
where the last equality follows by $\E\{\M\Psi(X,T)\M\Gamma^\T(X)\} = E\big[\E\{\M\Psi(X,T)\mid X\}\Gamma^\T(X)\big] = \E(\M\Gamma\M\Gamma^\T)$ due to the law of total expectation, and similarly $\E\{\M\Gamma(X)\M\Psi^\T(X,T)\} = \E(\M\Gamma\M\Gamma^\T)$. Summarizing the above upper bounds, one has $\|R_n\|_2 \leq \|Q_n\|_2+ \|\E(\M\Gamma\M\Gamma^\T)\|_2 \precsim 1$.
\par
\noindent\fbox{Proof of Lemma \ref{am:svb} (ii)}  Recalling \eqref{def:phi*}, if $ h\in \Lambda(p,c,\mathbb{X}\times\mathbb{T})$ for some $p,c > 0$, we have $\| h \|_{\mathbb{X}\times \mathbb{T}} \precsim 1$ and thus
\bee\nonumber
\|\phi^*\|&\leq \big\|Q_n^{-1}\big\|_2\big\|E\big[h(X,T)\Psi(X,T)\big]\big\|
\\
&=\big\|Q_n^{-1}\big\|_2\sup_{\|u\| = 1}\big|E\big[h(X,T)u^\T\Psi(X,T)\big]\big|
\\
&\leq \|h\|_{\mathbb{X}\times\mathbb{T}}\|Q_n^{-1}\|_2\sup_{\|u\| = 1}\sqrt{E\big[u^\T\Psi(X,T)\big]^2}
\\
&\precsim 1,
\ee
where the second inequality follows by Cauchy-Schwarz inequality, and the last inequality follows from the lower and upper bound of $Q_n$'s eigenvalues.
 \end{proof}

\begin{lemma}\label{lm:approx}
Suppose Assumptions \ref{am:bspline} and \ref{am:moment} hold. We have $\|m(x)\|_{\mathbb{X}}$, $\|\hat{m}(x)\|_{\mathbb{X}}$, $\|\hat{\Gamma}(x)\|_{\mathbb{X}}/\sqrt{K}$, $\|\M\Gamma(x)\|_{\mathbb{X}}/\sqrt{K}$ are all bounded away from $+\infty$ when $n$ grows, wpa1. 
\end{lemma}
\begin{proof}[of Lemma \ref{lm:approx}]
By Assumption \ref{am:moment}, we have
$
\|m\|_{\mathbb{X}} = \sup_{x\in \mathbb{X}}\big|\E(Y\mid X = x)\big| \precsim 1
$
as $n\rightarrow +\infty$. By Lemma \ref{am:psi} and the fact $\|\cdot\|$ in convex,
\bee\nonumber
\|\M\Gamma\|_{\mathbb{X}} &= \sup_{x\in\mathbb{X}}\big\|\E\{\Psi(X,T)\mid X=x\}\big\|
\\
&\leq \sup_{x\in\mathbb{X}}E\big\{\big\|\Psi(X,T)\big\|\mid X=x\big\}
\\
&\leq \sup_{(x,t)\in\mathbb{X}\times\mathbb{T}}\big\|\Psi(x,t)\big\|
\\
&\precsim \sqrt{K}.
\ee
\par
Finally, recalling \eqref{rate:m} and 
\eqref{rate:gamma}, by the triangle inequality we have $$\|\hat{m}\|_{\mathbb{X}} \leq \|m\|_{\mathbb{X}} + \|\hat{m} - m\|_{\mathbb{X}}\precsim 1 + o_\p(1),$$ which implies $\|\hat{m}\|_{\mathbb{X}}\precsim 1$, wpa1. Similar argument also yields $\|\hat{\M\Gamma}\|_{\mathbb{X}}/\sqrt{K} \precsim 1$ wpa1.
\end{proof}

\begin{lemma}\label{lm:QGB}
Suppose Assumptions \ref{A:CS}, \ref{am:bspline}, \ref{am:densX}  hold, and also \eqref{rate:gamma} holds. When $n\rightarrow +\infty$, we have  $\|\bar{G}_{n}-G_n\|_2 = o_\p(r_{\gamma}^2)$, and $\|\bar{R}_n\|_2 \precsim 1$ wpa1.
\end{lemma}
\begin{proof}[of Lemma \ref{lm:QGB}]
 First we decompose, 
\bee\label{GGdiff}
\bar{G}_{n}-G_n &= \p\big\{(\hat{\M\Gamma} - \Psi)(\hat{\M\Gamma} - \Psi)^\T\big\} + \rho  Q_n - E\big\{({\M\Gamma} - \Psi)({\M\Gamma} - \Psi)^\T\big\} - \rho Q_n
\\
& = \bar{R}_n - R_n
\\
&=\p\big\{(\hat{\M\Gamma} - \Psi )(\hat{\M\Gamma} - \Psi\}^\T\big\} - \p\big\{({\M\Gamma} - \Psi )({\M\Gamma} - \Psi)^\T\big\}
\\
&=\p\big\{(\hat{\M\Gamma} - \M\Gamma)(\hat{\M\Gamma} - {\M\Gamma}\}^\T\big\} + \p\big\{(\hat{\M\Gamma} - \M\Gamma )({\M\Gamma} - \Psi\}^\T\big\} + \p\big\{(\M\Gamma - \M\Psi)(\hat{\M\Gamma} -  \M\Gamma\}^\T\big\} 
\\
&+ E\big\{(\M\Gamma - \M\Psi)({\M\Gamma} -  \Psi)^\T\big\}-E\big\{({\M\Gamma} - \Psi )({\M\Gamma} - \Psi)^\T\big\}
\\
&= \p\big\{(\hat{\M\Gamma} - \M\Gamma)(\hat{\M\Gamma} - {\M\Gamma})^\T\big\}.
\ee
Note here $P\big\{(\hat{\M\Gamma} - \M\Gamma )({\M\Gamma} - \Psi)^\T\big\} = P\big\{(\M\Gamma - \M\Psi )(\hat{\M\Gamma} -  \M\Gamma)^\T\big\} =  0$ as
\bee\nonumber
&\p\big[\{\hat{\M\Gamma}(X) - \M\Gamma(X) \}\{{\M\Gamma}(X) - \Psi(X,T)\}^\T\big]
\\
&=\p\Big[\{\hat{\M\Gamma}(X) - \M\Gamma(X) \}\cdot E\big[\{{\M\Gamma}(X) - \Psi(X,T)\}^\T\mid X\big]\Big]
\\
&=\p\big[\{\hat{\M\Gamma}(X) - \M\Gamma(X) \}\{{\M\Gamma}(X) - {\M\Gamma}(X)\}^\T\big]
\\
&= 0,
\ee
where the first equality follows by the law of total expectation. Thus $\|\bar{G}_{n}-G_n\|_2 = \big\|\p\big\{(\hat{\M\Gamma} - \M\Gamma )(\hat{\M\Gamma} - {\M\Gamma})^\T\big\}\big\|_2 = o_{{P}}(r_{\gamma}^2)$ follows  from \eqref{rate:gamma}. In addition, by Lemma \ref{am:svb} and \eqref{GGdiff}, one has
$
\|\bar{R}_n\|_2 \leq \|R_n\|_2 + \|\bar{G}_{n}-G_n\|_2 \precsim 1 + o_{{P}}(1)\precsim 1
$ wpa1 since $r_{\gamma}\precsim 0$.
\end{proof}
\begin{lemma}\label{lm:iG2iG}
Suppose the general settings of Theorem \ref{thm:main}  hold. We have 
\begin{itemize}
\item[(i)] $\|\hat{Q}_n - Q_n\|_2 = \mathcal{O}_{{P}}\big(\sqrt{{K\log n}/{n}}\big)$, $\|\hat{G}_n - G_n\|_2 = \mathcal{O}_{{P}}(\sqrt{K\log n/n})$, and we also have 
$\|\hat{Q}_n\|_2\precsim 1$ and $\|P_n(\Gamma\Gamma^\T)\|_2\precsim 1$, wpa1;
\item[(ii)] $\|\hat{G}^{-1}_n - G_n^{-1}\|_2 = \mathcal{O}_{{P}}(\rho^{-2}\sqrt{K\log n/n})$; 
\item[(iii)]  We have $\|\tilde{\Sigma}^{-1}\|_2 \precsim \beta_n^{-1}$, $\|\tilde{\Sigma}_{\perp}^{-1}\|_2\precsim \rho^{-1}$, $\|\tilde{U}_{\perp}^\T U\|_2  \precsim \rho \beta_n^{-1}$. Additionally assume $\rho \prec \sqrt{K\log n/n}$, we have $\sigma_{\min}(\tilde{U}^\T U) \rightarrow 1$;
\item[(iv)] 
We have $\|\M\Sigma^{-1}\|_2 \precsim \beta_n^{-1}$, $\|\hat{\M\Sigma}^{-1}\|_2 \precsim \beta_n^{-1}$,  $\|\hat{\M\Sigma}_{\perp}^{-1}\|_2\precsim  \rho^{-1}$ wpa1. Additionally assume $\rho \prec \sqrt{K\log n/n}$, we have $\|\hat{U}^\T_{\perp}{U} \|_2 =\mathcal{O}_{{P}}\big(\beta_n^{-1}\sqrt{K\log n/n}\big)$. 
\item[(v)] {\color{black}Recall $\hat{A}_n$, $\hat{B}_n$ in Algorithm \ref{alg:sigma}, and let $A_n = \tilde{U}\tilde{\Sigma}^{-1}\tilde{U}^\T$, $B_n = E\big[\{\Psi(X,T) - \Gamma(X)\}\{\Psi(X,T) - \Gamma(X)\}^\T\{Y - \mu(X,T)\}^2\big]$ be their population counterparts. Further assuming the conditions in the confidence interval part of  Theorem \ref{thm:main}  hold, we have $\|\hat{A}_n\|_2$, $\|\hat{B}_n\|_2$, $\|A_n\|_2$, $\|B_n\|_2$ are all constantly bounded wpa1. In addition, we have, 
\bee\nonumber
\|\hat{A}_n - A_n\|_2 = o_P(1),\quad \|\hat{B}_n - B_n\|_2 = o_P(1).
\ee}
\end{itemize}
\end{lemma}

\begin{proof}[of Lemma \ref{lm:iG2iG}]
During the proofs, we will frequently use several classic matrix concentration and perturbation results. For the completeness, we first present these results and then get into  the main proof.
\begin{proposition}[Rudelson's matrix LLN \citep{rudelson1999random}]\label{rlnn}Let $\mathcal{R}_1,\dots,\mathcal{R}_n \in \RR^{d\times d}$ be i.i.d. random matrices with $d\geq 2$. Suppose $\mathcal{R} = E(\mathcal{R}_i)$ and $\|\mathcal{R}_i\|_2 \leq C$ a.s., for any $i\in[n]$, then
\bee\nonumber
\E\big\|P_n(\mathcal{R}) - \mathcal{R}^*\big\|_2 \precsim \frac{C\log d}{n} + \sqrt{\frac{C\|\mathcal{R}\|_2 \log d}{n}}.
\ee
\end{proposition} 
\begin{proposition}[Weyl's in equality \citep{weyl1912asymptotische}]\label{po:weyl}Let $\hat{\mathcal{R}}$ and $\mathcal{R}$ be $d \times d$ symmetric matrices. We have for any $i \in [d]$,
\bee\nonumber
\lambda_i(\mathcal{R})  +  \lambda_d(\hat{\mathcal{R}} - \mathcal{R}) \leq \lambda_{i}(\hat{\mathcal{R}}) \leq \lambda_i(\mathcal{R}) + \lambda_1(\hat{\mathcal{R}} - \mathcal{R}).
\ee
Therefore, if $\hat{\mathcal{R}} - \mathcal{R}$ is positive semi-definitive or $\hat{\mathcal{R}}\succeq \mathcal{R}$, one has $\lambda_i(\mathcal{R})\leq \lambda_{i}(\hat{\mathcal{R}})$ for any $i \in [n]$.

\end{proposition} 
\begin{proposition}[Davis-Kahan theorem \citep{davis1970rotation}]\label{dk}Let  symmetric matrix $\hat{\mathcal{R}}\in\RR^{d\times d}$ be the perturbed version of a symmetric matrix $\mathcal{R}\in\RR^{d\times d}$ such that 
\bee\nonumber
\hat{\mathcal{R}} = \mathcal{R} + \mathcal{E}.
\ee
Define their singular value decompositions
\bee\nonumber
\hat{\mathcal{R}} &= \hat{\mathcal{U}} \hat{\mathcal{S}}\hat{\mathcal{U}}^\T +  \hat{\mathcal{U}}_\perp \hat{\mathcal{S}}_\perp\hat{\mathcal{U}}_\perp^\T,
\\
{\mathcal{R}} &= {\mathcal{U}} {\mathcal{S}}{\mathcal{U}}^\T +  \hat{\mathcal{U}}_\perp {\mathcal{S}}_\perp{\mathcal{U}}_\perp^\T,
\ee
where $\hat{\mathcal{U}}$ and $\hat{\mathcal{S}}$ correspond to the top-$r$ singular vectors and  top-$r$ singular values of $\hat{\mathcal{R}}$, respectively; similar notation also holds for ${\mathcal{R}}$. We then have
\bee\nonumber
\|\hat{\mathcal{U}}_\perp^\T\mathcal{U}\|_2 \leq \frac{\|\hat{\mathcal{U}}^\T\mathcal{E}\|_2}{\sigma_r(\hat{\mathcal{R}}) - \sigma_{r + 1}({\mathcal{R}})}.
\ee
\end{proposition} 

\noindent\fbox{Proof of Lemma \ref{lm:iG2iG} (i)} By definition one has
\bee\label{dG:main}
\hat{G}_n - \bar{G}_n &= \Big[\p_n\big\{(\Psi - \hat{\M\Gamma})(\Psi - \hat{\M\Gamma})^\T\big\} - \p\big\{(\Psi - \hat{\M\Gamma})(\Psi - \hat{\M\Gamma})^\T\big\}\Big] +\rho\big(\hat{Q}_n - Q_n\big)
\ee
Recalling that by Lemma \ref{lm:approx} and Lemma \ref{lm:QGB}, we have wpa1,
\bee\label{con:nui}
&\|\M\Psi - \hat{\Gamma}\|_{\mathbb{X}\times \mathbb{T}}^2\leq (\|\M\Psi\|_{\mathbb{X}\times \mathbb{T}} + \|\hat{\Gamma}\|_{\mathbb{X}})^2
\\
&\quad\quad\quad\quad\quad\,\precsim K,
\\
&\text{and }\big\|\p\big\{(\Psi - \hat{\M\Gamma})(\Psi - \hat{\M\Gamma})^\T\big\}\big\|_2 \precsim 1.
\ee
Now we first condition on given $\hat{\M\Gamma}(\cdot)$ which satisfies \eqref{con:nui} as $n\rightarrow +\infty$. Since $\hat{\Gamma}(\cdot)$ is trained separately, one has $\{(\Psi_{i} - \hat{\M\Gamma}_{i})(\Psi_{i} - \hat{\M\Gamma}_{i})^\T\}_{i = 1}^n$ are i.i.d. now random matrices and 
\bee\nonumber
&\p\Big[(\Psi_{i} - \hat{\M\Gamma}_{i})(\Psi_{i} - \hat{\M\Gamma}_{i})^\T - \p\big\{(\Psi - \hat{\M\Gamma})(\Psi - \hat{\M\Gamma})^\T\big\}\Big] = 0.
\ee 
By \eqref{con:nui}, we also have
\bee\nonumber
\big\|(\M\Psi_i - \hat{\Gamma}_i)(\M\Psi_i - \hat{\Gamma}_i)^\T\big\|_2  &\leq \sup_{x\in\mathbb{X}}\big\|\M\Psi(x) - \hat{\Gamma}(x)\big\|^2
\\
&= \|\M\Psi - \hat{\Gamma}\|_{\mathbb{X}\times \mathbb{T}}^2 
\\
&\precsim K.
\ee 
By taking $\mathcal{R}_i = (\M\Psi_i - \hat{\Gamma}_i)(\M\Psi_i - \hat{\Gamma}_i)^\T$ in Proposition \ref{rlnn} and given $\hat{\Gamma}(\cdot)$, one has
\bee\label{PPPsiGamma}
\big\|\p_n\big\{(\Psi - \hat{\M\Gamma})(\Psi - \hat{\M\Gamma})^\T\big\} - \p\big\{(\Psi - \hat{\M\Gamma})(\Psi - \hat{\M\Gamma})^\T\big\}\big\|_2 &\precsim \frac{K\log K}{n} + \sqrt{\frac{K\log K}{n}}
\\
&\precsim \sqrt{\frac{K\log n}{n}},
\ee
as $(K\log K)/n \rightarrow 0$ under Assumption \ref{am:Kn}. Since the conditioned event \eqref{con:nui} happens wpa1, we can directly uncondition it and \eqref{PPPsiGamma} implies
\bee\nonumber
\big\|\p_n\big\{(\Psi - \hat{\M\Gamma})(\Psi - \hat{\M\Gamma})^\T\big\} - \p\big\{(\Psi - \hat{\M\Gamma})(\Psi - \hat{\M\Gamma})^\T\big\}\big\|_2 = \mathcal{O}_P\Big(\sqrt{\frac{K\log n}{n}}\Big).
\ee Similarly, we can show $$\|\hat{Q}_n - Q_n\|_2 = \mathcal{O}_{{P}}\big(\sqrt{{K\log n}/{n}}\big).$$ Note this also implies $\|\hat{Q}_n - Q_n\|_2 \rightarrow 0$ wpa1, since $\sqrt{K\log n/n}$ vanishes under Assumption \ref{am:Kn}. Thus we further have $$\|\hat{Q}_n\|_2 \leq  \|Q_n\|_2 + \|\hat{Q}_n - Q_n\|_2 \precsim \|Q_n\|_2 \precsim 1,$$ wpa$1$, by Lemma \ref{am:svb}. Similarly, we can also show $\|P_n(\Gamma\Gamma^\T)\|_2 \precsim 1$.
\par In summary, by \eqref{dG:main} and the bounds derived above, one has
 \bee\nonumber
 \|\hat{G}_n - \bar{G}_n\|_2 &\leq \Big\|\p_n\big\{(\Psi - \hat{\M\Gamma})(\Psi - \hat{\M\Gamma})^\T\big\} - \p\big\{(\Psi - \hat{\M\Gamma})(\Psi - \hat{\M\Gamma})^\T\big\}\Big\|_2 +\rho\big\|\hat{Q}_n - Q_n\big\|_2
 \\
 &= \mathcal{O}_{{P}}\big(\sqrt{K\log n/n}\big)
 \ee whenever $\rho\precsim 1$. Finally, by the triangle inequality, one has
\bee\label{concen:GG}
\|\hat{G}_n - G_n\|_2\leq \|\hat{G}_n - \bar{G}_n\|_2 + \|\bar{G}_n - {G}_n\|_2 =\mathcal{O}_{{P}}(\sqrt{K\log n/n}),
\ee 
follows from Lemma \ref{lm:QGB}, whenever $r_{\gamma}^2\precsim \sqrt{K\log n/n}$.
\par
\noindent\fbox{Proof of Lemma \ref{lm:iG2iG} (ii)} We first show $\hat{G}_n$ is invertible wpa1. By definition, we write 
\bee\label{GnQ}
\hat{G}_n = \p_n\big\{(\hat{\Psi} - \hat{\M\Gamma})(\hat{\Psi} - \hat{\M\Gamma})^\T\big\} + \rho\hat{Q}_n \succeq \rho\hat{Q}_n,
\ee
as both $\p_n\big\{(\hat{\Psi} - \hat{\M\Gamma})(\hat{\Psi} - \hat{\M\Gamma})^\T\big\}$ and $\hat{Q}_n$ are positive semi-definitive matrices. Then \eqref{GnQ} and Proposition \ref{po:weyl} imply that
\bee\label{GrQbound}
{\color{black}\lambda_{\min}(\hat{G}_n) \geq  \lambda_{\min}(\rho\hat{Q}_n) = \rho\lambda_{\min}(\hat{Q}_n).}
\ee
 Recall $\|\hat{Q}_n - Q_n\|_2 \rightarrow 0$ wpa1. By \eqref{wy:res:2} in Proposition \ref{po:gweyl} with $\hat{\mathcal{R}} = \hat{Q}_n$ and $\mathcal{R} = Q_n$,  one has $\lambda_{\min}(\hat{Q}_n)\geq \lambda_{\min}(Q_n) - \|\hat{Q}_n - Q_n\|_2 \succsim 1$ wpa1. Thus by $\lambda_{\min}(\hat{Q}_n) \succsim 1$ and  \eqref{GrQbound} we conclude that,
\bee\label{GGlambda}
\lambda_{\min}(\hat{G}_n)\succsim \rho,
\ee 
wpa1,  thus $\hat{G}_n$ is invertible wpa1. On the other hand, recalling \eqref{Rdecom}, one has $\lambda_{\min}(G_n)\geq \rho\lambda_{\min}({Q}_n)\succsim  \rho$ by Lemma \ref{am:psi}. Now wpa1, we can decompose,
\bee\label{inv:decom}
\hat{G}^{-1}_n - G_n^{-1} = G_n^{-1}(G_n - \hat{G}_n)\hat{G}^{-1}_n,
\ee
which combining with \eqref{GGlambda} implies, wpa1,
\bee\label{decom}
\|\hat{G}^{-1}_n - G_n^{-1}\|_2 &\leq \|G_n^{-1}\|_2\|G_n - \hat{G}_n\|_2\|\hat{G}^{-1}_n\|_2
\\
&=\lambda_{\min}^{-1}(\hat{G}_n)\lambda_{\min}^{-1}(G_n)\|G_n - \hat{G}_n\|_2
\\
&\precsim \rho^{-2}\|\hat{G}_n - G\|_2.
\ee
Finally by \eqref{concen:GG} and \eqref{decom}, we conclude  $\|\hat{G}^{-1}_n - G_n^{-1}\|_2 = \mathcal{O}_{{P}}(\rho^{-2}\sqrt{K\log n/n})$.
\par
\noindent\fbox{Proof of Lemma \ref{lm:iG2iG} (iii)} Recall $G_n = R_n + \rho Q_n$. By definition, it is easy to see both $R_n$ and $\rho Q_n$ are positive semi-definitive matrices. Thus we have $G_n \succeq R_n$, and  by Proposition \ref{po:weyl},
\bee\label{Gnb1}
\sigma_{\zeta}(G_n) \geq \sigma_{\zeta}(R_n) = \beta_n,
\ee 
which implies that, 
\bee\label{tsigmab}
\big\|\tilde{\Sigma}^{-1}\big\|_2 =  \sigma^{-1}_{\zeta}(G_n) \leq \beta_n^{-1}.
\ee On the other hand, since $G_n \succeq \rho Q_n$, we have
\bee\label{Gnb2}
\sigma_{\min}(G_n) \geq \rho \sigma_{\min}(Q_n) \succsim \rho,
\ee
also by Lemma \ref{am:svb}. Therefore, $G_n$ is invertible and $\|\tilde{\Sigma}_{\perp}^{-1}\|_2 = \sigma_{\min}^{-1}(G_n)\precsim \rho^{-1}$. Finally, by taking $\hat{\mathcal{R}} = G_n$, $\mathcal{R}_n = R_n$, and $\mathcal{E} = \rho Q_n$ in Proposition \ref{dk}, we have
\bee\label{dvkbound}
\|\tilde{U}_{\perp}^\T U\|_2  &\leq \frac{\rho\| Q_n\|_2}{\sigma_{\zeta}(G_n) - \sigma_{\zeta + 1}(R_n)} 
\\
&\precsim \rho \beta_n^{-1},
\ee
recalling that $\|Q_n\|_2\precsim 1$, $\sigma_{\zeta}(G_n)\succsim \beta_n$, and $\sigma_{\zeta + 1}(R_n) = 0$ due to Lemma \ref{lm:svdR} and the fact that $R_n$ is rank-$\zeta$. Finally \citet[Lemma 1]{cai2018rate} implies 
\bee\label{dvkbound2}
\sigma_{\min}^2(\tilde{U}^\T U) = 1 - \|\tilde{U}^\T U_{\perp}\|_2. \ee
By $\rho \prec \sqrt{K\log n/n}\prec \beta_n$  and \eqref{dvkbound}-\eqref{dvkbound2}, we have $\sigma_{\min}^2(\tilde{U}^\T U) \rightarrow 1$ as $n\rightarrow +\infty$.
\par
\noindent\fbox{Proof of Lemma \ref{lm:iG2iG} (iv)} Recalling definition \eqref{def:rn}, one has
$
\|\M\Sigma^{-1}\|_2 = \sigma_{\zeta}^{-1} = \beta_n^{-1}.
$
By taking $\hat{\mathcal{R}} = \hat{G}_n$ and $\mathcal{R} = G_n$ in Proposition \ref{po:gweyl}, one has 
\bee\label{sigmasigmaG}
\big|\hat{\sigma}_{\zeta} - \sigma_{\zeta}(G_n)\big| &\leq \|\hat{G}_n - G_n\|_2
\\
& = \mathcal{O}_{{P}}(\sqrt{K\log n/n}).
\ee
On the other hand, recalling \eqref{Gnb1} and \eqref{Gnb2}, we have
\bee\nonumber
\sigma_{\zeta}(G_n)\succsim \beta_n + \rho,
\ee
which combining with \eqref{sigmasigmaG}, implies $$\hat{\sigma}_{\zeta} \asymp \beta_n + \rho,$$ wpa1, under the assumed condition $\beta_n + \rho \succ \sqrt{K\log n/n}$. This impies 
\bee\label{tsigmab2}
\|\hat{\Sigma}^{-1}\|_2 = \hat{\sigma}_{\zeta}^{-1}\precsim (\beta_n + \rho)^{-1}
\ee wpa1. By \eqref{GGlambda}, we  have $\|\hat{\M\Sigma}_\perp^{-1}\|_2=  \lambda_{\min}^{-1}(\hat{G}_n)\precsim \rho^{-1}$ wpa1.
\par
By taking $\hat{\mathcal{R}} = \hat{G}_n$ and $\mathcal{R} = R_n$ in  Proposition \ref{dk}, we have
\bee\nonumber
\|\hat{U}_{\perp}^\T {U}\|_2  &\leq \frac{\|\hat{G}_n - R_n\|_2}{\lambda_{\zeta}(\hat{G}_n) - \lambda_{\zeta + 1}(R_n)} 
\\
&\leq \frac{\|\hat{G}_n - G_n\|_2 + \rho\|{Q}_n\|_2}{\hat{\sigma}_{\zeta}}
\\
&=\mathcal{O}_{{P}}\big\{(\beta_n + \rho)^{-1}(\sqrt{K\log n/n} + \rho)\big\},
\ee
recalling that $\hat{\sigma}_{\zeta}\asymp  \beta_n + \rho$, $\|Q_n\|_2 \precsim 1$, and $\|\hat{G}_n - G_n\|_2 = \mathcal{O}_{{P}}(\sqrt{K\log n / n})$. The final results follow, after taking the assumed condition $\rho\prec\sqrt{K\log n/n}\prec \beta_n$ into account.
\par
\noindent\fbox{Proof of Lemma \ref{lm:iG2iG} (v)} {\color{black}First when $\beta_n\asymp 1$, clearly we have
\bee\label{Anbound}
\|A_n\|_2 \leq \|\tilde{U}\|_2^2\|\tilde{\Sigma}^{-1}\|_2 \precsim  \beta_n^{-1}\precsim 1,
\ee
by \eqref{tsigmab}. On the other hand, we have
\bee\nonumber
\hat{A}_n - A_n &= (\hat{U}\hat{\Sigma}\hat{U}^\T)^{-1} - (\tilde{U}\tilde{\Sigma}\tilde{U}^\T)^{-1}
\\
&= (\hat{U}\hat{\Sigma}\hat{U}^\T)^{-1}(\tilde{U}\tilde{\Sigma}\tilde{U}^\T - \hat{U}\hat{\Sigma}\hat{U}^\T )(\tilde{U}\tilde{\Sigma}\tilde{U}^\T)^{-1},
\ee
and thus
\bee\nonumber
\|\hat{A}_n - A_n\|_2 &\leq \|\hat{\Sigma}^{-1}\|_2\|\tilde{\Sigma}^{-1}\|_2\|\tilde{U}\tilde{\Sigma}\tilde{U}^\T - \hat{U}\hat{\Sigma}\hat{U}^\T\|_2
\\
&\precsim \|\tilde{U}\tilde{\Sigma}\tilde{U}^\T - \hat{U}\hat{\Sigma}\hat{U}^\T\|_2,
\ee
by \eqref{tsigmab} and \eqref{tsigmab2}. In addition, one has
\bee\label{diffvvv}
\|\tilde{U}\tilde{\Sigma}\tilde{U}^\T - \hat{U}\hat{\Sigma}\hat{U}^\T\|_2 &\leq 2\|\hat{G}_n - G_n  + \tilde{U}_{\perp}\tilde{\Sigma}_{\perp}\tilde{U}_{\perp}^\T \|_2
\\
&\leq 2\|\hat{G}_n - G_n \|_2 + 2\|\tilde{U}_{\perp}\tilde{\Sigma}_{\perp}\tilde{U}_{\perp}^\T\|_2 
\\
&= \mathcal{O}_{{P}}(\sqrt{K\log n / n} + \rho)
\\
&=o_P(1),
\ee
where the last two equalities are by the previous derived bounds, and the conditions that $K\precsim n^{a_2}$ and $\rho\rightarrow 0$ for some $a_2 < 1$. The first inequality of \eqref{diffvvv} follows by that $\hat{G}_n$ can be seen as a perturbed version of $\tilde{U}\tilde{\Sigma}\tilde{U}^\T$ such that
\bee\nonumber
\hat{G}_n = \tilde{U}\tilde{\Sigma}\tilde{U}^\T + \big(\hat{G}_n - G_n  + \tilde{U}_{\perp}\tilde{\Sigma}_{\perp}\tilde{U}_{\perp}^\T \big).
\ee
Thus as the best rank-$\zeta$ approximation of $\hat{G}_n$, the $\hat{U}\hat{\Sigma}\hat{U}^\T$ satisfies the first inequality of \eqref{diffvvv} by Eckart--Young--Mirsky theorem \citep{eckart1936approximation} such that,
\bee\nonumber
\|\hat{G}_n - \hat{U}\hat{\Sigma}\hat{U}^\T\|_2 &\leq \|\hat{G}_n - \tilde{U}\tilde{\Sigma}\tilde{U}^\T\|_2 
\\
&= \|\hat{G}_n - G_n  + \tilde{U}_{\perp}\tilde{\Sigma}_{\perp}\tilde{U}_{\perp}^\T\|_2,
\ee
and thus $\|\tilde{U}\tilde{\Sigma}\tilde{U}^\T - \hat{U}\hat{\Sigma}\hat{U}^\T\|_2 \leq \|\tilde{U}\tilde{\Sigma}\tilde{U}^\T - \hat{G}_n\|_2+ \|\hat{G}_n - \hat{U}\hat{\Sigma}\hat{U}^\T\|_2\leq 2\|\hat{G}_n - G_n  + \tilde{U}_{\perp}\tilde{\Sigma}_{\perp}\tilde{U}_{\perp}^\T\|_2$. Summarizing the results above we have
\bee\nonumber
\|\hat{A}_n - A_n\|_2 = o_P(1), 
\ee
and thus by \eqref{Anbound}, $\|\hat{A}_n\|_2 \leq \|A_n\|_2 + \|\hat{A}_n - A_n\|_2 \precsim 1$ wpa1.
\par
We note $\|\mu\|_{\mathbb{X}\times \mathbb{T}} < +\infty$ under Assumption \ref{am:moment} (iii). By the new condition that
\bee\label{mumucon}
\|\hat{\mu} - {\mu}\|_{\mathbb{X}\times\mathbb{T}} = o_P(1),
\ee
we have wpa1,
\bee\label{hmuc}
\|\hat{\mu}\|_{\mathbb{X}\times\mathbb{T}} &\leq \|{\mu}\|_{\mathbb{X}\times\mathbb{T}} + \|\hat{\mu} - {\mu}\|_{\mathbb{X}\times\mathbb{T}}
\\
&\precsim 1.
\ee
Under our simplified two-fold training setting such that one fold trains nuisance functions and one fold trains the proposed estimator and $\hat{\sigma}$, we can write $\hat{B}_n$ as 
\bee\nonumber
\hat{B}_n = \frac{1}{n}\sum_{i = 1}^n\{Y_i - \hat{\mu}(X_i,T_i)\}^2\{\Psi(X_i,T_i) - \hat{\Gamma}(X_i)\}\{\Psi(X_i,T_i) - \hat{\Gamma}(X_i)\}^\T.
\ee
We also define 
$$
\bar{B}_n = P\Big[\{Y - \hat{\mu}(X,T)\}^2\{\Psi(X,T) - \hat{\Gamma}(X)\}\{\Psi(X,T) - \hat{\Gamma}(X)\}^\T\Big].
$$
Similar to \eqref{con:nui} and \eqref{PPPsiGamma}, we show the convergence of $\hat{B}_n$ by matrix concentration. It is easy to see $P(\hat{B}_n) = \bar{B}_n$, with $\hat{\Gamma}(X)$ trained separately. Based on Assumption \ref{am:moment}, Lemma \ref{lm:QGB} and \eqref{mumucon}, we have, wpa1,
\begin{align}\nonumber
\|\bar{B}_n\|_2 &= \sup_{\|\ell\| = 1}P\Big[\{Y - \hat{\mu}(X,T)\}^2 \big[\ell^\T\{\Psi(X,T) - \hat{\Gamma}(X)\}\big]^2\Big]
\\\nonumber
&\leq \sup_{\|\ell\| = 1}P\Big[\big[2\{Y - {\mu}(X,T)\}^2 + 2\{{\mu}(X,T) - \hat{\mu}(X,T)\}^2\big]\big[\ell^\T\{\Psi(X,T) - \hat{\Gamma}(X)\}\big]^2\Big]
\\\nonumber
&=\sup_{\|\ell\| = 1}P\Big[\big[2\text{Var}(Y\mid X,T) + 2\{{\mu}(X,T) - \hat{\mu}(X,T)\}^2\big]\big[\ell^\T\{\Psi(X,T) - \hat{\Gamma}(X)\}\big]^2\Big]
\\\nonumber
&\leq \Big\{\sup_{(x,t)\in\mathbb{X}\times\mathbb{T}}2\text{Var}(Y\mid X = x, T = t) + 2\|\hat{\mu}- {\mu}\|_{\mathbb{X}\times\mathbb{T}}\Big\}\sup_{\|\ell\| = 1}P\Big[\big[\ell^\T\{\Psi(X,T) - \hat{\Gamma}(X)\}\big]^2\Big]
\\\nonumber
&\precsim \sup_{\|\ell\| = 1}P\Big[\big[\ell^\T\{\Psi(X,T) - \hat{\Gamma}(X)\}\big]^2\Big]
\\\nonumber
&=\big\|\p\big\{(\Psi - \hat{\M\Gamma})(\Psi - \hat{\M\Gamma})^\T\big\}\big\|_2 
\\\nonumber
&\precsim 1,
\end{align}
where the last inequality follows by \eqref{con:nui}. By Proposition \ref{rlnn}, one has
\bee\nonumber
\|\hat{B}_n - \bar{B}_n\|_2 = \mathcal{O}_P\Big(\sqrt{\frac{K\log n}{n}}\Big) = o_P(1).
\ee
In addition, we have
\bee\nonumber
&\bar{B}_n - B_n
\\
& = P\Big[\big[\{Y - \hat{\mu}(X,T)\}^2 - \{Y - {\mu}(X,T)\}^2\big]\{\Psi(X,T) - \hat{\Gamma}(X)\}\{\Psi(X,T) - \hat{\Gamma}(X)\}^\T\Big]
\\
& + P\Big[ \{Y - {\mu}(X,T)\}^2\{\Psi(X,T) - {\Gamma}(X)\}\{{\Gamma}(X) - \hat{\Gamma}(X)\}^\T\Big]
\\
& + P\Big[ \{Y - {\mu}(X,T)\}^2\{{\Gamma}(X) - \hat{\Gamma}(X)\}\{\Psi(X,T) - \hat{\Gamma}(X)\}^\T\Big]
\\
& = \Delta_{B,1} + \Delta_{B,2} + \Delta_{B,3}.
\ee
We bound the spectral norms of three terms on the right-hand side of above display, respectively. We have
\bee\nonumber
\Delta_{B,1} = P\Big[\{2Y - \hat{\mu}(X,T)-{\mu}(X,T)\}\{{\mu}(X,T) - \hat{\mu}(X,T) \}\{\Psi(X,T) - \hat{\Gamma}(X)\}\{\Psi(X,T) - \hat{\Gamma}(X)\}^\T\Big],
\ee
and thus wpa1,
\bee\nonumber
&\|\Delta_{B,1}\|_2
\\
&=\sup_{\|\ell\| = 1}P\Big[\{2Y - \hat{\mu}(X,T)-{\mu}(X,T)\}\{{\mu}(X,T) - \hat{\mu}(X,T) \}[\ell^\T\{\Psi(X,T) - \hat{\Gamma}(X)\}]^2\Big]
\\
&=\sup_{\|\ell\| = 1}P\Big[\{2E(Y\mid X,T) - \hat{\mu}(X,T)-{\mu}(X,T)\}\{{\mu}(X,T) - \hat{\mu}(X,T) \}[\ell^\T\{\Psi(X,T) - \hat{\Gamma}(X)\}]^2\Big]
\\
&\leq \Big\{\|{\mu}\|_{\mathbb{X}\times\mathbb{T}} +  \|\hat{\mu}\|_{\mathbb{X}\times\mathbb{T}}\Big\}\|{\mu} - \hat{\mu}\|_{\mathbb{X}\times\mathbb{T}}\cdot\sup_{\|\ell\| = 1}P\Big[[\ell^\T\{\Psi(X,T) - \hat{\Gamma}(X)\}]^2\Big]
\\
&=o_P(1),
\ee
where the last inequality follows by the moment conditions in Assumption \ref{am:moment}, \eqref{hmuc}, $\|{\mu} - \hat{\mu}\|_{\mathbb{X}\times\mathbb{T}} = o_P(1)$, and \eqref{con:nui}. With similar arguments, we can further show
\bee\nonumber
\|\Delta_{B,2}\|_2 &= \sup_{\|\ell\| = 1}\Big|P\Big[\{Y - {\mu}(X,T)\}^2 [\ell^\T\{\Psi(X,T) - {\Gamma}(X)\}][\ell^\T\{{\Gamma}(X) - \hat{\Gamma}(X)\}]\Big]\Big|
\\
&\leq\sup_{\|\ell\| = 1}P\Big[\{Y - {\mu}(X,T)\}^2 \Big|[\ell^\T\{\Psi(X,T) - {\Gamma}(X)\}][\ell^\T\{{\Gamma}(X) - \hat{\Gamma}(X)\}]\Big|\Big]
\\
&\leq\sup_{(x,t)\in\mathbb{X}\times\mathbb{T}}\text{Var}(Y\mid X = x, T = t)\sup_{\|\ell_1\| = 1}P\Big[\ell_1^\T\{\Psi(X,T) - {\Gamma}(X)\}\Big]^2\sup_{\|\ell_2\| = 1}P\Big[\ell_2^\T\{{\Gamma}(X) - \hat{\Gamma}(X)\}\Big]^2
\\
&\leq\sup_{(x,t)\in\mathbb{X}\times\mathbb{T}}\text{Var}(Y\mid X = x, T = t)\|R_n\|_2\|P[\{{\Gamma}(X) - \hat{\Gamma}(X)\}\{{\Gamma}(X) - \hat{\Gamma}(X)\}^\T]\|_2
\\
&=o_P(1),
\ee
where the second inequality follows by Cauchy--Schwarz inequality, and the last equality follows by \eqref{rate:gamma2}. Similarly, we can also show $\|\Delta_{B,3}\|_2 = o_P(1)$ and $\|B_n\|_2 \precsim 1$. Summarizing all results above, we conclude wpa1,
\bee\nonumber
\|\hat{B}_n - {B}_n\|_2 &\leq \|\hat{B}_n - \bar{B}_n\|_2 + \|\Delta_{B,1}\|_2 + \|\Delta_{B,2}\|_2 + \|\Delta_{B,3}\|_2
\\
&=o_P(1),
\ee
and furthermore $\|\hat{B}_n\|_2 \leq \|B_n\|_2 + \|\hat{B}_n - {B}_n\|_2 \precsim 1$ wpa1.}
\end{proof}
\begin{lemma}\label{l3}
Suppose the general settings of Theorem \ref{thm:main}  hold. Define $\Delta_{1,1}$--$\Delta_{1,5}$ in \eqref{delta:decom1} through \eqref{delta:decom2}. 
\begin{itemize}
\item[(i)] Suppose $\ell_n\in\RR^K$ is a  vector that {can depend on} $\{(X_i,T_i)\}_{i = 1}^n$ when $n$ grows, and $\|\ell_n\| = 1$ for any $n > 0$. We have  $\|\Delta_{1,1}\| = \mathcal{O}_\p(\sqrt{K/n})$ and  $|\ell_n^\T \Delta_{1,1}| = \mathcal{O}_P(1/\sqrt{n})$;
\item[(ii)] Suppose $\ell_n\in\RR^K$ is a vector that  depends only on $n$, and $\|\ell_n\| = 1$ for any $n > 0$. We have  $\|\Delta_{1,2}\| = \mathcal{O}_P(K^{-p/(d + 1)} )$ and $|\ell_n^\T \Delta_{1,2}| = \mathcal{O}_P(K^{-p/(d + 1)})$;
\item[(iii)] Suppose $\ell_n\in\RR^K$ is a vector that  depends only on $n$, and $\|\ell_n\| = 1$ for any $n > 0$. We have  $\|\Delta_{1,3}\| = o_{{P}}(r_{\gamma}'\sqrt{K/n})$ and  $|\ell_n^\T \Delta_{1,3}| = o_P(r_\gamma /\sqrt{n})$;
\item[(iv)] Suppose $\ell_n\in\RR^K$ is a  vector that  depends  only on $n$, and $\|\ell_n\| = 1$ for any $n > 0$. We have  $\|\Delta_{1,4}\| = o_{{P}}(r_m\sqrt{K/n} + r_\gamma\sqrt{K/n})$ and  $|\ell_n^\T\Delta_{1,4}| = o_P(1/\sqrt{n})$;
\item[(v)] Suppose $\ell_n\in\RR^K$ is a vector that  depends  only on $n$, and $\|\ell_n\| = 1$ for any $n > 0$. We have  $\|\Delta_{1,5}\| = o_P(r_{m}r_{\gamma} + r^2_{\gamma} + {r_{\gamma}'} r_m\sqrt{K/n} + {r_{\gamma}'} r_\gamma\sqrt{K/n} )$ and  $|\ell_n^\T\Delta_{1,5}| = o_P(r_{\gamma}/\sqrt{n} + r_{m}r_{\gamma} + r^2_{\gamma})$.
\end{itemize}
\end{lemma}
\begin{proof}[of Lemma \ref{l3}]
We first present the following concentration result which will be repeatedly used during the proof. This proposition can be simply proved by high-order Markov inequality.
\begin{proposition}\label{clm1}
Given $n$, let $ v_1,\dots, v_n$ are i.i.d. copies of random vector $ v\in \RR^{d}$. Suppose $\E(\| v\|^2)\leq A_n$, where $A_n > 0$ is allowed to diverge. We have for any $J > 0$,
\bee\nonumber
\Pr\Big\{\big\|{P}_n( v) - {E}(v)\big\| > J \sqrt{d/n}\Big\} \leq \frac{4A_n}{J^2 d}.
\ee
\end{proposition}
\begin{proof}[of Proposition \ref{clm1}]
By the second-order Markov's inequaltiy, we have for any $J > 0$,
\bee\label{clm1:1}
\Pr\Big\{\big\|{P}_n(v) - {E}(v)\big\|>J \sqrt{d/n}\Big\} \leq \frac{n}{J^2 d}E\big\{\|{P}_n(v) - {E}(v)\|^2\big\}.
\ee 
Now we bound $E\big\{\|{P}_n(v) - {E}(v)\|^2\big\}$. We can write
\bee\label{clm1:2}
E\big\{\|{P}_n(v) - {E}( v)\|^2\big\} &= E\Big[\Big\|\frac{1}{n}\sum_{i = 1}^n\{v_i - {E}( v)\}\Big\|^2\Big]
\\
&=\frac{1}{n^2}E\Big[\Big\|\sum_{i = 1}^n\{v_i - {E}(v)\}\Big\|^2\Big]
\\
&=\frac{1}{n}E\big\{\big\| v_i - {E}(v)\big\|^2\big\}
\\
&\leq \frac{2}{n}E\big(\big\| v\big\|^2\big) + \big\|{E}(v)\big\|^2
\\
&\leq \frac{4}{n}E\big(\big\| v\big\|^2\big) 
\\
&\leq \frac{4A_n}{n}.
\ee

The third equality in the above display follows because
\bee\nonumber
\frac{1}{n^2}E\Bigg[\Big\|\sum_{i = 1}^n\big\{ v_i  - E\big(v\big)\big\}\Big\|^2\Bigg] &=\sum_{k = 1}^{d}\frac{1}{n^2}E\Bigg[\sum_{i = 1}^n\big\{ v_i^{(k)}  - E\big(v^{(k)}\big)\big\}\Bigg]^2
\\
&=\sum_{k = 1}^{d}\frac{1}{n^2}\text{Var}\Big(\sum_{i = 1}^n v_i^{(k)}\Big) 
\\
&=\sum_{k = 1}^{d}\frac{1}{n}\text{Var}\big(v^{(k)}\big) 
\\
&=\frac{1}{n}E\big\{\big\| v_i - {E}(v)\big\|^2\big\},
\ee
where we define $ v_i^{(k)}, v^{(k)}$ as the $k$th coordinates of $ v_i, v$, respectively; the second equality follows by $E\big(\sum_{i = 1}^n v_i^{(k)}\big) = nE\big(v^{(k)}\big)$; the third equality follows by $ v_1^{(k)},\dots, v_n^{(k)}$ are independent and identically distributed. Combining \eqref{clm1:1} and \eqref{clm1:2} yields the desired result.
\end{proof}
\par
Gven $n$ samples, during the following proof we will frequently condition on the following event $\mathcal{E}_n$: the nuisance functions $ (\hat{m},\hat{\M\Gamma})$ are already obtained from the separate data set for nuisance function training (see $\mathsection$\ref{sec:prea}), and  $ (\hat{m},\hat{\M\Gamma})$ satisfies the following conditions:
 \begin{itemize}
 \item $ \|\hat{m}\|_{\mathbb{X}}$  and $\|\hat{\M\Gamma}\|_\mathbb{X}/\sqrt{K}$ are bounded by some fixed constant $C > 0$;
 \item All of the following quantities: (i) ${r_{\gamma}'}^{-1}\|\hat{\M\Gamma} - \Gamma\|_{\mathbb{X}}/\sqrt{K}$; (ii) $r_m^{-1}\|\hat{m} -m\|_{\mathcal{L}_{\mathcal{P}}^2} $; (iii) $r_{\gamma}^{-1}\big\|{P}\big[\{\hat{\M\Gamma}(X) - \M\Gamma(X)\}\{\hat{\M\Gamma}(X) - \M\Gamma(X)\}^{\T}\big]\big\|^{1/2}_2$; (iv) $\|\hat{m} -m\|_{\mathbb{X}} $; (v) $\|\hat{\Gamma}^\T\phi^* - \Gamma^\T\phi^*\|_{\mathbb{X}}$  are bounded by $e_n$, where $e_n$ is a sequence vanishing to zero, when $n\rightarrow +\infty$.
 \end{itemize}
 \par
By Lemma \ref{lm:approx}, Assumption \ref{rate:sup}, and rates \eqref{rate:m}--\eqref{rate:gamma}, we can choose some fixed $C_1 > 0$ and a deterministic positive sequence $e_n \rightarrow 0$, such that $\mathcal{E}_n$ happens wpa1 when $n\rightarrow +\infty$. For simplicity, during the following proof, we use ${\Pr}_n(\cdot)$ to represent the probability that condition on $\mathcal{E}_n$. Note that since nuisance functions are trained separately from the samples involved in this proof, the expectation in this proof conditional on obtained nuisance functions can be simply represented by ${P}(\cdot)$. 
\par
Now we prove (i)--(iv) sequentially.
\par
\noindent{\fbox{Proof of Lemma \ref{l3} (i)}} For simplicity, let $ m_{i,1}= \big\{Y_i - \mu(T_i,X_i)\big\}\big\{\Psi(T_i,X_i) - \Gamma(X_i)\big\}$. Thus $\Delta_{1,1} = \p_n(m_{i,1})$. We have
\bee\label{mi1:def}
E\big(m_{i,1}\big) &=  E\big[\big\{Y  - \mu(X,T) \big\}\big\{\Psi(X,T) - \Gamma(X)\big\}\big]
\\
& = E\big[\big\{E\big(Y\mid X,T\big)  - \mu(X,T) \big\}\big\{\Psi(X,T) - \Gamma(X)\big\}\big]
\\
&=  0,
\ee
where the second equality follows by the law of total expectation. On the other hand, 
\bee\nonumber
E\big(\| m_{i,1}\|^2\big) &= E\Big[\big\{Y - \mu(X,T)\big\}^2\big\|\Psi(X,T) - \Gamma(X)\big\|^2\Big]
\\
&=E\Big\{\text{Var}(Y\mid X,T)\big\|\Psi(X,T) - \Gamma(X)\big\|^2\Big\}
\\
&\leq \sup_{(x,t)\in \mathbb{X}\times \mathbb{T}}\text{Var}(Y\mid X = x,T = t) \cdot E\Big\{\big\|\Psi(X,T) - \Gamma(X)\big\|^2\Big\}.
\ee
Under Assumption \ref{am:moment}, one has  $\sup_{(x,t)\in \mathbb{X}\times \mathbb{T}}\text{Var}(Y\mid X = x,T=t) \precsim 1$. On the other hand,
\bee\nonumber
E\Big\{\big\|\Psi(X,T) - \Gamma(X)\big\|^2 \Big\}\leq 2\|\Psi\|_{\mathbb{X}\times\mathbb{T}}^2 + 2\|\Gamma\|_{\mathbb{X}}^2 \precsim K
\ee
by Lemma \ref{lm:approx}. Thus $E(\|m_{i,1}\|^2) \precsim K$. Let $ v_i =  m_{i,1}$ in Claim \ref{clm1}. With results above, we have for any fixed $J = J_0 > 0$
\bee\nonumber
\Pr\Big(\big\|\Delta_{1,1}\big\| > J_0\sqrt{K/n}\Big) &= \Pr\Big(\big\|{P}_n( m_{i,1}) \big\| > J_0 \sqrt{K/n}\Big)
\\
&\leq \frac{C_1K}{J_0^2K} = \frac{C_1}{J_0^2}
\ee
where $C_1$ is some fixed constant independent with $n$ and $J_0$. Therefore, when $J_0$ is sufficiently large, the right-hand side of the above display can by arbitrarily small, and thus $\|\Delta_{1,1}\| = \mathcal{O}_\p(\sqrt{K/n})$.
\par
Now we bound $|\ell_n^\T \Delta_{1,1}|$. We first condition on $\{(X_i,T_i)\}_{i = 1}^n$. We note that given $\{(X_i,T_i)\}_{i = 1}^n$, $\{\ell_n^\T m_{i,1}\}_{i = 1}^n$ are independent random variables since $\{Y_{i}\}_{i = 1}^n$ are independent and $\ell_n$  depends  only on $n$ and $\{(X_i,T_i)\}_{i = 1}^n$. Furthermore, they are mean-zero,
\bee\nonumber
E\big[\ell_n^\T m_{i,1}\mid \{(X_i,T_i)\}_{i = 1}^n\big] &= E\big(\tilde{\ell}_{n}^\T m_{i,1}\mid X_i,T_i\big)
\\
&=\big\{E\big(Y\mid X_i,T_i\big)  - \mu(X_i,T_i) \big\}\big\{\tilde{\ell}_{n}^\T\Psi(X_i,T_i) - \tilde{\ell}_{n}^\T\Gamma(X_i)\big\}
\\
& = 0.
\ee
Here we use $\tilde{\ell}_{n}$ to denote $\ell_n$ given specific $\{(X_i,T_i)\}_{i = 1}^n$. Then the conditional variance for each $\ell_n^\T m_{i,n}$ is
\bee\nonumber
\text{Var}\Big[\ell_n^\T m_{i,n}\mid \{(X_i,T_i)\}_{i = 1}^n\Big] &= E\Big\{(\tilde{\ell}_n^\T m_{i,n})^2 \mid X_i,T_i\Big\}
\\
& = E\Big[\big\{Y_i - \mu(T_i,X_i)\big\}^2\big\{\tilde{\ell}_n^\T\Psi(T_i,X_i) - \tilde{\ell}_n^\T\Gamma(X_i)\big\}^2\mid X_i,T_i\Big]
\\
&\leq \text{Var}(Y\mid X_i,T_i)\cdot\big\{\tilde{\ell}_n^\T\Psi(T_i,X_i) - \tilde{\ell}_n^\T\Gamma(X_i)\big\}^2
\\
&\leq\Big[2\big\{\tilde{\ell}_n^\T\Psi(T_i,X_i)\big\}^2 + 2\big\{\tilde{\ell}_n^\T\Gamma(X_i)\big\}^2\Big]\sup_{(x,t)\in\mathbb{X}\times\mathbb{T}}\text{Var}(Y\mid X = x, T = t) 
\\
&\leq C_2\Big[\big\{\tilde{\ell}_n^\T\Psi(T_i,X_i)\big\}^2 + \big\{\tilde{\ell}_n^\T\Gamma(X_i)\big\}^2\Big]\quad (\text{Assumption \ref{am:moment}})
\ee
for some fixed $C_2$ independent of $n$  and $\{(X_i,T_i)\}_{i = 1}^n$. Therefore, given $\{(X_i,T_i)\}_{i = 1}^n$, we have that $P_n(\ell_n^\T m_{i,n})$ is mean zero and has variance,
\bee\nonumber
\text{Var}\Big[P_n(\ell_n^\T m_{i,n})\mid \{(X_i,T_i)\}_{i = 1}^n\Big] & = \text{Var}\Big[\frac{1}{n}\sum_{i = 1}^n{\ell}_n^\T m_{i,n}\mid \{(X_i,T_i)\}_{i = 1}^n\Big] 
\\
&= \frac{1}{n^2}\sum_{i = 1}^n\text{Var}\Big[\ell_n^\T m_{i,n}\mid \{(X_i,T_i)\}_{i = 1}^n\Big]
\\
&\leq \frac{C_2}{n}\cdot\frac{1}{n}\sum_{i = 1}^n \Big[\big\{\tilde{\ell}_n^\T\Psi(T_i,X_i)\big\}^2 + \big\{\tilde{\ell}_n^\T\Gamma(X_i)\big\}^2\Big].
\ee
By Chebyshev's inequality,  given any $\{(X_i,T_i)\}_{i = 1}^n$,
\bee\label{delta11:main}
|P_n(\ell_n^\T m_{i,n})| = \mathcal{O}_{P}\Bigg[n^{-1/2}\sqrt{\frac{1}{n}\sum_{i = 1}^n \Big[\big\{\tilde{\ell}_n^\T\Psi(T_i,X_i)\big\}^2 + \big\{\tilde{\ell}_n^\T\Gamma(X_i)\big\}^2\Big]}\Bigg]
\ee
\par
Now we do not condition on specific $\{(X_i,T_i)\}_{i = 1}^n$, and  consider the following positive random variable, which has a constant upper bound wpa1,
\bee\label{delta11:main2}
\xi_1 &= \sqrt{\frac{1}{n}\sum_{i = 1}^n \Big[\big\{{\ell}_n^\T\Psi(T_i,X_i)\big\}^2 + \big\{{\ell}_n^\T\Gamma(X_i)\big\}^2\Big]}
\\
&= \sqrt{P_n\Big[\big\{{\ell}_n^\T\Psi(T,X)\big\}^2\Big] + P_n\Big[\big\{{\ell}_n^\T\Gamma(X)\big\}^2\Big]}
\\
&\leq \sqrt{ \sup_{\|\ell\| = 1} P_n\Big[\big\{{\ell}^\T\Psi(T,X)\big\}^2\Big] + \sup_{\|\ell'\| = 1}P_n\Big[\big\{{\ell'}^\T\Gamma(X)\big\}^2\Big]}
\\
& = \sqrt{\big\|\hat{Q}_n\big\|_2 + \big\|P_n(\Gamma\Gamma^\T)\big\|_2}
\\
&\precsim 1,
\ee
where the last equality follows by Lemma \ref{lm:iG2iG}. Combining \eqref{delta11:main} and \eqref{delta11:main2}, we directly uncondition the given $\{(X_i,T_i)\}_{i = 1}^n$, and conclude $|P_n(\ell_n^\T m_{i,n})| = \mathcal{O}_P(1/\sqrt{n})$.
\par
\
\par
\noindent{\fbox{Proof of Lemma \ref{l3} (ii)}}  First, by definition, we observe
\bee\nonumber
&\tau(X,T) - \E\{\tau(X,T)\mid X\}  - \{\Psi(X,T) - \Gamma(X)\}^\T{\phi^*}
\\
&=\tau(X,T) - \E\{\tau(X,T)\mid X\} - E\big[\tau(X,T) - \E\{\tau(X,T)\mid X\}\mid X\big]  - \{\Psi(X,T) - \Gamma(X)\}^\T{\phi^*}
\\
&=\tilde{\tau}(X,T) - \E\{\tilde{\tau}(X,T)\mid X\}- \{\Psi(X,T) - \Gamma(X)\}^\T{\phi^*}.
\ee
We now let 
\bee\nonumber
 m_{i,2}& = \big[ \tilde{\tau}(X_i,T_i) - \E\{\tilde{\tau}(X,T)\mid X_i\}  - \{\Psi(X_i,T_i) - \Gamma(X_i)\}^\T{\phi^*}\big]\big\{\Psi(X_i,T_i) - \Gamma(X_i)\big\}
 \\
 & =  \big[ \tau(X_i,T_i) - \E\{\tau(X,T)\mid X_i\}  - \{\Psi(X_i,T_i) - \Gamma(X_i)\}^\T{\phi^*}\big]\big\{\Psi(X_i,T_i) - \Gamma(X_i)\big\};
\ee 
thus $\Delta_{1,2} = P_n(m_{i,2})$. We then have 
\bee\label{m2:e}
E( m_{i,2})&= E\Big[\big[\tilde{\tau}(X,T) - \E\{\tilde{\tau}(X,T)\mid X\}- \{\Psi(X,T) - \Gamma(X)\}^\T{\phi^*}\big]\Psi(X,T)\Big]
\\
&-E\Big[\big[\tilde{\tau}(X,T) - \E\{\tilde{\tau}(X,T)\mid X\}- \{\Psi(X,T) - \Gamma(X)\}^\T{\phi^*}\big]  \Gamma(X)\Big]
\\
&= E\Big[\big[\tilde{\tau}(X,T) - \E\{\tilde{\tau}(X,T)\mid X\}- \{\Psi(X,T) - \Gamma(X)\}^\T{\phi^*}\big]\Psi(X,T)\Big]
\\
&-E\Big[\underbrace{E\Big[\big[\tilde{\tau}(X,T) - \E\{\tilde{\tau}(X,T)\mid X\}- \{\Psi(X,T) - \Gamma(X)\}^\T{\phi^*}\big]  \mid X\Big]}_{ = 0}\Gamma(X)\Big]
\\
&= E\Big[E\big\{ \Gamma^\T(X){\phi^*}- \tilde{\tau}(X,T)\mid X\big\}\Psi(X,T)\Big],
\ee
by the law of total expectation. Therefore, by Lemma \ref{am:svb} and Proposition \ref{lm:psieapprox}, we have
\bee\label{delta222}
\big\|{E}(\Delta_{1,2})\big\| &= \big\|\E(m_{i,2})\big\| 
\\
&= \sup_{\|\ell\| = 1}\Big|E\Big[E\big\{ \Gamma^\T(X){\phi^*}- \tilde{\tau}(X,T)\mid X\big\}\ell^\T\Psi(X,T)\Big]\Big|
\\
&\leq \sqrt{E\Big[E\big\{ \Gamma^\T(X){\phi^*}- \tilde{\tau}(X,T)\mid X\big\}\Big]^2} \sup_{\|\ell\| = 1}\sqrt{E\Big[\ell^\T\Psi(X,T)\Big]^2}\quad\text{(Cauchy--Schwarz inequality)}
\\
&= \sqrt{E\Big[E\big\{ \Psi^\T(X,T){\phi^*}- \tilde{\tau}(X,T)\mid X\big\}\Big]^2} \big\|Q_n\big\|^{1/2}_2
\\
&\leq\sup_{x\in\mathbb{X}}\big|E\big\{\Psi^\T(X,T){\phi^*}- \tilde{\tau}(X,T)\mid X = x\big\}\big|\big\|Q_n\big\|^{1/2}_2
\\
&\leq \sup_{x\in\mathbb{X}}\big|E\big\{\Psi^\T(X,T){\phi^*}- \tilde{\tau}(X,T)\mid X = x\big\}\big|\big\|Q_n\big\|^{1/2}_2
\\
&\leq \|\tilde{\tau} - \{\phi^*\}^\T\Psi\|_{\mathbb{X}\times\mathbb{T}}\big\|Q_n\big\|^{1/2}_2
\\
&= \mathcal{O}(K^{-p/(d + 1)}).
\ee
\par
Recalling $\Psi^\T_n(T_i,X_i) \phi^*$ is a sieve approximation to $\tilde{\tau}(T_i,X_i)$, we then have
\bee\label{mi2square}
&E\big(\big\| m_{i,2} \big\|^2\big)
\\
&\leq \sup_{(x,t)\in\mathbb{X}\times\mathbb{T}}\Big|\big[ \tilde{\tau}(X = x,T = t) - \E\{\tilde{\tau}(X,T)\mid X = x\}  - \{\Psi(X = x,T = t) - \Gamma(X = x)\}^\T{\phi^*}\big]\Big|^2
\\
&\cdot E\Big(\big\|\Psi(X,T) - \Gamma(X)\big\|^2\Big)
\\
&\leq \Big[2\big\| \tilde{\tau} - \Psi^\T_n {\phi^*}\big\|_{ \mathbb{X}\times \mathbb{T}}^2 + \sup_{x\in\mathbb{X}}2\Big|\E\{\tilde{\tau}(X,T) - \Psi^\T_n(X,T) {\phi^*}\mid X = x\}\Big|^2\Big]
\cdot\Big(2\|\M\Psi \|_{\mathbb{X}\times \mathbb{T}}^2 + 2\|\Gamma\|_{\mathbb{X}}^2\Big)
\\
&\leq \Big(2\big\| \tilde{\tau} - \Psi^\T_n {\phi^*}\big\|_{ \mathbb{X}\times \mathbb{T}}^2 + \sup_{x\in\mathbb{X}}2\Big|\E\big(\big\| \tilde{\tau} - \Psi^\T_n {\phi^*}\big\|_{\mathbb{X}\times \mathbb{T}}^2\mid X = x\big)\Big|\Big)
\cdot\Big(2\|\M\Psi \|_{\mathbb{X}\times \mathbb{T}}^2 + 2\|\Gamma\|_{\mathbb{X}}^2\Big)
\\
&\leq \Big(4\big\| \tilde{\tau} - \Psi^\T_n {\phi^*}\big\|_{ \mathbb{X}\times \mathbb{T}}^2\Big)
\cdot\Big(2\|\M\Psi \|_{\mathbb{X}\times \mathbb{T}}^2 + 2\|\Gamma\|_{\mathbb{X}}^2\Big)
\\
&=\mathcal{O}(K^{1-2p/(d+1)}),
\ee
by Lemma \ref{am:psi}, Proposition \ref{lm:psieapprox} and Lemma \ref{lm:approx}. By taking $v_i = m_{i,2}$, $A_n = C_3K^{1-2p/(d+1)}$, $d = K$, and $J = C_3K^{-p/(d+ 1)}$ in Claim \ref{clm1}, then for any $c > 0$, there exists large enough yet fixed $C_3> 0$ such that
\bee\nonumber
\Pr\Big\{\big\|\Delta_{1,2} - {E}(\Delta_{1,2})\big\| > C_3K^{-p/(d + 1)} \sqrt{K/n}\Big\} \leq c.
\ee
Therefore $\big\|\Delta_{1,2} - {E}(\Delta_{1,2})\big\| = \mathcal{O}_P\big(K^{-p/(d + 1)} \sqrt{K/n}\big) = o_{P}(K^{-p/(d + 1)} )$. Then by \eqref{delta222}, one has
\bee\nonumber
\big\|\Delta_{1,2}\big\| \leq \big\|\Delta_{1,2} - {E}(\Delta_{1,2})\big\| + \|E(\Delta_{1,2})\| = \mathcal{O}_P(K^{-p/(d + 1)} )
\ee
\par
When $\ell_n$  depends only on $n$ and $\|\ell_n\| = 1$, then similar to \eqref{delta222}, we have
\bee\label{d12:mean-bias}
|E(\ell_n^\T m_{i,2})| &= \Big|E\Big[E\big\{ \Gamma^\T(X){\phi^*}- \tilde{\tau}(X,T)\mid X\big\}\ell_n^\T\Psi(X,T)\Big]\Big|
\\
&=\mathcal{O}(K^{-p/(d + 1)}).
\ee
Also, similar to \eqref{mi2square} and by Lemma \ref{am:svb}, we have
\bee\nonumber
 E(\ell_n^\T m_{i,2})^2 &\leq \sup_{(x,t)\in\mathbb{X}\times\mathbb{T}}\Big|\big[ \tilde{\tau}( x, t) - \E\{\tilde{\tau}(X,T)\mid X = x\}  - \{\Psi(x,t) - \Gamma( x)\}^\T{\phi^*}\big]\Big|^2
\\
&\quad \cdot E\Big[\ell_n^\T\Big\{\Psi(X,T) - \Gamma(X)\Big\}\Big]^2
\\
&\leq \Big(4\big\| \tilde{\tau} - \Psi^\T_n {\phi^*}\big\|_{ \mathbb{X}\times \mathbb{T}}^2\Big)
\cdot\Big(\|Q_n\|_2 + \|E(\Gamma\Gamma^\T)\|_2\Big)
\\
&=\mathcal{O}(K^{-2p/(d+1)}).
\ee
Take $v_i = \ell_n^\T m_{i,2}$, $d = 1$, $A_n = C_4K^{-2p/(d + 1)}$, $J = 2\sqrt{C_4}K^{-p/(d + 1)}/c$ in Proposition \ref{clm1}. We then have for any fixed $c > 0$ and sufficiently large $C_4 > 0$,
\bee\nonumber
\Pr\Big\{\big|P_n(\ell_n^\T m_{i,2}) - E(\ell_n^\T m_{i,2})\big| > 2\sqrt{C_4}/c K^{-p/(d + 1)}n^{-1/2}\Big\} \leq c,
\ee
which combining with \eqref{d12:mean-bias} implies that $|\ell_n^\T\Delta_{1,2}| =|P_n(\ell_n^\T m_{i,2})| \leq  |P_n(\ell_n^\T m_{i,2}) - E(\ell_n^\T m_{i,2})| + |E(\ell_n^\T m_{i,2})| = \mathcal{O}_P(K^{-p/(d + 1)} n ^{-1/2})+ \mathcal{O}_P(K^{-p/(d + 1)}) = \mathcal{O}_P(K^{-p/(d + 1)})$.

\par
\
\par
\noindent{\fbox{Proof of Lemma \ref{l3} (iii)}} We now condition on $\mathcal{E}_n$ with  well-conditioned $(\hat{m},\hat{\Gamma})$. We let $ m_{i,3} = [Y_i - m(X_i) - \{\Psi(T_i,X_i) - \Gamma(X_i)\}^\T{\phi^*}]\{\Gamma(X_i) - \hat{\Gamma}(X_i)\}$ and thus $\Delta_{1,3} = P_n(m_{i,3})$. Recall the nuisance functions and proposed estimators are trained independently. The conditional expectations then can be bounded as follows
\bee\label{mi31}
&\p(m_{i,3})
\\
& =\p\Big[\big[Y - m(X) - \{\Psi(X,T) - \Gamma(X)\}^\T{\phi^*}\big]\big\{\Gamma(X) - \hat{\Gamma}(X)\big\}\Big]
\\
&=\p\Big[E\big[Y - m(X) - \{\Psi(X,T) - \Gamma(X)\}^\T{\phi^*}\mid X\big]\big\{\Gamma(X) - \hat{\Gamma}(X)\big\}\Big]
\\
& = 0,
\ee
where the last equality follows by $E(Y\mid X) = m (X)$ and $E\{\Psi^\T(X,T)\phi^*\mid X\} = \Gamma^\T(X)\phi^*$; thus $E(\Delta_{1,3}) = 0$.
We also have
\bee\label{mi322}
&\p\big(\| m_{i,3}\|^2\big)
\\
&=\p\Big[\big[Y - m(X) - \{\Psi(X,T) - \Gamma(X)\}^\T{\phi^*}\big]^2\big\|\Gamma(X) - \hat{\Gamma}(X)\big\|^2\Big]
\\
&\leq \p\Big[\big[Y - m(X) - \{\Psi(X,T) - \Gamma(X)\}^\T{\phi^*}\big]^2\Big] \big\|\Gamma - \hat{\Gamma}\big\|_{\mathbb{X}}^2
\\
&=\mathcal{O}(1) \cdot \mathcal{O}(K{r_{\gamma}'}^2e_n) 
\\
&= \mathcal{O}(K{r_{\gamma}'}^2e_n),
\ee
where the inequality follows by
\bee\label{m13o1}
\p\Big[\big[Y - m(X) - \{\Psi(X,T) - \Gamma(X)\}^\T{\phi^*}\big]^2\Big] &\leq 2\E\Big[\big\{Y - m(X)\big\}^2\Big] + 2E\Big[ \big[\{\Psi(X,T) - \Gamma(X)\}^\T{\phi^*}\big]^2\Big]
\\
& =  2E\Big\{\var(Y\mid X)\Big\} + 2(\phi^*)^\T R_n\phi^*
\\
&\leq 2E\Big\{\var(Y\mid X)\Big\} + 2\|\phi^*\|^2 \|R_n\|_2
\\
& = \mathcal{O}(1),\quad\text{(Assumptions \ref{am:moment} and Lemma \ref{am:svb})}
\ee
 and last two equalities are by the event defined in $\mathcal{E}_n$. Let $v_i =  m_{i,3}$, $d=K$, $A_n = C_5K{r_{\gamma}'}^2e_n$, and $J = c_5r_{\gamma}'$ in Proposition \ref{clm1}. We have when $n\rightarrow +\infty$,
\bee\label{pmm3}
{\Pr}_n\Big(\big\|\Delta_{1,3}\big\| >  c_5r_{\gamma}'\sqrt{K/n}\Big) &\leq \frac{4C_5K{r_{\gamma}'}^2e_n}{{ c_5^2r_{\gamma}'}^2K} 
\\
&= 4C_5e_n/ c_5^2
\\
&\rightarrow 0,
\ee
for some sufficiently large yet fixed $C_5 > 0$ and any fixed $c_5 > 0$. We thus have, conditional on $\mathcal{E}_n$, $\|\Delta_{1,3}\|  = o_{{P}}(r_{\gamma}'\sqrt{K/n})$. Recall $\mathcal{E}_n$ is a event happens wpa1. We thus can directly uncondition $\mathcal{E}_n$ and have $\|\Delta_{1,3}\| = o_{{P}}(r_{\gamma}'\sqrt{K/n})$. To be more specific, for any fixed $C_6 > 0$, we have
\bee\label{temp:uncon}
&{\Pr}\Big(\big\|\Delta_{1,3}\big\| > c_5r_{\gamma}'\sqrt{K/n}\Big)
\\
&={\Pr}\Big(\Big\{\big\|\Delta_{1,3}\big\| > c_5r_{\gamma}'\sqrt{K/n}\Big\}\cap\mathcal{E}_n\Big)  + {\Pr}\Big(\Big\{\big\|\Delta_{1,3}\big\| > c_5r_{\gamma}'\sqrt{K/n}\Big\}\cap\mathcal{E}^c_n\Big)
\\
&\leq {\Pr}\Big(\Big\{\big\|\Delta_{1,3}\big\| > c_5r_{\gamma}'\sqrt{K/n}\Big\}\cap\mathcal{E}_n\Big)  + {\Pr}(\mathcal{E}^c_n)
\\
&={\Pr}_n\Big(\big\|\Delta_{1,3}\big\| > c_5r_{\gamma}'\sqrt{K/n}\Big) \cdot{\Pr}(\mathcal{E}_n) + {\Pr}(\mathcal{E}^c_n)
\\
&\rightarrow 0,
\ee
by ${\Pr}(\mathcal{E}^c_n) \rightarrow 0$, ${\Pr}(\mathcal{E}_n) \leq 1$, and \eqref{pmm3}.
\par
Now we bound $|\ell_n^\T \Delta_{1,3}|$. Similar to \eqref{mi31}, we have $\p(\ell_n^\T m_{i,3}) = 0$. Similar to \eqref{mi322}, we have
\bee\label{lmi3:decom}
&\p(\ell_n^\T m_{i,3})^2
\\
&=\p\Big[\big[Y - m(X) - \{\Psi(X,T) - \Gamma(X)\}^\T{\phi^*}\big]^2\big[\ell_n^\T\big\{\Gamma(X) - \hat{\Gamma}(X)\big\}\big]^2\Big]
\\
&\leq 3\p\Big[\big\{Y - m(X)\big\}^2\big[\ell_n^\T\big\{\Gamma(X) - \hat{\Gamma}(X)\big\}\big]^2\Big]
\\
&+3\p\Big[\big[ \{\Psi(X,T) - \Gamma(X)\}^\T{\phi^*} - [\tilde{\tau}(X,T) - E\{\tilde{\tau}(X,T)\mid X\}]\big]^2\big[\ell_n^\T\big\{\Gamma(X) - \hat{\Gamma}(X)\big]^2\Big]
\\
&+3\p\Big[\big[ \tilde{\tau}(X,T) - E\{\tilde{\tau}(X,T)\mid X\}\big]^2\big[\ell_n^\T\big\{\Gamma(X) - \hat{\Gamma}(X)\big]^2\Big],
\ee 
by the triangle inequality. The first term in \eqref{lmi3:decom} can be bounded by the law of total expectation,
\begin{align}\nonumber
&\p\Big[\big\{Y - m(X)\big\}^2\big[\ell_n^\T\big\{\Gamma(X) - \hat{\Gamma}(X)\big\}\big]^2\Big]
\\\nonumber
&= \p\Big[\var(Y\mid X)\big[\ell_n^\T\big\{\Gamma(X) - \hat{\Gamma}(X)\big\}\big]^2\Big]
\\\nonumber
&\leq \sup_{x \in\mathbb{X}}\var(Y\mid X = x) \cdot\p\Big[\ell_n^\T\big\{\Gamma(X) - \hat{\Gamma}(X)\big\}\Big]^2
\\\nonumber
&\precsim \p\Big[\ell_n^\T\big\{\Gamma(X) - \hat{\Gamma}(X)\big\}\Big]^2
\\\nonumber
&\leq \sup_{\|\ell\| = 1}\p\Big[\ell^\T\big\{\Gamma(X) - \hat{\Gamma}(X)\big\}\Big]^2
\\\nonumber
&= \big\|P\big[\big\{\Gamma(X) - \hat{\Gamma}(X)\big\}\big\{\Gamma(X) - \hat{\Gamma}(X)\big\}^\T\big]\big\|_2
\\\label{mi13sss}
& = \mathcal{O}(r_{\gamma}^2e_n^2),
\end{align}
where the second inequality follows by Assumption \ref{am:moment} and $\mathcal{E}_n$. For the second term on the right-hand side of \eqref{lmi3:decom}, we first note
\bee\nonumber
&\sup_{(x,t)\in\mathbb{X}\times\mathbb{T}}\Big|\big[ \tilde{\tau}(x,t) - \E\{\tilde{\tau}(X,T)\mid X = x\}  - \{\Psi( x,t) - \Gamma(x)\}^\T{\phi^*}\big]\Big|^2
\\
& = \mathcal{O}(K^{-2p/(d + 1)}),
\ee
by Proposition \ref{lm:psieapprox} similar to \eqref{mi2square}. Then we have
\bee\nonumber
&\p\Big[\big[ \{\Psi(X,T) - \Gamma(X)\}^\T{\phi^*} - [\tilde{\tau}(X,T) - E\{\tilde{\tau}(X,T)\mid X\}]\big]^2\big[\ell_n^\T\big\{\Gamma(X) - \hat{\Gamma}(X)\big]^2\Big]
\\
&\leq \sup_{(x,t)\in\mathbb{X}\times\mathbb{T}}\Big|\big[ \tilde{\tau}(x,t) - \E\{\tilde{\tau}(X,T)\mid x\}  - \{\Psi(x,t) - \Gamma( x)\}^\T{\phi^*}\big]\Big|^2\cdot\p\Big[\big[\ell_n^\T\big\{\Gamma(X) - \hat{\Gamma}(X)\big]^2\Big]
\\
&=\mathcal{O}(K^{-2p/(d + 1)}r_{\gamma}^2e_n^2).
\ee
For the third term on the right-hand side of \eqref{lmi3:decom}, we note $\|\tilde{\tau}\|_{\mathbb{X}\times\mathbb{T}} \precsim 1$ as $\tilde{\tau}$ is in the H\"older class, and thus $\|E(\tilde{\tau}\mid \cdot)\|_{\mathbb{X}}\leq \sup_{x \in\mathbb{X}}E(\|\tilde{\tau}\|_{\mathbb{X}\times\mathbb{T}}\mid X = x) = \|\tilde{\tau}\|_{\mathbb{X}\times\mathbb{T}}\precsim 1$. Then we have
\bee\nonumber
&\p\Big[\big[ \tilde{\tau}(X,T) - E\{\tilde{\tau}(X,T)\mid X\}\big]^2\big[\ell_n^\T\big\{\Gamma(X) - \hat{\Gamma}(X)\big\}\big]^2\Big]
\\
&\leq \big\| \tilde{\tau} - E\{\tilde{\tau}(X,T)\mid X = \cdot\}\big\|_{\mathbb{X}\times\mathbb{T}}^2\cdot\p\big[\ell_n^\T\big\{\Gamma(X) - \hat{\Gamma}(X)\big\}\big]^2
\\
&\leq \Big[2\big\| \tilde{\tau}\big\|_{\mathbb{X}\times\mathbb{T}}^2 + 2\big\|E\{\tilde{\tau}(X,T)\mid X = \cdot\}\big\|_{\mathbb{X}}^2\Big]\cdot\p\Big[\big[\ell_n^\T\big\{\Gamma(X) - \hat{\Gamma}(X)\big]^2\Big]
\\
&\precsim \p\Big[\ell_n^\T\big\{\Gamma(X) - \hat{\Gamma}(X)\Big]^2
\\
&=\mathcal{O}(r_{\gamma}^2e_n^2),
\ee  
similar to \eqref{mi13sss}. In summary, we have
\bee\nonumber
\p(\ell_n^\T m_{i,3})^2 = \mathcal{O}(r_{\gamma}^2e_n^2).
\ee
By taking $v_i = \ell_n^\T m_{i,3}$, $d = 1$, $A_n = C_6r_{\gamma}^2e_n^2$, $J = c_6r_\gamma$ in Proposition \ref{clm1}, we have
\bee\nonumber
{\Pr}_n\Big(|\ell_n^\T\Delta_{1,3}| > c_6r_\gamma /\sqrt{n}\Big) \leq \frac{4C_6 r_{\gamma}^2e_n^2}{r_{\gamma}^2c_6^2} = 4C_6 e_n^2/c_6^2 \rightarrow 0.
\ee
Thus similar to \eqref{temp:uncon}, we can directly uncondition $\mathcal{E}_n$ and conclude  $|\ell_n^\T\Delta_{1,3}| =o_P(r_\gamma /\sqrt{n})$.
\par
\
\par
\noindent\fbox{Proof of Lemma \ref{l3} (iv)} We still condition on $\mathcal{E}_n$. We let $ m_{i,4} = \big[m(X_i) - \hat{m}(X_i) - \{\Gamma(X_i) - \hat{\Gamma}(X_i)\}^\T{\phi^*}\big]\big\{\Psi(X_i,T_i)-\Gamma(X_i)\big\}$, thus $\Delta_{1,4} = P_n(m_{i,4})$. By the law of total expectation we have, 
\bee\label{mi41}
&\p[ m_{i,4}]
\\
&=\p\Big[\big[m(X) - \hat{m}(X) - \{\Gamma(X) - \hat{\Gamma}(X)\}^\T{\phi^*}\big]\big\{\Psi(X,T)-\Gamma(X)\big\}\Big] 
\\
&=\p\Big[\big[m(X) - \hat{m}(X) - \{\Gamma(X) - \hat{\Gamma}(X)\}^\T{\phi^*}\big]\p\big\{\Psi(X,T)-\Gamma(X)\mid X\big\}\Big]
\\
&= 0.
\ee
By the triangle inequality, we have
\bee\nonumber
&\p\| m_{i,4}\|^2
\\
&\leq 2\p \Big[\{m(X) - \hat{m}(X)\}^2\big\|\Psi(X,T)-\Gamma(X)\big\|^2\Big] + 2\p \Big[\big[\{\Gamma(X) - \hat{\Gamma}(X)\}^\T{\phi^*}\big]^2\big\|\Psi(X,T)-\Gamma(X)\big\|^2\Big]
\\
&\leq 2\big\|\Psi-\Gamma\big\|^2_{\mathbb{X}\times \mathbb{T}}\Big[\p \big\{m(X) - \hat{m}(X)\big\}^2 + \| \phi^*\|^2\p\big[\{\Gamma(X) - \hat{\Gamma}(X)\}^\T\big(\phi^*/\|\phi^*\|\big)\big]^2 \Big]
\\
&\leq 2\big\|\Psi-\Gamma\big\|^2_{\mathbb{X}\times \mathbb{T}}\Big\{\|\hat{m} - m\|_{\mathcal{L}_{\mathcal{P}}^2}^2 + \| \phi^*\|^2\big\|{P}\big[\{\hat{\M\Gamma}(X) - {\M\Gamma}(X)\}\{\hat{\M\Gamma}(X) - \M\Gamma(X)\}^{\T}\big]\big\|_2  \Big\}
\\
&= \mathcal{O}(Kr_m^2e_n^2 + Kr_{\gamma}^2e_n^2),
\ee
where  the last equality follows by $\mathcal{E}_n$. By taking $v_i = \ell_n^\T m_{i,4}$, $d = K$, $A_n = C_7(Kr_m^2e_n^2 + Kr_{\gamma}^2e_n^2)$, $J = c_7r_m + c_7r_\gamma$ in Proposition \ref{clm1}, for some sufficiently large yet fixed $C_7 > 0$ and any given $c_7 > 0$, 
\bee\nonumber
&{\Pr}_n\Big\{\|\Delta_{1,4}\| > c_7(r_m + r_\gamma)\sqrt{K/n}\Big\} 
\\
&\leq \frac{4C_7 (Kr_m^2e_n^2 + Kr_{\gamma}^2e_n^2)}{c_7^2(r_m + r_\gamma)^2K}
\\
&\precsim e_n^2
\\
& \rightarrow 0.
\ee
which  implies $\|\Delta_{1,4}\| = o_{{P}}(r_m\sqrt{K/n} + r_\gamma\sqrt{K/n})$ by directly unconditioning $\mathcal{E}_n$, similar to \eqref{temp:uncon}.
\par
On the other hand, similar to \eqref{mi41} one has $P(\ell_n^\T m_{i,4}) = 0$. Also we have
\bee\nonumber
& P(\ell_n^\T m_{i,4})^2
\\
& = P\Big[\big[m(X) - \hat{m}(X) - \{\Gamma(X) - \hat{\Gamma}(X)\}^\T{\phi^*}\big]^2\big\{\ell_n^\T\Psi(X,T)-\ell_n^\T\Gamma(X)\big\}^2\Big]
\\
&\leq  2\Big[\|m - \hat{m}\|_{\mathbb{X}}^2 + \|\{\Gamma - \hat{\Gamma}\}^\T{\phi^*}\|_{\mathbb{X}}^2 \Big]\cdot P\Big[\ell_n^\T\Psi(X,T)-\ell_n^\T\Gamma(X)\Big]^2
\\
&\leq  4e_n^2 \|R_n\|_2
\\
& = \mathcal{O}(e_n^2),
\ee
where the second inequality is due to $\mathcal{E}_n$, and the last equality follows by Lemma \ref{am:svb} such that $\|R_n\|_2\precsim 1$. By taking $v_i = \ell_n^\T m_{i,4}$, $d = 1$, $A_n = C_8e_n^2$, and $J = c_8$ for some sufficiently large yet fixed $C_8 > 0$ and any given $c_8 > 0$, we then have
\bee\nonumber
{\Pr}_n\Big(|\ell_n^\T\Delta_{1,4}| > c_8 1/\sqrt{n}\Big) \leq \frac{4C_8e_n^2}{c_8^2}\rightarrow 0,
\ee
which, after unconditioning $\mathcal{E}_n$ similar to \eqref{temp:uncon}, we have $|\ell_n^\T\Delta_{1,4}| = o_P(1/\sqrt{n})$.
\par
\
\par
\noindent\fbox{Proof of Lemma \ref{l3} (v)}  We still condition on $\mathcal{E}_n$. Let $ m_{i,5} = [m(X_i) - \hat{m}(X_i) - \{\Gamma(X_i) - \hat{\Gamma}(X_i)\}^\T{\phi^*}]\big\{\Gamma(X_i)-\hat{\Gamma}(X_i)\big\}$, and thus $\Delta_{1,5} = P_n(m_{i,5})$. First we have
\begin{align}\nonumber
&\|\p(m_{i,5})\|
\\\nonumber
& = \sup_{\|\ell\| = 1}\Big|P\Big[[m(X) - \hat{m}(X) - \{\Gamma(X) - \hat{\Gamma}(X)\}^\T{\phi^*}]\big\{\ell^\T\Gamma(X)-\ell^\T\hat{\Gamma}(X)\big\}\Big]\Big|
\\\nonumber
&\leq\sup_{\|\ell\| = 1}\Big|P\Big[\{m(X) - \hat{m}(X)\}\big\{\ell^\T\Gamma(X)-\ell^\T\hat{\Gamma}(X)\big\}\Big]\Big|
\\\nonumber
&+\|\phi^*\|\sup_{\|\ell\| = 1}\Big|P\Big[\big[ \{\Gamma(X) - \hat{\Gamma}(X)\}^\T{\phi^*}/\|\phi^*\|\big]\big\{\ell^\T\Gamma(X)-\ell^\T\hat{\Gamma}(X)\big\}\Big]\Big|
\\\nonumber
&\leq \|\hat{m} - m\|_{\mathcal{L}_{\mathcal{P}}^2} \big\|P[\{\Gamma(X) - \hat{\Gamma}(X)\}\{\Gamma(X) - \hat{\Gamma}(X)\}^\T]\big\|_{2}^{1/2}
\\\nonumber
&+\|\phi^*\|\sup_{\|\ell\| = 1}\Big|P\Big[({\phi^*}/\|\phi^*\|)^\T\{\Gamma(X) - \hat{\Gamma}(X)\}\big\{\Gamma(X)-\hat{\Gamma}(X)\big\}^\T\ell\Big]\Big|
\\\nonumber
&\leq \|\hat{m} - m\|_{\mathcal{L}_{\mathcal{P}}^2} \big\|P[\{\Gamma(X) - \hat{\Gamma}(X)\}\{\Gamma(X) - \hat{\Gamma}(X)\}^\T]\big\|_{2}^{1/2}
+\|\phi^*\|\big\|P[\{\Gamma(X) - \hat{\Gamma}(X)\}\{\Gamma(X) - \hat{\Gamma}(X)\}^\T]\big\|_{2}
\\\nonumber
&=\mathcal{O}(e_n^2r_{m}r_{\gamma} +e_n^2 r^2_{\gamma})
\\\label{mi51}
&=o(r_{m}r_{\gamma} + r^2_{\gamma}),
\end{align}
where the second inequality follows by Cauchy--Schwarz inequality; the last equality follows by $\mathcal{E}_n$, and Proposition \ref{lm:psieapprox} such that $\|\phi^*\|\precsim 1$. On the other hand,
\bee\nonumber
&\p(\| m_{i,5}\|^2)
\\
&\leq 2\p\Big[\{m(X) - \hat{m}(X)\}^2\|\Gamma(X)-\hat{\Gamma}(X)\|^2\Big] + 2\p \Big[\big[\{\Gamma(X) - \hat{\Gamma}(X)\}^\T{\phi^*}\big]^2\|\Gamma(X)-\hat{\Gamma}(X)\|^2\Big]
\\
&\leq 2\|\Gamma-\hat{\Gamma}\|^2_{\mathbb{X}}\cdot\Big[\|m - \hat{m}\|_{\mathcal{L}^2_{\mathcal{P}}}^2 + \|\phi^*\|^2\big\|{P}\big[\{\hat{\M\Gamma}_n(X) - {\M\Gamma}_n(X)\}\{\hat{\M\Gamma}_n(X) - \M\Gamma(X)\}^{\T}\big]\big\|_2\Big]
\\
&=\mathcal{O}(e_n^4K {r_{\gamma}'}^2 r_m^2 + e_n^4K {r_{\gamma}'}^2 r_\gamma^2). 
\ee
Then by taking $v_i = m_{i,5}$, $d = K$, $A_n = C_9e_n^4K {r_{\gamma}'}^2 r_m^2 + C_9e_n^4K {r_{\gamma}'}^2 r_\gamma^2$, and $J = c_9({r_{\gamma}'} r_m + {r_{\gamma}'} r_\gamma)$ in Proposition \ref{clm1}, we have, for any fixed $c_9 > 0$ and some sufficiently large yet fixed $C_9 > 0$,
\bee\nonumber
{\Pr}_n\Big(\big\|\Delta_{1,5} - \p(m_{i,5})\big\| > c_9({r_{\gamma}'} r_m + {r_{\gamma}'} r_\gamma) \sqrt{K/n}\Big) &\leq \frac{4C_9e_n^4K {r_{\gamma}'}^2 r_m^2 + 4C_9e_n^4K {r_{\gamma}'}^2 r_\gamma^2}{c^2_9({r_{\gamma}'} r_m + {r_{\gamma}'} r_\gamma)^2K}
\\
&\precsim e_n^4
\\
&\rightarrow 0.
\ee
Thus under $\mathcal{E}_n$, we have 
\bee\nonumber
\|\Delta_{1,5}\| &\leq \|P(m_{i,5})\| + \|\Delta_{1,5} - \p(m_{i,5})\| 
\\
&= o_P(r_{m}r_{\gamma} + r^2_{\gamma} + {r_{\gamma}'} r_m\sqrt{K/n} + {r_{\gamma}'} r_\gamma\sqrt{K/n} ).
\ee Similar to \eqref{temp:uncon}, we can  uncondition $\mathcal{E}_n$, and the same bound of $\|\Delta_{1,5}\|$ still holds.
\par
On the other hand, by \eqref{mi51} we further have
\bee\nonumber
|P(\ell_n^\T m_{i,5})| &\leq \sup_{\|\ell\| = 1}|P(\ell^\T m_{i,5})|
\\
& = \|P( m_{i,5})\|
\\
&=\mathcal{O}(e_n^2r_{m}r_{\gamma} +e_n^2 r^2_{\gamma})
\\
&=o(r_{m}r_{\gamma} + r^2_{\gamma}).
\ee
In addition, under $\mathcal{E}_n$, we have
\bee\nonumber
P(\ell_n^\T m_{i,5})^2 &= P\Big[[m(X) - \hat{m}(X) - \{\Gamma(X) - \hat{\Gamma}(X)\}^\T{\phi^*}]^2\big\{\ell_n^\T\Gamma(X)-\ell_n^\T\hat{\Gamma}(X)\big\}^2\Big]
\\
&\leq 2\Big\{\|m - \hat{m}\|_{\mathbb{X}}^2 + \|(\Gamma - \hat{\Gamma})^\T{\phi^*}\|_{\mathbb{X}}^2 \Big\} \cdot P\Big\{\ell_n^\T\Gamma(X)-\ell_n^\T\hat{\Gamma}(X)\Big\}^2
\\
&\precsim e_n^2\big\|P[\{\hat{\Gamma} (X)- \Gamma(X)\}\{\hat{\Gamma}(X) - \Gamma(X)\}^\T]\big\|_2
\\
& = \mathcal{O}(e_n^2 r_{\gamma}^2).
\ee
Then by taking $v_i = \ell_n^\T m_{i,5}$, $d = 1$, $A_n = C_{10}e_n^2 r_{\gamma}^2$, and $J = c_{10}$ in Proposition \ref{clm1}, we have, for any fixed $c_{10} > 0$ and some sufficiently large yet fixed $C_{10} > 0$,
\bee\nonumber
{\Pr}_n\Big\{\big|\ell_n^\T\Delta_{1,5} - \p(\ell_n^\T m_{i,5})\big| > c_{10}r_{\gamma}/\sqrt{n}\Big\} \leq \frac{4C_{10}e_n^2 r_{\gamma}^2}{c_{10}^2r_{\gamma}^2}\rightarrow 0.
\ee
Thus under $\mathcal{E}_n$, $|\ell_n^\T\Delta_{1,5}| \leq |\ell_n^\T\Delta_{1,5} - \p(\ell_n^\T m_{i,5})| + |\p(\ell_n^\T m_{i,5})| = o_P(r_{\gamma}/\sqrt{n} + r_{m}r_{\gamma} + r^2_{\gamma})$. Similar to \eqref{temp:uncon}, we can then uncondition $\mathcal{E}_n$, and the same bound of $|\ell_n^\T\Delta_{1,5}|$ still holds.
\par
\QED
\end{proof}
\subsection{Proof of Proposition \ref{thm:id}}\label{sec:partIIint}
\noindent\fbox{Proof of Proposition \ref{thm:id} (i)} We first consider  a more general minimization problem,
\bee\label{thm:id:1}
q &= \argmin_{h\in \mathcal{L}^2_{\mathcal{P}}( X,T)} E\Big\{Y - m(X) - h( X,T)\Big\}^2
\\
&=\argmin_{h\in \mathcal{L}^2_{\mathcal{P}}( X,T)} E\Big\{Y - E(Y\mid X) - h( X,T)\Big\}^2.
\ee
Treat $Y - E(Y\mid X)$  as a  random variable. Then based on the least-square form of \eqref{thm:id:1}, we know \eqref{thm:id:1} is minimized, if and only if $q(X,T)$ is the conditional mean of $Y - E(Y\mid X)$ given $X$ and $T$ a.s., or equivalently, 
\bee\label{qet}
q(X,T) = E\{Y - E(Y\mid X)\mid X,T\}\text{ a.s..}
\ee  We rigorously show the above statement by contradiction. Suppose $q(X,T)$ is not a.s. $E\{Y - E(Y\mid X)\mid X,T\}$. Then we have 
\bee
q (X,T) &= E\{Y - E(Y\mid X)\mid X,T\} + u(X,T) \text{ a.s.},
\ee
for some $u(X,T) \neq 0$ with probability larger than $0$. We then have
\bee\label{poYXmain}
 &E\Big\{Y - E(Y\mid X) - q( X,T)\Big\}^2
 \\
 &=E\Big[Y - E(Y\mid X)-E\{Y - E(Y\mid X)\mid X,T\} - u(X,T)\Big]^2
 \\
 &=E\Big[Y - E(Y\mid X) - E\{Y - E(Y\mid X)\mid X,T\}\Big]^2
 \\
 & - 2E\Big[\Big[Y - E(Y\mid X) - E\{Y - E(Y\mid X)\mid X,T\}\Big]u(X,T)\Big] + E\big\{u^2(X,T)\big\}
 \\
 & = E\Big[Y - E(Y\mid X) - E\{Y - E(Y\mid X)\mid X,T\}\Big]^2 + E\big\{u^2(X,T)\big\}
 \\
 & > E\Big[Y - E(Y\mid X) - E\{Y - E(Y\mid X)\mid X,T\}\Big]^2,
\ee
where by the law of total expectation,
\bee\nonumber
&E\Big[\Big[Y - E(Y\mid X) - E\{Y - E(Y\mid X)\mid X,T\}\Big]u(X,T)\Big]
\\
&=E\Big[E\Big[Y - E(Y\mid X) - E\{Y - E(Y\mid X)\mid X,T\}\mid X,T\Big]u(X,T)\Big]
\\
& = E\{0\cdot u(X,T)\} 
\\
& = 0.
\ee
Thus \eqref{poYXmain}  implies that $q(X,T)$ satisfies \eqref{thm:id:1} if and only if \eqref{qet} holds. We then have
\bee\label{qequiv}
q(X,T) &= E\big\{Y - E(Y\mid  X)\mid  X,T\big\}  
\\
&= E(Y \mid  X,T) - E(Y\mid  X)
\\
&=  E(Y \mid X,T) - E\big\{E(Y\mid X,T)\mid X\big\}
\\
& = \tau(X,T) - E\{\tau(X,T)\mid X\},\text{ a.s..}
\ee
 The last equality follows by $\tau(X,T) = E(Y\mid X,T) - E(Y\mid X,T = 0)$ under Assumptions \ref{A:UNC} and \ref{A:NI}. Since $\tau\in \mathcal{L}_{\mathcal{P}}^2(X,T)$, it is easy to verify that $q$ is also in $ \mathcal{L}_{\mathcal{P}}^2(X,T)$ based on  \eqref{qequiv} and the Cauchy-Schwarz inequality,
 \bee\nonumber
 E[\tau(X,T) - E\{\tau(X,T)\mid X\}]^2 &\leq 2E\{\tau^2(X,T)\} + 2E[E\{\tau(X,T)\mid X\}]^2
 \\
 &\leq 2E\{\tau^2(X,T)\} + 2E[E\{\tau^2(X,T)\mid X\}]
 \\
 & = 4E\{\tau^2(X,T)\}
 \\
 & < +\infty.
 \ee
 \par Comparing \eqref{thm:id:1} and \eqref{opt:1},  if  $h\in \mathcal{L}_{\mathcal{P}}^2(X,T)$ is a minimum of $L_c(h)$, $h(X,T) - \E\{h(X,T)\mid X\}$ must minimize the general problem in \eqref{thm:id:1}, i.e.,
\bee\nonumber
h(X,T) - \E\{h(X,T)\mid X\} &=q(X,T) 
\\
&=  \tau(X,T) - \E\{\tau(X,T)\mid X\}\quad \text{a.s.},
\ee
which is equivalent to 
\bee\label{thm:id:2}
h(X,T) = \tau(X,T) + \big[\E\{h(X,T)\mid X\} - \E\{\tau(X,T)\mid X\}\big]\quad \text{a.s.}
\ee 
by \eqref{qequiv}.
Thus any $h\in \mathcal{L}_{\mathcal{P}}^2(X,T)$ minimizing \eqref{opt:1} must satisfy 
\bee\nonumber
h(X,T) = \tau(X,T) + s(X)\quad \text{a.s.},
\ee
for some $s\in\mathcal{L}_{\mathcal{P}}^2(X)$ such that $s(x) = E\{h(X,T)\mid X = x\} - E\{\tau(X,T)\mid X = x\}$.
\par
On the other hand, for arbitrary $s\in\mathcal{L}_{\mathcal{P}}^2(X)$, if $h(X,T) = \tau(X,T) + s(X)$ a.s, $h$ must satisfy
\bee\nonumber
h(X,T) - \E\{h(X,T)\mid X\} &= \tau(X,T) - \E\{\tau(X,T)\mid X\}
\\
&= q(X,T),\,\, \text{a.s.},
\ee
recalling $q$ is the solution of the general minimization problem \eqref{thm:id:1}. It is also easy to see that $h\in\mathcal{L}_{\mathcal{P}}^2(X,T)$ as both $s$ and $\tau$ have bounded $\mathcal{L}_{\mathcal{P}}^2$ norms. Therefore, $h$ is a minimum of $L_c(h)$.
\par
Summarizing the above results, we conclude that the solution set in $\mathcal{L}_{\mathcal{P}}^2(X,T)$ that minimizes $L_c(h)$ is exactly 
$\mathcal{S} = \{h \mid h(X,T) = \tau(X,T) + s(X)\,\text{a.s., for any }  s(x)\in \mathcal{L}^2_{\mathcal{P}}(X)\}$. 
\par
\noindent\fbox{Heuristic proof of Proposition \ref{thm:id} (ii)}\label{sec:exp}
Before heading to the formal proof of  Proposition \ref{thm:id} (ii), we first give an explanation about how \eqref{eq:shape} narrows $\mathcal{S}$ to $\mathcal{S}^{\natural}$. This explanation also gives some intuitions for the formal proof. 
\par
Notice that any other solution $h$  in $\mathcal{S}$ but not in $\mathcal{S}^{\natural}$, satisfies that  $h(X,T) = \tau(X,T) + s(X)$ a.s., with some $s(X)$ that is not a.s. zero, i.e.,  $\Pr\big[\Omega = \{X\text{ satisfies } s(X)\neq 0\}\big] > 0$. Therefore, under the positivity assumption such that $\Pr(T = 0\mid X = x) > \epsilon'$ for some fixed $\epsilon' > 0$ and any $x\in\mathbb{X}$, we have 
\bee
\Pr\big\{ h(X,0) = \tau(X,0) + s(X) = s(X)\neq 0\big\} &\geq \Pr\big[\{T = 0\}\cap\Omega\big] = \Pr(T = 0\mid \Omega)\Pr(\Omega)
\\
&\geq \epsilon' \cdot \Pr(\Omega) > 0,
 \ee
 which is in conflict with  \eqref{eq:shape}  such that $h(X,0)  = 0$ with probability $1$.
 \par
\noindent\fbox{Formal proof of Proposition \ref{thm:id} (ii)}  Write the optimization problem   for the bianry treatment as
\bee\label{Lbh11}
\argmin_{h\in\{h\mid  h(\cdot,1)\in\mathcal{L}_{\mathcal{P}}^2(X) \text{ and }h(X,0) = 0 \text{ a.s.}\}}E\big[Y - m(X)-\{T - e(X)\}h(X,1)\big]^2.
\ee 
Since the objective function above only involves $h(\cdot,1)$, we consider a simplified problem
\bee\label{Lbh22}
\argmin_{h'\in\mathcal{L}_{\mathcal{P}}^2(X)}E\big[Y - m(X)-\{T - e(X)\}h'(X)\big]^2.
\ee 
Let  $\mathcal{S}'$ be the solution set of \eqref{Lbh22}, and let $\tilde{\mathcal{S}}^{\natural}$ be the solution set of $
\argmin_{h\in\mathcal{L}_b}L_b(h)
$, and thus is also the solution set of \eqref{Lbh11}. In the following, we show 
\bee\label{final:target:po1}
\tilde{\mathcal{S}}^{\natural} = \{h\mid h(X,T) = \tau(X,T)\ \text{a.s.}\} = {\mathcal{S}}^{\natural}
\ee and thus finish the proof.
\par Comparing \eqref{Lbh11} and \eqref{Lbh22}, we have that any $h$ being a solution of \eqref{Lbh11} must satisfy 
\bee\nonumber
h(X,1) = h'(X)\text{ a.s..}
\ee
On the other hand, by the constraint of \eqref{Lbh11}, we also have $h(X,0) = 0$ a.s., when $h\in \tilde{\mathcal{S}}^{\natural}$. Therefore we have 
\bee\label{SSrelation}
\tilde{\mathcal{S}}^{\natural} \subseteq \mathcal{S}^{\natural\natural}  = \{h\mid h(X,1) = h'(X)\text{ a.s. for any $h'\in\mathcal{S}'$, and }h(X,0) = 0\text{ a.s.}\}.
\ee It is also easy to check that if $h\in \mathcal{S}^{\natural\natural}$, $h(\cdot,1)$ must minimize the objective function in \eqref{Lbh11}. Meanwhile, if $h\in \mathcal{S}^{\natural\natural}$, we also have $h(\cdot,1)\in\mathcal{L}_{\mathcal{P}}^2(X)$ as $h'\in\mathcal{L}_{\mathcal{P}}^2(X)$ and $h(X,0) = 0$ a.s., thus $h$ satisfies all the constraints and is a solution of \eqref{Lbh11}. So $ \mathcal{S}^{\natural\natural}\subseteq \tilde{\mathcal{S}}^{\natural}$, and by \eqref{SSrelation}, we have 
\bee\label{Snncon}
\tilde{\mathcal{S}}^{\natural} = \mathcal{S}^{\natural\natural} = \{h\mid h(X,1) = h'(X)\text{ a.s. for any $h'\in\mathcal{S}'$, and }h(X,0) = 0\text{ a.s.}\}.
\ee
\par
For simplicity, we denote $\tau(x) = \tau(x,1)$. Therefore, we have $\tau(\cdot)\in\mathcal{L}_{\mathcal{P}}^2(X)$ as $\tau(\cdot,1)\in\mathcal{L}_{\mathcal{P}}^2(X)$. We first prove that, 
\bee\label{S'con}
\mathcal{S}' = \{h'\mid h'(X) = \tau(X)\text{ a.s.}\}. 
\ee
Any $h'\in\mathcal{L}_{\mathcal{P}}^2(X)$ can be written as $h'(x) = \tau(x) + s'(x)$, where $s'(x) = \tau(x) - h'(x)\in\mathcal{L}_{\mathcal{P}}^2(X)$ since both $h'(x)$ and $\tau(x)$ are in $\mathcal{L}_{\mathcal{P}}^2(X)$. Then solving \eqref{Lbh22} is equivalent to solving 
\bee\label{binary:target}
\argmin_{s'\in\mathcal{L}_{\mathcal{P}}^2(X)}E\big[Y - m(X)-\{T - e(X)\}\{s'(X) +\tau(X)\}\big]^2.
\ee
The above square loss function can be decomposed into
\bee\label{po:1:1}
&E\big[Y - m(X)-\{T - e(X)\}\{s'(X) +\tau(X)\}\big]^2
\\
&=E\Big[\big[Y - m(X)-\{T - e(X)\}\tau(X)\big] - \big[\{T - e(X)\}s'(X)\big]\Big]^2
\\
&=E\big[Y - m(X)-\{T - e(X)\}\tau(X)\big]^2
\\
&-2E\Big[\big[Y - m(X)-\{T - e(X)\}\tau(X)\big]\big[\{T - e(X)\}s'(X)\big]\Big]
\\
&+E\big[\{T - e(X)\}s'(X)\big]^2.
\ee
For the second term on the right-hand side of \eqref{po:1:1},
\bee\nonumber
&E\Big[\big[Y - m(X)-\{T - e(X)\}\tau(X)\big]\big[\{T - e(X)\}s'(X)\big]\Big]
\\
&=E\Big[E\big[Y - m(X)-\{T - e(X)\}\tau(X)\mid X,T\big]\big[\{T - e(X)\}s'(X)\big]\Big]
\\
&=E\Big[0\cdot\big[\{T - e(X)\}s'(X)\big]\Big]
\\
& = 0,
\ee
where the first equality follows by the law of total expectation, and the second inequality follows by
\bee\nonumber
&E\big[Y - m(X)-\{T - e(X)\}\tau(X)\mid X,T\big]
\\
&=\mu(X,T) - m(X) - \{T - e(X)\}\tau(X)
\\
&=\E(Y\mid X,T) - E(Y\mid X) - \{T - \Pr(T = 1\mid X)\}\big\{E(Y\mid X,T = 1)-E(Y\mid X,T = 0)\big\}
\\
&=\E(Y\mid X,T) - \{\Pr(T = 1\mid X)E(Y\mid X,T = 1) +\Pr(T = 0\mid X)E(Y\mid X,T = 0)\}
\\
&- \{T - \Pr(T = 1\mid X)\}\big\{E(Y\mid X,T = 1)-E(Y\mid X,T = 0)\big\}
\\
&=T\E(Y\mid X,T = 1) + (1 - T)\E(Y\mid X,T = 0)
\\
& -\Pr(T = 1\mid X)E(Y\mid X,T = 1)  - \{1-\Pr(T = 1\mid X)\}E(Y\mid X,T = 0)
\\
&- \{T - \Pr(T = 1\mid X)\}\big\{E(Y\mid X,T = 1)-E(Y\mid X,T = 0)\big\}
\\
&=T\E(Y\mid X,T = 1)  - T\E(Y\mid X,T = 0)
\\
& -\Pr(T = 1\mid X)E(Y\mid X,T = 1)  +\Pr(T = 1\mid X)E(Y\mid X,T = 0)
\\
&- \{T - \Pr(T = 1\mid X)\}\big\{E(Y\mid X,T = 1)-E(Y\mid X,T = 0)\big\}
\\
&=0,
\ee
where the first equality follows by the definition of $\mu(x,t)$, the third equality follows by the law of total expectation, and the fourth equality follows by $\Pr(T = 0\mid X = x) = 1 - \Pr(T = 1\mid X = x)$ and 
\bee\nonumber
\E(Y\mid X,T) &= I(T = 1)\E(Y\mid X,T = 1) + I(T = 0)\E(Y\mid X,T = 0)
\\
&= T\E(Y\mid X,T = 1) + (1 - T)\E(Y\mid X,T = 0).
\ee
For the third term on the right-hand side of \eqref{po:1:1}, if $s(X)$ is not $0$ a.s.,
\bee\nonumber
&E\big[\{T - e(X)\}^2\{s'(X)\}^2\big]
\\
& = E\big[E\big[\{T - e(X)\}^2\mid X\big]\{s'(X)\}^2\big]
\\
&>(\epsilon')^3 E\{s'(X)\}^2
\\
& >0,
\ee
where the first equality follows by the law of total expectation, the first inequality follows by that  $E\big[\{T - e(X)\}^2\mid X = x\big] \geq \Pr(T = 1\mid X = x)E\big[\{1 - e(X)\}^2\mid X = x,T = 1\big] \geq (\epsilon')^3$ for any $x\in\mathbb{X}$ as $e(x)\in(\epsilon',1 - \epsilon')$, and the last inequality is because  $E\{s'(X)\}^2 > 0$ when $s'(X)$ is not $0$ a.s..
\par Summarizing the above results,  when $s'(X)$ is not $0$ a.s., we have
\bee\nonumber
&E\big[Y - m(X)-\{T - e(X)\}\{s'(X) +\tau(X)\}\big]^2
\\
&=E\big[Y - m(X)-\{T - e(X)\}\tau(X)\big]^2
\\
&+E\big[\{T - e(X)\}s'(X)\big]^2
\\
&>E\big[Y - m(X)-\{T - e(X)\}\tau(X)\big]^2.
\ee
That is, \eqref{binary:target} is solved if and only if $s'(X) = 0$ a.s., which is equivalent to that \eqref{Lbh22} is solved if and only if $h'(x)$ satisfies that $h'(X) = \tau(X)$ a.s.. Thus \eqref{S'con} is verified as desired. Combining \eqref{Snncon} and \eqref{S'con}, we have
\bee\nonumber
\tilde{\mathcal{S}}^{\natural}  = \{h\mid h(X,1) = \tau(X,1)\text{ a.s., and }h(X,0) = 0\text{ a.s..}\}.
\ee

Now we  aim at showing that, under the positivity assumption,
\bee\label{po1:target}
\{h\mid h(X,1) = \tau(X,1)\text{ a.s., and }h(X,0) = 0\text{ a.s.}\} = \{h\mid h(X,T) = \tau(X,T)\ \text{a.s.}\},
\ee
and thus show \eqref{final:target:po1} and finish the proof. First suppose $h \in \{h\mid h(X,1) = \tau(X,1)\text{ a.s., and }h(X,0) = 0\text{ a.s.}\}$. Then we have two subsets of $\mathbb{X}$, $\mathcal{X}_1,\mathcal{X}_2\subseteq \mathbb{X}$ such that $h(x,1) = \tau(x,1)$ when $x \in \mathcal{X}_1$, $h(x,0) = 0$ when $x \in \mathcal{X}_0$, and $\Pr(X\in\mathcal{X}_1) = \Pr(X\in\mathcal{X}_2) = 1$. Thus we have when $(x,t)\in \mathcal{X}_1\cap \mathcal{X}_2\times \{0,1\}$,
\bee\nonumber
h(x,t) = \tau(x,t).
\ee
On the other hand, since  $\Pr[(X,T)\in \mathcal{X}_1\cap \mathcal{X}_2\times \{0,1\}] = \Pr(X\in\mathcal{X}_1\cap \mathcal{X}_2) = 1$, we have $h(X,T) = \tau(X,T)$ a.s., which implies
\bee\label{po:final1}
\{h\mid h(X,1) = \tau(X,1)\text{ a.s., and }h(X,0) = 0\text{ a.s.}\} \subseteq \{h\mid h(X,T) = \tau(X,T)\ \text{a.s.}\}.
\ee
Second, suppose $h \in \{h\mid h(X,T) = \tau(X,T)\ \text{a.s.}\}$. Then there exists some $\Omega_1\in\mathbb{X}\times \mathbb{T}$ such that $\Pr\{(X,T)\in \Omega_1\} = 1$, and  when $(x,t)\in \Omega_1$ one has $h(x,t) = \tau(x,t)$ and $h(x,t) \neq \tau(x,t)$ otherwise. We now define two marginal sets,
\bee\nonumber
\mathcal{X}_3^{(0)} = \{x\mid (x,0) \in \Omega_1\},\text{ and }\mathcal{X}_3^{(1)} = \{x\mid (x,1) \in \Omega_1\}.
\ee
Now we prove $\Pr(X\in \mathcal{X}_3^{(0)}) = \Pr(X\in \mathcal{X}_3^{(1)}) = 1$ by contradiction. Assume $\Pr(X\in \mathcal{X}_3^{(0)}) < 1$. We have
\bee\label{keypo1}
\Pr\big[(X,T)\in (\mathbb{X}\setminus \mathcal{X}_3^{(0)})\times \{0\}\big] &= \Pr(X\in \mathbb{X}\setminus \mathcal{X}_3^{(0)})\Pr(T = 0\mid \text{given }X\in \mathbb{X}\setminus \mathcal{X}_3^{(0)})
\\
&\geq\big\{1 - \Pr(X\in \mathcal{X}_3^{(0)})\big\}\cdot \epsilon'
\\
& > 0.
\ee
By definition, we know  $(\mathbb{X}\setminus \mathcal{X}_3^{(0)})\times \{0\} \cap \Omega_1 = \varnothing$. Thus \eqref{keypo1} implies that with probability larger than $0$, we have $(X,T)\in (\mathbb{X}\times\mathbb{T})\setminus\Omega_1$, which is in conflict with $\Pr\{(X,T)\in \Omega_1\} = 1$. So we conclude $\Pr(X\in \mathcal{X}_3^{(0)}) = 1$. Recall that when $x\in\mathcal{X}_3^{(0)}$, we have $(x,0)\in \Omega_1$ and thus
$
h(x,0) = \tau(x,0) = 0.
$ We then have
\bee\nonumber
h(X,0) = 0\text{ a.s..}
\ee
With the same argument, we can show  $\Pr(X\in \mathcal{X}_3^{(1)}) = 1$ and thus $h(X,1) = \tau(X,1)\text{ a.s..}$ So we have $h \in\{h\mid h(X,1) = \tau(X,1)\text{ a.s., and }h(X,0) = 0\text{ a.s.}\}$ and thus
\bee\label{po:final2}
\{h\mid h(X,T) = \tau(X,T)\ \text{a.s.}\}\subseteq \{h\mid h(X,1) = \tau(X,1)\text{ a.s., and }h(X,0) = 0\text{ a.s.}\}.
\ee
Combining \eqref{po:final1} and \eqref{po:final2}, we thus show \eqref{po1:target} and \eqref{final:target:po1} and thereby complete the proof.
\QED

\subsection{Proof of Proposition~\ref{po:nonunique}}
Since the density function of $(X,T)$ is uniformly upper bounded, we have that a.s., 
$$
\check{\tau}(X,T\mid s) = \tau(X,T) + s(X),
$$ 
for any $s\in\mathcal{L}_{\mathcal{P}}^2(X)$, which implies that $\check{\tau}(x,t\mid s)\in \mathcal{S}$. By Proposition~\ref{thm:id}(i), we  have that $\check{\tau}(x,t\mid s)\in \mathcal{L}_{\mathcal{P}}^2(X,T)\cap\{h\mid h(x,0) = 0\}\subseteq \mathcal{L}_{\mathcal{P}}^2(X,T)$, solves \eqref{opt:1} and thus solves \eqref{opt:1:rec}.
\QED

\subsection{Proof of Theorem \ref{pro:tech}}\label{SSthm1pf}
With a basic decomposition, one has 
\bee\label{pf:po1:1}
&L_c(h) 
\\
&=  E\big[\tau(X,T) - E\{\tau(X,T)\mid X\}-[h(X,T) - E\{h(X,T)\mid X\}]\big]^2
\\
&+ E\big[Y - m(X)-[\tau(X,T) - E\{\tau(X,T)\mid X\}]\big]^2
\\
&+ 2 E\Big[\big[\tau(X,T) - E\{\tau(X,T)\mid X\}-[h(X,T) - E\{h(X,T)\mid X\}]\big]\big[Y - m(X)-[\tau(X,T) - E\{\tau(X,T)\mid X\}]\big]\Big].
\ee
Under Assumptions \ref{A:UNC} and \ref{A:NI} and by the law of total expectation, one has
\bee\nonumber
\E\big\{Y - m(X)\mid X,T\big\} &= \E\big(Y\mid X,T\big) - \E\big(Y\mid X,T = 0\big) - \big\{E(Y\mid X) - \E\big(Y\mid X,T = 0\big)\big\}
\\
&=\tau(X,T) - \big[E\{E(Y\mid X,T)\mid X\} - \E\big(Y\mid X,T = 0\big)\big]
\\
& = \tau(X,T) - E\big\{\E\big(Y\mid X,T\big) - \E\big(Y\mid X,T = 0\big)\mid X\big\}
\\
& = \tau(X,T) - E\big\{\tau(X,T)\mid X\big\},
\ee
and therefore,
\bee\nonumber
&E\Big[\big[\tau(X,T) - E\{\tau(X,T)\mid X\}-[h(X,T) - E\{h(X,T)\mid X\}]\big]\big[Y - m(X)-[\tau(X,T) - E\{\tau(X,T)\mid X\}]\big]\Big]
\\
&=E\Big[\big[\tau(X,T) - E\{\tau(X,T)\mid X\}-[h(X,T) - E\{h(X,T)\mid X\}]\big]
\\
&\cdot E\big[Y - m(X)-[\tau(X,T) - E\{\tau(X,T)\mid X\}]\mid X,T\big]\Big]
\\
&=E\Big[\big[\tau(X,T) - E\{\tau(X,T)\mid X\}-[h(X,T) - E\{h(X,T)\mid X\}]\big]
\\
&\cdot \underbrace{\big[\E\big\{Y - m(X)\mid X,T\big\}-[\tau(X,T) - E\{\tau(X,T)\mid X\}]\big]}_{ = 0}\Big]
\\
& = 0.
\ee
Combining \eqref{pf:po1:1} with the above display, one has
\bee\nonumber
L_c(h) &= E\big[\tau(X,T) - E\{\tau(X,T)\mid X\}-[h(X,T) - E\{h(X,T)\mid X\}]\big]^2
\\
&+ E\big[Y - m(X)-[\tau(X,T) - E\{\tau(X,T)\mid X\}]\big]^2.
\ee
For simplicity, we define an operator $\Pi(\cdot)$ for any $h\in\mathcal{L}^2_{\mathcal{P}}(X,T)$ such that $$\Pi(h)(x,t) = h(x,t) - E\{h(X,T)\mid X = x\}.$$ Therefore, we have,
\bee\nonumber
L_{c,\ell_2}(h\mid \rho) &= L_c(h) +\rho \|h\|_{\mathcal{L}^2_{\mathcal{P}}}^2
\\
&=E\big[\tau(X,T) - E\{\tau(X,T)\mid X\}-[h(X,T) - E\{h(X,T)\mid X\}]\big]^2
\\
&+ E\big[Y - m(X)-[\tau(X,T) - E\{\tau(X,T)\mid X\}]\big]^2 
\\
& +\rho \|h\|_{\mathcal{L}^2_{\mathcal{P}}}^2
\\
&=E\big[\Pi(h)(X,T) - [\tau(X,T) - E\{\tau(X,T)\mid X\}]\big]^2 
\\
&+ \rho\|h\|^2_{\mathcal{L}_{\mathcal{P}}^2}
\\
&+E\big[Y - m(X)-[\tau(X,T) - E\{\tau(X,T)\mid X\}]\big]^2,
\ee
which  implies that     $\argmin_{h\in\mathcal{L}_{\mathcal{P}}^2(X,T)}L_{c,\ell_2}(h\mid \rho)$ is equivalent to
\bee\label{thm:trans}
&\argmin_{h\in\mathcal{L}_{\mathcal{P}}^2(X,T)}E\big[\Pi(h)(X,T) - [\tau(X,T) - E\{\tau(X,T)\mid X\}]\big]^2 + \rho\|h\|^2_{\mathcal{L}_{\mathcal{P}}^2}
\\
&=\argmin_{h\in\mathcal{L}_{\mathcal{P}}^2(X,T)}\mathcal{F}(h),
\ee
where we define $\mathcal{F}(h)= E\big[\Pi(h)(X,T) - \tilde{\tau}(X,T)\big]^2 + \rho E\big\{h(X,T)\big\}^2$ and $\tilde{\tau}(x,t) = \tau(x,t) - E\{\tau(X,T)\mid X =  x\}$.
\par
Next, we prove $L_{c,\ell_2}(h\mid \rho)$ has a unique minimum among $\mathcal{L}^2_{\mathcal{P}}(X,T)$. By above derivations, we only need to show the minimum of \eqref{thm:trans} is unique. An argument similar to the proof of Lemma 1.1 in \citet{tikhonov1995numerical} can shows that $\mathcal{F}(h)$ has a unique solution. In particular, we first show $\Pi(\cdot)$ is a self-adjoint operator; see, e.g., \citet[$\mathsection$12.11]{rudin} for the definition of a self-adjoint operator. For any $h_1, h_2\in \mathcal{L}^2_{\mathcal{P}}(X,T)$, we have
\bee\nonumber
\big< \Pi(h_1),h_2\big>_{\mathcal{L}^2_{\mathcal{P}}(X,T)} &= E\Big[\Pi(h_1)(X,T)h_2(X,T)\Big]
\\
&=E\Big[\big[h_1(X,T) - E\{h_1(X,T)\mid X\}\big]h_2(X,T)\Big]
\\
&=E\Big[\big[h_1(X,T)h_2(X,T)\Big] - E\Big[E\{h_1(X,T)\mid X\}h_2(X,T)\Big]
\\
&=E\Big[\big[h_1(X,T)h_2(X,T)\Big] - E\Big[E\{h_1(X,T)\mid X\}E\{h_2(X,T)\mid X\}\Big],
\ee
where $\big<\cdot,\cdot\big>_{\mathcal{L}^2_{\mathcal{P}}(X,T)}$ denotes the inner product of the Hilbert space $\mathcal{L}^2_{\mathcal{P}}(X,T)$, and the last equality follows by the law of total expectation. By symmetry, we can also show
\bee\nonumber
\big<h_1,\Pi(h_2)\big>_{\mathcal{L}^2_{\mathcal{P}}(X,T)} &= E\Big[h_1(X,T)h_2(X,T)\Big] - E\Big[E\{h_1(X,T)\mid X\}E\{h_2(X,T)\mid X\}\Big]
\\
&=\big< \Pi(h_1),h_2\big>_{\mathcal{L}^2_{\mathcal{P}}(X,T)},
\ee
which by definition \citep[ e.g.,][$\mathsection$12.11]{rudin}, implies that $\Pi(\cdot)$ is a self-adjoint operator from  $\mathcal{L}^2_{\mathcal{P}}(X,T)$ to $\mathcal{L}^2_{\mathcal{P}}(X,T)$, i.e. the adjoint operator of $\Pi(\cdot)$ is still  $\Pi(\cdot)$. Then, similar to the proof of Lemma 1.1 in \citet{tikhonov1995numerical}, we can denote the second-order Fr\'echet derivate of $\mathcal{F}(h)$ by $D^2 \mathcal{F}(h)$ and have that 
\bee\nonumber
D^2 \mathcal{F}(h) &= 2\Pi\big\{\Pi(h)\big\} + 2\rho h
\\
&=2\Pi(h) + 2\rho h
\\
&=2h- 2E\{h(X,T)\mid X = \cdot\} + 2\rho h,
\ee
where the second equality follows by $\Pi\{\Pi(\cdot)\} = \Pi(\cdot)$ after checking the definition. We now have
\bee\label{hahagaoding}
&\big<D^2 \mathcal{F}(h),h\big>_{\mathcal{L}^2_{\mathcal{P}}(X,T)}
\\
& = 2\Big[E\{h(X,T)\}^2 -E\big[h(X,T)E\{h(X,T)\mid X \}\big]\Big] + 2\rho E\{h(X,T)\}^2
\\
& = 2\Big[E\{h(X,T)\}^2 -E\big[E\{h(X,T)\mid X\}\big]^2\Big] + 2\rho E\{h(X,T)\}^2
\\
&\geq 2\rho E\{h(X,T)\}^2
\\
&>0,
\ee
where the first equality follows by the law of expectation, the first inequality follows by $E\big[E\{h(X,T)\mid X\}\big]^2 \leq E\{h(X,T)\}^2$ due to Cauchy-Schwarz inequality.  \eqref{hahagaoding} implies that $\mathcal{F}(h)$ is a strictly convex functional, and thus $\mathcal{F}(h)$ has a unique minimum among $\mathcal{L}^2_{\mathcal{P}}(X,T)$; see, e.g.,  \citet{zeidler2013nonlinear}. Thus \eqref{id:eq2} has a unique solution, namely, $\tau_\rho(X,T)$ in $\mathcal{L}^2_{\mathcal{P}}(X,T)$.
\par
Next we derive the concrete form of $\tau_\rho$. Since we already know $\tau_\rho(X,T)$ is unique and is the minimum of $\mathcal{F}(h)$, we know that
\bee\nonumber
\mathcal{F}(\tau_\rho + a\cdot\alpha) \geq \mathcal{F}(\tau_\rho)
\ee
for any $a\in \RR$ and $\alpha\in\mathcal{L}_{\mathcal{P}}^2(X,T)$. This implies 
\bee\nonumber
 0 & = \frac{d}{d a}\mathcal{F}(\tau_\rho+a\cdot\alpha) \,\Big|_{a = 0}
\\
&= \E\left[\{Y - m(X) - \tau_\rho(X,T) + \E\{\tau_\rho(X,T)\mid X\}\}\cdot\{-\alpha(X,T) + \E\{\alpha(X,T)\mid X\}\}\right] + \rho\E\{\tau_\rho(X,T)\alpha(X,T)\}
\\
&= \E\Big[ [-\mu(X,T) + m(X) + (\rho+1)\tau_\rho(X,T) - \E\{\tau_\rho(X,T)\mid X\}]\{\alpha(X,T) \}\Big]. 
\ee
Let $\alpha(X,T) =-\mu(X,T) + m(X) + (\rho+1)\tau_\rho(X,T) - \E\{\tau_\rho(X,T)\mid X\} $, then above display implies that
\bee\label{S113}
 & (\rho+1)\tau_\rho(X,T) - \E\{\tau_\rho(X,T)\mid X\} = \mu(X,T) - m(X)
\\
\Rightarrow & \E\left[ (\rho+1)\tau_\rho(X,T) - \E\{\tau_\rho(X,T)\mid X\}\mid X\right] = 0
\\
\Rightarrow & \E[\tau_\rho(X,T)\mid X] = 0,
\ee
with $\rho > 0$. Then the first equation and third equation of \eqref{S113} together imply that
\bee\nonumber
\tau_\rho(X,T) &= \frac{\mu(X,T)-m(X)}{\rho + 1} 
\\
&= \frac{\E(Y^{(T)} - Y^{(0)}\mid X) - \E_{T'\sim\varpi(\cdot\mid X)}(Y^{(T')} - Y^{(0)}\mid X)}{\rho + 1}
\\
&=\frac{\tau(X,T)- \E_{T'\sim\varpi(\cdot\mid X)}\{\tau(X,T')\mid X\}}{1+\rho} 
\\
&= \frac{\tilde{\tau}(X,T)}{1+\rho}.
\ee
Since we have proved $\tau_{\rho}(X,T)$ is unique, the above display gives the concrete form of $\tau_\rho(X,T)$.

\QED
\subsection{Proof of Theorem~\ref{thm:onetwoid}}
We have a.s.,
\bee\label{has,eq}
h(X,T) = \tau(X,T)  + {s} (X),
\ee
for some $ {s}(x)  \in\mathcal{L}^2_{\mathcal{P}}(X)$. 
 We first prove a claim such that, with probability $1$,
\bee\label{hXteq}
h(X,t) =  \tau(X,t) +  {s} (X),
\ee
holds almost everywhere over $t\in\mathbb{T}$. We  prove this   claim by contradiction. Suppose this claim does not hold. Then there exists $\Theta\subseteq \mathbb{X}$ such that $\Pr(X\in\Theta) = p_X > 0$, and for any $x\in \Theta_X$ there exists $\mathbb{T}_x\subseteq \mathbb{T}$ such that $\lambda(\mathbb{T}_x)  > 0$ and 
$$
h(x,t) \neq  \tau(x,t) +  {s}(x)\text{ for all }t\in\mathbb{T}_x,
$$
where $\lambda(\cdot)$ is the Lebesgue measure over $\mathbb{R}$. Now we denote $\Theta({\bar\epsilon}) = \{x\mid x\in\Theta,\lambda(\mathbb{T}_x) > {\bar\epsilon}\}$, and  events: $\mathcal{E} = \{X\in\Theta\}$ and $\mathcal{E}_{\bar\epsilon} = \{X\in\Theta({\bar\epsilon})\}$. By definition, we have $\lim_{{\bar\epsilon} \rightarrow 0}\mathcal{E}_{\bar\epsilon} = \mathcal{E}$, and by the continuity of probability measure \citep[Theorem 1.8]{wasserman2004all},
$$
\lim_{{\bar\epsilon} \rightarrow 0}\Pr\{X\in\Theta({{\bar\epsilon}})\} = \lim_{{\bar\epsilon} \rightarrow 0}\Pr(\mathcal{E}_{{\bar\epsilon}}) = \lim_{{\bar\epsilon} \rightarrow 0}\Pr(\mathcal{E}) = \Pr(X\in\Theta) = p_X.
$$ 
This implies that there exists some ${\bar\epsilon}' > 0$ such that $\Pr\{X\in\Theta({{\bar\epsilon}'})\} \geq p_X/2$. Then we have 
\bee\label{eq:contract}
\Pr\big\{ h(X,T) \neq  \tau(X,T) +  {s}(X)\big\} &\geq \Pr\big\{X\in\Theta({{\bar\epsilon}')}\text{ and }T\in \mathbb{T}_{X} \big\} 
\\
&=\Pr\big\{X\in\Theta({{\bar\epsilon}'})\big\}\Pr\big\{T\in \mathbb{T}_{X} \mid X\in\Theta({{\bar\epsilon}'}) \big\}
\\
&\geq \epsilon{\bar\epsilon}' p_X/2 > 0,
\ee
under Assumption~\ref{A:CS}. Equation \eqref{eq:contract} is in contradiction with \eqref{has,eq}, and thus we prove \eqref{hXteq}. Since \eqref{hXteq} holds a.e. over $\mathbb{T}$ and $\tau$ is continuous at $t=0$, we have that a.s.,
\bee\label{has,eq2}
h(X,0) = \tau(X,0) + s(X) = s(X). 
\ee
Thus we have by \eqref{has,eq} and \eqref{has,eq2}, a.s.,
\bee\nonumber
\mathscr{C} (h)(X,T) = (1+\rho)\cdot (1+\rho)^{-1}\big\{h(X,T) - h(X,0)\big\}
 =  {\tau}(X,T).
\ee
\subsection{Proof of Proposition \ref{po:nui:equal2}}\label{sec:po4}
First we note that when estimating   
 $
\hat{\Gamma}(x)$ through  $\hat{\Gamma}(x) = \E_{\hat{\varpi}}\{ \Psi(X,T)\mid X = x\}$, we have
\bee\nonumber
\hat{\Gamma}(x) &= \E_{\hat{\varpi}}\{ \psi(T)\otimes \Psi(X)\mid X = x\}
\\
&=\E_{\hat{\varpi}}\{ \psi(T)\mid X = x\}\otimes \Psi(x),
\ee
where we denote $\Psi(x) = \psi(x^{(1)})\otimes \cdots\otimes\psi(x^{(d)})$. Correspondingly, $\Gamma(x) = E_{{\varpi}}\{\psi(T)\mid X = x\}\otimes \Psi(x)$ and thus 
\bee\nonumber
\hat{\Gamma}(x) - \Gamma(x) &= [E_{\hat{\varpi}}\{\psi(T)\mid X = x\} - E_{\varpi}\{\psi(T)\mid X = x\}]\otimes \Psi(x)
\\
&= \Big[\int_{\mathbb{T}} \psi(t)\{\hat{\varpi}( t\mid  x) - {\varpi}( t\mid x)\}dt\Big]\otimes \Psi(x)
\\
&= \int_{\mathbb{T}} \big[\psi(t)\otimes \Psi(x)\big]\{\hat{\varpi}(t\mid x) - {\varpi}( t\mid x)\}dt
\\
&= \int_{\mathbb{T}}\Psi(x,t)\{\hat{\varpi}(t\mid x) - {\varpi}(t\mid x)\}dt.
\ee 
On the other hand,
\bee\nonumber
\sup_{x\in\mathbb{X}}\int_{\mathbb{T}}\{\hat\varpi(t\mid x) - \varpi(t\mid x)\}^2dt &=\sup_{x\in\mathbb{X}}\|\hat{\varpi}(\cdot\mid x) - {\varpi}(\cdot\mid x)\|^2_{\mathcal{L}^2_{}} 
\\
&= o_{{P}}(r_{\varpi}^2).
\ee 
First we  have,
\begin{align}\nonumber
&\big\|{P}\big[\{\hat{\M\Gamma}(X) - {\M\Gamma}(X)\}\{\hat{\M\Gamma}(X) - \M\Gamma(X)\}^{\T}\big]\big\|_2 
\\\nonumber
&= \sup_{\|\ell\| = 1}P\big[\ell^\T\{\hat{\M\Gamma}(X) - {\M\Gamma}(X)\}\big]^2
\\\nonumber
&=\sup_{\|\ell\| = 1}P\Big[\int_{\mathbb{T}}\ell^\T\Psi(X,t)\{\hat{\varpi}(t\mid X) - {\varpi}(t\mid X)\}dt\Big]^2
\\\nonumber
&\leq \sup_{\|\ell\| = 1}P\Big[\int_{\mathbb{T}}\{\ell^\T\Psi(X,t)\}^2dt\int_{\mathbb{T}}\{\hat{\varpi}(t\mid X) - {\varpi}(t\mid X)\}^2dt\Big]
\\\nonumber
&\leq \sup_{x\in\mathbb{X}}\big\|\hat{\varpi}(\cdot \mid x) - {\varpi}(\cdot \mid x)\big\|_{\mathcal{L}^2}^2 \cdot \sup_{\|\ell\| = 1}P\Big[\int_{\mathbb{T}}\{\ell^\T\Psi(X,t)\}^2dt\Big]
\\\nonumber
&\leq \sup_{x\in\mathbb{X}}\big\|\hat{\varpi}(\cdot \mid x) - {\varpi}(\cdot \mid x)\big\|_{\mathcal{L}^2}^2 \cdot \sup_{\|\ell\| = 1}\Big[\int_{\mathbb{X}\times\mathbb{T}}\{\ell^\T\Psi(x,t)\}^2dxdt\Big]\cdot C_f \quad(\text{Assumption \ref{am:densX}})
\\\nonumber
&\precsim \sup_{x\in\mathbb{X}}\big\|\hat{\varpi}(\cdot \mid x) - {\varpi}(\cdot \mid x)\big\|_{\mathcal{L}^2}^2 \quad(\text{Lemma \ref{am:psi}})
\\\nonumber
&=o_P(r_{\varpi}^2),
\end{align}
which directly yields $\big\|{P}\big[\{\hat{\M\Gamma}(X) - {\M\Gamma}(X)\}\{\hat{\M\Gamma}(X) - \M\Gamma(X)\}^{\T}\big]\big\|^{1/2}_2  = o_{P}(r_{\varpi})$. 
Second, by the property of Kronecker product \citep{schacke2004kronecker}, we have 
\bee\nonumber
\big\|\hat{\M\Gamma} - \M\Gamma\big\|_{\mathbb{X}}^2 
&= \sup_{x\in\mathbb{X}}\Big\|\int_{\mathbb{T}} \Psi(x,t)\{\hat\varpi(t\mid x) - \varpi( t\mid x)\}dt\Big\|^2
\\
&\leq \|\Psi\|_{\mathbb{X}\times \mathbb{T}}^2 \cdot\sup_{x\in\mathbb{X}}\int_{\mathbb{T}}\big\{\hat\varpi(t\mid  x) - \varpi(t\mid x)\big\}^2dt
\\
&=o_\p(K r_{\varpi}^2),
\ee
where the last equality is because that by Lemma \ref{am:psi}, we have $\|\Psi\|_{\mathbb{X}\times\mathbb{T}}^2 \precsim K$. We thus show $\|\hat{\M\Gamma} - \M\Gamma\|_{\mathbb{X}} /\sqrt{K}= o_{P}(r_{\varpi})$. 
\QED
\subsection{Proof of Theorem \ref{thm:main}}
 
For fixed $(x_0,t_0)\in\mathbb{X}\times\mathbb{T}$, we decompose
\bee\label{thm2P1}
&\big|\hat{\tau}(x_0,t_0) - \tau(x_0,t_0)\big|
\\
&=\Big|\big\{\Psi(x_0,t_0) - \Psi(x_0, 0)\big\}^\T\big\{( 1+ \rho)\hat{\phi} - \phi^*\big\}   + \big\{\Psi(x_0,t_0) - \Psi(x_0,0)\big\}^\T\phi^* - \tau(x_0,t_0)\Big|
\\
&\leq \Big|\big\{\Psi(x_0,t_0) - \Psi(x_0,0)\big\}^\T\big\{( 1+ \rho)\hat{\phi} - \phi^*\big\}\Big|+ \Big|\big\{\Psi(x_0,t_0) - \Psi(x_0,0)\big\}^\T\phi^* - \tau(x_0,t_0)\Big|
\\
&=T_1 + T_2.
\ee
The second term is the bias term, which can be bounded by Proposition \ref{lm:psieapprox}. In particular, recalling Theorem \ref{pro:tech} that $\tau(x_0,t_0) = \tilde{\tau}(x_0,t_0) - \tilde{\tau}(x_0,0)$, we have
\bee\label{bias:l2}
T_2 &\leq \Big|\Psi^\T(x_0,t_0)\phi^* - \tilde{\tau}(x_0,t_0)\Big| + \Big|\Psi^\T(x_0,0)\phi^* - \tilde{\tau}(x_0,0)\Big|
\\
&\leq 2\|(\phi^*)^\T\Psi - \tilde{\tau}\|_{\mathbb{X}\times\mathbb{T}} 
\\
&\precsim K^{-p/(d + 1)}
\ee
by Proposition \ref{lm:psieapprox}. 
\par
Next we bound the term $T_1$. Since $
 U U^\T + U_\perp U_\perp^\T = I$, we observe that
\bee\label{l2:simplify}
T_1&=\big\{\Psi(x_0,t_0) - \Psi(x_0,0)\big\}^\T\big\{( 1+ \rho)\hat{\phi} - \phi^*\big\} 
\\
&=\big\{\Psi(x_0,t_0) - \Psi(x_0,0)\big\}^\T(UU^\T + U_\perp U_\perp^\T)\big\{( 1+ \rho)\hat{\phi} - \phi^*\big\}
\\
&=\Big\{\M\Psi^\T(x_0,t_0)U - \M\Psi^\T(x_0,0)U\Big\}U^\T\big\{( 1+ \rho)\hat{\phi} - \phi^*\big\}  
+ \big\{\Psi^\T(x_0,t_0)U_\perp - \Psi^\T(x_0,0)U_\perp\big\}U_\perp^\T\big\{( 1+ \rho)\hat{\phi} - \phi^*\big\} 
\\
&= \Big\{\M\Psi^\T(x_0,t_0)U - \M\Psi^\T(x_0,0)U\Big\}U^\T\big\{( 1+ \rho)\hat{\phi} - \phi^*\big\}
\\
& = v_n^\T U^\T\big\{( 1+ \rho)\hat{\phi} - \phi^*\big\} ,
\ee 
where we denote $v_n = U^\T\M\Psi(x_0,t_0) - U^\T\M\Psi(x_0,0)$ for simplicity, and thus
\bee\label{vn:bound}
\|v_n\| \leq \|U\|_2\big \{\|\Psi(x_0,t_0)\| + \|\Psi(x_0,0)\|\big\} \precsim \sqrt{K},
\ee
by Lemma \ref{am:psi}. The second equality of \eqref{l2:simplify} is because that by Lemma \ref{lm:svdR}, $\Psi^\T(x,t)U_\perp$ are functions free of $t$ and thus $\Psi^\T(x_0,t)U_\perp$ is the same for all $t\in\mathbb{T}$, which implies that
\bee\nonumber
\Psi^\T(x_0,t_0)U_\perp - \Psi^\T(x_0,0)U_\perp &=  \Psi^\T(x_0,0)U_\perp- \Psi^\T(x_0,0)U_\perp 
\\
&=  0.
\ee
\par
Now we focus on  bounding $|T_1|$. 
Recall the form of $\hat{\phi}$ in \eqref{def:phi} and the simplifying setting of training nuisance functions by a single separate dataset ($\mathsection$\ref{sec:prea}). We then have
\bee\label{decom:main}
T_1& = v_n^\T U^\T\big\{( 1+ \rho)\hat{\phi} - \phi^*\big\}
\\
&= v_n^\T U^\T\hat{G}_n^{-1}\Big[(1+\rho)\p_n\big[\big\{Y - \hat{m}(X)\}\{\Psi(X,T) - \hat{\Gamma}(X)\big\} \big] - \hat{G}_n\phi^*\Big]
\\
&=v_n^\T U^\T\hat{G}_n^{-1}\p_n\Big[\big[Y - \hat{m}(X) - \{\Psi(X,T) - \hat{\Gamma}(X)\}^\T{\phi^*}\big]\big\{\Psi(X,T) - \hat{\Gamma}(X)\big\}\Big]
\\
&\quad  + \rho \cdot v_n^\T U^\T\hat{G}_n^{-1}\p_n\Big[\big[Y - \hat{m}(X)- \{\Psi(X,T) - \hat{\Gamma}(X)\}^\T{\phi^*}  \big]\big\{\Psi(X,T) - \hat{\Gamma}(X)\big\}\Big]  
\\
&\quad  + \rho \cdot v_n^\T U^\T\hat{G}_n^{-1}\p_n\Big[ \{\Psi(X,T) - \hat{\Gamma}(X)\}\big\{\Psi(X,T) - \hat{\Gamma}(X)\big\}^\T{\phi^*}  - \Psi(X,T)   \Psi^\T(X,T) {\phi^*}\Big]  
 \\
&=v_n^\T U^\T\hat{G}_n^{-1}\Delta_1 + \rho\cdot v_n^\T U^\T\hat{G}_n^{-1}\Delta_1  + v_n^\T U^\T\hat{G}_n^{-1}\Delta_2.
\ee
In the following, we bound $|v_n^\T  U^\T\hat{G}_n^{-1}\Delta_1|$ and $|v_n^\T  U^\T\hat{G}_n^{-1}\Delta_2|$, respectively. Note that as $\rho\rightarrow 0$, the second term above, $\rho\cdot v_n^\T U^\T\hat{G}_n^{-1}\Delta_1$, is negligible compared to the first term, $v_n^\T U^\T\hat{G}_n^{-1}\Delta_1$. By Lemma \ref{lm:iG2iG}, we know wpa1, $\hat{G}_n$ and ${G}_n$ are full-rank and their inverse have the singular value decomposition,
\bee\label{inv:svd}
\hat{G}^{-1}_n &= \begin{pmatrix}
\hat{U} & \hat{U}_\perp
\end{pmatrix}
\begin{pmatrix}
\hat{\Sigma}^{-1} & 
\\
& \hat{\Sigma}^{-1}_{\perp}
\end{pmatrix}
\begin{pmatrix}
\hat{U}^\T
\\
\hat{U}_{\perp}^\T
\end{pmatrix},
\\
G_n^{-1} &= 
\begin{pmatrix}
\tilde{U} & \tilde{U}_\perp
\end{pmatrix}
\begin{pmatrix}
\tilde{\Sigma}^{-1} & 
\\
& \tilde{\Sigma}^{-1}_\perp
\end{pmatrix}
\begin{pmatrix}
\tilde{U}^\T
\\
\tilde{U}_{\perp}^\T
\end{pmatrix}.
\ee
\noindent{\fbox{Bound of $\big|v_n^\T  U^\T\hat{G}_n^{-1}\Delta_1\big|$}} We first bound $|v_n^\T  U^\T\hat{G}_n^{-1}\Delta_1|$. With straightforward algebra, we further write the following decomposition of $\Delta_1$,
\bee\label{delta:decom1}
\Delta_1&=\p_n\Big[\big[Y - m(X) - \tau(X,T) + \E\{\tau(X,T)\mid X\}\big]\big\{\Psi(X,T) - \Gamma(X)\big\}\Big]
\\
&+\p_n\Big[\big[\tau(X,T) - \E\{\tau(X,T)\mid X\}  - \{\Psi(X,T) - \Gamma(X)\}^\T{\phi^*}\big]\big\{\Psi(X,T) - \Gamma(X)\big\}\Big]
\\
&+\p_n\Big[\big[Y - m(X) - \{\Psi(X,T) - \Gamma(X)\}^\T{\phi^*}\big]\big\{\Gamma(X) - \hat{\Gamma}(X)\big\}\Big]
\\
&+\p_n\Big[\big[m(X) - \hat{m}(X) - \{\Gamma(X) - \hat{\Gamma}(X)\}^\T{\phi^*}\big]\big\{\Psi(X,T)-\Gamma(X)\big\}\Big]
\\
&+\p_n\Big[\big[m(X) - \hat{m}(X) - \{\Gamma(X) - \hat{\Gamma}(X)\}^\T{\phi^*}\big]\big\{\Gamma(X)-\hat{\Gamma}(X)\big\}\Big]
\\
&= \Delta_{1,1} + \Delta_{1,2} + \Delta_{1,3} + \Delta_{1,4} + \Delta_{1,5},
\ee
which further yields the decomposition,
\bee\nonumber
v_n^\T U^\T\hat{G}_n^{-1}\Delta_1 = \sum_{j = 1}^5 v_n^\T U^\T\hat{G}_n^{-1}\Delta_{1,j}.
\ee
We bound $|v_n^\T U^\T\hat{G}_n^{-1}\Delta_1|$ by deriving the bounds of  $|v_n^\T U^\T\hat{G}_n^{-1}\Delta_{1,1}|$ through $|v_n^\T U^\T\hat{G}_n^{-1}\Delta_{1,5}|$. 
\begin{itemize}
\item \textit{Bounding $|v_n^\T U^\T\hat{G}_n^{-1}\Delta_{1,1}|$}: Recall $\mu(x,t) = \E(Y\mid X=x,T=t)$. We first note $\Delta_{1,1}$ can be further simplified to, 
\bee\label{delta:decom2}
\Delta_{1,1}  = \p_n\big[\big\{Y -\mu(X,T)\big\}\big\{\Psi(X,T) - \Gamma(X)\big\}\big],
\ee
as, by definition for any $i = 1,\dots,n$,
\begin{align}\nonumber
&Y_i - m(X_i) - \tau(T_i,X_i) + \E\{\tau(X,T)\mid X = X_i\} 
\\\nonumber
&= Y_i - \E(Y\mid X_i) - \mu(X_i,T_i) + \mu(X_i,0) + E\big\{\mu(X,T) - \mu(X, 0)\mid X = X_i\big\} 
\\\nonumber
&= Y_i - \E(Y\mid X_i) - \mu(X_i,T_i) + \mu(X_i,0) + \E(Y\mid X_i) - \mu(X_i,0)
\\\nonumber
& = Y_i - \mu(X_i,T_i).
\end{align}
Again by $\begin{pmatrix}U & U_\perp\end{pmatrix}\begin{pmatrix}U & U_\perp\end{pmatrix}^\T = I$, 
we have
\bee\label{Gndelta112}
\Delta_{1,1} = \begin{pmatrix}U & U_\perp\end{pmatrix}\begin{pmatrix}U^\T \\ U_\perp^\T\end{pmatrix}\Delta_{1,1} = UU^\T\Delta_{1,1},
\ee
because,
\bee\label{Udelta11}
U_\perp^\T \Delta_{1,1}  &= \p_n\Big[\big[Y - m(X) - \{\Psi(X,T) - \Gamma(X)\}^\T{\phi^*}\big]\cdot\big\{U_\perp^\T\Psi(X,T) -U_\perp^\T \Gamma(X)\big\}\Big]
\\
&=\p_n\Big[\big[Y - m(X) - \{\Psi(X,T) - \Gamma(X)\}^\T{\phi^*}\big]\cdot 0\Big]
\\
&=  0.
\ee
The second equality above is due to that, by Lemma \ref{lm:svdR}, $U_\perp^\T\Psi(X,T) $ is a vector of functions free of $T$ and thus, 
\bee\label{Gndelta112f}
U_\perp^\T \Gamma(x) = \E[U_\perp^\T \Psi(X,T)\mid X = x] = U_\perp^\T \Psi(x,t),
\ee
for any $(x,t)\in \mathbb{X}\times\mathbb{T}$. Summarizing the results above, we have wpa1,
\bee\label{decom:D11}
&v_n^\T U^\T\hat{G}_n^{-1}\Delta_{1,1} 
\\
&=   v_n^\T U^\T{G}_n^{-1}\Delta_{1,1} + v_n^\T U^\T(\hat{G}_n^{-1} - {G}_n^{-1})\Delta_{1,1}
\\
&= v_n^\T U^\T{G}_n^{-1}\Delta_{1,1} + v_n^\T U^\T{G}_n^{-1}({G}_n - \hat{G}_n)\hat{G}_n^{-1}\Delta_{1,1}
\\
&= v_n^\T U^\T{G}_n^{-1}UU^\T\Delta_{1,1} + v_n^\T U^\T{G}_n^{-1}({G}_n - \hat{G}_n)\hat{G}_n^{-1}UU^\T\Delta_{1,1}
\\
& = v_n^\T U^\T\tilde{U}\tilde{\Sigma}^{-1}\tilde{U}^\T UU^\T\Delta_{1,1}  + v_n^\T U^\T\tilde{U}_{\perp}\tilde{\Sigma}_{\perp}^{-1}\tilde{U}_{\perp}^\T UU^\T\Delta_{1,1}  
\\
&+ v_n^\T U^\T\tilde{U}\tilde{\Sigma}^{-1}\tilde{U}^\T({G}_n - \hat{G}_n)\hat{U}\hat{\Sigma}^{-1}\hat{U}^\T UU^\T\Delta_{1,1}  + v_n^\T U^\T\tilde{U}\tilde{\Sigma}^{-1}\tilde{U}^\T({G}_n - \hat{G}_n)\hat{U}_\perp\hat{\Sigma}^{-1}_{\perp}\hat{U}_{\perp}^\T UU^\T\Delta_{1,1}
\\
& + v_n^\T U^\T\tilde{U}_{\perp}\tilde{\Sigma}_{\perp}^{-1}\tilde{U}_{\perp}^\T({G}_n - \hat{G}_n)\hat{U}\hat{\Sigma}^{-1}\hat{U}^\T UU^\T\Delta_{1,1}  + v_n^\T U^\T\tilde{U}_{\perp}\tilde{\Sigma}_{\perp}^{-1}\tilde{U}_{\perp}^\T({G}_n - \hat{G}_n)\hat{U}_\perp\hat{\Sigma}^{-1}_{\perp}\hat{U}_{\perp}^\T UU^\T\Delta_{1,1}
\\
& = \sum_{m = 1}^6 \ell_{m,n}^\T \Delta_{1,1},
\ee
where the third equality follows by  \eqref{Gndelta112}; the fourth equality follows by \eqref{inv:svd}; and we define 
\bee\nonumber
&\ell_{1,n} = (v_n^\T U^\T\tilde{U}\tilde{\Sigma}^{-1}\tilde{U}^\T UU^\T)^\T,
\\
&\ell_{2,n} = (v_n^\T U^\T\tilde{U}_{\perp}\tilde{\Sigma}_{\perp}^{-1}\tilde{U}_{\perp}^\T UU^\T)^\T,
\\
&\ell_{3,n} = \{v_n^\T U^\T\tilde{U}\tilde{\Sigma}^{-1}\tilde{U}^\T({G}_n - \hat{G}_n)\hat{U}\hat{\Sigma}^{-1}\hat{U}^\T UU^\T\}^\T,
\\
&\ell_{4,n} = \{v_n^\T U^\T\tilde{U}\tilde{\Sigma}^{-1}\tilde{U}^\T({G}_n - \hat{G}_n)\hat{U}_\perp\hat{\Sigma}^{-1}_{\perp}\hat{U}_{\perp}^\T UU^\T\}^\T,
\\
& \ell_{5,n} = \{v_n^\T U^\T\tilde{U}_{\perp}\tilde{\Sigma}_{\perp}^{-1}\tilde{U}_{\perp}^\T({G}_n - \hat{G}_n)\hat{U}\hat{\Sigma}^{-1}\hat{U}^\T UU^\T\}^\T,
\\
&\ell_{6,n} = \{v_n^\T U^\T\tilde{U}_{\perp}\tilde{\Sigma}_{\perp}^{-1}\tilde{U}_{\perp}^\T({G}_n - \hat{G}_n)\hat{U}_\perp\hat{\Sigma}^{-1}_{\perp}\hat{U}_{\perp}^\T UU^\T\}^\T.
\ee
We note $\ell_{1,n}$ and $\ell_{2,n}$  depend only on $n$, and $\ell_{3,n}$--$\ell_{6,n}$ depend on both $n$ and $(X_1,T_1),\dots,(X_n,T_n)$ since the $\hat{G}_n$ is involved. For each $m = 1,\dots,6$, by taking $\ell_n = \ell_{m,n}/\|\ell_{m,n}\|$ in Lemma \ref{l3} (i), we have
\bee\label{decom:D112}
|\ell_{m,n}^\T \Delta_{1,1}| = \mathcal{O}_P(\|\ell_{m,n}\|/\sqrt{n}).
\ee
We now bound $\|\ell_{m,n}\|$ for each $m = 1,\dots,6$ as $n\rightarrow +\infty$, by Lemma \ref{lm:iG2iG},
\bee\label{l16}
&\|\ell_{1,n}\| \leq \|v_n\|\| U^\T\tilde{U}\|_2\|\tilde{\Sigma}^{-1}\|_2\|\tilde{U}^\T UU^\T\|_2 = \mathcal{O}(\|v_n\|\beta_n^{-1}),
\\
&\|\ell_{2,n}\| \leq \|v_n\|\| U^\T\tilde{U}_{\perp}\|_2\|\tilde{\Sigma}_{\perp}^{-1}\|_2\|\tilde{U}_{\perp}^\T U\|_2\|U^\T\|_2 = \mathcal{O}(\|v_n\|\rho\beta_n^{-2}),
\\
&\|\ell_{3,n}\| \leq \|v_n\|\| U^\T\tilde{U}\|_2\|\tilde{\Sigma}^{-1}\|_2\|\tilde{U}^\T\|_2\|{G}_n - \hat{G}_n\|_2\|\hat{U}\|_2\|\hat{\Sigma}^{-1}\|_2\|\hat{U}^\T UU^\T\|_2 = \mathcal{O}_P(\|v_n\|\beta_n^{-2}\sqrt{K\log n/n}),
\\
&\|\ell_{4,n}\| \leq \|v_n\|\|U^\T\tilde{U}\|_2\|\tilde{\Sigma}^{-1}\|_2\|\tilde{U}^\T\|_2\|{G}_n - \hat{G}_n\|_2\|\hat{U}_\perp\|_2\|\hat{\Sigma}^{-1}_{\perp}\|_2\|\hat{U}_{\perp}^\T U\|_2\|U^\T\|_2 = \mathcal{O}_P(\|v_n\|\beta_n^{-2}\rho^{-1}{K\log n/n}),
\\
& \|\ell_{5,n}\| \leq \|v_n\| \|U^\T\tilde{U}_{\perp}\|_2\|\tilde{\Sigma}_{\perp}^{-1}\|_2\|\tilde{U}_{\perp}^\T\|_2\|{G}_n - \hat{G}_n\|_2\|\hat{U}\|_2\|\hat{\Sigma}^{-1}\|_2\|\hat{U}^\T UU^\T\|_2 = \mathcal{O}_P(\|v_n\|\beta_n^{-2}\sqrt{K\log n/n}),
\\
&\|\ell_{6,n}\| \leq \|v_n\| \|U^\T\tilde{U}_{\perp}\|_2\|\tilde{\Sigma}_{\perp}^{-1}\|_2\|\tilde{U}_{\perp}^\T\|_2\|{G}_n - \hat{G}_n\|_2\|\hat{U}_\perp\|_2\|\hat{\Sigma}^{-1}_{\perp}\|_2\|\hat{U}_{\perp}^\T U\|_2\|U^\T\|_2=\mathcal{O}_P(\|v_n\|\beta_n^{-2}\rho^{-1}{K\log n/n}).
\ee
We denote the following rate induced by $\ell_{1,n}$--$\ell_{6,n}$:
\bee\label{def:zeta}
\zeta_n &= \|v_n\|\beta_n^{-1} +\|v_n\|\rho\beta_n^{-2}  
\\
\zeta'_n &=    \|v_n\|\beta_n^{-2}\sqrt{K\log n/n} + \|v_n\|\beta_n^{-2}\rho^{-1}{K\log n/n}.
\ee
Combining \eqref{decom:D11}, \eqref{decom:D112}, and \eqref{l16}, we conclude
\bee\label{bound:delta11}
|v_n^\T U^\T\hat{G}_n^{-1}\Delta_{1,1} | &= \mathcal{O}_P\left((\zeta_n + \zeta_n')n^{-1/2}\right).
\ee
\item \textit{Bounding $|v_n^\T U^\T\hat{G}_n^{-1}\Delta_{1,2}|$}: With the same derivations as \eqref{Gndelta112}--\eqref{Gndelta112f}, we have $\Delta_{1,2} = UU^\T\Delta_{1,2}$. and thus similar to \eqref{decom:D11},
\bee\nonumber
v_n^\T U^\T\hat{G}_n^{-1}\Delta_{1,2} = \sum_{m = 1}^6 \ell_{m,n}^\T \Delta_{1,2}.
\ee 
For $m = 1,2$, we can apply Lemma \ref{l3} (ii) and bound
\bee\nonumber
\big|\ell_{m,n}^\T \Delta_{1,2}\big| = \|\ell_{m,n}\|\big|(\ell_{m,n}/\|\ell_{m,n}\|)^\T \Delta_{1,2}\big| = \mathcal{O}_P(\|\ell_{m,n}\|K^{-p/(d + 1)}),
\ee
since  $\ell_{1,n}/\|\ell_{1,n}\|$ and $\ell_{2,n}/\|\ell_{2,n}\|$  depend  only on $n$. For $m = 3,\dots,6$, we have
\bee\nonumber
\big|\ell_{m,n}^\T \Delta_{1,2}\big| \leq \|\ell_{m,n}\|\|\Delta_{1,2}\| = \mathcal{O}_P(\|\ell_{m,n}\|K^{-p/(d + 1)}).
\ee
Combining the  above results with \eqref{decom:D112}, we conclude
\bee\label{bounddelta12}
|v_n^\T U^\T\hat{G}_n^{-1}\Delta_{1,2} | = \mathcal{O}_P\Big((\zeta_n + \zeta_n') K^{-p/(d + 1)} \Big).
\ee
\item \textit{Bounding $|v_n^\T U^\T\hat{G}_n^{-1}\Delta_{1,3}|$}: By Lemma \ref{lm:svdR} (iv), we have $U_{\perp}^\T\{\Gamma(x) - \hat{\Gamma}(x)\} = 0$ for any $x\in\mathbb{X}$, which implies
\bee\label{dl3}
U_\perp^\T \Delta_{1,3}  &= \p_n\Big[\big[Y - m(X) - \{\Psi(X,T) - \Gamma(X)\}^\T{\phi^*}\big]\big\{U_\perp^\T\Gamma(X) - U_\perp^\T\hat{\Gamma}(X)\big\}\Big]
\\
& = 0.
\ee 
Then with similar derivations as \eqref{Gndelta112}--\eqref{Gndelta112f}, we have $\Delta_{1,3} = UU^\T\Delta_{1,3}$. And thus similar to \eqref{decom:D11},
\bee\nonumber
v_n^\T U^\T\hat{G}_n^{-1}\Delta_{1,3}  =\sum_{m = 1}^6 \ell_{m,n}^\T \Delta_{1,3}.
\ee
For $m = 1,2$, we can apply Lemma \ref{l3} (iii) and bound
\bee\nonumber
\big|\ell_{m,n}^\T \Delta_{1,3}\big| = \|\ell_{m,n}\|\big|(\ell_{m,n}/\|\ell_{m,n}\|)^\T \Delta_{1,3}\big| = o_P(\|\ell_{m,n}\|r_\gamma/\sqrt{n}),
\ee
since  $\ell_{1,n}/\|\ell_{1,n}\|$ and $\ell_{2,n}/\|\ell_{2,n}\|$  depend  only on $n$. For $m = 3,\dots,6$, we have
\bee\nonumber
\big|\ell_{m,n}^\T \Delta_{1,3}\big| \leq \|\ell_{m,n}\|\|\Delta_{1,3}\| = o_P(\|\ell_{m,n}\|r_{\gamma}'\sqrt{K/n}).
\ee
Combining the above results with \eqref{decom:D112}, we conclude
\bee\label{final:bound:delta3}
|v_n^\T U^\T\hat{G}_n^{-1}\Delta_{1,3} |  = o_P(r_\gamma\zeta_n/\sqrt{n} + r_\gamma'\zeta_n'\sqrt{K/n}).
\ee
\item \textit{Bounding $|v_n^\T U^\T\hat{G}_n^{-1}\Delta_{1,4}|$}: With the same derivations as \eqref{Gndelta112}--\eqref{decom:D11}, we have
\bee\nonumber
v_n^\T U^\T\hat{G}_n^{-1}\Delta_{1,4} = \sum_{m = 1}^6 \ell_{m,n}^\T \Delta_{1,4}.
\ee 
For $m = 1,2$, we can apply Lemma \ref{l3} (iv) and bound
\bee\nonumber
\big|\ell_{m,n}^\T \Delta_{1,4}\big| = \|\ell_{m,n}\|\big|(\ell_{m,n}/\|\ell_{m,n}\|)^\T \Delta_{1,4}\big| = o_P(\|\ell_{m,n}\|/\sqrt{n}),
\ee
since  $\ell_{1,n}/\|\ell_{1,n}\|$ and $\ell_{2,n}/\|\ell_{2,n}\|$  depend only on $n$. For $m = 3,\dots,6$, we have
\bee\nonumber
\big|\ell_{m,n}^\T \Delta_{1,4}\big| \leq \|\ell_{m,n}\|\|\Delta_{1,4}\| = o_P(r_m\|\ell_{m,n}\|\sqrt{K/n} + r_\gamma\|\ell_{m,n}\|\sqrt{K/n}).
\ee
Combining the above results with \eqref{decom:D112}, we conclude
\bee\label{bounddelta14}
|v_n^\T U^\T\hat{G}_n^{-1}\Delta_{1,4} | = o_P\left(\zeta_n/\sqrt{n} + r_m\zeta_n'\sqrt{K/n} + r_\gamma\zeta_n'\sqrt{K/n}\right). 
\ee

\item \textit{Bounding $|v_n^\T U^\T\hat{G}_n^{-1}\Delta_{1,5}|$}: Similar to \eqref{dl3}, we have $U_\perp^\T \Delta_{1,5} = 0$ and thus with the same derivations as \eqref{Gndelta112}--\eqref{decom:D11}, we have
\bee\nonumber
v_n^\T U^\T\hat{G}_n^{-1}\Delta_{1,5} = \sum_{m = 1}^6 \ell_{m,n}^\T \Delta_{1,5}.
\ee 
For $m = 1,2$, we can apply Lemma \ref{l3} (v) and bound
\bee\nonumber
\big|\ell_{m,n}^\T \Delta_{1,5}\big| = \|\ell_{m,n}\|\big|(\ell_{m,n}/\|\ell_{m,n}\|)^\T \Delta_{1,5}\big| = o_P(r_{\gamma}\|\ell_{m,n}\|/\sqrt{n} + \|\ell_{m,n}\|r_{m}r_{\gamma} + \|\ell_{m,n}\|r^2_{\gamma}),
\ee
since  $\ell_{1,n}/\|\ell_{1,n}\|$ and $\ell_{2,n}/\|\ell_{2,n}\|$  depend only on $n$. For $m = 3,\dots,6$, we have
\bee\nonumber
\big|\ell_{m,n}^\T \Delta_{1,5}\big| \leq \|\ell_{m,n}\|\|\Delta_{1,5}\| = o_P(\|\ell_{m,n}\|r_{m}r_{\gamma} + \|\ell_{m,n}\|r^2_{\gamma} + \|\ell_{m,n}\|{r_{\gamma}'} r_m\sqrt{K/n} + \|\ell_{m,n}\|{r_{\gamma}'} r_\gamma\sqrt{K/n}).
\ee
Combining the  above results with \eqref{decom:D112}, we conclude
\bee\label{delta15bound}
&|v_n^\T U^\T\hat{G}_n^{-1}\Delta_{1,5} | 
\\
&= o_P\Big(\zeta_n\big(r_{\gamma}/\sqrt{n} + r_{m}r_{\gamma} + r^2_{\gamma}\big)\Big) 
+o_P\Big(\zeta_n'\big(r_{m}r_{\gamma} + r^2_{\gamma} + {r_{\gamma}'} r_m\sqrt{K/n} + {r_{\gamma}'} r_\gamma\sqrt{K/n}\big)\Big). 
\ee
\end{itemize}
After summarizing all bounds above and dropping some negligible terms, we conclude
\bee\label{p12}
&|v_n^\T U^\T\hat{G}_n^{-1}\Delta_1| 
\\
 &=\mathcal{O}_P\left((n^{-1/2} + K^{-p/(d + 1)})\zeta_n + (n^{-1/2} + K^{-p/(d+1)})\zeta_n'\right)
\\
&+o_P\Big(\big(1/\sqrt{n} + r_{\gamma}/\sqrt{n} + r_{m}r_{\gamma} + r^2_{\gamma}\big)\zeta_n\Big) 
\\
&+o_P\Big(\big(r_{m}r_{\gamma} + r^2_{\gamma} + {r_{\gamma}'} r_m\sqrt{K/n} + {r_{\gamma}'} r_\gamma\sqrt{K/n} + r_m\sqrt{K/n} + r_\gamma\sqrt{K/n} + r'_\gamma\sqrt{K/n}\big)\zeta_n'\Big)
\\
 &=\mathcal{O}_P\left((n^{-1/2} + K^{-p/(d + 1)})\zeta_n + (n^{-1/2} + K^{-p/(d+1)})\zeta_n'\right)
\\
&+o_P\Big(\big(1/\sqrt{n}   + r_{m}r_{\gamma} + r^2_{\gamma}\big)\zeta_n + \big(r_{m}r_{\gamma} + r^2_{\gamma}   + r_m\sqrt{K/n} + r_\gamma\sqrt{K/n} + r'_\gamma\sqrt{K/n}\big)\zeta_n'\Big)   
\ee because $\sqrt{K/n}\prec 1$ and $r_{m}r_{\gamma} + r^2_{\gamma} + {r_{\gamma}'} r_m\sqrt{K/n} + {r_{\gamma}'} r_\gamma\sqrt{K/n} \precsim r_{\gamma}'+r_m + r_\gamma$, due to Assumption \ref{am:Kn} and $r_m,r_{\gamma},r_{\gamma}'\precsim 1$.
\par
\noindent{\fbox{Bound of $\big|v_n^\T  U^\T\hat{G}_n^{-1}\Delta_2\big|$}} Recall that we have
\bee\label{decom:d2}
\big|v_n^\T  U^\T\hat{G}_n^{-1}\Delta_2\big| &= \rho \left| v_n^\T U^\T\hat{G}_n^{-1}\p_n\Big[ \{\Psi(X,T) - \hat{\Gamma}(X)\}\big\{\Psi(X,T) - \hat{\Gamma}(X)\big\}^\T{\phi^*}  - \Psi(X,T)   \Psi^\T(X,T) {\phi^*}\Big]\right|.
\ee
We now decompose 
\begin{align}\nonumber
&v_n^\T U^\T\hat{G}_n^{-1}\p_n\Big[ \{\Psi(X,T) - \hat{\Gamma}(X)\}\big\{\Psi(X,T) - \hat{\Gamma}(X)\big\}^\T{\phi^*}  - \Psi(X,T)   \Psi^\T(X,T) {\phi^*}\Big]
\\\nonumber
&=-v_n^\T U^\T\hat{G}_n^{-1}\p_n\Big[   \hat{\Gamma}(X) \Psi^\T(X,T){\phi^*}  \Big] - v_n^\T U^\T\hat{G}_n^{-1}\p_n\Big[      \Psi(X,T)\hat{\Gamma}^\T(X)\phi^* \Big] \\\nonumber
&\quad + v_n^\T U^\T\hat{G}_n^{-1}\p_n\Big[   \hat{\Gamma}(X)\hat{\Gamma}^\T(X)\phi^*\Big]
\\\nonumber
&=-v_n^\T U^\T\hat{G}_n^{-1}\p_n\Big[   \hat{\Gamma}(X) \Psi^\T(X,T){\phi^*}  \Big] 
\\\nonumber
&\quad - v_n^\T U^\T\hat{G}_n^{-1}\p_n\Big[    \Big\{  \Psi(X,T) - \Gamma(X)\Big\}\Big\{\hat{\Gamma} (X) - \Gamma(X)\Big\}^\T\phi^* \Big] 
\\\nonumber
&\quad - v_n^\T U^\T\hat{G}_n^{-1}\p_n\Big[    \Big\{  \Psi(X,T) - \Gamma(X)\Big\} \Gamma^\T(X) \phi^* \Big] 
\\\nonumber
&\quad + v_n^\T U^\T\hat{G}_n^{-1}\p_n\Big[  \left\{ \Gamma(X) - \hat{\Gamma}(X)\right\}\Big\{\hat{\Gamma}(X) - \Gamma(X)\Big\}^\T \phi^*\Big]
\\\nonumber
&\quad + v_n^\T U^\T\hat{G}_n^{-1}\p_n\Big[  \left\{ \Gamma(X) - \hat{\Gamma}(X)\right\} \Gamma^\T(X)  \phi^*\Big].
\end{align}
We bound the five terms on the right-hand side above respectively. In the following, some of the  arguments for bounding the corresponding terms, are similar to the previous ones and thus we omit the details. 

First, we have
\bee
 v_n^\T U^\T\hat{G}_n^{-1}\p_n\Big[   \hat{\Gamma}(X) \Psi^\T(X,T){\phi^*}  \Big]  =&  v_n^\T U^\T\hat{G}_n^{-1}\p_n\Big[   \left\{\hat{\Gamma}(X) -\Gamma(X)\right\} \tilde{\tau}(X,T) \Big]    
\\
&+v_n^\T U^\T\hat{G}_n^{-1}\p_n\Big[   \left\{\hat{\Gamma}(X) -\Gamma(X)\right\} \Big\{\Psi^\T(X,T){\phi^*}  - \tilde{\tau}(X,T)\Big\} \Big]
\\
&+v_n^\T U^\T\hat{G}_n^{-1}\p_n\Big[   \Gamma(X) \Big\{\Psi^\T(X,T){\phi^*} -\tilde{\tau}(X,T)\Big\} \Big]
\\
&+v_n^\T U^\T\hat{G}_n^{-1}\p_n\Big[   \Gamma(X)  \tilde{\tau}(X,T)  \Big]
\\
 = &\mathcal{O}_P\left((\zeta_n + \zeta_n')n^{-1/2}\right) 
+\mathcal{O}_P\Big((\zeta_n + \zeta_n') K^{-p/(d + 1)} \Big) 
\\
& +  o_P(r_\gamma\zeta_n/\sqrt{n} + r_\gamma'\zeta_n'\sqrt{K/n}),
\ee
where the derivation of the last equation is similar to \eqref{bound:delta11}, \eqref{bounddelta12} and \eqref{final:bound:delta3}.

Second, similar to \eqref{bounddelta14}, we have
\bee\nonumber
 \left|v_n^\T U^\T\hat{G}_n^{-1}\p_n\Big[    \Big\{  \Psi(X,T) - \Gamma(X)\Big\}\Big\{\hat{\Gamma} (X) - \Gamma(X)\Big\}^\T\phi^* \Big] \right| 
 = o_P\left(\zeta_n/\sqrt{n}  + r_\gamma\zeta_n'\sqrt{K/n}\right).
\ee

Third, similar to \eqref{bounddelta12}, we have
\bee\nonumber
&\left|v_n^\T U^\T\hat{G}_n^{-1}\p_n\Big[    \Big\{  \Psi(X,T) - \Gamma(X)\Big\} \Gamma^\T(X) \phi^* \Big] \right|
\\
&=\left|v_n^\T U^\T\hat{G}_n^{-1}\p_n\Big[    \Big\{  \Psi(X,T) - \Gamma(X)\Big\} \E\Big(\Psi(X,T) \phi^* - \tilde{\tau}(X,T)\mid X\Big) \Big] \right|
\\
& = \mathcal{O}_P\Big((\zeta_n + \zeta_n') K^{-p/(d + 1)} \Big),
\\
&\left| v_n^\T U^\T\hat{G}_n^{-1}\p_n\Big[  \left\{ \Gamma(X) - \hat{\Gamma}(X)\right\} \Gamma^\T(X)  \phi^*\Big]\right| 
\\
&= \mathcal{O}_P\Big((\zeta_n + \zeta_n') K^{-p/(d + 1)} \Big).
\ee
Finally, similar to \eqref{delta15bound}, we have
\bee\nonumber
 \left|v_n^\T U^\T\hat{G}_n^{-1}\p_n\Big[  \left\{ \Gamma(X) - \hat{\Gamma}(X)\right\}\Big\{\hat{\Gamma}(X) - \Gamma(X)\Big\}^\T \phi^*\Big]\right|
 =o_P\Big(\zeta_n\big(r_{\gamma}/\sqrt{n}  + r^2_{\gamma}\big)\Big)  
+o_P\Big(\zeta_n'\big(  r^2_{\gamma} +   {r_{\gamma}'} r_\gamma\sqrt{K/n}\big)\Big).
\ee
\par
In summary, we have
\bee\label{P13}
\big|v_n^\T  U^\T\hat{G}_n^{-1}\Delta_2\big| &=\rho\cdot\Big[\mathcal{O}_P\left((\zeta_n + \zeta_n')n^{-1/2}\right)  
+\mathcal{O}_P\Big((\zeta_n + \zeta_n') K^{-p/(d + 1)} \Big) 
\\
&\quad\quad\quad  +  o_P(r_\gamma\zeta_n/\sqrt{n} + r_\gamma'\zeta_n'\sqrt{K/n})
  + o_P\left(\zeta_n/\sqrt{n}   + r_\gamma\zeta_n'\sqrt{K/n}\right)
\\
&\quad\quad\quad + o_P\Big(\zeta_n\big(r_{\gamma}/\sqrt{n}  + r^2_{\gamma}\big)\Big)  
+o_P\Big(\zeta_n'\big(  r^2_{\gamma} +   {r_{\gamma}'} r_\gamma\sqrt{K/n}\big)\Big)\Big].
\ee

\noindent{\fbox{Summary of convergence rates}} We now prove the convergence rate results in Theorem \ref{thm:main}. Combining \eqref{thm2P1}, \eqref{bias:l2}, \eqref{decom:main}, \eqref{p12}, and \eqref{P13} and noting that $\rho\precsim 1$, we finally have
\bee\label{bigboy}
&\big|\hat{\tau}(x_0,t_0) - \tau(x_0,t_0)\big| 
\\
&\leq r(n,K,\beta_n,\rho,r_m,r_\gamma,r_{\gamma}' )
\\
&:=\mathcal{O}_P\left((n^{-1/2} + K^{-p/(d + 1)})\zeta_n + (n^{-1/2} + K^{-p/(d+1)})\zeta_n'\right)
\\
&\quad +o_P\Big(\big(1/\sqrt{n}   + r_{m}r_{\gamma} + r^2_{\gamma}\big)\zeta_n + \big(r_{m}r_{\gamma} + r^2_{\gamma}   + r_m\sqrt{K/n} + r_\gamma\sqrt{K/n} + r'_\gamma\sqrt{K/n}\big)\zeta_n'\Big) .
\ee
Here we recall the definitions:
\bee\nonumber
  v_n &= U^\T\M\Psi(x_0,t_0) - U^\T\M\Psi(x_0,0),
\\
\zeta_n &= \|v_n\|\beta_n^{-1} +\|v_n\|\rho\beta_n^{-2}  ,
\\
\zeta'_n &=    \|v_n\|\beta_n^{-2}\sqrt{K\log n/n} + \|v_n\|\beta_n^{-2}\rho^{-1}{K\log n/n},
\ee
and $\|v_n\| \leq \|U\|_2\|\Psi(x_0,t_0) - \Psi(x_0,0) \| = \|\Psi(x_0,t_0) - \Psi(x_0,0) \|$ depends on $x_0$ and $t_0$, and it has a general bound as shown in \eqref{vn:bound} such that $\|v_n\|\precsim \sqrt{K}$. 
\par
When $\beta_n\asymp 1$, $r_{m}, r_\gamma,r_{\gamma'}\precsim n^{-1/4}$ and $\rho\precsim 1$, and $p > d + 1$, the above bound can be simplified to
\begin{align}\nonumber
&\big|\hat{\tau}(x_0,t_0) - \tau(x_0,t_0)\big| 
\\\nonumber
&=\mathcal{O}_P\Big\{\|v_n\|\big(n^{-1/2}+K^{-p/(d + 1)}\big)\big(1 +\rho^{-1}K\log n/n\big)  \Big\}
\\\nonumber
&\quad +o_P\Big(\|v_n\|\sqrt{K\log n /n}(n^{-1/2} + \sqrt{K/n}\cdot n^{-1/4})+ \|v_n\|{(\log n)}\rho^{-1}K^{3/2}n^{-7/4}\Big)
\\\nonumber
&=\mathcal{O}_P\Big(\sqrt{K/n}+K^{1/2-p/(d + 1)}\Big)
\\\nonumber
&\quad+\mathcal{O}_P\Big(\big(\sqrt{K/n}+K^{1/2-p/(d + 1)}\big)\big(\rho^{-1}K\log n/n\big)  \Big)\\\label{simple:rate}
&\quad+o_P\Big( \sqrt{ \log n }K^{3/2} n^{-5/4}+  {(\log n)}\rho^{-1}K^{2}n^{-7/4}\Big),
\end{align}
where the last equality is derived by $\|v_n\| \precsim \sqrt{K}$ by \eqref{vn:bound}. 

In the following, we balance the rate of \eqref{simple:rate} by selecting $K\asymp n^{(d + 1)/2p}$. We first focus on the first term of \eqref{simple:rate},
\bee\nonumber
\sqrt{K/n}+K^{1/2-p/(d + 1)} \asymp n^{-1/2 + (d + 1)/(4p)}.
\ee
Now we show the second term in \eqref{simple:rate} is negligible. When $\rho\succ n^{-1 +(d + 1)/(2p)}\log n$, we have
\bee\nonumber
\rho^{-1}K\log n/n&\prec n^{1  - (d + 1)/(2p)}\cdot n^{(d + 1)/(2p)}\cdot n^{-1}
\\
& =1.
\ee
This implies $$\big(\sqrt{K/n}+K^{1/2-p/(d + 1)}\big)\big(\rho^{-1}K\log n/n\big) \prec \big(\sqrt{K/n}+K^{1/2-p/(d + 1)}\big),$$ which is negligible compared with the rate of the first term in \eqref{simple:rate}. 
Finally 
\bee\label{rate:final}
{(\log n)}\rho^{-1}K^{2}n^{-7/4} &\prec  n^{1 - (d + 1)/(2p)}\cdot n^{(d + 1)/p } \cdot n^{- 7/4}
\\
&=n^{-1/2 + (d + 1)/(4p) }\cdot n^{- 1/4 + (d + 1)/(4p)}
\\
&\prec n^{-1/2 + (d + 1)/(4p) }
\\
\sqrt{ \log n }K^{3/2} n^{-5/4}&\precsim  \sqrt{\log n}\cdot n^{3(d + 1)/(4p) - 5/4}
\\
& = \sqrt{\log n}\cdot n^{  - 1/2}\cdot n^{3\{(d+1)/(4p)-1/4\}}
\\
&\prec n^{-1/2 + (d + 1)/(4p)},
\ee
since $-1/4 + (d + 1)/(4p) < 0$ due to $d + 1 <p$. Thus the third term in \eqref{simple:rate} is also negligible compared with the first term in \eqref{simple:rate}. In summary, the first term in \eqref{simple:rate} is optimized as $\mathcal{O}_P(n^{-1/2 + (d + 1)/(4p)})$ when $K\asymp n^{(d + 1)/(2p)}$. Moreover, when selecting $n^{-1 + (d + 1)/(2p)}\log n\prec \rho \precsim n^{-1/2}$, other terms than the first term in \eqref{simple:rate} are negligible, and thus the whole rate is minimized to $\mathcal{O}_P(n^{-1/2 + (d + 1)/(4p)})$.
\par
\noindent\fbox{Limiting distribution} In this part, we show the central limiting theorem (CLT) result for our proposed estimator. Now we have $\rho\rightarrow 0$. In  the previous parts, if carefully tracking the derivations, we can show by \eqref{simple:rate} that
\bee\label{clt:basic}
\big|\hat{\tau}(x_0,t_0) - \tau(x_0,t_0)\big| =  v_n^\T U^\T\tilde{U}\tilde{\Sigma}^{-1}\tilde{U}^\T UU^\T\Delta_{1,1} + R_n,
\ee
where, 
\bee\label{rndef}
R_n  &=\mathcal{O}_P\Big\{\|v_n\| K^{-p/(d + 1)} \big(1 +\rho^{-1}K\log n/n\big)  \Big\}
\\\nonumber
&\quad +o_P\Big(\|v_n\|\sqrt{\log n } K n^{-5/4} + \|v_n\|{(\log n)}\rho^{-1}K^{3/2}n^{-7/4}\Big)
\\
 &=\mathcal{O}_P\Big\{\|v_n\| K^{-p/(d + 1)} \big(1 +\rho^{-1}K\log n/n\big)  \Big\} +o_P\Big(\|v_n\|\sqrt{\log n } K n^{-5/4}  \Big).
\ee
This is because only the term $v_n^\T U^\T\tilde{U}\tilde{\Sigma}^{-1}\tilde{U}^\T UU^\T\Delta_{1,1}$ in \eqref{decom:D11} produces the rate term $\mathcal{O}_P\big\{\|v_n\|n^{-1/2}\big\}$, on the right-hand side of the second equality in \eqref{simple:rate} when $\rho\rightarrow 0 $. In the following, we will first show the limiting distribution of the first term $v_n^\T U^\T\tilde{U}\tilde{\Sigma}^{-1}\tilde{U}^\T UU^\T\Delta_{1,1}$ on the right-hand side of \eqref{clt:basic}. We then show that the second term in \eqref{clt:basic} is negligible. Finally, we will show the consistency of the our asymptotic variance estimator, which helps us to construct the confidence interval. 
\par
First we have
\bee\label{clt:main}
{\sqrt{n}}\tilde{\sigma}^{-1}v_n^\T U^\T\tilde{U}\tilde{\Sigma}^{-1}\tilde{U}^\T UU^\T\Delta_{1,1} &= {\sqrt{n}}\tilde{\sigma}^{-1}\p_n\Big[\kappa_n(X,T)\big[Y - m(X) - \tau(X,T) + \E\{\tau(X,T)\mid X\}\big]\Big]
\\
&=\sum_{i = 1}^n\xi_n(X_i,T_i,Y_i),
\ee
where we define
\bee\label{def:tsigma}
\tilde{\sigma} &= \sqrt{E\big[\kappa^2_n(X,T)\big\{Y - \mu(X,T)\big\}^2\big]}
\\
\kappa_n(x,t) &= v_n^\T U^\T\tilde{U}\tilde{\Sigma}^{-1}\tilde{U}^\T UU^\T\big\{\Psi(x,t) - \Gamma(x)\big\}
\\
\xi_n(x,t,y) &= \frac{1}{\sqrt{n}\tilde{\sigma}}\kappa_n(x,t)\big[y - m(x) - \tau(x,t) + \E\{\tau(X,T)\mid X = x\}\big]
\\
& = \frac{1}{\sqrt{n}\tilde{\sigma}}\kappa_n(x,t)\big[y - E(Y\mid X = x) - \mu(x,t) + \mu(x,0) + E\{\mu(X,T)\mid X = x\} -  E\{\mu(X,0)\mid X = x\}\big]
\\
& = \frac{1}{\sqrt{n}\tilde{\sigma}}\kappa_n(x,t)\big[y - E(Y\mid X = x) - \mu(x,t) + \mu(x,0) + E(Y\mid X = x) -  \mu(x,0)\big]
\\
& = \frac{1}{\sqrt{n}\tilde{\sigma}}\kappa_n(x,t)\big\{y - \mu(x,t)\big\},
\ee
recalling that $\mu(x,t) = E(Y\mid X = x, T = t)$. We now verify \eqref{clt:main} satisfies the Lindberg's condition for the CLT. We emphasize $v_n^\T U^\T\tilde{U}\tilde{\Sigma}^{-1}\tilde{U}^\T UU^\T$ is a deterministic vector independent of samples. Therefore, $\xi_n(X_i,T_i,Y_i)$ with $i = 1,\dots,n$, are i.i.d. samples with
\bee\label{fl1}
E\{\xi_n(X,T,Y)\} &= \frac{1}{\sqrt{n}\tilde{\sigma}} E\Big[E\big[\kappa_n(X,T)\big\{Y - \mu(X,T)\big\}\mid X,T\big]\Big]
\\
&=  \frac{1}{\sqrt{n}\tilde{\sigma}} E\Big[\kappa_n(X,T) E\big[Y - \mu(X,T)\mid X,T\big]\Big]
\\
& = 0.
\ee
Therefore we have
\bee\label{fl2}
\text{Var}\Big\{\sum_{i = 1}^n\xi_n(X_i,T_i,Y_i)\Big\} = \frac{1}{\tilde{\sigma}^2}\text{Var}\Big[\kappa_n(X,T)\big\{Y - \mu(X,T)\big\}\Big] = 1.
\ee
Next we derive the lower bound of $\tilde{\sigma}$. By definition, we have 
\bee\label{lower:sigma}
&\tilde{\sigma}^2 
\\
&= v_n^\T U^\T\tilde{U}\tilde{\Sigma}^{-1}\tilde{U}^\T UU^\T E\Big[\big\{\Psi(X,T) - \Gamma(X)\big\}\big\{\Psi(X,T) - \Gamma(X)\big\}^\T\big\{Y - \mu(X,T)\big\}^2\Big] UU^\T\tilde{U}\tilde{\Sigma}^{-1}\tilde{U}^\T U v_n. 
\ee
Observe that by the law of total expectation,
\bee\label{lower:sigma2}
&E\Big[\big\{\Psi(X,T) - \Gamma(X)\big\}\big\{\Psi(X,T) - \Gamma(X)\big\}^\T\big\{Y - \mu(X,T)\big\}^2\Big]
\\
&=E\Big[\big\{\Psi(X,T) - \Gamma(X)\big\}\big\{\Psi(X,T) - \Gamma(X)\big\}^\T\text{Var}(Y\mid X,T)\Big]
\\
&\succeq c_1E\Big[\big\{\Psi(X,T) - \Gamma(X)\big\}\big\{\Psi(X,T) - \Gamma(X)\big\}^\T\Big]
\\
& = c_1 R_n,
\ee
since $\text{Var}(Y\mid X,T)$ is uniformly lower bounded by some fixed $c_1 > 0$ under Assumption \ref{am:moment}. Combining \eqref{lower:sigma} and \eqref{lower:sigma2}, we have
\bee\label{sigmat:1}
\tilde{\sigma}^2 &\geq c_1 v_n^\T U^\T\tilde{U}\tilde{\Sigma}^{-1}\tilde{U}^\T UU^\T  R_nUU^\T\tilde{U}\tilde{\Sigma}^{-1}\tilde{U}^\T U v_n
\\
&= c_1 v_n^\T U^\T\tilde{U}\tilde{\Sigma}^{-1}\tilde{U}^\T UU^\T  U\Sigma U^\T UU^\T\tilde{U}\tilde{\Sigma}^{-1}\tilde{U}^\T U v_n
\\
& = c_1 v_n^\T M_{clt} v_n,
\ee
where $M_{clt} =  U^\T\tilde{U}\tilde{\Sigma}^{-1}\tilde{U}^\T UU^\T  U\Sigma U^\T UU^\T\tilde{U}\tilde{\Sigma}^{-1}\tilde{U}^\T U$ is a $\zeta\times \zeta$ matrix. Its smallest singular value can be bounded by Lemma \ref{lm:iG2iG},
\bee\label{Mclt}
\sigma_{\min}(M_{clt}) &\geq \sigma_{\min}( U^\T\tilde{U})\sigma_{\min}(\tilde{\Sigma}^{-1}\tilde{U}^\T UU^\T  U\Sigma U^\T UU^\T\tilde{U}\tilde{\Sigma}^{-1}\tilde{U}^\T U)
\\
&\geq \sigma_{\min}( U^\T\tilde{U})\sigma_{\min}(\tilde{\Sigma}^{-1})\sigma_{\min}(\tilde{U}^\T UU^\T  U\Sigma U^\T UU^\T\tilde{U}\tilde{\Sigma}^{-1}\tilde{U}^\T U)
\\
&\geq \cdots
\\
&\geq \sigma_{\min}( U^\T\tilde{U})\sigma_{\min}(\tilde{\Sigma}^{-1})\sigma_{\min}(\tilde{U}^\T U)\sigma_{\min}(U^\T  U)\sigma_{\min}(\Sigma)\sigma_{\min}( U^\T U)\sigma_{\min}(U^\T\tilde{U})\sigma_{\min}(\tilde{\Sigma}^{-1})\sigma_{\min}(\tilde{U}^\T U)
\\
&\succsim 1.
\ee
In \eqref{Mclt}, we treat $M_{clt}$ as the multiplication of nine $\zeta \times \zeta$ matrices $ U^\T\tilde{U},\tilde{\Sigma}^{-1},\dots,\tilde{U}^\T U$, where each of them is full-rank with smallest singular values bounded away from $0$; see Lemma \ref{lm:iG2iG} and note $U^\T U = I_\zeta$. Then the first three inequalities in \eqref{Mclt} follow by repeatedly using the fact that 
$$
\sigma_{\min}(AB) \geq \sigma_{\min}(A)\sigma_{\min}(B),
$$ 
for two full-rank and square matrices $A$ and $B$; see, e.g., \citet[Problem III.6.14]{bhatia2013matrix}. Combining \eqref{sigmat:1} and \eqref{Mclt}, we conclude that
\bee\label{sigmabound}
\tilde{\sigma}^2 &\geq   \|v_n\|^2\cdot c_1 (v_n/\|v_n\|)^\T M_{clt} (v_n/\|v_n\|)
\\
&\geq \|v_n\|^2\sigma_{\min}(M_{clt})
\\
&\succsim \|v_n\|^2.
\ee
We thus derive the lower bound $\tilde{\sigma}\succsim \|v_n\|$. On the other hand, we aim at deriving
\bee\label{thm2:lin:1}
\sum_{i = 1}^n\E\Big[\big|\xi_n(X_i,T_i,Y_i)\big|^2  1\big\{\big|\xi_n(X_i,T_i,Y_i)\big|>\delta\big\}\Big] \rightarrow 0,
\ee
to verify the Lindberg's condition. By H\"{o}lder's inequality for fixed $c_0 >0$ and any $\delta > 0$,
\bee\label{holder:in}
&\E\Big[\big|\xi_n(X,T,Y)\big|^2  1\big\{\big|\xi_n(X,T,Y)\big|>\delta\big\}\Big] 
\\
&\leq \Big[\E\Big\{\big|\xi_n(X,T,Y)\big|^{2\cdot{(2+c_0)}/{2}}\Big\}\Big]^{2/(2+c_0)} \Big[\E\Big[ 1\big\{\big|\xi_n(X,T,Y)\big|>\delta\big\}\Big]\Big]^{1-2/(2+c_0)}
\\
&= \Big[\E\Big\{\big|\xi_n(X,T,Y)\big|^{2+c_0}\Big\}\Big]^{2/(2+c_0)} \Big[\Pr\Big\{|\xi_n(X,T,Y)|>\delta\Big\}\Big]^{1-2/(2+c_0)}.
\ee
For the first factor on the right-hand side of \eqref{holder:in}, we have
\bee\label{firtfactor}
&\E\Big\{\big|\xi_n(X,T,Y)\big|^{2+c_0}\Big\} 
\\
&= n^{-(2+c_0)/2}\tilde{\sigma}^{-(2 + c_0)}E\big[ \kappa^{2+c_0}_n(X,T)\big|Y - \mu(X,T)\big|^{2+c_0}\big]
\\
&\leq n^{-(2+c_0)/2}\tilde{\sigma}^{-(2 + c_0)} \cdot\big|\sup_{(x,t)\in\mathbb{X}\times\mathbb{T}} \kappa_n(x,t)\big|^{c_0}\cdot E\big[ \kappa^{2}_n(X,T)\big|Y - \mu(X,T)\big|^{2+c_0}\big]
\\
&\precsim n^{-(2+c_0)/2}K^{c_0/2},
\ee
where the last inequality follows by $\tilde{\sigma}\succsim \|v_n\|$ and the following bounds.
\begin{itemize}
\item By Lemmas \ref{am:psi} and \ref{lm:approx}, we have
\bee\label{kappaclt}
\sup_{(x,t)\in\mathbb{X}\times\mathbb{T}} \kappa_n(x,t) &\leq  \sup_{(x,t)\in\mathbb{X}\times\mathbb{T}} \big|v_n^\T U^\T\tilde{U}\tilde{\Sigma}^{-1}\tilde{U}^\T UU^\T\big\{\Psi(x,t) - \Gamma(x)\big\} \big|
\\
&\leq \|v_n\|\|U\|_2\|\tilde{U}\|_2\|\tilde{\Sigma}^{-1}\|_2\|\tilde{U}\|_2\|U\|_2^2\big(\|\Psi\|_{\mathbb{X}\times\mathbb{T}} + \|\Gamma\|_{\mathbb{X}}\big)
\\
&\precsim \sqrt{K}\|v_n\|.
\ee
\item  Similar to \eqref{lower:sigma}, we have
\bee\nonumber
&E\big\{ \kappa^{2}_n(X,T)\big|Y - \mu(X,T)\big|^{2+c_0}\big\} 
\\
&= v_n^\T U^\T\tilde{U}\tilde{\Sigma}^{-1}\tilde{U}^\T UU^\T E\Big[\big\{\Psi(X,T) - \Gamma(X)\big\}\big\{\Psi(X,T) - \Gamma(X)\big\}^\T\big\{Y - \mu(X,T)\big\}^{2 +c_0}\Big] UU^\T\tilde{U}\tilde{\Sigma}^{-1}\tilde{U}^\T U v_n
\\
&=v_n^\T U^\T\tilde{U}\tilde{\Sigma}^{-1}\tilde{U}^\T UU^\T E\Big[\big\{\Psi(X,T) - \Gamma(X)\big\}\big\{\Psi(X,T) - \Gamma(X)\big\}^\T E\Big[\big\{Y - \mu(X,T)\big\}^{2 +c_0}\mid X,T\Big]\Big] UU^\T\tilde{U}\tilde{\Sigma}^{-1}\tilde{U}^\T U v_n
\\
&\precsim v_n^\T U^\T\tilde{U}\tilde{\Sigma}^{-1}\tilde{U}^\T UU^\T E\Big[\big\{\Psi(X,T) - \Gamma(X)\big\}\big\{\Psi(X,T) - \Gamma(X)\big\}^\T \Big] UU^\T\tilde{U}\tilde{\Sigma}^{-1}\tilde{U}^\T U v_n
\\
& = v_n^\T U^\T\tilde{U}\tilde{\Sigma}^{-1}\tilde{U}^\T UU^\T R_n UU^\T\tilde{U}\tilde{\Sigma}^{-1}\tilde{U}^\T U v_n
\\
&\leq \|v_n\| \|U\|_2\cdots\| v_n\|
\\
&\precsim \|v_n\|^2,
\ee 
where the second equality follows by the law of total expectation, the first inequality follows by \eqref{clt:con}, and the last inequality is similar to \eqref{kappaclt}. 
\end{itemize}
\par
We note the bound \eqref{firtfactor} actually holds for not only $c_0$ but also all $c\in[0,c_0]$ with similar arguments. Thus when $c = 0$, one has $\E|\xi_n(X,T,Y)|^{2} \precsim n^{-1}$. Then for the second factor on the right-hand side of \eqref{holder:in}, by Chebyshev's inequality,
\bee\label{firtfactor2}
\Pr\Big\{|\xi_n(X,T,Y)|>\delta\Big\} \leq \frac{E|\xi_n(X,T,Y)|^2}{\delta^2} \precsim n^{-1},
\ee
with fixed $\delta > 0$. Combining \eqref{holder:in}, \eqref{firtfactor}, and \eqref{firtfactor2}, we conclude
\bee\label{fl3}
\sum_{i = 1}^n\E\Big[\big|\xi_n(X,T,Y)\big|^2  1\big\{\big|\xi_n(X,T,Y)\big|>\delta\big\}\Big]  & \precsim n \cdot(n)^{-1}\cdot K^{c_0/(2 + c_0)} \cdot n^{-1 + 2/(2 + c_0)} 
\\
& = (K/n)^{c_0/(2 + c_0)}
\\
&\rightarrow 0,
\ee
under Assumption \ref{am:Kn}. Thus \eqref{thm2:lin:1} is verified. By \eqref{fl1}, \eqref{fl2}, and \eqref{fl3}, the conditions of the Lindeberg-Feller CLT  
 are verified. Thus we have
\bee\label{Rnneg2}
{\sqrt{n}}\tilde{\sigma}^{-1}v_n^\T U^\T\tilde{U}\tilde{\Sigma}^{-1}\tilde{U}^\T UU^\T\Delta_{1,1} \leadsto \mathcal{N}(0,1).
\ee
By \eqref{rndef} and \eqref{sigmabound}, $K\asymp n^{\epsilon_{clt} + (d + 1)/2p}$ with $\epsilon_{clt} < 1/2 - (d + 1)/(2p)$, we have
\bee\label{Rnneg}
 {\sqrt{n}}\tilde{\sigma}^{-1}R_n 
&=  \mathcal{O}_P\Big\{  {\sqrt{n}} K^{-p/(d + 1)} \big(1 +\rho^{-1}K\log n/n\big)  \Big\}
  +o_P\Big( \sqrt{\log n } K n^{-3/4}  \Big) = o_P(1),
\ee
with similar rate comparison arguments as \eqref{simple:rate}--\eqref{rate:final}. For example, the rate term $\rho^{-1}K\log n/n \precsim (\log n)n^{ -(1 + \delta)\epsilon_{clt} +  1 - (d+1)/(2p) +  (d + 1)/(2p) + \epsilon_{clt}-1}\precsim (\log n)n^{-\delta\epsilon_{clt}}\rightarrow 0$. Other terms can be bounded similarly. Finally combining \eqref{clt:basic}, \eqref{Rnneg2}, \eqref{Rnneg}, and applying the Slutsky's theorem leads to \eqref{main:clt}. 
\par
Notably, when $\delta > 0$ is sufficiently large, we can have
\bee
n^{-1 + (d+1)/(2p) + (1+\delta)\epsilon_{clt}} =  n^{-1/2 + \{-1/2 + (d+1)/(2p) + (1 +\delta)\epsilon_{clt}\}}\prec n^{-1/2},
\ee
where we recall $\epsilon_{clt}\in[0,1/2 -(d + 1)/(2p))$ and $n^{-1/2}\precsim\sqrt{K\log n/n}$ as $K$ is growing. So we can always select $\rho = n^{-1/2}$ to satisfy the general condition $n^{-1 + (d+1)/(2p) + (1+\delta)\epsilon_{clt}}\precsim\rho\precsim \sqrt{K\log n/n}$.
\par
\noindent\fbox{Confidence interval} We finally prove the confidence interval part of Theorem \ref{thm:main}. First we simplify $\sigma_n^2$. From \eqref{l2:simplify}, we have
\bee\label{vsimple}
v_n^\T U^\T = \Psi^\T(x_0,t_0) - \Psi^\T(x_0,0).
\ee
On the other hand, by Lemma \ref{lm:svdR} (i), we have $U_\perp^\T\{\Psi(x,t) - \Gamma(x)\} = 0$ and thus
\bee\label{msimple}
&E\Big[\big\{\Psi(X,T) - \Gamma(X)\big\}\big\{\Psi(X,T) - \Gamma(X)\big\}^\T\big\{Y - \mu(X,T)\big\}^2\Big] 
\\ 
&= UU^\T E\Big[\big\{\Psi(X,T) - \Gamma(X)\big\}\big\{\Psi(X,T) - \Gamma(X)\big\}^\T\big\{Y - \mu(X,T)\big\}^2\Big] UU^\T
\\
& + U_\perp U_\perp^\T E\Big[\big\{\Psi(X,T) - \Gamma(X)\big\}\big\{\Psi(X,T) - \Gamma(X)\big\}^\T\big\{Y - \mu(X,T)\big\}^2\Big] U_\perp U_\perp^\T
\\
&= UU^\T E\Big[\big\{\Psi(X,T) - \Gamma(X)\big\}\big\{\Psi(X,T) - \Gamma(X)\big\}^\T\big\{Y - \mu(X,T)\big\}^2\Big] UU^\T
\\
& + U_\perp  E\Big[\underbrace{U_\perp^\T\big\{\Psi(X,T) - \Gamma(X)\big\}}_{ = 0}\underbrace{\big\{\Psi(X,T) - \Gamma(X)\big\}^\T U_\perp}_{ = 0} \big\{Y - \mu(X,T)\big\}^2\Big]  U_\perp^\T
\\
& = UU^\T E\Big[\big\{\Psi(X,T) - \Gamma(X)\big\}\big\{\Psi(X,T) - \Gamma(X)\big\}^\T\big\{Y - \mu(X,T)\big\}^2\Big] UU^\T.
\ee
By \eqref{lower:sigma}, \eqref{vsimple}, and \eqref{msimple}, $\tilde{\sigma}$ can be simplified to
\bee\nonumber
\tilde{\sigma}^2 = \{ \Psi(x_0,t_0) - \Psi(x_0,0)\}^\T A_n B_n A_n \{ \Psi(x_0,t_0) - \Psi(x_0,0)\},
\ee
where $A_n = \tilde{U}\tilde{\Sigma}^{-1}\tilde{U}^\T$ and $B_n = E\big[\{\Psi(X,T) - \Gamma(X)\}\{\Psi(X,T) - \Gamma(X)\}^\T\{Y - \mu(X,T)\}^2\big]$. Recall that in Algorithm \ref{alg:sigma}, our variance estimator is
\bee\nonumber
\hat{\sigma}^2 =  \tilde{v}_n^\T \hat{A}_n\hat{B}_n\hat{A}_n v_n\tilde{v}_n,
\ee
where we define,
\bee\nonumber
\hat{A}_n &= \hat{U}\hat{\Sigma}^{-1}\hat{U}^\T, 
\\
\hat{B}_n &= P_n\big[\{\Psi(X,T) - \hat{\Gamma}(X)\}\{\Psi(X,T) - \hat{\Gamma}(X)\}^\T\{Y - \hat{\mu}(X,T)\}^2\big]
\\
\tilde{v}_n &=\Psi(x_0,t_0) - \Psi(x_0,0).
\ee
We then decompose
\bee\nonumber
\hat{\sigma}^2 - \tilde{\sigma}^2& = \tilde{v}_n^\T \hat{A}_n\hat{B}_n\hat{A}_n \tilde{v}_n - \tilde{v}_n^\T A_n B_n A_n \tilde{v}_n
\\
&=\tilde{v}_n^\T (\hat{A}_n - A_n)\hat{B}_n\hat{A}_n \tilde{v}_n + \tilde{v}_n^\T  A_n(\hat{B}_n - {B}_n)\hat{A}_n \tilde{v}_n + \tilde{v}_n^\T A_n B_n (\hat{A}_n - {A}_n) \tilde{v}_n.
\ee
By Lemma \ref{lm:iG2iG} (v), we have the spectral norms of $\hat{A}_n,A_n,\hat{B}_n,B_n$ are constantly bounded, while $\|\hat{A}_n - A_n\|_2, \|\hat{B}_n - B_n\|_2 \rightarrow 0$ wpa1. We thus have wpa1,
\bee\label{sigmabound2}
&|\hat{\sigma}^2 - \tilde{\sigma}^2| 
\\
&\leq \|\tilde{v}_n\|^2 \|\hat{A}_n - A_n\|_2\|\hat{B}_n\|_2\|\hat{A}_n\|_2 + \|\tilde{v}_n\|^2  \|A_n\|_2\| \hat{B}_n - {B}_n\|_2\|\hat{A}_n\|_2  + \|\tilde{v}_n\|^2\| A_n \|_2\|B_n \|_2 \|\hat{A}_n - {A}_n\|_2
\\
&\prec\|\Psi(x_0,t_0) - \Psi(x_0,0)\|^2
\\
& = \{\Psi(x_0,t_0) - \Psi(x_0,0)\}^\T \{\Psi(x_0,t_0) - \Psi(x_0,0)\} 
\\
& = \{\Psi(x_0,t_0) - \Psi(x_0,0)\}^\T UU^\T \{\Psi(x_0,t_0) - \Psi(x_0,0)\} +  \{\Psi(x_0,t_0) - \Psi(x_0,0)\}^\T U_\perp U_\perp^\T \{\Psi(x_0,t_0) - \Psi(x_0,0)\}
\\
& = \{\Psi(x_0,t_0) - \Psi(x_0,0)\}^\T UU^\T \{\Psi(x_0,t_0) - \Psi(x_0,0)\} 
\\
&= \|v_n\|^2,
\ee
where the second and third equalities can be derived similar to \eqref{l2:simplify}. With \eqref{sigmabound} and \eqref{sigmabound2}, we have wpa1,
\bee\nonumber
\frac{\hat{\sigma}^2}{\tilde{\sigma}^2} = 1 + \frac{\hat{\sigma}^2 - \tilde{\sigma}^2}{\tilde{\sigma}^2} \rightarrow 1,
\ee
and thus $\hat{\sigma}^{-1}\tilde{\sigma}\rightarrow 1$ wpa1. Then by \eqref{main:clt} and Slutsky's theorem, we finally have
\bee\nonumber
{\sqrt{n}}\hat{\sigma}^{-1}\big\{\hat{\tau}(x_0,t_0) - \tau(x_0,t_0)\big\} &=\hat{\sigma}^{-1}\tilde{\sigma} \cdot {\sqrt{n}}\tilde{\sigma}^{-1}\big\{\hat{\tau}(x_0,t_0) - \tau(x_0,t_0)\big\}
\\
&\leadsto \mathcal{N}(0,1).
\ee
\QED
 
%
\end{document}